\documentclass[12pt]{article}

\usepackage{amsmath,enumerate,amsbsy,amsfonts,amssymb,mathabx,amscd,graphicx,multirow,xcolor,subcaption, pdflscape}

\usepackage{afterpage}
\usepackage{caption,float}
\usepackage[round,authoryear]{natbib} 
\usepackage{authblk,booktabs, multirow}
\usepackage{mathrsfs,longtable}
\usepackage[flushleft]{threeparttable}
\usepackage{algorithmic,algorithm}

\RequirePackage[colorlinks,citecolor=blue,urlcolor=blue]{hyperref}
\usepackage{xurl}
\usepackage{epsfig}
\usepackage[section]{placeins}

\makeatletter
\let\old@NAT@hyper@\NAT@hyper@
\renewcommand{\NAT@hyper@}[1]{\textup{\old@NAT@hyper@{#1}}}
\let\oldNAT@open\NAT@open
\renewcommand{\NAT@open}{\textup{\oldNAT@open}}
\let\oldNAT@close\NAT@close
\renewcommand{\NAT@close}{\textup{\oldNAT@close}}
\let\oldNAT@separator\NAT@separator
\renewcommand{\NAT@separator}{\textup{\oldNAT@separator}}
\let\oldNAT@aysep\NAT@aysep
\renewcommand{\NAT@aysep}{\textup{\oldNAT@aysep}}
\let\oldNAT@cmt\NAT@cmt
\renewcommand{\NAT@cmt}{\textup{\oldNAT@cmt}}
\makeatother

\providecommand{\keywords}[1]{\textbf{\textit{Keywords: }} #1}

\newcommand{\Date}[1]{\def\@Date{#1}}
\def\today{\number\day\ifcase\month\or
January\or February\or March\or April\or May\or June\or
July\or August\or September\or October\or November\or December\fi\number\year}

\makeatletter
\newcommand{\figcaption}{\def\@captype{figure}\caption}
\newcommand{\tabcaption}{\def\@captype{table}\caption}
\makeatother

\newcommand{\bOmega}{\boldsymbol \Omega}
\newcommand{\bomega}{\boldsymbol \omega}
\newcommand{\bDelta}{\boldsymbol \Delta}
\newcommand{\bSigma}{\boldsymbol \Sigma}

\newcommand{\bUpsilon}{\boldsymbol \Upsilon}

\newcommand{\bxi}{\boldsymbol \xi}

\newcommand{\btheta}{\boldsymbol \theta}

\newcommand{\bdelta}{\boldsymbol \delta}

\newcommand{\vep}{\varepsilon}
\newcommand{\bvep}{\boldsymbol \varepsilon}

\newcommand{\bGamma}{\boldsymbol \Gamma}

\newcommand{\bfeta}{\boldsymbol \eta}
\newcommand{\bfzeta}{\boldsymbol \zeta}

\newcommand{\bXi}{{\boldsymbol \Xi}}

\newcommand{\bu}{{\mathbf u}}
\newcommand{\bU}{{\mathbf U}}
\newcommand{\bv}{{\mathbf v}}
\newcommand{\bh}{{\mathbf h}}
\newcommand{\bb}{{\mathbf b}}
\newcommand{\ba}{{\mathbf a}}
\newcommand{\bc}{{\mathbf c}}
\newcommand{\be}{{\mathbf e}}

\newcommand{\bs}{{\mathbf s}}
\newcommand{\bt}{{\mathbf t}}
\newcommand{\bz}{{\mathbf z}}
\newcommand{\bw}{{\mathbf w}}

\newcommand{\md}{{\rm d}}
\newcommand{\cT}{{\cal T}}
\newcommand{\cU}{{\cal U}}

\newcommand{\cS}{{\cal S}}
\newcommand{\cD}{{\cal D}}

\newcommand{\cG}{{\mathrm{ g}}}
\newcommand{\cN}{{\mathcal N}}

\newcommand{\cE}{{\mathcal E}}
\newcommand{\bA}{{\bf A}}
\newcommand{\bB}{{\bf B}}

\newcommand{\bD}{{\bf D}}
\newcommand{\bS}{{\bf S}}
\newcommand{\bI}{{\bf I}}
\newcommand{\bM}{{\bf M}}
\newcommand{\bH}{{\bf H}}

\newcommand{\bP}{{\bf P}}
\newcommand{\bR}{{\bf R}}

\newcommand{\bbR}{\mathbb{R}}
\newcommand{\bX}{{\mathbf X}}
\newcommand{\bE}{{\mathbf E}}
\newcommand{\bQ}{{\mathbf Q}}

\newcommand{\bZ}{{\mathbf Z}}

\newcommand{\cov}{{\rm{Cov}}}
\newcommand{\var}{{\rm{Var}}}
\newcommand{\diag}{{\rm{diag}}}
\newcommand{\cv}{\text{cv}}
\newcommand{\T}{{\scriptscriptstyle {\top}}}

\newcommand{\0}{{\bf 0}}

\newcommand{\calK}{{\mathcal{K}}}
\newcommand{\bbP}{{\mathbb{P}}}
\newcommand{\bbE}{{\mathbb{E}}}
\newcommand{\cK}{{\mathcal{K}}}

\usepackage{authblk}

\def\JBES{{\sl Journal of Business \& Economic Statistics}}

\def\AS{{\sl The Annals of Statistics}}

\def\AP{{\sl The Annals of Probability}}

\newtheorem{theorem}{Theorem}
\newtheorem{lemma}{Lemma}[section]

\newtheorem{proposition}{Proposition}
\newtheorem{remark}{Remark}

\newtheorem{Condition}{Condition}

\newcommand{\blind}{1}

\allowdisplaybreaks

\begin{document}

\def\spacingset#1{\renewcommand{\baselinestretch}%
{#1}\small\normalsize} \spacingset{1}

%%%%%%%%%%%%%%%%%%%%%%%%%%%%%%%%%%%%%%%%%%%%%%%%%%%%%%%%%%%%%%%%%%%%%%%%%%%%%%

\if1\blind
{
\spacingset{1.25}
\title{\bf \Large %White noise test for high-dimensional functional time series
Testing for functional white noise in high dimensions
}
\author[1,2,3]{Jinyuan Chang}
\author[4]{Qing Jiang}
\author[5]{Xinghao Qiao}
\author[1]{Lin Yang}
\affil[1]{\it Joint Laboratory of Data Science and Business Intelligence, Institute of Statistical
Interdisciplinary Research,  Southwestern University of Finance and Economics, Chengdu, China}
\affil[2]{\it \small State Key Laboratory of Mathematical Sciences, Academy of Mathematics and Systems Science, Chinese Academy of Sciences, Beijing, China}
\affil[3]{\it \small China Center for Economic Research, Peking University, Beijing, China}
\affil[4]{Faculty of Arts and Sciences, Beijing Normal University, Zhuhai, China}
\affil[5]{\it Faculty of Business and Economics, The University of Hong Kong, Hong Kong}
\setcounter{Maxaffil}{0}

\renewcommand\Affilfont{\itshape\small}
%\date{\today}
\date{\vspace{-5ex}}
\maketitle
%\vspace{-2cm}
} \fi

\if0\blind
{\spacingset{2}
\bigskip
\bigskip
\bigskip
\begin{center}
{\Large\bf Testing for functional white noise in high dimensions}
\end{center}
\medskip
} \fi

\bigskip

\spacingset{1.5}
\begin{abstract}
White noise testing is a fundamental problem in time series analysis. Yet it remains largely unsolved for high-dimensional functional time series, despite the growing attention this area has received in recent years, as existing tests are confined to either univariate functional time series or high-dimensional scalar time series. In this paper, we develop a general error-contamination framework for testing white noise in high-dimensional functional time series. We propose a supremum-type test statistic based on cross-autocovariance functions and develop a parametric bootstrap procedure to approximate its null distribution. By imposing a general high-level condition, we derive a new Gaussian approximation result that ensures size control, and establish an asymptotic power guarantee. We then apply our framework to two concrete applications: (i) white noise test for discretely observed functional time series, and (ii) residual-based goodness-of-fit test for functional factor model. For each problem, we verify the corresponding high-level condition to ensure the theoretical validity of our proposed method. Extensive simulations show that our proposed method achieves good finite-sample performance. The practical utility of our proposed method is further illustrated through applications to two real datasets.
\end{abstract}

\noindent
\keywords{Discretely observed functional time series; Error contamination; Gaussian approximation; Goodness-of-fit;  High-dimensional functional white noise;  Parametric bootstrap}

%\newpage
\spacingset{1.7}
\section{Introduction}
\label{sec.intro}

%Functional data analysis is a statistical discipline that concerned with data represented as curves, surfaces, or images that vary over a continuous domain. When such functional observations are collected sequentially over time, they form a functional time series, capturing both the functional nature of the data and their temporal dynamics.
Functional time series analysis provides a statistical framework for modeling random functions observed sequentially over time, capturing both their functional structure and temporal dependence.
Recent technological developments have led %lead 
to the emergence of multivariate, and even high-dimensional functional time series datasets. Examples include yield curves \cite[]{Hays2012} and age-specific mortality rates \cite[]{Shang2017} across multiple countries, daily energy consumption trajectories \cite[]{Cho2013} from a large number of households,  cumulative intraday returns \cite[]{Zhang2016} for a collection of stocks and hourly readings of PM2.5 concentrations collected at different monitoring stations \cite[]{tan2024}. 
These data can be represented as  
\begin{equation}
\label{data}
\bvep_t(\cdot)=\{\vep_{t,1}(\cdot),\ldots,\vep_{t,p}(\cdot)\}^{\T}\,, ~~~t=1, \dots, n\,,
\end{equation}
defined on a compact interval $\cU.$ In high-dimensional settings, the dimension $p$ is large relative to, and may be greater than, the length of time series $n.$

Recent years have seen rapid progress in the development of estimation and inference methods for high-dimensional functional time series,
see, e.g., \cite{tang2022}, \cite{Guo2023}, \cite{Zhou2023}, \cite{Chang2024}, \cite{tan2024}, \cite{Li2024}, \cite{Chang2025a}, \cite{li2025arxiv}, \cite{guo2026}, and \cite{leng2026}.
Within this broader context, a foundational problem in time series analysis is testing for white noise, which is crucial not only for exploring the temporal dependence structure of the series but also for  performing diagnostic checks of time series models.
The existing literature on white noise tests for functional time series has been confined to the univariate case and can be categorized into two groups: time-domain tests based on autocovariance operators \citep{GK2007,Horvath2013,KRS2017,Rice2020,Bucher2023}
and spectral-domain tests based on spectral density operators \citep{Zhang2016,Bagchi2018,CR2020,Hlavka2021}. However, these tests are primarily developed for the univariate functional time series. Their validity may be severely challenged in high-dimensional settings with large $p$,
resulting in poor performance and thus underscoring the need for new test statistics and techniques to accurately characterize the asymptotic behavior.
Additionally, these tests only consider the idealized fully observed functional scenario, whereas in practice, functional time series are typically discretely observed with errors.

In contrast, the existing literature on white noise tests for high-dimensional scalar time series can be loosely divided into two categories. The first comprises supremum-type tests, which are particularly powerful against sparse alternatives and rely on different correlations among series \cite[]{Chang2017,Tsay2020,Chen2025}. The second, designed to detect dense alternatives, consists of sum-type tests that aggregate dependence information across dimensions \cite[]{Li2019,Zhao2024,Feng2022}. However, extending these methods to the infinite-dimensional functional domain presents substantial methodological and theoretical complexities, as it requires constructing test statistics under appropriate functional norms and addressing technical obstacles within abstract Hilbert spaces. 
Moreover, practical scenarios involving discretely observed and noisy curves introduce additional measurement and estimation errors, further complicating the analysis.

In this paper, we propose a general error-contamination framework with methodological development and theoretical guarantees for testing white noise in high-dimensional functional time series. More specifically, we construct a supremum-type test statistic based on cross-autocovariance functions and develop a computationally feasible parametric bootstrap procedure to approximate its null distribution. By imposing a general high-level condition, we derive a new Gaussian approximation result ensuring  size control,  and demonstrate the asymptotic power.
Moreover, we apply our framework to two concrete applications: (i) white noise test for discretely observed functional time series, and (ii) residual-based goodness-of-fit test for functional factor model \cite[]{guo2026}. Specifically, we illustrate how each application is seamlessly incorporated into the error-contamination framework, and verify the high-level condition to provide theoretical validity of our proposed methods. 
Simulation results show that our proposed methods  exhibit good size control and power in these two applications.
%In the white noise testing with discretely observed data, simulations show that our proposed method performs competitively with existing methods in the univariate case and  exhibits good size control and power under moderate- to high-dimensional settings. 
%{\color{red} In the goodness-of-fit testing application, we develop a procedure to determine the number of factors by sequentially applying the proposed white noise test to residuals, supported by empirical evidence of its effectiveness.}  
The main contributions of our paper are threefold. 

First, we  present the first general methodology with theoretical guarantees for testing high-dimensional functional white noise. Our approach is fully functional, without relying on any dimension reduction techniques such as functional principal component analysis or pre-fixed basis expansion, thereby avoiding any incurred information loss. The only existing inference work in high-dimensional functional time series literature is \cite{Zhou2023}, which developed tests for mean functions and thus differs from our white noise testing formulation. 
Moreover, their work and existing white noise testing methods for univariate functional time series consider only the unrealistic fully observed functional scenario, which restricts their practical utility. To overcome this, we introduce a unified error-contamination framework that encompasses the practical scenario of discretely observed functional time series as a concrete example.

Second, our error-contamination framework is applicable to residual-based goodness-of-fit tests for high-dimensional functional time series models such as functional factor model \cite[]{guo2026} and vector functional autoregressive model \cite[]{Guo2023}, which are crucial for model diagnostics and for determining model complexity. As an illustrative example, our proposed goodness-of-fit test for functional factor model 
naturally fits  within our error-contamination framework and its theoretical validity is established by verifying the high-level condition. In contrast, the existing literature on white noise tests for high-dimensional scalar time series has rarely examined goodness of fit theoretically in depth. \cite{Chang2017} provide only a simplified theoretical analysis for a general model without any illustrative example, while \cite{Feng2022} consider goodness-of-fit test merely in real data analysis without any theoretical justification.

Third, we derive novel Gaussian approximation theory for the supremum-type statistic of the sum of high-dimensional dependent processes and develop a new parametric bootstrap procedure, providing a suite of useful methodological and technical tools that can be broadly applied to other testing and inference problems for high-dimensional functional time series. Existing Gaussian approximation results, however, primarily address high-dimensional independent vectors \cite[]{Cher2013,Cher2017, Cher2019,Deng2020,Cher2023}, high-dimensional independent processes \cite[]{Cher2014} and high-dimensional scalar time series \cite[]{Zhang2017, Zhang2018, ChangChenWu2023}.
The most related work is \cite{Zhou2023}, which develops  Gaussian approximation theory for mean functions while allowing $p$ to grow polynomially with $n$. In contrast, our theory considers cross-autocovariance functions and accommodates an exponentially diverging dimension $p$. Furthermore, whereas their theory is based on Fourier expansion,
our framework is fully functional that works directly with the functional realizations.

The rest of the paper is organized as follows. Section \ref{sec.error} introduces the proposed white noise testing procedure and establishes the theoretical guarantees within a general error-contamination framework. Section \ref{sec.app} illustrates the applicability of the proposed framework through two concrete applications, namely white noise test for discretely observed functional time series and goodness-of-fit test for  functional factor model.
The finite-sample performance of our methodology is evaluated through extensive simulations in Section \ref{sec.simu}, while Section \ref{sec.real} demonstrates its practical utility via the analysis of two real-world datasets.  All technical proofs are relegated to the supplementary material.

{\it Notation}. For any positive integer $q\geq2$, we write $[q]=\{1,\ldots,q\}$. 
Denote  by
${\cal L}_2(\cU)$ the Hilbert space of square integrable functions defined on a compact set $\cU$, and write $\cU^2=\cU\times\cU.$  
Let $I(\cdot)$ denote the indicator function and ${\rm vec}(\cdot)$ be the vectorization for matrices.  
Denote by $\lfloor x \rfloor $ and $\lceil x\rceil$ the largest integer not greater than $x$ and the smallest integer not smaller than $x$, respectively. For two sequences of positive numbers $\{a_n\}$ and $\{b_n\}$, we write $a_n\lesssim b_n$ or $b_n\gtrsim a_n$ if $\limsup_{n\rightarrow\infty}a_n/b_n\leqslant c_0$ for some   constant $c_0>0$. For a real-valued random variable $\xi$, we define
$
\|\xi\|_{\psi_2} = \inf  [\lambda > 0 : \mathbb{E} \{\psi_2  ( {|\xi|}/{\lambda} ) \} \leq 1  ]$, where $\psi_2(x)= \exp(x^{2})-1$ for any $x>0$.  
For a  $p$-dimensional vector $\bv=(v_1,\ldots,v_p)^{\T}$, let  $|\bv|_{\max} = \max_{j\in[p]}|v_j|$.
Denote by $\mathbb{S}^{q-1}$ the $q$-dimensional unit sphere. For any $\mathcal{B}\in \mathcal{S}:=\mathcal{L}_2(\cU^2)$, we denote the supremum norm by $\|\mathcal{B}\|_\infty=\sup_{(u,v)\in\cU^2}|\mathcal{B}(u,v)|$ and the Hilbert--Schmidt norm by $\|\mathcal{B}\|_{\mathcal{S}}=\{\int_{\cU}\int_{\cU}\mathcal{B}(u,v)^2\,{\rm d}u{\rm d}v\}^{1/2}.$
For any $\boldsymbol{\mathcal{B}}=(\mathcal{B}_{jk})_{m\times n}$ with each $\mathcal{B}_{jk}\in\mathcal{S}$, we denote its elementwise maximum norm by
%functional version of $L_\infty$-norm by
$\|\boldsymbol{\mathcal{B}}\|_{\infty,\max} = \max_{j\in[m],k\in[n]}\|\mathcal{B}_{jk}\|_{\infty}$. For a symmetric matrix $\bB$, we denote its smallest and largest eigenvalues by $\lambda_{\min}(\bB)$ and $\lambda_{\max}(\bB)$, respectively.

\section{A general framework for testing functional white noise in high dimensions}\label{sec.error}
\subsection{Methodology}
We first present a general framework for white noise testing in high-dimensional functional time series.  Let $\{\bvep_t(\cdot)\}_{t=1}^{n}$ in (\ref{data}) denote a sequence of $p$-dimensional weakly stationary functional time series with mean zero and 
each $\vep_{t,j}(\cdot) \in {\cal L}_2(\cU)$. Consider the following hypothesis testing problem:
\begin{align*}%\label{eq:test}
H_0: \{\bvep_t(\cdot)\} \text{ is white noise } \quad\text{versus} \quad H_1: \{\bvep_t(\cdot)\} \text{ is not white noise}\,,
\end{align*}
which tests whether the observations are serially uncorrelated across time while allowing for contemporaneous cross-correlation among the $p$ functional components. In practical applications, $\bvep_t(\cdot)$'s are rarely fully observed. Instead, we can estimate them using appropriate methods, yielding estimated curves $\hat \bvep_t(\cdot) = \{\hat \varepsilon_{t,1}(\cdot), \dots, \hat \varepsilon_{t,p}(\cdot)\}^\T$ that satisfy
\begin{equation}
\label{err.m}
\hat \bvep_{t}(\cdot) = \bvep_{t}(\cdot) + \bdelta_{t}(\cdot)\,,~~~ t \in [n]\,, 
\end{equation}
where $\bdelta_{t}(\cdot)=\{\delta_{t,1}(\cdot),\ldots,\delta_{t,p}(\cdot)\}^{\T}$ represents the fully functional estimation error. Then $\{\hat \bvep_t(\cdot)\}$ can be treated as input data, upon which the white noise test is performed. We next present three common scenarios that satisfy \eqref{err.m}.
\begin{itemize}
    \item When  $\{\bvep_{t}(\cdot)\}$ is fully observed, it holds that $\hat{\bvep}_{t}(\cdot)=\bvep_{t}(\cdot)$.
    \item When $\{\bvep_t(\cdot)\}$ is discretely observed with errors, we can apply nonparametric smoothing methods to the observed data and thus obtain smoothed estimates $\hat{\bvep}_t(\cdot).$ 
    See Section \ref{sec.partobs} for details.
    \item When $\{\bvep_t(\cdot)\}$ represents the functional errors in high-dimensional functional modeling, such as 
    functional factor model and vector functional  autoregressive model, we can assess the adequacy of the fitted model by applying the white noise test to the resulting functional residuals $\hat{\bvep}_t(\cdot)$. See an illustrative example in Section \ref{sec.ffm}.
 \end{itemize}

Notice that under the null hypothesis $H_0$, all cross-autocovariance functions at nonzero lags vanish, i.e.,
$$
\bSigma^{(\ell)}(u,v)=\{\Sigma^{(\ell)}_{jj'}(u,v)\}_{p\times p}=\cov\{\bvep_{t}(u), \bvep_{t+\ell}(v)\} = {\bf 0}_{p\times p}
$$
for all $\ell\neq0$  and $(u,v)\in\cU^2$. Given the inputs $\{\hat{\bvep}_t(\cdot)\}_{t=1}^n$, we can estimate  $\bSigma^{(\ell)}(u,v)$ by its empirical counterpart 
\begin{align*}%\label{eq:sigest}
    \widehat{\bSigma}^{(\ell)}(u,v)=\{\widehat\Sigma^{(\ell)}_{jj'}(u,v)\}_{p\times p} = \frac{1}{n-\ell} \sum_{t=1}^{n-\ell} \{\hat\bvep_{t}(u)-\bar{\hat\bvep}(u)\} \{\hat\bvep_{t+\ell}(v)-\bar{\hat\bvep}(v)\}^{\T}
\end{align*}
with $\bar{\hat\bvep}(u)=\{\bar{\hat\vep}_1(u),\ldots,\bar{\hat\vep}_p(u)\}^{\T}=n^{-1}\sum_{t=1}^{n} \hat\bvep_{t}(u)$. To aggregate information across lags and  component pairs, we define the test statistic
\begin{align*}%\label{eq:teststat}
T_n  = \max_{\ell\in[L]} \sqrt{n}\|\widehat\bSigma^{(\ell)}\|_{\infty,\max}\,, 
\end{align*}
where $L \geq 1$ is a prescribed integer. Then large values of $T_n$ indicate departure from white noise.
To approximate the null distribution of $T_n$, we first construct a Gaussian process analogue. Let the vectorized empirical autocovariances $$\bxi_{n}(u,v)=\sqrt{n} ([ {\rm vec}\{\widehat\bSigma^{(1)}(u,v)\}]^{\T}, \ldots,[{\rm vec}\{\widehat\bSigma^{(L)}(u,v)\}]^{\T} )^{\T}$$ and define a centered Gaussian process $\boldsymbol{\cG}(u,v)=\{\cG_1(u,v),\ldots,\cG_{Lp^2}(u,v)\}^{\T}$ with covariance structure matching that of $\bxi_{n}(u,v)$, i.e., $$\bXi_{n}(u,v,\tilde{u},\tilde{v})=
\cov\{\boldsymbol{\cG}(u,v),\boldsymbol{\cG}(\tilde{u},\tilde{v})\}=\cov\{\bxi_{n}(u,v), \bxi_{n}(\tilde{u},\tilde{v})\}\,.$$
The Gaussian analogue of $T_n$ is then defined as 
\begin{equation}\label{eq:TnG}
T^{\rm G}_{n} = \|\boldsymbol{\cG}\|_{\infty,\max}\,.
\end{equation}
Under some mild conditions, our newly established Gaussian approximation result shows that $\sup_{x\in \mathbb{R}}|\bbP_{H_0}({T}_n\leq x) - \bbP({T}_n^{\rm G}\leq x )| = o(1)$ provided that $\log p\ll n^{c}$ for some universal constant $c>0$. Hence, for a given significance level $\alpha\in(0,1)$, the corresponding critical value for the proposed test statistic $T_n$ can be selected as ${\rm cv}_{\alpha} = \inf\{x>0: \mathbb{P}(T_n^{\rm G}>x)\leq \alpha\}$.

Since $\bXi_n(u,v,\tilde{u},\tilde{v})$ is unknown, we can estimate it via a kernel smoother applied to the process $\{\bfeta_t(u,v)\}_{t=1}^{n-L}$ with
\begin{align*}
{\bfeta}_t(u,v)
&= \{({\rm vec} [\{\hat\bvep_t(u)-\bar{\hat\bvep}(u)\}\{\hat{\bvep}_{t+1}(v)-\bar{\hat{\bvep}}(v)\}^{\T}])^{\T}, \\
&~~~~~~~~~~~~~~~~~~\ldots,
({\rm vec} [\{\hat\bvep_t(u)-\bar{\hat{\bvep}}(u)\}\{\hat\bvep_{t+L}(v)-\bar{\hat{\bvep}}(v)\}^{\T}] )^{\T}\}^{\T}\,.
\end{align*}
 Denote the resulting long-run covariance estimator \citep{Andrews1991} by
\begin{align*}
%\label{est.F}
\bXi^{*}_{n}(u,v,\tilde{u},\tilde{v})=\sum_{i=-(n-L)+1}^{{n-L}-1} \mathcal{W}\left(\frac{i}{b_{n}}\right) {\bH}^{*}_i(u,v,\tilde{u},\tilde{v}) \,,
\end{align*}
where $\mathcal{W}(\cdot)$ is a symmetric kernel function that is continuous at 0 with $\mathcal{W}(0)=1$, $b_{n}>0$ is the bandwidth diverging with $n$, and  
\begin{align*}
   {\bH}^{*}_i(u,v,\tilde u,\tilde v)
=
\frac{1}{n-L}
\sum_{t=\max(1,1-i)}^{\min(n-L,n-L-i)}
\{\bfeta_{t+i}(u,v)-\bar{\bfeta}(u,v)\}
\{\bfeta_t(\tilde u,\tilde v)-\bar{\bfeta}(\tilde u,\tilde v)\}^{\T}
\end{align*} 
with $\bar{{\bfeta}}(u,v)=(n-L)^{-1}\sum_{t=1}^{{n-L}}{\bfeta}_{t}(u,v)$. Conditionally on $\cD_n=\{\hat{\bvep}_t(\cdot) \}_{t=1}^n$,  we can define a centered Gaussian process ${\boldsymbol\cG}^*(u,v) = \{{\rm g}^*_1(u,v),\ldots,{\rm g}^*_{Lp^2}(u,v)\}^{\T}$ with covariance function $\bXi^{*}_{n}(u,v,\tilde{u},\tilde{v})$. Let
\begin{align*}%\label{eq:TnGstar}
{T}_{n}^{\rm G*}= \|\boldsymbol{\cG}^*\|_{\infty,\max}\,. 
\end{align*}
Our theoretical analysis indicates  that the distribution of $T_n^{\rm G}$ can be approximated by the conditional distribution of $T_n^{\rm G*}$ given $\cD_n$, which means
 the critical value ${\rm cv}_\alpha$ can be approximated by 
$$\hat{\rm cv}_{\alpha} = \inf\{x>0:\mathbb{P}({T}_{n}^{\rm G*}> x\,|\,\cD_n)\leq \alpha\}\,. $$ 
Therefore, for a given significance level $\alpha\in(0,1)$, we reject the null hypothesis $H_0$ if   $T_{n}>\hat{\rm cv}_{\alpha}$. Theorems \ref{thm.Tn.H0} and \ref{thm.Tn.H1} in Section \ref{sec.theory} provide the validity of our proposed testing procedure, which allow the dimension $p$ to diverge exponentially with the sample size $n$.

To implement our proposed testing procedure, we need to generate a $(Lp^2)$-variate Gaussian process $\boldsymbol{\cG}^*$. 
Directly generating a realization of the Gaussian process
$\boldsymbol\cG^*$ from $\bXi_n^*$ is, however, computationally
prohibitive when $p$ is large. Specifically, one possible
approach relies on a multivariate Karhunen-Lo\`eve expansion \citep{HappGreven2018} from
$\bXi_n^*$, which requires solving a high-dimensional operator-based eigenproblem  and storing the resulting eigenfunctions (i.e., function-valued vectors with $Lp^2$ components). Alternatively, one may first evaluate $\bXi_n^*$ at $K$ grid 
points in $\cU^2$ to construct an
$(Lp^2K)\times(Lp^2K)$ covariance matrix. A Cholesky decomposition of this matrix is then computed and used to generate the discretized realization of $\boldsymbol\cG^*$ at these evaluation points \citep{PorcuEtAl2016}. Consequently, both approaches
incur rapidly increasing computational and memory costs as $p$ grows.
To circumvent this
difficulty,  
 we suggest a novel parametric bootstrap procedure to obtain ${\boldsymbol\cG}^*$. Specifically, generate a Gaussian random vector $\boldsymbol\varrho =( \varrho_1 ,\dots,\varrho_{{n-L}})^{\T} \sim \cN({\bf 0},\bXi)$ independent of $\cD_n$, where $\bXi=(\Xi_{ij})_{[n-L]\times [n-L]}$  with  $\Xi_{ij}= \mathcal{W}\{(i-j)/{b_{n}}\}$, and define 
\begin{align}\label{eq:cG}
    {\boldsymbol\cG}^*(u,v) = \frac{1}{\sqrt{n-L}}\sum_{t=1}^{{n-L}}\varrho_t \{{\bfeta}_{t}(u,v)-\bar{{\bfeta}}(u,v)\}\,.
\end{align}
It is easy to see that $\mathbb{E}\{ {\boldsymbol\cG}^*(u,v) {\boldsymbol\cG}^{*}(\tilde u,\tilde v)^{\T} \,|\,\cD_n\} =\bXi^*_{n}(u,v,\tilde u,\tilde v)$. As noted by \cite{Andrews1991}, if the kernel function $\mathcal{W}(\cdot)$  satisfies   $\int_{-\infty}^{\infty}\mathcal W(x)e^{-\sqrt{-1}\lambda x}\,{\rm d}x \ge 0$ for any $\lambda \in \mathbb{R}$, such defined matrix $\bXi$ is then positive semi-definite.  Widely used kernels,
including the Bartlett, Parzen, and Quadratic Spectral kernels, all satisfy this
condition. See Section \ref{sec.simu} for the explicit forms of these kernels.

\subsection{Theoretical properties}\label{sec.theory}
We impose the following conditions for the asymptotic analysis of the test statistic $T_n$.

\setcounter{Condition}{0} 
\renewcommand{\theCondition}{C\arabic{Condition}}
\renewcommand{\theHCondition}{C.\arabic{Condition}}

\begin{Condition}\label{c.alpha}
Assume that $\{\bvep_{t}(\cdot)\}$ is $\alpha$-mixing with the mixing coefficient
\begin{align*}
\alpha(m)\equiv\sup_{t}\sup_{A\in\mathcal{F}_{-\infty}^{t},B\in\mathcal{F}_{t+m}^{\infty}}|\mathbb{P}(A \cap B)-\mathbb{P}(A)\mathbb{P}(B)| \,,~~~ m\geq 1  \,,
\end{align*}
where $\mathcal{F}_{-\infty}^{t}$ and $\mathcal{F}_{t+m}^{\infty}$ are the $\sigma$-fields generated, respectively, by $\{\bvep_s(\cdot)\}_{s\leq t}$ and $\{\bvep_s(\cdot)\}_{s\geq t+m}$. Furthermore, there exist two  universal constants $C_1>1$ and $C_2>0$ %independent of $(n,p,L,\cU)$ 
such that $\alpha(m)\leq C_1\exp(-C_2m)$ for all $m\geq 1$.
\end{Condition}

Consider the oracle counterpart $\tilde{\bxi}_{n}(u,v)$ of $\bxi_n(u,v)$ defined as 
\begin{align}\label{tilde.xi.def}
    \tilde{\bxi}_{n}(u,v)=\sqrt{n} ([ {\rm vec}\{\widetilde\bSigma^{(1)}_{\vep\vep}(u,v)\}]^{\T}, \ldots,[{\rm vec}\{\widetilde\bSigma^{(L)}_{\vep\vep}(u,v)\}]^{\T} )^{\T}\,,
\end{align} 
where 
\begin{align}\label{tilde.Sigma.def}
    \widetilde{\bSigma}_{\vep\vep}^{(\ell)}(u,v)&=\frac{1}{n-\ell}\sum_{t=1}^{n-\ell} \{\bvep_{t}(u)-\bar{\bvep}(u)\} \{\bvep_{t+\ell}(v)-\bar{\bvep}(v)\}^{\T}\,.
\end{align}

\begin{Condition}\label{c.bvar} 
Write $\tilde{\bxi}_n(u,v)=\{\tilde{\xi}_{n,1}(u,v),\ldots,\tilde{\xi}_{n,Lp^2}(u,v)\}^{\T}$. There exists a universal constant $C_3>0$ such that $\var\{ \tilde{\xi}_{n,r}(u,v) \}\geq C_{3}$ for any $r\in[Lp^2]$ and $(u,v)\in\cU^2$. 
\end{Condition}

\begin{Condition}\label{c.subgaussian}
For each $t\in[n]$ and $j\in[p]$, the process
$\{\varepsilon_{t,j}(u):u\in\mathcal U\}$ admits a separable
modification with respect to the Euclidean metric on $\mathcal U$.
Moreover, 
{\rm(i)} there exist two universal constants $C_4>0$ and  $\kappa\in(0,1]$ such that $\max_{t\in[n]}\max_{j\in[p]}\|\varepsilon_{t,j}(u) - \varepsilon_{t,j}(v)\|_{\psi_2} \leq C_4 |u - v|^\kappa$
for all $u, v \in \mathcal{U}$; 
{\rm(ii)} there exists some universal constant $C_5>0$ such that  $\max_{t\in[n]}\max_{j\in[p]}\|\varepsilon_{t,j}(u_0)\|_{\psi_2} \leq C_5$ for some fixed point $u_0 \in \mathcal{U}$.
\end{Condition}

\begin{Condition}\label{c.kernelF}
The symmetric kernel function  $\mathcal{W}(\cdot):\mathbb{R}\rightarrow [-1,1]$ is continuously differentiable with bounded derivatives on $\mathbb{R}$. Moreover, it satisfies {\rm(i)} $\mathcal{W}(0)=1$, and {\rm(ii)}
$|\mathcal{W}(x)|\lesssim |x|^{-\vartheta}$ as $|x|\rightarrow \infty$ for some universal constant $\vartheta>1$. 
%Let the bandwidth  $b_{n}\asymp n^{\rho}$ for some constant $\rho$, where $0<\rho<(\vartheta-1)/(3\vartheta-2)$.
\end{Condition}

\begin{remark}
The $\alpha$-mixing assumption in Condition {\rm\ref{c.alpha}} can be relaxed to allow $\alpha(m)\leq C_1\exp(-C_2m^{\tau})$ for $m\geq 1$ with some constant $\tau\in(0,1]$. Our theoretical results remain valid for any $\tau\in(0,1]$, and we set $\tau=1$ only to simplify the presentation of technical proofs.
Condition {\rm\ref{c.bvar}} restricts the variances of $\tilde{\xi}_{n,r}(u,v)$'s to be uniformly bounded away from zero, which is required for applying the Nazarov's inequality \citep[{\rm Lemma A.1}]{Cher2017}. 
In practice, Condition~{\rm\ref{c.bvar}} can always be satisfied by considering a
slightly perturbed process
$\{\bvep_{t}(\cdot)+\bz_{t}\}$, where $\{\bz_t\}$ are mutually  independent
mean-zero random vectors with covariance matrix $\sigma_z^2\bI_p$ for some
sufficiently small constant $\sigma_z^2>0$, and are independent of
$\{\bvep_t(\cdot)\}$. Therefore, we have $\cov\{\bvep_{t}(u)+\bz_{t},\bvep_{t+\ell}(v)+ \bz_{t+\ell} \}=\cov\{\bvep_{t}(u),\bvep_{t+\ell}(v) \}$ for   any $\ell\neq0$,   
%the nonzero-lagged autocovariance structure is preserved,
and testing whether
$\{\bvep_{t}(\cdot)\}$ is white noise is equivalent to testing whether
$\{\bvep_{t}(\cdot)+\bz_{t}\}$ is white noise.
The separability requirement in Condition~{\rm\ref{c.subgaussian}} is imposed to facilitate the application of chaining arguments in our theoretical analysis.
In addition, Condition {\rm\ref{c.subgaussian}(i)} requires that the increments of $\{\vep_{t,j}(\cdot)\}_{t\in[n],j\in[p]}$ satisfy a sub-Gaussian property \citep[{\rm Section 8.1}]{Ver2018}.
In particular, for Gaussian processes $\{\varepsilon_{t,j}(\cdot)\}$, its increments satisfy $\|\varepsilon_{t,j}(u) - \varepsilon_{t,j}(v)\|_{\psi_2} \lesssim {d}_{t,j}(u,v),$ where ${d}_{t,j}(u,v) \equiv (\mathbb{E}[\{\vep_{t,j}(u)-\vep_{t,j}(v)\}^{2}])^{1/2}$ is the associated metric. 
Furthermore, if ${d}_{t,j}(u,v)\lesssim |u-v|^{\kappa}$ for some universal constant $ \kappa\in(0,1]$ across all $t,j$ and $u,v$, as in \cite{Xiao2009} and \cite{MeerschaertWangXiao2013}, then Condition {\rm\ref{c.subgaussian}(i)} is satisfied automatically. 
Specially, if $\{\vep_{t,j}(\cdot)\}$ is standard Brownian motion, $d_{t,j}(u,v)= |u-v|^{1/2}$ and $\kappa=1/2$. 
Moreover, Conditions {\rm\ref{c.subgaussian}(i) and \ref{c.subgaussian}(ii)} imply that $\varepsilon_{t,j}(u)$ is a sub-Gaussian random variable for any $t,j$ and $u$.
Condition {\rm\ref{c.kernelF}} is standard in the literature on nonparametric estimation of the long-run covariance \cite[]{Andrews1991,ChangQiuYaoZou2018}. 
% The nonnegative-definiteness requirement in Condition {\rm\ref{c.kernelF}}  guarantees that
% $\bXi=\{\mathcal W((i-j)/b_n)\}_{[n-L],[n-L]}$ is positive semidefinite and hence
% defines a valid covariance matrix for the Gaussian multipliers.  
 For kernel function with bounded support, such as the Parzen kernel and the Bartlett kernel, we have $\vartheta=\infty$ in Condition {\rm\ref{c.kernelF}}. 
\end{remark}

Let $\tilde{T}_n$ and $\tilde{T}_n^{{\rm G}*}$ denote the oracle counterparts of  ${T}_n$ and ${T}_n^{{\rm G}*}$, respectively, obtained  by replacing  $\hat{\bvep}_t(\cdot)$ with   $\bvep_t(\cdot)$. More specifically, 
\begin{align}\label{eq:tildeTn}
    \tilde{T}_n  = \max_{\ell\in[L]} \sqrt{n}\|\widetilde\bSigma^{(\ell)}_{\vep\vep}\|_{\infty,\max}\,,
\end{align}
where $\widetilde\bSigma^{(\ell)}_{\vep\vep}$ is defined in \eqref{tilde.Sigma.def}. Write
\begin{align*} 
\tilde{\boldsymbol{\cG}}^*_{\vep\vep}(u,v) =\frac{1}{\sqrt{n-L}}\sum_{t=1}^{{n-L}}\varrho_t \{\tilde{\bfeta}_t^{\vep\vep}(u,v)-\bar{\tilde{\bfeta}}^{\vep\vep}(u,v)\}\,,
\end{align*}
where $\{\varrho_t\}_{t=1}^{n-L}$ is constructed as in \eqref{eq:cG},
\begin{align*}
\tilde{\bfeta}_t^{\vep\vep}(u,v)
&= \{({\rm vec} [\{\bvep_t(u)-\bar{\bvep}(u)\}\{{\bvep}_{t+1}(v)-\bar{{\bvep}}(v)\}^{\T}])^{\T}, \\
&~~~~~~~~~~~~~~~~~~\ldots,
({\rm vec} [\{\bvep_t(u)-\bar{{\bvep}}(u)\}\{\bvep_{t+L}(v)-\bar{{\bvep}}(v)\}^{\T}] )^{\T}\}^{\T} \,,
\end{align*} 
and $\bar{\tilde{\bfeta}}^{\vep\vep}(u,v)= ({n-L})^{-1}\sum_{t=1}^{{n-L}}\tilde{\bfeta}_t^{\vep\vep}(u,v)$.  
We  then  define 
\begin{align*}
     \tilde{T}_n^{{\rm G}*}=  \|\tilde{\boldsymbol{\cG}}^*_{\vep\vep}\|_{\infty,\max}\,.
\end{align*} 
Based on \eqref{err.m}, in order to control  
$|T_n-\tilde{T}_n|$, we need to involve some intermediate statistics
$\widetilde{\bSigma}_{\vep\delta}^{(\ell)}(u,v)$, $\widetilde{\bSigma}_{\delta\vep}^{(\ell)}(u,v)$, and $\widetilde{\bSigma}_{\delta\delta}^{(\ell)}(u,v)$ defined as  
\begin{align*}
     \widetilde{\bSigma}^{(\ell)}_{\vep\delta}(u,v)
 &=    \frac{1}{n-\ell}\sum_{t=1}^{n-\ell}\{\bvep_{t}(u)-\bar{\bvep}(u)\} \{\bdelta_{t+\ell}(v)-\bar{\bdelta}(v)\}^{\T}\,, \\
  \widetilde{\bSigma}^{(\ell)}_{\delta\vep}(u,v)
 &=    \frac{1}{n-\ell}\sum_{t=1}^{n-\ell} \{\bdelta_{t}(u)-\bar{\bdelta}(u)\} \{\bvep_{t+\ell}(v)-\bar{\bvep}(v)\}^{\T}\,,  \\
\widetilde{\bSigma}^{(\ell)}_{\delta\delta}(u,v)
 &= \frac{1}{n-\ell}\sum_{t=1}^{n-\ell} \{\bdelta_{t}(u)-\bar{\bdelta}(u)\} \{\bdelta_{t+\ell}(v)-\bar{\bdelta}(v)\}^{\T}\,. 
\end{align*} 
 %constructed in the same manner as $\widetilde{\bSigma}_{\vep\vep}^{(\ell)}(u,v)$. For conciseness, their explicit forms  are deferred to the supplementary material; see \eqref{eq:def.sigmadv}. 
Similarly, to study the discrepancy between $T_n^{\rm G*}$ and $\tilde{T}_n^{\rm G*}$,
we also need to consider the bootstrap counterparts $\tilde{\boldsymbol{\cG}}^*_{\vep\delta}(u,v)$, $\tilde{\boldsymbol{\cG}}^*_{\delta\vep}(u,v)$, and  $\tilde{\boldsymbol{\cG}}^*_{\delta\delta}(u,v)$ defined as
\begin{align*}
  \tilde{\boldsymbol{\cG}}^*_{\vep\delta}(u,v) &= \frac{1}{\sqrt{n-L}}\sum_{t=1}^{{n-L}} \varrho_t \{ \tilde{\bfeta}^{\vep\delta}_t(u,v)  -\bar{\tilde{\bfeta}} ^{\vep\delta}(u,v)\}  \,, \\
  \tilde{\boldsymbol{\cG}}^*_{\delta\vep}(u,v) &= \frac{1}{\sqrt{n-L}}\sum_{t=1}^{{n-L}} \varrho_t \{ \tilde{\bfeta}^{\delta\vep}_t(u,v)  -\bar{\tilde{\bfeta}} ^{\delta\vep}(u,v)\}     \,, \\
  \tilde{\boldsymbol{\cG}}^*_{\delta\delta}(u,v) &=\frac{1}{\sqrt{n-L}}\sum_{t=1}^{{n-L}} \varrho_t \{ \tilde{\bfeta}^{\delta\delta}_t(u,v)  -\bar{\tilde{\bfeta}} ^{\delta\delta}(u,v)\}   \,, 
\end{align*}
where $\bar{\tilde{\bfeta}}^{\vep\delta}(u,v)=(n-L)^{-1}\sum_{t=1}^{n-L}\tilde{\bfeta}_t^{\vep\delta}(u,v)$, $\bar{\tilde{\bfeta}}^{\delta\vep}(u,v)=(n-L)^{-1}\sum_{t=1}^{n-L}\tilde{\bfeta}_t^{\delta\vep}(u,v)$, and $\bar{\tilde{\bfeta}}^{\delta\delta}(u,v)= (n-L)^{-1}\sum_{t=1}^{n-L}\tilde{\bfeta}_t^{\delta\delta}(u,v)$ with 
\begin{align*}
    \tilde{\bfeta}_t^{\vep\delta}(u,v)
= &~\{({\rm vec} [\{\bvep_t(u)-\bar{\bvep}(u)\}\{{\bdelta}_{t+1}(v)-\bar{{\bdelta}}(v)\}^{\T}])^{\T},\\
&~~~~~\ldots, 
({\rm vec} [\{\bvep_t(u)-\bar\bvep(u)\}\{\bdelta_{t+L}(v)-\bar\bdelta(v)\}^{\T}] )^{\T}\}^{\T}\,,\\
\tilde{\bfeta}_t^{\delta\vep}(u,v)
=&~ \{({\rm vec} [\{\bdelta_t(u)-\bar{\bdelta}(u)\}\{{\bvep}_{t+1}(v)-\bar{{\bvep}}(v)\}^{\T}])^{\T},\\
&~~~~~ \ldots, 
({\rm vec} [\{\bdelta_t(u)-\bar\bdelta(u)\}\{\bvep_{t+L}(v)-\bar\bvep(v)\}^{\T}] )^{\T}\}^{\T}\,,\\
\tilde{\bfeta}_t^{\delta\delta}(u,v)
=&~ \{({\rm vec} [\{\bdelta_t(u)-\bar{\bdelta}(u)\}\{{\bdelta}_{t+1}(v)-\bar{{\bdelta}}(v)\}^{\T}])^{\T},\\
&~~~~~ \ldots, 
({\rm vec} [\{\bdelta_t(u)-\bar\bdelta(u)\}\{\bdelta_{t+L}(v)-\bar\bdelta(v)\}^{\T}] )^{\T}\}^{\T} \,.
\end{align*}   
%analogously to $\tilde{\boldsymbol{\cG}}^*_{\vep\vep}(u,v)$. See \eqref{eq:def.cgdv} in the supplementary material for details. 
In addition, we define
\begin{align*}
    \Delta_{\widetilde{\bSigma}}^{\varepsilon\delta}:=&~ \max_{\ell\in [L]} 
\{\|\widetilde{\bSigma}_{\varepsilon\delta}^{(\ell)}\|_{\infty,\max} + \|\widetilde{\bSigma}_{\delta\vep}^{(\ell)}\|_{\infty,\max} +  \|\widetilde\bSigma_{\delta\delta}^{(\ell)} \|_{\infty,\max} \}\,,\\
\Delta_{\tilde{\boldsymbol{\cG}}}^{\varepsilon\delta}:= &~\|\tilde{\boldsymbol{\cG}}_{\varepsilon\delta}^{*}\|_{\infty,\max}+\|\tilde{\boldsymbol{\cG}}_{\delta\vep}^{*}\|_{\infty,\max}+\|\tilde{\boldsymbol{\cG}}_{\delta\delta}^{*}\|_{\infty,\max}    \,.
\end{align*}
% {\color{red} Consequently, $|T_n-\tilde{T}_n| \leq n^{1/2}\Delta_{\widetilde{\bSigma}}^{\varepsilon\delta}$ and $|T_n^{\rm G*}-\tilde{T}_n^{\rm G*}|\leq \Delta_{\tilde{\boldsymbol{\cG}}}^{\varepsilon\delta}$. The high-level requirement stated in Condition~{\rm\ref{c.rate}} below allows us to establish the theoretical guarantees for the proposed
% method. See Theorems~\ref{thm.Tn.H0} and~\ref{thm.Tn.H1} for details. 
% }

% Before presenting the asymptotic properties of the test statistic $T_n$ under the null and alternative hypotheses, we impose the following high-level conditions on $\Delta_{\widetilde{\bSigma}}^{\varepsilon\delta}$ and $\Delta_{\tilde{\boldsymbol{\cG}}}^{\varepsilon\delta}$. 

\begin{Condition}\label{c.rate}
There exist some universal constants $\gamma_1 > 1/2$ and $\gamma_2,\gamma_3 >0$ such that $\Delta_{\widetilde{\bSigma}}^{\varepsilon\delta}=O_{\rm p}\{n^{-\gamma_1}(\log p)^{\gamma_2}\}$ and $ \Delta_{\tilde{\boldsymbol{\cG}}}^{\varepsilon\delta} =  {O}_{{\rm p}}\{n^{-\gamma_3}(\log p)^{\gamma_2}\}$.
\end{Condition}

\begin{remark}
To investigate the asymptotic behavior of our proposed test statistic $T_n$, it suffices to study the theoretical properties of $\tilde{T}_n$ specified in \eqref{eq:tildeTn}, provided that the impact of the estimation errors $\bdelta_t(\cdot)$'s on $T_n$ is properly controlled.
To this end, the high-level Condition {\rm\ref{c.rate}} is imposed to ensure that the discrepancies between $T_n$ and $\tilde{T}_n$, as well as between their  bootstrap analogues $T_n^{\rm G*}$ and $\tilde{T}_n^{\rm G*}$, are asymptotically negligible.
This condition is automatically satisfied for fully observed functional time series with $\bdelta_t(\cdot)={\bf 0}$. 
By contrast, for discretely observed functional time series and for functional  residuals obtained from   fitted    models, it needs to be verified  under specific settings. In Section {\rm\ref{sec.app}}, we will verify Condition {\rm\ref{c.rate}} for two applications by specifying suitable choices of the constants $(\gamma_1,\gamma_2,\gamma_3)$.  See Remarks {\rm \ref{remark.app1}} and {\rm\ref{remark.app2}} for details.
\end{remark}

\begin{theorem}
\label{thm.Tn.H0}
 Let $p\geq n^{\upsilon}$ and $b_n\asymp n^{\rho}$ for some universal constants $\upsilon>0$ and $ \rho\in (0,1)$.  Under  Conditions {\rm\ref{c.alpha}--\ref{c.rate}}, if $\rho<(\vartheta-1)/(3\vartheta-2)$ and $\log  p \ll n^{\iota}$  for some constant $\iota>0$ depending only on $(\rho,\vartheta, \gamma_1,\gamma_2,\gamma_3)$, it then holds  that
$
\mathbb{P}_{{  H_0}}({T}_{n} >\hat{\mathrm{cv}}_{\alpha})\rightarrow \alpha$ as $n\rightarrow\infty$. 
\end{theorem}

Theorem {\rm\ref{thm.Tn.H0}} establishes the validity  of our proposed test in the sense that it maintains the nominal significance level asymptotically under the null hypothesis. To investigate the asymptotic power of the proposed test, we write 
$\widetilde{\bXi}_{n}(u,v,\tilde{u},\tilde{v}) =\cov\{\tilde\bxi_{n}(u,v), \tilde\bxi_{n}(\tilde{u},\tilde{v})\}$ for 
$\tilde{\bxi}_{n}(u,v)$ specified in \eqref{tilde.xi.def}. Let   $\varrho = \sup_{(u,v)\in\cU^2}{\mathrm{\Psi}}(u,v) $, where $\mathrm{\Psi}(u,v)$ is the largest diagonal element  of $\widetilde{\bXi}_{n}(u,v,u,v)$. Theorem {\rm \ref{thm.Tn.H1}} shows that the proposed test is consistent under local alternatives.

\begin{theorem}\label{thm.Tn.H1}
%Suppose that the conditions of Theorem {\rm \ref{thm.Tn.H0}} hold.
Let the assumptions of Theorem {\rm\ref{thm.Tn.H0}}   hold. Additionally, if  $
\max _{ \ell\in[L]} \|\bSigma^{(\ell)}\|_{\infty,\max}  
\geq \check{C} \varrho^{1/2} n^{-1 / 2}(\log p)^{1/2}
$ for some sufficiently large universal constant $\check{C}>0$, it then holds that 
$
\mathbb{P}_{ H_1}({T}_{n}>\hat{\mathrm{cv}}_{\alpha})\rightarrow 1$ as $n\rightarrow\infty$.
\end{theorem}

\begin{remark}
    Theorems {\rm \ref{thm.Tn.H0}} and {\rm\ref{thm.Tn.H1}} accommodate high-dimensional settings in which $p$ is allowed to grow exponentially with $n$. It is worth noting that the assumption $p \geq n^{\upsilon}$ is fairly mild in the high-dimensional literature, where $\upsilon>0$ can be chosen arbitrarily small. This assumption is not essential for our theoretical analysis and is adopted only to simplify the presentation. In particular, many arguments in our theoretical proofs involve comparing $\log p$ and $\log n$. Without this assumption, the technical proofs would become more involved.
\end{remark}

\section{Applications}\label{sec.app}
In this section, we illustrate the proposed testing framework through two concrete applications, namely white noise test for discretely observed functional time series in Section \ref{sec.partobs} and goodness-of-fit test for functional factor model in Section \ref{sec.ffm}. 

\subsection{Test for discretely observed functional time series}\label{sec.partobs}

For fully observed functional time series, $\hat\bvep_t(\cdot) = \bvep_t(\cdot)$ in \eqref{err.m}. In this case, the testing procedure in Section \ref{sec.error} can be performed by simply substituting $\hat\bvep_t(\cdot)$ with $\bvep_t(\cdot)$. Condition \ref{c.rate} is automatically satisfied, and the asymptotic properties of the statistic $T_n$ follow immediately from Theorems \ref{thm.Tn.H0} and \ref{thm.Tn.H1}, ensuring the validity of the testing procedure.

However,  $\bvep_t(\cdot)$ is not directly observable in practice. Instead, it is discretely observed and subject to measurement errors. We therefore need to smooth the discrete observations using nonparametric methods to obtain the estimated curves $\hat{\bvep}_t(\cdot)$, which naturally places the testing problem within our error-contamination framework.
For each $t\in[n]$ and $j\in[p]$, assume the function $\vep_{t,j}(\cdot)$  is observed with errors at $N_{t,j}$ random   points  $u_{t,j,1},\ldots,u_{t,j,N_{t,j}}\in \cU$, where $\{u_{t,j,k}\}_{t\in[n],j\in[p],k\in[N_{t,j}]}$ are independent and identically distributed (i.i.d.) random variables.
Let $\vep_{t,j,k}^*$ be the observed value of $\vep_{t,j}(u_{t,j,k})$ satisfying
\begin{align}\label{eq:dismodel}
\vep_{t,j,k}^*=\vep_{t,j}(u_{t,j,k})+\varsigma_{t,j,k} \,,
\end{align}
where the random errors
$\varsigma_{t,j,k}$'s  are i.i.d. with $\bbE(\varsigma_{t,j,k})=0$ and ${\rm Var}(\varsigma_{t,j,k}) <\infty$.

To reconstruct functional trajectories, we apply the frequently used local linear smoother \cite[]{ZC2007} to the observed data $\{(u_{t,j,k}, \vep_{t,j,k}^*)\}_{k \in [N_{t,j}]}$ for each $t\in[n]$ and $j\in[p]$. For a given kernel $\cK(\cdot)$ and bandwidth $h_{t,j}>0$, we obtain the estimated curves $\hat{\varepsilon}_{t,j}(u)=\hat{f}_{0,t,j}$, where
\begin{align*}
(\hat{f}_{0,t,j}, \hat{f}_{1,t,j}) = \arg\min_{f_{0,t,j},f_{1,t,j}}\sum_{k=1}^{N_{t,j}}\{\varepsilon_{t,j,k}^* - f_{0,t,j}-f_{1,t,j}(u_{t,j,k}-u)\}^2 \calK_{h_{t,j}}(u_{t,j,k}-u)\,,
\end{align*}
and $\calK_h(\cdot)=\calK(\cdot/h)/h$. 
%For each $j\in[p]$, individual functions across $t\in[n]$ often admit similar smoothness properties and sometimes similar patterns. It is thus reasonable to use a common bandwidth $h_j$ for all of them. 
Under this setting, the estimation error $\bdelta_{t}(\cdot)=\hat{\bvep}_{t}(\cdot) - \bvep_{t}(\cdot)$ satisfies the additive structure \eqref{err.m}. Therefore, the white noise test can be carried out exactly as described in Section \ref{sec.error}, using the smoothed curves $\{\hat{\bvep}_t(\cdot)\}$.

To provide theoretical guarantees of the testing procedure  for discretely observed functional time series based on Theorems~\ref{thm.Tn.H0} and \ref{thm.Tn.H1}, we need to verify the high-level Condition~\ref{c.rate}.
To this end, we impose the following regularity conditions.

\setcounter{Condition}{0} 
\renewcommand{\theCondition}{D\arabic{Condition}}
\renewcommand{\theHCondition}{D.\arabic{Condition}}
 
\begin{Condition}\label{cond.subgauss}
There exists some universal constant $C_1'> 0$ such that $ \|\varsigma_{t,j,k}\|_{\psi_2}\leq C_1'$
for each $t\in[n]$, $j\in[p]$  and $k \in [N_{t,j}]$. 
\end{Condition}

\begin{Condition}\label{cond.Tij}
For each $t\in[n]$ and $j\in[p]$, the sampling frequency $N_{t,j}$ and the bandwidth $h_{t,j}$ satisfy  $N_{t,j}\asymp N\rightarrow\infty$,  $h_{t,j}\asymp h\rightarrow 0$ and $Nh\rightarrow \infty$ as $n\rightarrow\infty$.
\end{Condition}

\begin{Condition}\label{cond.fu}
{\rm (i)} Let $\{u_{t,j,k}\}_{t\in[n], j\in[p], k\in[N_{t,j}]}$ be i.i.d. copies of a random variable $U$ defined on $\cU$
with density $f_U(\cdot)$ satisfying that $(C_2')^{-1}\leq \inf_{u\in\cU} f_U(u)\leq \sup_{u\in\cU}f_U(u)\leq C_2' $ for some universal constant $C_2'>0$. 
{\rm (ii)} For each $t\in[n]$ and $j\in[p]$, the random elements
$\vep_{t,j}(\cdot)$, $\{u_{t,j,k}\}_{k\in [N_{t,j}]}$ and
$\{\varsigma_{t,j,k}\}_{k\in [N_{t,j}]}$ are mutually independent. 
\end{Condition}

\begin{Condition}\label{cond.kern}
{\rm (i)}
The kernel $\cK(\cdot)$ is a symmetric probability density function on compact support $[-1,1]$ with $ C_3'\leq \int u^2\cK(u)\,{\rm d} u<\infty$
for some universal constant $C_3' >0$. 
{\rm (ii)} The kernel $\cK(\cdot)$ is Lipschitz continuous on $\mathbb R$, i.e., there exists some universal   constant $C_4'>0$ such that $|\cK(u)-\cK(v)|\leq C_4'|u-v|$ for any $u,v\in\mathbb R$.
\end{Condition}

Condition~\ref{cond.subgauss} assumes that the random errors are sub-Gaussian. 
Conditions \ref{cond.Tij}--\ref{cond.kern} are standard in the literature on nonparametric smoothing for functional data   \cite[]{ZC2007,ZW2016}. 
We next present the following proposition to verify that Condition~\ref{c.rate} holds for discretely observed functional time series.

\begin{proposition}\label{thm.partial}
 Let $p\geq n^{\upsilon}$ and $b_n\asymp n^{\rho}$ for some universal constants $\upsilon>0$ and $  \rho \in (0,1)$. 	
Under Conditions~{\rm \ref{c.subgaussian}}, {\rm  \ref{c.kernelF}} and   {\rm\ref{cond.subgauss}--\ref{cond.kern}},
%if $\log(N\vee p) \ll Nh$, 
it holds that
\begin{align*}
  \Delta_{\widetilde{\bSigma}}^{\varepsilon\delta} =&\, 
 O_{\rm p} \big[    \{  (Nh)^{- 1/2}  + h^{\kappa}   \}  \log(N\vee p)  \big]\,,\\  
 \Delta_{\tilde{\boldsymbol{\cG}}}^{\varepsilon\delta}=&\  {O}_{\rm p} \big[ n^{\rho/2} \{  (Nh)^{- 1/2}  + h^{\kappa}   \}  \log^2(N\vee p)  \big] \,,
\end{align*}  
provided that $\log(N\vee p) \ll Nh$.
\end{proposition}

\begin{remark}\label{remark.app1}
Recall that $p\geq n^{\upsilon}$ for $\upsilon>0$ is assumed in Theorems~{\rm\ref{thm.Tn.H0}} and  {\rm\ref{thm.Tn.H1}}. 
In our analysis, we assume   that the underlying functions $\{\vep_{t,j}(\cdot)\}$ satisfy the sub-Gaussian increment property (see Condition {\rm\ref{c.subgaussian}}) without imposing any smoothness conditions, which yields the term $h^{\kappa}$. Balancing the trade-off between $(Nh)^{-1/2}$ and $h^{\kappa}$ suggests choosing the bandwidth $h\asymp N^{-1/(2\kappa+1)}$.
Furthermore, if $N\gtrsim n^{a}$ for some constant $a>(2\kappa+1)/(2\kappa)$, then Condition {\rm\ref{c.rate}} holds
with $\gamma_1=a\kappa/(2\kappa+1)>1/2$, $\gamma_2=2$ and $\gamma_3=\gamma_1-\rho/2>0$.
With verifiable Condition {\rm\ref{c.rate}} under the possibly suboptimal rates established in Proposition {\rm \ref{thm.partial}} and other conditions required in Theorems {\rm \ref{thm.Tn.H0}} and {\rm \ref{thm.Tn.H1}}, an application of Theorems {\rm \ref{thm.Tn.H0}} and {\rm \ref{thm.Tn.H1}} can establish  the validity of the proposed white noise testing procedure for discretely observed functional time series.  
\end{remark}

\subsection{Goodness-of-fit test for functional factor model}\label{sec.ffm}
The error-contamination framework introduced in Section \ref{sec.error} naturally accommodates residuals from high-dimensional functional models, providing a unified approach for residual-based goodness-of-fit tests. In this section, we illustrate its application to the functional factor model for high-dimensional functional time series \cite[]{guo2026}, which models the $p$-dimensional weakly stationary  functional time series $\bX_t(\cdot)=\{X_{t,1}(\cdot), \ldots,X_{t,p}(\cdot)\}^{\T}$ as 
\begin{equation}\label{eq:ffm}
\bX_{t}(\cdot) =  \bA \bZ_{t}(\cdot) +\bvep_{t}(\cdot) \,,~~~~~ t\in[n] \,.
\end{equation}
Here the observed series $\bX_t(\cdot)$ is decomposed as the sum of two components: one common dynamic component driven by an $r$-dimensional weakly stationary latent functional factor time series $\bZ_{t}(\cdot)=\{Z_{t,1}(\cdot),\ldots,Z_{t,r}(\cdot)\}^{\T}$ with corresponding $p \times r$ factor loading matrix $\bA,$  and one idiosyncratic white noise component $\bvep_{t}(\cdot)=\{\varepsilon_{t,1}(\cdot),\ldots,\varepsilon_{t,p}(\cdot)\}^{\T}$,  and $r$ is assumed known and finite. Model \eqref{eq:ffm} remains unchanged if $(\bA,\bZ_t(\cdot))$ is replaced by $(\bA\bH,\bH^{-1}\bZ_t(\cdot))$ for any invertible matrix $\bH\in\mathbb R^{r\times r}$. Consequently, $\bA$ and $\bZ_t(\cdot)$ can not be determined uniquely. However, the factor loading space spanned by the columns of $\bA$ is unique. We therefore focus on estimating the column space of $\bA$.
% {\color{red} 
% Since $\bA$ and $\bZ_t(\cdot)$ cannot in general be identified separately, standard identification restrictions are imposed, under which the loading space spanned by the columns of $\bA$, rather than $\bA$ itself, is identifiable. This loading space can be represented by an orthonormal basis.}
One of the most frequently used procedures for assessing model adequacy  is to test whether the error process $\{\bvep_t(\cdot)\}$ is white noise. 
 Since $\{\bvep_t(\cdot)\}$ is unobserved in practice, we instead project $\bX_t(\cdot)$ onto the orthogonal complement of the estimated loading space and work with the estimated idiosyncratic residuals   $\hat{\bvep}_{t}(\cdot) = (\bI_{p}-\widehat{\bA}\widehat{\bA}^{\T})\bX_t(\cdot)$ for $t \in [n]$, where the columns of $\widehat{\bA}\in\mathbb{R}^{p\times r}$ form an estimated orthonormal basis for the column  space of $\bA$. 
% {\color{red}Since $\{\bvep_t(\cdot)\}$ is unobserved in practice, we instead work with the residuals
% \begin{align*}
% \hat{\bvep}_{t}(\cdot) = \bX_{t}(\cdot) -  \widehat{\bA}\widehat{\bZ}_{t}(\cdot)\,, ~~~t\in[n]\,,
% \end{align*}
% where $\widehat{\bA}$ and $\widehat{\bZ}_t(\cdot)$  are some consistent estimators for $\bA$ and $\bZ_t(\cdot)$, respectively.}
These residuals naturally satisfy the error-contamination structure \eqref{err.m}, where $\bdelta_{t}(\cdot)=\hat\bvep_{t}(\cdot)-\bvep_{t}(\cdot)$ represents the estimation error for the idiosyncratic component.   
Hence, the white noise test developed in Section \ref{sec.error} can be directly applied to $\{\hat\bvep_t(\cdot)\}$, providing a formal goodness-of-fit test for the functional factor model (\ref{eq:ffm}).

To implement the testing procedure based on $\{\hat\bvep_t(\cdot)\}$, we need to obtain  $\widehat\bA$.  
Define the lag-$\ell$ autocovariance functions $\bSigma_{xx}^{(\ell)}(u,v)=\cov\{\bX_{t}(u) ,\bX_{t+\ell}(v)\}$  for $\ell\geq 1$  and $(u,v)\in\cU^2$. 
Since  the column space of $\bA$ can be recovered by the $r$ leading eigenvectors that correspond to the $r$ largest eigenvalues of $\sum_{\ell=1}^{\ell_0} \int_{\cU}\int_{\cU} \bSigma_{xx}^{(\ell)}(u,v)\bSigma_{xx}^{(\ell)}(u,v)^\T \, {\rm d}u{\rm d}v$ under some regularity conditions,
we adopt the estimation procedure with identity weight matrix proposed in \cite{guo2026}.    
Specifically, for some prescribed integer $\ell_0 \geq 1$, define the $p \times p$ matrix
\begin{align}\label{eq:hatM}
    \widehat{\bM} = \sum_{\ell=1}^{\ell_0} \int_{\cU}\int_{\cU} \widetilde{\bSigma}_{xx}^{(\ell)}(u,v)   \widetilde{\bSigma}_{xx}^{(\ell)}(u,v)^{\T} \, {\rm d}u{\rm d}v\,,
\end{align}
 where $\widetilde{\bSigma}_{xx}^{(\ell)}(u,v)=(n-\ell)^{-1}\sum_{t=1}^{n-\ell}\{\bX_{t}(u)-\widebar{\bX}(u)\}\{\bX_{t+\ell}(v)-\widebar{\bX}(v)\}^{\T}$ is the sample estimator for $\bSigma_{xx}^{(\ell)}(u,v)$ and $\widebar{\bX}(\cdot)=n^{-1}\sum_{t=1}^n\bX_t(\cdot)$. By performing the eigen-decomposition of $\widehat{\bM}$, we obtain the leading $r$ eigenvectors  
 $\hat{\bomega}_1, \ldots, \hat{\bomega}_r$, which form the columns of $\widehat \bA$.
%{\color{red}The estimated functional factors are then given by $\widehat{\bZ}_t(\cdot)=\widehat{\bA}^{\T}\bX_t(\cdot)$. Accordingly, the residuals are
% $\hat{\bvep}_t(\cdot)= \bX_t(\cdot) - \widehat{\bA}\widehat{\bA}^{\T}\bX_t(\cdot)$ for $t \in [n]$.  }

Similarly, to verify   Condition~\ref{c.rate} in this setting, we impose Conditions \ref{cond.facerror}--\ref{cond.facmat} below. Throughout this section, the processes $\{\bvep_t(\cdot)\}$ and $\{\bZ_t(\cdot)\}$ are assumed to admit separable modifications with respect to the Euclidean metric on $\cU$.

% To verify the high-level Condition~\ref{c.rate}, we impose the following conditions. 

\setcounter{Condition}{0}
\renewcommand{\theCondition}{F\arabic{Condition}} 
\renewcommand{\theHCondition}{F.\arabic{Condition}}
 
 \begin{Condition}\label{cond.facerror}
The zero-mean random processes $\{\bvep_t(\cdot)\}_{t\in[n]}$ satisfy that {\rm (i)} for some $u_0\in\cU$,  $ \max_{t\in[n]}\sup_{\bv\in\mathbb{S}^{p-1}}\|\bv^{\T}\bvep_t(u_0)\|_{\psi_2}\leq C_1''$ for some universal constant $C_1''>0$;
{\rm (ii)} for any $(u,v)\in\cU^2$, $\max_{t\in[n]}\sup_{\bv\in\mathbb{S}^{p-1}}\|\bv^{\T}\{\bvep_{t}(u)-\bvep_{t}(v)\}\|_{\psi_2} \leq C_{2}'' |u-v|^{\kappa}$ for some universal constant $C_{2}''>0$;
{\rm (iii)} for any $s,t\in[n]$ and $(u,v)\in\cU^2$, $\mathbb{E}\{ \bvep_s(u)\bZ_t(v)^{\T} \}=\mathbf{0}$.
%$\{\bvep_t(\cdot)\}_{t\in[n]}$ and $\{\bZ_t(\cdot)\}_{t\in[n]}$ are uncorrelated.
\end{Condition}

\begin{Condition}\label{cond.facprocess}	
The zero-mean random processes $\{\bZ_t(\cdot)\}_{t\in[n]}$ satisfy that {\rm (i)} for some $u_0\in\cU$, $ \max_{t\in[n]}\max_{i\in[r]} \| Z_{t,i}(u_0)\|_{\psi_2}\leq C_{3}''$ for some universal constant $C_{3}''>0$; 
{\rm (ii)} for any $(u,v)\in\cU^2$,  $\max_{t\in[n]}\max_{i\in[r]} \| Z_{t,i}(u)-Z_{t,i}(v)\|_{\psi_2} \leq C_{4}'' |u-v|^{\kappa}$ for some universal  constant $C_{4}''>0$;
{\rm (iii)} the joint process $\{(\bZ_t(\cdot)^{\T},\bvep_t(\cdot)^{\T})^{\T}\}$  
satisfies  the  $\alpha$-mixing condition specified in Condition {\rm \ref{c.alpha}};
{\rm (iv)} 
the $r$ eigenvalues of $\sum_{\ell=1}^{\ell_0}
\iint \bSigma_{zz}^{(\ell)}(u,v)
\bSigma_{zz}^{(\ell)}(u,v)^{\T}
\,{\rm d}u{\rm d}v$ are uniformly bounded away from zero, where $\bSigma_{zz}^{(\ell)}(u,v)=\cov\{\bZ_{t}(u),\bZ_{t+\ell}(v)\}$.

% the $r$ eigenvalues of $\sum_{\ell=1}^{\ell_0}
% \iint \bSigma_{zz}^{(\ell)}(u,v)
% \bSigma_{zz}^{(\ell)}(u,v)^{\T}
% \,{\rm d}u{\rm d}v$ satisfy $\theta_1 \geq \cdots \geq \theta_r > C_{5}''$  for some universal constant $C_{5}''>0$ , where $\bSigma_{zz}^{(\ell)}(u,v)=\cov\{\bZ_{t}(u),\bZ_{t+\ell}(v)\}$.
\end{Condition}

\begin{Condition}\label{cond.facmat}
 The factor loading matrix $\bA=(\ba_1,\ldots,\ba_p)^{\T}$ satisfies that 
 {\rm (i)} $\lambda_{\min}(\bA^{\T}\bA) \asymp p  \asymp \lambda_{\max}(\bA^{\T}\bA)$;   
{\rm(ii)} 
there exists some universal constant $C_{5}''>0$ such that $|\ba_{j}|_{\max}\leq C_{5}''$ for all $j\in[p]$. 
\end{Condition}

Conditions {\rm \ref{cond.facerror}}(i) and {\rm \ref{cond.facerror}(ii)} require the sub-Gaussian properties of the high-dimensional process $\bvep_{t}(\cdot)$ for all $t\in[n]$. Combined with the preceding separability property,
Conditions~{\rm \ref{cond.facerror}(i)} and {\rm \ref{cond.facerror}(ii)} imply
Condition~{\rm \ref{c.subgaussian}}. 
%which imply that Condition \ref{c.subgaussian} is satisfied automatically. 
Condition~{\rm\ref{cond.facerror}}(iii) assumes that the idiosyncratic
process $\{\bvep_t(\cdot)\}$ and the factor process $\{\bZ_t(\cdot)\}$
are uncorrelated in all leads and lags \citep{ChenYangZhang2022,QiaoEtAl2026}.
%thereby providing a clear separation between the factor
%and idiosyncratic components at the second-order level
%\citep{DozGiannoneReichlin2012,BantisClementsUrquhart2023}.} 
Conditions {\rm \ref{cond.facprocess}(i)--\ref{cond.facprocess}(iii)} impose the marginal sub-Gaussianities and $\alpha$-mixing condition for the latent factor processes $\{\bZ_t(\cdot)\}$. 
Condition {\rm \ref{cond.facprocess}(iv)} ensures that $\bZ_t(\cdot)$ has exactly $r$ components \cite[]{guo2026}. 
Condition {\rm \ref{cond.facmat}} requires the factors to be pervasive in the sense that they influence a large fraction of the functional outcomes. See similar conditions in factor model literature 
\cite[]{Bai2003,Fan2013}.
%  Condition {\rm \ref{cond.facmat}(i)} can be viewed as a simple representation of the usual pervasiveness condition that  $\lambda_{1}   \asymp \lambda_{r}\asymp p$, where $\lambda_{1}\geq \ldots \geq \lambda_{r}$ are the eigenvalues of $\bA^{\T}\bA$. More precisely, there exists a orthgonal matrix $\bGamma$ and $\bf{\Lambda} = {\diag}(\lambda_{1},\ldots,\lambda_{r})$  such that $\bA^{\T}\bA = \bGamma\bf{\Lambda} \bGamma^{\T}$. We can set $\widetilde{\bA}= p^{1/2}\bA \bGamma \bf{\Lambda}^{-1/2}$ and $\widetilde{\bZ}_t(\cdot)= p^{-1/2} \bf{\Lambda}^{1/2} \bGamma^{\T} \bZ_t(\cdot)$ such that $\widetilde{\bA} \widetilde{\bZ}_t(\cdot)= \bA\bZ_t(\cdot)$.}
Condition {\rm \ref{cond.facmat}(i)} can    be  relaxed to allow for weak factors by assuming $\lambda_{\min}(\bA^{\T}\bA) \asymp p^{1-\chi}  \asymp \lambda_{\max}(\bA^{\T}\bA)$ for some constant $\chi\in[0,1]$. Under this weaker condition, the rates in Proposition {\rm \ref{prop.facmodel}} below would depend on $\chi$, complicating the verification of Condition~\ref{c.rate}.

\begin{proposition}\label{prop.facmodel}
Let $p\geq n^{\upsilon}$ and $b_n\asymp n^{\rho}$ for some universal constants $\upsilon>0$ and $  \rho \in (0,1)$.  
Under model~\eqref{eq:ffm} with Conditions {\rm \ref{c.kernelF}} and   {\rm\ref{cond.facerror}--\ref{cond.facmat}},  it holds that 
\begin{align*}
   \Delta_{\widetilde{\bSigma}}^{\varepsilon\delta} =  &~
O_{\rm p}\big\{ (n^{-1}+n^{-1/2}p^{-1/2}) \log p  \big\}\,,\\
   \Delta_{\tilde{\boldsymbol{\cG}}}^{\varepsilon\delta}    = &~ {O}_{\rm p}\big\{ n^{\rho/2} (n^{-1/2}+p^{-1/2} )(\log p)^{5/2} \big\}\,,  
\end{align*}  
provided that $\log p\ll n^{1/5}$.
\end{proposition}

\begin{remark}\label{remark.app2}
For a goodness-of-fit test, the null hypothesis specifies a concrete data-generating model. In our framework, the null assumes that the functional factor model~\eqref{eq:ffm} is correctly specified. With the convergence results established in Proposition~{\rm\ref{prop.facmodel}}, we set $\gamma_1=\min\{1,(1+\upsilon)/2\}>1/2$, $\gamma_2=5/2$ and $\gamma_3=\min\{(1-\rho)/2,(\upsilon-\rho)/2\}>0$.
Provided that $\rho<\min(1,\upsilon)$,  Condition~{\rm\ref{c.rate}} holds when model~\eqref{eq:ffm} is correctly specified. Together with the other conditions required in Theorem~{\rm\ref{thm.Tn.H0}}, an application of Theorem~{\rm\ref{thm.Tn.H0}} ensures that the proposed goodness-of-fit test maintains size control under the correctly specified  model~\eqref{eq:ffm}.  However, when the model is misspecified, i.e., the true process does not follow  model~\eqref{eq:ffm}, the verification of Condition~{\rm\ref{c.rate}} is substantially more complicated. 
%The true process need not follow the factor model and may have a much more general structure. 
%For example, one may have
%$
%\bX_t(\cdot)=f\{\bZ_t(\cdot)\}+\bs_t(\cdot)$
%for some complex function $f$, so that the induced error relative to model~\eqref{eq:ffm} takes the form
%$
%f\{\bZ_t(\cdot)\}-\bA\bZ_t(\cdot)+\bs_t(\cdot)$.
%Since the structural discrepancy between $f\{\bZ_t(\cdot)\}$ and $\bA\bZ_t(\cdot)$ is unknown and may be arbitrarily complex, 
%verifying Condition~{\rm\ref{c.rate}}  becomes model-specific and technically challenging. 
Consistent with this difficulty, existing studies on white-noise-based goodness-of-fit tests have primarily  provided theoretical guarantees under correctly specified models \citep{Zhang2016,Kim2024}. In particular, \cite{Kim2024} explicitly note the difficulty of deriving asymptotic power for such tests when the model is not correctly specified. Consequently, a general power analysis under model misspecification is beyond the scope of this paper and is left for future research.
  The numerical results in Section~{\rm\ref{subsec:simu.ff}} show that the proposed test exhibits good empirical power in finite samples.  
\end{remark}

\begin{remark}
Our procedure can be naturally extended to the functional factor model with discrete observations by applying the methodology in Section {\rm\ref{sec.ffm}} to the reconstructed curves obtained from the pre-smoothing step described in Section {\rm\ref{sec.partobs}}. However, the pre-smoothing step introduces additional estimation errors, making the verification of Condition~{\rm\ref{c.rate}} substantially more involved.
\end{remark}

% {\color{red}
% We have developed the theory assuming that the number of factors $r$ is known or can be identified correctly. In practice, $r$ is unknown and the most commonly adopted ratio-based estimator \cite[]{lam2012,guo2026} often suffers from underestimating $r$.
% Building upon our goodness-of-fit test, we can develop a sequential testing procedure to determine $r$.  
%  Specifically, we start from $r_1=1$,  estimate the corresponding factor loading matrix $\widehat{\bA}^{(r_1)}$,  then compute the residuals $\hat{\bvep}_t^{(r_1)}(\cdot)=[\bI_p -\widehat{\bA}^{(r_1)}\{\widehat{\bA}^{(r_1)}\}^{\T} ]\bX_t(\cdot)$, and perform a white noise test on $\{\hat{\bvep}_t^{(r_1)}(\cdot)\}$.  
% If the null  is not rejected, we set  $\hat{r}=r_1$ and stop; otherwise, we increase $r_1$ by one and repeat the above procedure.  }

\section{Numerical studies}\label{sec.simu}	
In this section, we evaluate the finite sample performance of our proposed testing procedures for discretely observed functional time series in Section~\ref{subsec:simu.cov} and  functional factor model in Section~\ref{subsec:simu.ff}.   All tests  are implemented at the 0.05 significance level using 1000 Monte Carlo replications, with  $\cU=[0,1]$ and $1000$ bootstrap replications used to determine the critical value in our procedures. Three kernel functions are used to compute the critical value, 
%estimate the long-run covariance function of $\{ {\boldsymbol{\eta}}_t(u,v)\}_{t=1}^{n-L}$, 
namely,  (i) Quadratic Spectral (QS) kernel $\mathcal{K}_{\rm{QS}}(x)=25(12\pi^2x^2)^{-1}\{ (6\pi x/5)^{-1}\sin(6\pi x/5) - \cos(6\pi x/5)\}$, (ii) Parzen (Pz) kernel $\mathcal{K}_{\rm{Pz}}(x)=(1-6x^2+6|x|^3)I(0\leqslant|x|\leqslant 1/2) + 2(1-|x|)^3I(1/2<|x|\leqslant 1)$, and (iii) Bartlett (Bt) kernel $\mathcal{K}_{\rm{Bt}}(x)= (1-|x|)I(|x|\leqslant 1)$. 
The bandwidths for these kernels are selected via the data-driven bandwidth selection procedures of \cite{Andrews1991} and \cite{Chang2017}.

\subsection{Test for discretely observed functional time series}\label{subsec:simu.cov} 
We consider functional time series of length $n\in \{200,400\}$ and dimension $p\in \{1,10,50,100\}$, covering univariate, moderate-, and high-dimensional settings. 
%For lag orders  $L\in \{2,4,6\}$, the corresponding test statistics are denoted by {\color{red} $T_n^{(2)}$, $T_n^{(4)}$, and $T_n^{(6)}$}, respectively. 
For the univariate case ($p=1$), we compare our proposed method with several existing procedures, including those of \cite{GK2007}, \cite{Zhang2016}, \cite{KRS2017}, and \cite{CR2020}, denoted by GK, ZH, KRS, and CR, respectively. 
For methods involving the lag order $L$, namely the proposed test $T_n$, GK, and KRS, we consider $L\in\{2,4,6\}$.
In addition, for the GK test, the projection dimension  is chosen based on the cumulative percentage of variance explained, using a threshold of 85\%. 
%The lag order is set to $L_{\rm GK}\in\{2,4,6\}$, yielding   ${\rm GK}_2$, ${\rm GK}_4$, and ${\rm GK}_6$. 
For the ZH test, the block length  is determined using the minimum volatility method. 
%For the KRS test, we similarly set $L_{\rm KRS}\in \{2,4,6\}$, resulting in ${\rm KRS}_2$, ${\rm KRS}_4$ and ${\rm KRS}_6$. 
For the CR test, we employ the Parzen (${\rm CR}_{\rm Pz}$) and Bartlett (${\rm CR}_{\rm Bt}$) kernels with bandwidth $h=n^{1/(2q+1)}$, where $q$ denotes the order of the kernel.  
To assess the performance of the proposed method for discretely observed functional time series, we generate discrete observations according to model \eqref{eq:dismodel}. The observation time points $u_{t,k}$ and measurement errors $\varsigma_{t,j,k}$ are independently sampled from $\text{Unif}(0,1)$ and $\mathcal{N}(0,0.5^2)$, respectively. 
We set $N_{t,j}=N\in\{25,51\}$ and apply the Epanechnikov kernel to obtain local linear smoothing estimators $\hat{\vep}_{t,j}(\cdot)$, with the bandwidth selected by 5-fold cross-validation over the candidate set $\{0.1, 0.15,\ldots,0.4\}$.
%where each validation set is defined as the collection of functional observations across a randomly partitioned subset of the discrete observation time points.
Moreover, six data-generating models for the latent curves $\bvep_t(\cdot)$ are considered as follows:

\begin{description}
\item[{\bf Model 1.}] To mimic the infinite-dimensionality of random curves, we generate $\varepsilon_{t,j}(u)=\bs(u)^{\T}\btheta_{t,j}$ for $t\in[n]$, $j\in[p]$ and $u\in \cU$, where $\bs(u)$ is a 25-dimensional Fourier basis function and $\btheta_t=(\btheta_{t,1}^{\T},\ldots,\btheta_{t,p}^{\T})^{\T}\in\bbR^{25p}$ is sampled from a mean-zero multivariate Gaussian distribution with block covariance matrix $\bOmega\in \bbR^{25p\times 25p}$, 
whose $(j,k)$-th block is $\bOmega_{jk}=\omega_{jk}\bD$ with $\bD=\diag(4^2,1,3^{-2},\ldots,25^{-2})$ 
and $\omega_{jk}=0.995^{|j-k|}$ for $j,k\in[p]$.

\item[{\bf Model 2.}] 
Let $\boldsymbol{\varepsilon}_t(u)={\bf D}^{1/2} \{\boldsymbol{\sigma}_t(u)\circ \ba_t(u)\}$, where $\circ$ is the Hadamard product,
${\bf D}=\{D_{jk}\}_{p\times p}$ has unit diagonal, entries $D_{jk}=0.97$ for $r(q-1)+1\leq j\neq k\leq rq$ with $r=\lceil p/3\rceil$ and $q=1,\ldots,\lfloor p/r\rfloor$, and zeros otherwise. Let $\ba_t(u)=\{a_{t,1}(u),\ldots,a_{t,p}(u)\}^{\T}$  be  i.i.d. $p$-dimensional standard Brownian motions on $\cU$  and $\boldsymbol{\sigma}_t(u)=\{\sigma_{t,1}(u),\ldots,\sigma_{t,p}(u)\}^{\T}$, where
$  \sigma_{t,j}^2(u)=u+ 0.1\int_{\cU} \exp\{(u^2+s^2)/2\} \sigma_{t-1,j}^2(s)   a_{t-1,j}^2(s)\,{\rm d}s$.

\item[{\bf Model 3.}] 
Let $\bvep_{t}(u) = \int_{\cU} \bA(u,v)\bvep_{t-1}(v)\,{\rm d}v + \bfzeta_{t}(u) $, where $\bfzeta_t(u)$ is generated in the same way as $\bvep_t(\cdot)$ in Model 1, and $\bA(u,v) = \{A_{jk}(u,v)\}_{p\times p}$ with $A_{jj}(u,v) = c\exp\{-(u^2+v^2)/2\}$ and $A_{jk}(u,v) = 0.5A_{jj}(u,v) I(|j-k|=1)$ for  $j\neq k$. The constant $c$ is chosen so that $\|A_{jj}\|_{\cS}=0.3$.

\item[{\bf Model 4.}]  
Let $\bvep_t(u)$ follow the same autoregressive structure as in Model 3, except that $\bfzeta_t(u)$ is generated from the same mechanism used for $\bvep_t(\cdot)$ in Model 2 and the constant $c$ is chosen so that $\|A_{jj}\|_{\cS}=0.5$.
% Let  $\bvep_{t}(u) = \int_{\cU} \bA(u,v)\bvep_{t-1}(v)\,{\rm d}v + \bfzeta_{t}(u) $, where $\bfzeta_t(u)$ is generated as   in Model 2,   
% and $\bA(u,v) = \{A_{jk}(u,v)\}_{p\times p}$ with $A_{jj}(u,v) = c\exp\{-(u^2+v^2)/2\}$ and $A_{jk}(u,v) = 0.5A_{jj}(u,v) I(|j-k|=1)$ for  $j\neq k$. 
% The constant $c$ is chosen so that $\|A_{jj}\|_{\cS}=0.5$.

% \item[{\bf Model 5.}]
% Let $\bvep_t(u) = \bB\bfzeta_t(u)$, where $\bfzeta_t(u)=\{\zeta_{t,1}(u),\ldots,\zeta_{t,p}(u)\}^{\T}$. For $k \in [k_0]$ with  $k_0=\lceil 3p/4 \rceil,$ 
% let $\zeta_{t,k}(u)=4u(1-u)a_{t,k}$, 
% where $(a_{1,k},\ldots,a_{n,k})^{\T}\sim {\cal N}({\bf 0}, \bSigma^a)$, with $\bSigma^a$ having unit diagonal, entries $0.8|i-j|^{-0.5}$ for $1\leqslant |i- j| \leqslant\lceil n/2\rceil$, and zeros otherwise. 
% For the remaining $p-k_0$ elements,  $\{ \zeta_{t,k_0+1}(u),\ldots,\zeta_{t,p}(u)\}^{\T}$ are generated from the same Gaussian basis-expansion mechanism used for $\bvep_t(\cdot)$  in Model 1 with dimension $p-k_0$ and $\bD=\diag(1^{-2},\ldots,25^{-2})$.
% The coefficient matrix $\bB=(B_{jk})_{p\times p}$ is generated such that,  $B_{jk}\sim \text{Unif}(-1,1)$ with probability 2/3 and $B_{jk}=0$ with probability 1/3 independently for $1\leq j\neq k\leq p$, and $B_{jj}=0.8$ for $j\in[p]$.
% \item[{\bf Model 6.}] 
% $\bvep_t(u)$ is generated in the same way as in Model 5, except that $k_0=\lceil p/2\rceil$ and $\{\zeta_{t,k_0+1}(u),\ldots,\zeta_{t,p}(u)\}^{\T}$ is generated using the same mechanism as $\bvep_t(\cdot)$ in Model 2 with dimension $p-k_0$.

\item[{\bf Model 5.}]
Let $\bvep_t(u) = \bB\bfzeta_t(u)$, where $\bfzeta_t(u)=\{\zeta_{t,1}(u),\ldots,\zeta_{t,p}(u)\}^{\T}$. For $k \in [k_0]$ with  $k_0=\lceil 3p/4 \rceil,$ 
let $\zeta_{t,k}(u)=4u(1-u)a_{t,k}$, 
where $(a_{1,k},\ldots,a_{n,k})^{\T}\sim {\cal N}({\bf 0}, \bSigma^a)$, with $\bSigma^a$ having unit diagonal, entries $0.8|i-j|^{-0.5}$ for $1\leqslant i,j \leqslant\lceil n/2\rceil$, and zeros otherwise. 
For the remaining $p-k_0$ elements,  $\{ \zeta_{t,k_0+1}(u),\ldots,\zeta_{t,p}(u)\}^{\T}$ are generated from the same Gaussian basis-expansion mechanism used for $\bvep_t(\cdot)$  in Model 1 with dimension $p-k_0$ and $\bD=\diag(1^{-4},\ldots,25^{-4})$.
The coefficient matrix $\bB=(B_{jk})_{p\times p}$ is generated such that,  $B_{jk}\sim \text{Unif}(-1,1)$ with probability 2/3 and $B_{jk}=0$ with probability 1/3 independently for $1\leq j\neq k\leq p$, and $B_{jj}=0.8$ for $j\in[p]$.

\item[{\bf Model 6.}] 
$\bvep_t(u)$ is generated in the same way as in Model 5, except that $k_0=\lceil 2p/3\rceil$ and $\{\zeta_{t,k_0+1}(u),\ldots,\zeta_{t,p}(u)\}^{\T}$ is generated using the same mechanism as $\bvep_t(\cdot)$ in Model 2 with dimension $p-k_0$.
\end{description}

Models~1--2 and Models~3--6 correspond to the null and alternative hypotheses, respectively. Specifically, Model~2 extends the univariate functional white noise case in  \cite{Zhang2016} to the multivariate setting.
Models 3 and 4 consider $p$-dimensional generalizations of the univariate functional autoregression of order 1 adopted by \cite{Horvath2013} for functional white noise test. Models~5 and 6 are inspired by \cite{Chang2017}, which studied the white noise test for high-dimensional scalar time series.

\begin{table}[!htb]
\caption{Comparison of empirical sizes (Models 1--2) and powers (Models 3--6) between our proposed test based on the QS kernel and other tests for $p=1$ at the 0.05 nominal level. All numbers are multiplied by 100.}
\small
\renewcommand{\arraystretch}{0.7}
\centering
\setlength{\aboverulesep}{0pt}
\setlength{\belowrulesep}{0pt} 
\resizebox{\textwidth}{!}{ 
\begin{tabular}{ccccccccccccccc}
\toprule
& \multirow{2}{*}{$n$} & \multirow{2}{*}{$N$}
& \multicolumn{3}{c}{${T}_n$}
& \multicolumn{3}{c}{${\rm GK}$}
& \multicolumn{3}{c}{${\rm KRS}$}
& \multirow{2}{*}{${\rm CR}_{\rm Pz}$}
& \multirow{2}{*}{${\rm CR}_{\rm Bt}$}
& \multirow{2}{*}{${\rm ZH}$} \\
\cmidrule(lr){4-6}
\cmidrule(lr){7-9}
\cmidrule(lr){10-12}
& & &
$L=2$ & $L=4$ & $L=6$
& $L=2$ & $L=4$ & $L=6$
& $L=2$ & $L=4$ & $L=6$
& & & \\
\midrule
Model 1 & 200 & 25 & 7.1 & 5.5 & 6.0 & 4.9 & 5.4 & 4.8 & 5.8 & 5.6 & 5.2 & 5.5 & 5.8 & 4.8 \\
& & 51 & 6.5 & 5.6 & 5.5 & 5.8 & 5.4 & 5.2 & 5.8 & 5.8 & 5.3 & 5.6 & 5.7 & 6.2 \\
& 400 & 25 & 5.2 & 5.2 & 5.1 & 5.6 & 4.2 & 4.5 & 5.4 & 4.7 & 4.2 & 4.8 & 5.3 & 6.2 \\
& & 51 & 5.0 & 5.1 & 5.2 & 5.1 & 5.0 & 4.6 & 5.6 & 4.9 & 4.1 & 4.9 & 5.1 & 4.8 \\ 
\midrule
Model 2 & 200 & 25 & 4.3 & 3.0 & 3.2 & 4.7 & 4.9 & 5.2 & 4.6 & 4.7 & 4.9 & 6.0 & 5.4 & 3.9 \\
& & 51 & 3.8 & 3.4 & 3.5 & 6.1 & 5.6 & 5.0 & 5.6 & 5.0 & 4.6 & 7.1 & 7.3 & 5.9 \\
& 400 & 25 & 3.6 & 4.8 & 4.5 & 5.9 & 5.8 & 5.6 & 5.0 & 5.6 & 5.2 & 6.0 & 5.7 & 4.2 \\
& & 51 & 4.8 & 4.6 & 4.9 & 7.5 & 6.5 & 6.5 & 7.0 & 6.2 & 5.4 & 6.6 & 7.7 & 6.4 \\ 
\midrule
Model 3 & 200 & 25 & 97.1 & 95.5 & 93.9 & 95.6 & 86.2 & 78.7 & 96.7 & 92.2 & 87.2 & 98.7 & 97.5 & 90.7 \\
& & 51 & 97.2 & 95.8 & 93.7 & 96.2 & 92.0 & 88.4 & 96.7 & 92.7 & 87.5 & 98.9 & 97.7 & 98.1 \\
& 400 & 25 & 100 & 100 & 100 & 100 & 99.9 & 99.3 & 100 & 99.9 & 99.8 & 100 & 100 & 95.6 \\
& & 51 & 100 & 100 & 100 & 100 & 100 & 99.9 & 100 & 99.9 & 99.8 & 100 & 100 & 100 \\ 
\midrule
Model 4 & 200 & 25 & 90.7 & 80.9 & 71.6 & 99.7 & 99.6 & 99.2 & 100 & 100 & 100 & 100 & 100 & 98.4 \\
& & 51 & 90.1 & 79.7 & 69.9 & 98.9 & 98.4 & 97.1 & 100 & 100 & 100 & 100 & 100 & 99.8 \\
& 400 & 25 & 100 & 100 & 100 & 100 & 100 & 100 & 100 & 100 & 100 & 100 & 100 & 99.4 \\
& & 51 & 100 & 100 & 100 & 100 & 100 & 100 & 100 & 100 & 100 & 100 & 100 & 100 \\ 
\midrule
Model 5 & 200 & 25 & 86.8 & 82.4 & 79.8 & 94.6 & 92.5 & 89.9 & 96.8 & 96.6 & 95.5 & 98.1 & 98.2 & 84.9 \\
& & 51 & 87.4 & 83.6 & 80.6 & 98.1 & 97.6 & 96.6 & 97.4 & 96.7 & 96.1 & 98.2 & 98.8 & 90.3 \\
& 400 & 25 & 99.7 & 99.7 & 99.4 & 100 & 99.8 & 99.8 & 100 & 100 & 100 & 100 & 100 & 98.8 \\
& & 51 & 99.9 & 99.9 & 99.8 & 100 & 100 & 100 & 100 & 100 & 100 & 100 & 100 & 99.1 \\ 
\midrule
Model 6 & 200 & 25 & 86.7 & 82.7 & 80.3 & 94.5 & 92.2 & 90.6 & 96.8 & 96.4 & 95.4 & 98.1 & 98.1 & 86.6 \\
& & 51 & 87.5 & 84.1 & 80.0 & 98.1 & 97.7 & 96.6 & 97.7 & 96.7 & 96.3 & 98.2 & 98.8 & 90.4 \\
& 400 & 25 & 99.7 & 99.7 & 99.5 & 100 & 99.8 & 99.8 & 100 & 100 & 100 & 100 & 100 & 98.5 \\
& & 51 & 99.9 & 99.7 & 99.7 & 100 & 100 & 100 & 100 & 100 & 100 & 100 & 100 & 99.3 \\ 
\bottomrule
\end{tabular}
}
\label{tab:p1}
\end{table}

Table \ref{tab:p1} reports the empirical sizes and powers of the proposed and competing tests for $p=1$. Since the results for different kernels exhibit similar patterns, we present only those based on the QS kernel here, and defer the results for other kernels to  Table \ref{tab:p1_supp} in the supplementary material. 
For Models 1 and 2, all tests maintain accurate size control, with empirical sizes generally close to the nominal level across different time-series lengths ($n\in\{200,400\}$) and numbers of observation points ($N\in\{25,51\}$). 
For Models 3--6, the empirical powers of all tests increase with $n$.  
When $n=200$, the proposed test has slightly lower power than some competing methods under Models 3-6. However, as $n$ increases, these differences become smaller and negligible. Overall, the proposed test performs  comparably to existing methods for testing univariate functional white noise.

\begin{table}[htb]
\caption{Empirical sizes (Models 1--2) and powers (Models 3--6) of the proposed test ${T}_n$ based on the QS kernel for $p\in\{10,50,100\}$ at the 0.05 nominal level. All numbers   are multiplied by 100.}
\centering
\renewcommand\arraystretch{0.7}
\setlength{\aboverulesep}{0pt}
\setlength{\belowrulesep}{0pt} 
\resizebox{\textwidth}{!}{
\begin{tabular}{cc cccccc  cccccc}
\toprule
& &\multicolumn{6}{c}{$n=200$} &  \multicolumn{6}{c}{$n=400$} \\
\cmidrule(lr){3-8}\cmidrule(lr){9-14}
& &\multicolumn{2}{c}{$p=10$} & \multicolumn{2}{c}{$p=50$} & \multicolumn{2}{c}{$p=100$} & \multicolumn{2}{c}{$p=10$} & \multicolumn{2}{c}{$p=50$} & \multicolumn{2}{c}{$p=100$} \\ 
& & $N$=25 & $N$=51 & $N$=25 & $N$=51 & $N$=25 & $N$=51  & $N$=25 & $N$=51 & $N$=25 & $N$=51 & $N$=25 & $N$=51 \\
\midrule
Model 1 & $L=2$  
& 4.3 & 5.1 & 3.9 & 3.9 & 2.7 & 3.1 & 4.6 & 4.7 & 5.1 & 4.0 & 3.7 & 3.6 \\
& $L=4$  
& 4.1 & 4.1 & 3.5 & 3.4 & 2.3 & 3.2 & 5.0 & 4.1 & 4.9 & 3.8 & 3.9 & 4.7 \\
& $L=6$  
& 4.6 & 4.0 & 4.0 & 3.5 & 2.5 & 2.8 & 4.7 & 3.9 & 4.4 & 3.9 & 3.2 & 4.0 \\
\midrule
Model 2 & $L=2$  
& 3.5 & 2.1 & 2.2 & 2.0 & 1.5 & 2.0 & 3.8 & 3.9 & 2.4 & 2.5 & 2.8 & 3.5 \\
& $L=4$  
& 2.5 & 1.7 & 1.7 & 1.4 & 1.4 & 1.9 & 3.1 & 2.8 & 2.3 & 2.9 & 1.6 & 3.2 \\
& $L=6$  
& 2.4 & 1.9 & 1.4 & 1.4 & 1.2 & 1.2 & 3.1 & 2.1 & 2.2 & 2.9 & 2.4 & 3.0 \\
\midrule
Model 3 & $L=2$  
& 100 & 100 & 99.7 & 100 & 99.9 & 100 & 100 & 100 & 100 & 100 & 100 & 100 \\
& $L=4$  
& 100 & 100 & 99.6 & 100 & 99.7 & 100 & 100 & 100 & 99.9 & 100 & 100 & 100 \\
& $L=6$  
& 100 & 100 & 99.6 & 100 & 99.5 & 100 & 100 & 100 & 99.9 & 100 & 100 & 100 \\
\midrule
Model 4 & $L=2$  
& 99.3 & 99.0 & 99.0 & 100 & 98.7 & 100 & 99.7 & 99.8 & 99.5 & 100 & 99.6 & 100 \\
& $L=4$  
& 99.3 & 99.1 & 99.1 & 100 & 98.3 & 100 & 99.7 & 99.9 & 99.4 & 100 & 99.5 & 100 \\
& $L=6$  
& 99.5 & 99.1 & 98.6 & 100 & 98.0 & 100 & 99.8 & 99.8 & 99.3 & 100 & 99.5 & 99.9 \\
\midrule
Model 5 & $L=2$  
& 83.6 & 86.0 & 72.2 & 73.8 & 62.1 & 67.7 & 99.4 & 99.6 & 96.8 & 97.6 & 94.3 & 94.6 \\
& $L=4$  
& 78.6 & 80.9 & 65.3 & 67.0 & 52.9 & 57.2 & 99.4 & 99.6 & 94.1 & 96.0 & 91.3 & 92.0 \\
& $L=6$  
& 75.8 & 77.9 & 61.3 & 61.5 & 48.3 & 50.1 & 99.1 & 99.5 & 92.5 & 94.6 & 89.3 & 90.0 \\
\midrule
Model 6 & $L=2$  
& 91.7 & 92.4 & 97.2 & 97.7 & 96.5 & 98.2 & 99.9 & 100 & 99.9 & 100 & 99.8 & 99.9 \\
& $L=4$  
& 90.3 & 90.8 & 95.4 & 96.6 & 94.5 & 97.1 & 99.9 & 100 & 99.9 & 100 & 99.9 & 99.8 \\
& $L=6$  
& 89.6 & 89.6 & 94.4 & 95.4 & 92.8 & 96.2 & 100 & 100 & 100 & 99.9 & 99.5 & 99.9 \\
\bottomrule
\end{tabular}
}
\label{tab:dis.H0} 
\end{table}

The simulation results for multivariate functional time series are reported in Table \ref{tab:dis.H0}, and Tables \ref{tab:dis.H0_supp} and \ref{tab:dis.HA_supp} in the supplementary material, which summarize the empirical sizes and powers of our proposed test.
Some apparent patterns are observable.
First, for Models 1 and 2, the proposed test exhibits reasonable size control, and 
for Models 3--6, our test achieves strong empirical powers.
% Second, the rejection rates (i.e., empirical sizes and powers) {\color{red}generally} tend to decrease as $p$ increases, showing the  impact of dimensionality on the bootstrap-based approximation. 
Second, increasing $n$ improves the performance of the proposed test, and a larger $N$ also enhances performance in most cases. Specifically, the empirical sizes become closer to the nominal level, while the empirical powers approach one, providing evidence for the consistency of the proposed test. Additionally, the empirical size and power results tend to be less favorable for larger values of $p,$ reflecting the impact of dimensionality on the bootstrap-based approximation.
% Nevertheless, increasing $n$ improves the performance of the proposed test,
% whereas increasing $N$ also enhances performance in most cases. The empirical sizes
% become closer to the nominal level and the empirical powers approach one,
% supporting the consistency of the proposed test.}
% {\color{red}However, increasing $n$ or $N$ improves test performance {\color{red}in most cases}, with the empirical sizes becoming closer to the nominal level and the empirical powers approaching one, supporting the consistency of our proposed test.} 
Finally, the results are generally insensitive to the choice of kernel function and lag order $L,$ indicating robust performance. 
Overall, the proposed test performs well across all settings. The slightly undersize phenomenon in the moderate- and high-dimensional scenarios may be due to the choice of bandwidth, which is often a difficult issue in practice.

\subsection{Goodness-of-fit test for   functional factor model}\label{subsec:simu.ff}
To evaluate the finite-sample performance of the proposed goodness-of-fit test for the functional factor model, we consider both correctly specified and misspecified settings in terms of the number of factors.
We first generate multivariate functional time series of length $n\in\{200,400\}$ and dimension $p\in \{10,50,100\}$ according to the following model with the true number of factors $r_0=3$.

\begin{description}
    \item [{\bf Model 7.}] Let $\bX_t(u)=\bA \bZ_{t}(u) + \bvep_t(u)$, where $\bZ_{t}(u)=\{Z_{t,1}(u),\ldots,Z_{t,r_0}(u)\}^{\T}$ and  $\bA\in \mathbb{R}^{p\times r_0}$ has entries independently sampled from $\text{Unif}(-\sqrt{3},\sqrt{3})$. Generate $Z_{t,l}(u)=\sum_{i=1}^{25}\zeta_{t,l,i}\phi_i(u)$ for $l\in[r_0]$, where $\{\phi_i(u)\}_{i=1}^{25}$ is a 25-dimensional Fourier basis function and the basis coefficients $\boldsymbol{\zeta}_{t,i}=(\zeta_{t,1,i}, \ldots, \zeta_{t,r_0,i})^{\T} ={\bf V}\boldsymbol{\zeta}_{t-1,i}+\boldsymbol{\epsilon}_{t,i}$ with ${\bf V} = (0.45^{|\ell-\ell'|+1})_{r_0\times r_0}$ and   $\boldsymbol{\epsilon}_{t,i}=(\epsilon_{t,1,i},\ldots,\epsilon_{t,r_0,i})^{\T}$ consisting of independent $\cN(0,i^{-1.5})$ components. For each $j\in[p]$, we generate $\varepsilon_{t,j}(u)=\sum_{i=1}^{10}2^{-(i+1)}\tilde{Z}_{t,j,i}\phi_{i}(u)$, where $\tilde{Z}_{t,j,i}$'s are independent standard normal random variables. 
\end{description} 

For estimation, we set $\ell_0=4$ in \eqref{eq:hatM} and consider two cases:
(i) $r=3$, where the model is correctly specified, the idiosyncratic components are white noise, and rejection rates reflect empirical sizes;
(ii) $r=1$, where the number of factors is underspecified, so that the idiosyncratic components retain latent dynamic factor structure and exhibit temporal dependence, and rejection rates reflect the empirical powers for detecting misspecification.

Table \ref{tab:factor} and Table \ref{tab:factor_supp} in the supplementary material report the rejection rates of the proposed test for Model 7 under the correctly specified scenario ($r=3$) and misspecified scenario ($r=1$). 
The results show that, under the correctly specified scenario, the test is slightly conservative for $n=200$, particularly when $p$ is large. 
As $n$ increases to $400$, the empirical sizes move  closer to the nominal level across all dimensions, indicating improved size control. 
Under misspecification, the empirical powers for $n=200$ are relatively high across all dimensions. 
As expected, enlarging  $n$ improves the empirical powers,  which approach one across all dimensions.
% This shows the consistency of the proposed test in detecting model misspecification.
These findings support the use of the proposed test as a goodness-of-fit diagnostic for the functional factor model by assessing residual serial dependence.

\begin{table}[H]
\caption{Rejection rates of our proposed test ${T}_n$  for Model 7 based on the QS kernel  at the 0.05 nominal level. All numbers  are multiplied by 100.}
\footnotesize
\centering
\renewcommand\arraystretch{0.6} 
\setlength{\aboverulesep}{2pt}
\setlength{\belowrulesep}{2pt}
\begin{tabular}{cccccccc}
% \begin{tabular}{p{1.7em}p{1.8em}cccccc}
\toprule
& &\multicolumn{3}{c}{$r=3$ (size)} &\multicolumn{3}{c}{$r=1$ (power)} \\\cmidrule(lr){3-5}\cmidrule(lr){6-8}
$n$&  $p$ & $L=2$   & $L=4$   & $L=6$ & $L=2$   & $L=4$   & $L=6$  \\ 
\midrule 
200&  10  & 3.2     & 2.0     & 2.6    & 96.3     & 94.6      &  92.9  \\
 & 50  & 1.2    & 1.2    & 0.8    & 99.4    & 98.7     &   97.3   \\
 &100 & 1.1    & 1.0     & 0.5    & 99.4      &  98.4    &  97.7   \\ \midrule 
400&  10   & 3.9     & 3.8     & 3.6 & 100 & 100 & 99.9\\
& 50& 2.6    & 1.9    & 1.8  & 100  & 100  & 100\\
& 100 & 2.0      & 1.3    & 1.1  & 100 & 99.9 & 99.8 \\
\bottomrule
\end{tabular}
\label{tab:factor}
\end{table}

\section{Real data analysis}\label{sec.real}
In this section, we apply our testing procedure to two real-world datasets from finance and demography. 
Both datasets consist of functional time series with discrete observations, which we first smooth using a local linear method to  recover the underlying functional trajectories. We then perform the proposed white noise test  on the reconstructed curves.
If the white noise hypothesis is rejected, we further fit a high-dimensional functional time series model (i.e., functional factor model in Section~\ref{sec.yield} and  vector functional autoregressive model in Section~\ref{sec.mortality}), and apply our test to assess the goodness-of-fit and determine the model complexity (i.e., the number of factors or lags). 
The  $p$-values are computed based on Quadratic Spectral kernel, Parzen kernel, and Bartlett kernel defined in Section \ref{sec.simu}, lag order  $L=2$, and $1000$ bootstrap replications.

\subsection{Yield data}\label{sec.yield}
Our first dataset contains sovereign zero-coupon yield curves, which  is downloaded from the Bloomberg terminal. It consists of end-of-week zero-coupon rates for $p=22$ countries (list is available in Table \ref{tab.list} of the supplementary material) from January 2023 to December 2024, with a total of $n=104$ observations. For each country and observation date, the yield curve is observed  at a set of discrete maturities $120\times u_{k} \in \{3,6,12,24,36,48,60,72,84,96,108,120\}$ months, covering a broad range from 3 months to 10 years. We first conduct the white noise test and obtain the corresponding $p$-values of zero for all three kernels, 
%$p_{\rm QS}= p_{\rm Pz}= p_{\rm Bt}=0$, 
indicating that the yield data are not a white noise sequence at the 0.05 significance level.
% {\color{red}We then apply the functional factor model in Section \ref{sec.ffm}  and perform 
% the sequential testing procedure to determine the number of factors. The results show that the null hypothesis is significantly rejected for $r=1$ and $r=2$ across all three kernels, whereas for $r=3$, the corresponding $p$-values are $p_{\rm QS}= 0.075$, $p_{\rm Pz}= 0.081$ and $p_{\rm Bt}=0.063$. Thus the estimated number of factors is $\hat{r}=3$ at the 0.05 significance level. By contrast, the ratio-based estimator proposed in \cite{lam2012} and \cite{guo2026} yields $\hat{r}=1$, which seems to underestimate the  number of factors.}
We then apply the functional factor model~\eqref{eq:ffm} in Section~\ref{sec.ffm} and examine several candidate numbers of factors. The results show that the null hypothesis is rejected for $r=1$ and $r=2$ across all three kernels at the 0.05 significance level, providing evidence that one or two factors are insufficient to capture the temporal dependence in the data. For $r=3$, the corresponding $p$-values based on Quadratic Spectral kernel, Parzen kernel and Bartlett kernel are, respectively, 0.075, 0.081 and 0.063, which suggests that the functional factor model~\eqref{eq:ffm} with three factors provides an adequate fit to the data.

\begin{figure}[htbp]
\centering
\includegraphics[scale=0.8]{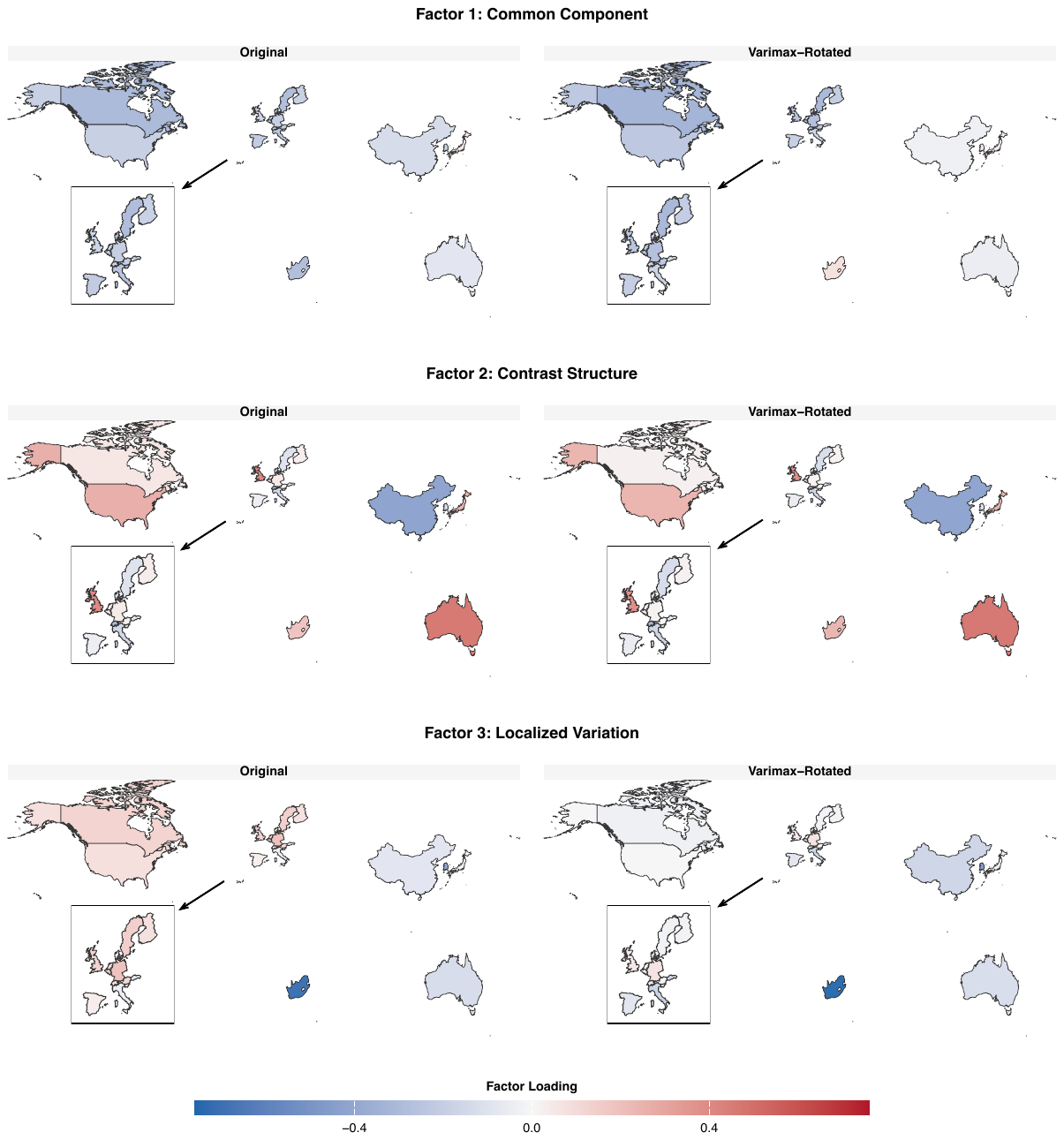}
\vspace{-4cm}
\caption{Spatial heatmaps based on estimated factor loading matrix (left column) and varimax-rotated loading matrix (right column) of 22 countries.}
\label{fig.yield}
\end{figure}

Figure \ref{fig.yield} displays spatial heatmaps of the estimated factor loading matrix and its varimax-rotated matrix, obtained by maximizing the sum of the variances of the squared loadings.
The varimax rotation clarifies the factor structure, enabling a clearer identification of each factor’s role in driving the dynamics of yield curves.
The first factor represents a global common component and influences the yield dynamics relatively uniformly across most countries. 
%although it assigns a more negative loading to Russia.
The second factor mainly captures heterogeneity across economies, distinguishing Western economies from China through loadings of opposite signs. 
Additionally, South Africa loads heavily on the third factor, suggesting that local sovereign risk primarily affects its yield dynamics.

\subsection{Age-specific mortality data}\label{sec.mortality}
This section applies our proposed testing method to a dataset of age-specific mortality rates for males, which is obtained from \cite{hmd2023}.
This dataset covers the period 1960--2013  from 24 countries, listed in Table \ref{tab.morlist} of the supplementary material. 
Let $W_{t,j,k}$ $(t\in[54], j\in[24], k\in[96])$ be the observed log mortality rate of $X_{t,j}(u_k)$ for males aged $k-1$ in the $j$-th country during year $1959+t$. 
After implementing local linear smoothers on $\{W_{t,j,k}\}_{k \in [96]}$, we conduct the proposed white noise test to the reconstructed curve  series to examine whether $\bX_t(\cdot)$ exhibits temporal dependence. The resulting $p$-values based on the three kernels are all zero,
%$p_{\rm QS}= p_{\rm Pz}= p_{\rm Bt}=0$,   
strongly rejecting the null hypothesis of white noise at the 0.05 significance level.  
To capture the temporal dependence,
we then employ our test to evaluate the goodness-of-fit of a vector functional autoregressive model of order 1:
\begin{align}\label{eq:fvar1}
	\bX_t(u) = \int_{\cU} \bA(u,v)\bX_{t-1}(v)\,{\rm d}v + \bvep_t(u)\,, \quad t=2,\ldots,n\,,
\end{align}
where $\bvep_t(\cdot)$ is a white noise sequence, and $\bA(\cdot,\cdot)=\{A_{jk}(\cdot,\cdot)\}_{p \times p}$ is a sparse functional transition matrix.
We estimate the model  using the three-step procedure in \cite{Guo2023} (see Section \ref{threesteps} of the supplementary material for details), and  obtain the residuals $\hat{\bvep}_t(\cdot)$.
We then apply the white noise test to $\{\hat{\bvep}_t(\cdot)\}$ to assess whether $\{\bvep_t(\cdot)\}$ is white noise, and obtain the corresponding $p$-values 
of 0.187, 0.161, and 0.201 for the Quadratic Spectral, Parzen,
and Bartlett kernels, respectively,
%$p_{\rm QS}=0.187$, $p_{\rm Pz}=0.161$, and $p_{\rm Bt}=0.201$, 
all exceeding the 0.05 significance level. These results provide empirical support for the adequacy  of the fitted vector functional autoregressive model with lag order 1.

We adopt the functional network Granger causality formulation under \eqref{eq:fvar1} to capture the Granger-type causal relationships among mortality rates across countries. 
Specifically, if $\|A_{jk}\|_{\cS}\neq 0$,  $\{X_{t,k}(\cdot)\}$ is said to be the  Granger-type causal for  $\{X_{t,j}(\cdot)\}$ or equivalently there is a directed edge from node $k$ to node $j$.
For better visualization,  Figure \ref{fig.mortality} presents a directed network with indegree 3 (i.e., each node receives connections from three other  nodes, excluding self‑loops).
It is obvious that the network reveals  regional clustering. In particular, former Soviet Union countries, such as RUS (Russia), UKR (Ukraine), EST (Estonia), and BLR (Belarus), form a tightly connected cluster, reflecting persistent historical and socioeconomic ties \cite[]{MV2002, Denny2010}. 
Moreover, HUN (Hungary) and ITA (Italy) are both placed in the center of the network,  yet their causal roles differ. 
HUN  acts as a  geographical bridge between Eastern and Western Europe, while ITA  serves as a reference trajectory for international mortality evolution.

\vspace{-1.5em}
\begin{figure}[htbp]
    \centering
        \includegraphics[scale=0.4]{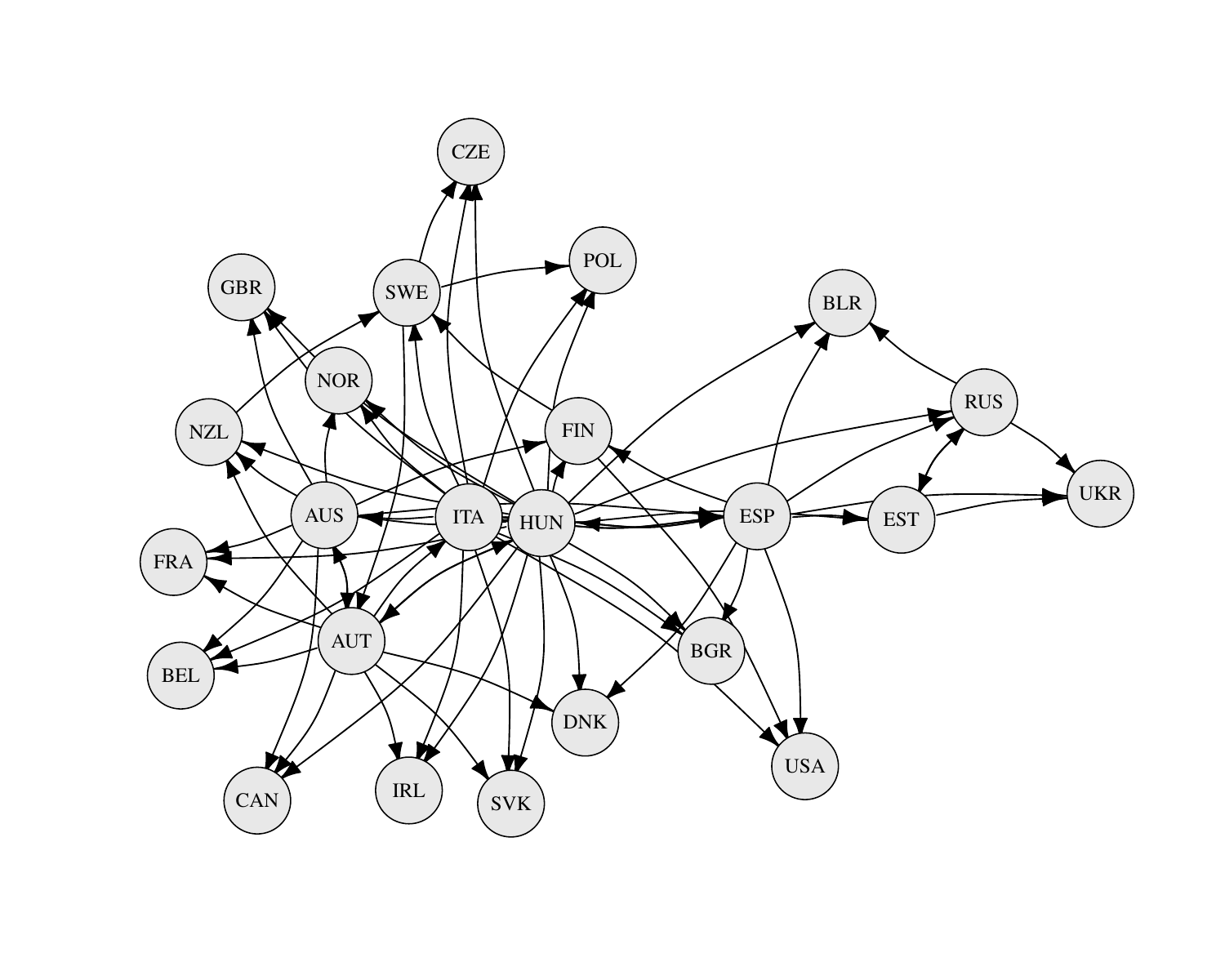}
\vspace{-3em}
     \caption{The estimated directed network with indegree $=3$ for $p=24$ countries.}
    \label{fig.mortality}
\end{figure}

\section{Discussion}
%Motivated by the rapid emergence of high-dimensional functional time series, 
In this paper,
we develop a unified error-contamination testing framework for functional white noise in high dimensions, which has not yet been explored in the literature.
To overcome both methodological and theoretical challenges induced by temporal dependence, high dimensionality, infinite-dimensional functional objects and estimation errors, we construct a supremum-type test statistic, and establish the Gaussian approximation theory 
%for high-dimensional dependent processes 
under a high-level condition. We further verify this condition in two concrete applications.
%To obtain the critical value, we develop computationally feasible parametric bootstrap procedure to test the nullity of cross-autocovariance functions.
%In addition, we apply our framework to two concrete testing scenarios and rigorously verify the corresponding high-level condition. 
%Finally, extensive simulations and empirical applications are provided to  illustrate the size and power of the proposed tests.

We identify several important directions for future research. First, for the tuning parameter $b_n$, we adopt the bandwidth selection method of \cite{Andrews1991}, which performs well in our numerical studies. Nevertheless, developing
a testing-driven bandwidth criterion to further improve size and power remains a valuable open problem.
Second, our test involves the choice of lag order $L$. Although the test is stable across different choices of $L$ in simulations, 
selecting $L$ for testing efficiency deserves further study. Alternatively, a spectral-domain approach could also be developed to avoid lag selection, as in univariate functional time series \cite[]{Zhang2016,Bagchi2018}. 
%Developing such tests for high-dimensional functional time series is a natural extension of existing spectral-domain tests for univariate functional time series \cite[]{Zhang2016,Bagchi2018}.}
Third, the proposed supremum-type test statistic $T_n$ is well suited to sparse alternatives with strong cross-autocovariance signals. For dense alternatives with weak signals accumulating over $\cU^2$, an $L_2$-type statistic $T_n'$ based on the Hilbert--Schmidt norm may be more powerful. Combining these two statistics is a promising direction for future research.
Finally, it is of great interest to investigate whether the proposed  framework can be adapted to other problems, such as,  testing stationarity in high-dimensional functional time series.  These topics are beyond the scope of the current paper and will be pursued elsewhere.

% \bigskip
% \noindent{\bf Data availability statement:}
% The data that support our findings are available from the corresponding author, upon request.

%\bibliographystyle{Chicago}
{
\setlength{\bibsep}{0pt} 
\bibliographystyle{dcu}
\bibliography{WNTest_Ref} 

\begin{thebibliography}{11}
	
	 
	
	\bibitem[Boucheron et al.(2013)]{Boucheron2013_supp}
	Boucheron, S., Lugosi, G. and Massart, P. (2013). {\sl Concentration inequalities: A Nonasymp- totic Theory of Independence}, Oxford University Press.
	
	\bibitem[Chang et al.(2024)]{ChangChenWu2023_supp}
	Chang, J., Chen, X. and Wu, M. (2024). Central limit theorems for high dimentional dependent data, {\sl Bernoulli} {\bf30}: 712--742.
	
	\bibitem[Chang et al.(2013)]{Chang2013_supp}
	Chang, J., Tang, C. and Wu, Y. (2013). Marginal empirical likelihood and sure independence feature screening, {\AS}  {\bf 41}: 2123--2148.


    \bibitem[Chernozhukov et al.(2017)]{Cher2017_supp}
	Chernozhukov, V., Chetverikov, D. and Kato, K. (2017). Central limit theorem and bootstrap approximations in high dimensions, {\AP} {\bf45}: 2309--2352.

    \bibitem[Chernozhukov  et al.(2022)]{Cher2022_supp}
	Chernozhukov, V., Chetverikov, D., Kato, K. and Koike, Y. (2022). Improved central limit theorem and bootstrap approximations in high dimensions, {\AS}  {\bf 50}: 2562--2586.


    \bibitem[Fan et al.(2018)]{Fan2018_supp}
Fan, J., Wang, W. and Zhong, Y. (2018). An $\ell_{\infty}$ eigenvector perturbation bound and its application to robust covariance estimation, {\sl Journal of Machine Learning Research} {\bf 18}: 1--42.
	
	\bibitem[Guo et al.(2025)]{Guo_etal2025_supp}
	Guo, S., Li, D., Qiao, X. and Wang, Y. (2025). From sparse to dense functional data in high dimensions: Revisiting phase transitions from a non-asymptotic perspective, {\sl Journal of Machine Learning Research} {\bf26}: 1--40.
	
	\bibitem[Guo et al.(2026)]{GQWW_2026_supp}
	Guo, S., Qiao, X., Wang, Q. and Wang, Z. (2026). Factor modeling for high-dimensional functional time series, {\JBES} {\bf 44}: 106--119.
	
	 
	
	\bibitem[Lam et al.(2011)]{Lam2011_supp}
	Lam, C., Yao, Q.  and Bathia, N. (2011). Estimation of latent factors for high-dimensional time series, {\sl Biometrika}  {\bf 98}: 901--918.
	
	\bibitem[Ledoux  and Talagrand(1991)]{Ledoux1991_supp}
	Ledoux, M.  and Talagrand, M. (1991). {\sl Probability in Banach spaces: Isoperimetry and processes}, Springer.
	
		\bibitem[Li et al.(2025)]{LQW2025_supp}
Li, D., Qiao, X.  and Wang, Y. (2025). Factor-guided estimation of large covariance matrix function with conditional functional sparsity, {\sl Journal of Econometrics} {\bf 251}: 106070.
 
	
	\bibitem[Meerschaert et al.(2013)]{MWX2013_supp}
	Meerschaert, M., Wang, W. and Xiao, Y. (2013). Fernique-type inequalities and moduli of continuity for anisotropic Gaussian random fields, {\sl Transactions of the American Mathematical Society} {\bf365}: 1081--1107.
	
	
	\bibitem[Talagrand(2014)]{Talagrand2014_supp}
	Talagrand, M. (2014). {\sl Upper and lower bounds for stochastic processes}, Springer.
	
	\bibitem[Vershynin(2018)]{Vershynin2018_supp}
	Vershynin, R. (2018). {\sl High-dimensional probability: An introduction with applications
		in data science}, Cambridge University Press.
	 
\end{thebibliography}
}

\newpage
 
\setcounter{equation}{0}
\setcounter{table}{0}
\setcounter{figure}{0}
\setcounter{lemma}{0}
\newtheorem{props}{Proposition}
\setcounter{props}{0}
\setcounter{page}{1}
\renewcommand{\thelemma}{\thesection\arabic{lemma}}
\renewcommand{\theequation}{S.\arabic{equation}}
\renewcommand{\theprops}{\thesection\arabic{props}}
\renewcommand{\thepage}{S\arabic{page}}
\renewcommand{\thetable}{S\arabic{table}}
\renewcommand{\thefigure}{S\arabic{figure}}

\spacingset{1.68}
\appendix
\begin{center}
{\bf\Large Supplementary material to ``Testing for functional white noise in high dimensions'' }
	% \begin{center}
	% 	{\noindent Jinyuan Chang, ~ Qing Jiang, ~Xinghao Qiao, ~Lin Yang}
	% \end{center}
% \bigskip
% \medskip\\
% {\bf Jinyuan Chang, ~\, Qing Jiang, ~\,Xinghao Qiao$^*$, ~\,Lin Yang}
\end{center} 

Throughout the supplementary material, we use $C \in (0,\infty) $ to denote a generic universal positive finite constant that does not depend on $(n,p,\cU)$, and  may be different in different uses. 
Denote $|x|_{+}=\max(0,x)$, $x\vee y=\max(x,y)$ and $x\wedge y=\min(x,y)$. For any positive integer $q\geq1$, we write $[q]=\{1,\ldots,q\}$ and $[q]_0=\{0,1,\ldots,q\}$.    Let $I(\cdot)$ denote the indicator function and ${\rm vec}$ be the vectorization for matrices.  
Denote by $\lfloor x \rfloor $ the largest integer not greater than $x$.
For two sequences of positive numbers $\{a_n\}$ and $\{b_n\}$, we write $a_n\lesssim b_n$ or $b_n\gtrsim a_n$ if $\limsup_{n\rightarrow\infty}a_n/b_n\leqslant c_0$ for some positive constant $c_0$. Let $a_n\asymp b_n$ if $a_n\lesssim b_n$ and $b_n\lesssim a_n$ hold simultaneously.
We write $a_n\ll b_n$ or $b_n\gg a_n $ if $\limsup_{n\rightarrow\infty}a_n/b_n=0$. Let $\cdot^{\T}$ be the transpose of $\cdot$.
For two $p$-dimensional real valued vectors $\ba=(a_1,\dots,a_p)^{\T}$ and $\bw=(w_1,\ldots,w_p)^{\T}$, we say $\ba\leq \bw$ if $a_j\leq w_j$ for any $j\in[p]$.
For a  random variable $Y$, let $\|Y\|_r=\{ \mathbb{E}(|Y|^r)\}^{1/r}$ for a positive integer $r$.  
Given $\gamma > 0$, we define the function $\psi_{\gamma}(x) := \exp(x^{\gamma}) - 1$ for any $x > 0$.
%, and let $\psi_{\gamma}^{-1}$ denote its inverse function.
For a real-valued random variable $\xi$, we define
$\|\xi\|_{\psi_{\gamma}} := \inf[\lambda > 0 : \mathbb{E}\{\psi_{\gamma}(|\xi|/\lambda)\} \leq 1]$.
For a fixed compact $r$-dimensional interval $\mathcal{V}$ equipped with a distance $d$, let $N_{\epsilon} := N(\mathcal{V}, d; \epsilon)$ denote the covering number of $\mathcal{V}$ with respect to the pseudo-metric $d$ and radius $\epsilon$. 
For a $m\times n$-matrix $\bA=(a_{ij})_{m\times n}$ and a $p$-dimensional vector $\bv=(v_1,\ldots,v_p)^{\T}$, let $\|\bA\|_{\rm op} = \sqrt{\lambda_{\max}(\bA\bA^{\T})}$, $\|\bA\|_\infty = \max_{i\in[m]}\sum_{j=1}^n |a_{ij}| $, $\|\bA\|_{\rm F}=\sqrt{\sum_{i=1}^m\sum_{j=1}^n a_{ij}^2}$, $\|\bA\|_{\max} = \max_{i\in[m]}\max_{j\in[n]} |a_{ij}| $, $|\bv|_2 = \sqrt{\sum_{j=1}^p v_j^2}$ and $|\bv|_{\max} = \max_{j\in[p]}|v_j|$.
%Denote by $\mathbb{S}^{q-1}$ the $q$-dimensional unit sphere. 
Denote  by ${\cal L}_2(\cU)$ the Hilbert space of square integrable functions defined on a compact set $\cU$, and $\cU^2=\cU\times\cU$  the Cartesian product of $\cU$.
For any $\mathcal{B}\in \mathcal{S}:=\mathcal{L}_2(\cU^2)$, we denote the supremum norm and  Hilbert-Schmidt norm by  $\|\mathcal{B}\|_\infty=\sup_{(u,v)\in\cU^2}|\mathcal{B}(u,v)|$  and $\|\mathcal{B}\|_{\mathcal{S}}=\{\int_{\cU}\int_{\cU}\mathcal{B}(u,v)^2\,{\rm d}u{\rm d}v\}^{1/2}$, respectively.
For any $\boldsymbol{\mathcal{B}}=(\mathcal{B}_{jk})_{m\times n}$ with each $\mathcal{B}_{jk}\in\mathcal{S}$, we denote its functional versions of $L_1$-norm by $\|\boldsymbol{\mathcal{B}}\|_{\cS,1} = \max_{k\in[n]}\sum_{j=1}^m\|\mathcal{B}_{jk}\|_{\cS}$. Similarly, we define $\|\boldsymbol{\mathcal{B}}\|_{\cS,\infty}=\max_{j\in[m]}\sum_{k=1}^n\|\mathcal{B}_{jk}\|_{\cS}$, $\|\boldsymbol{\mathcal{B}}\|_{\cS,\max}=\max_{j\in[m]}\max_{k\in[n]}\|\mathcal{B}_{jk}\|_{\cS}$ and   $\|\boldsymbol{\mathcal{B}}\|_{\infty,\max} = \max_{j\in[m]}\max_{k\in[n]}\|\mathcal{B}_{jk}\|_{\infty}$.  
Let $n_\ell=n-\ell$ for any  integer $\ell \geq 0$.

%\newpage

\section{Proofs of Theorems \ref{thm.Tn.H0} and \ref{thm.Tn.H1}}

\subsection{Preliminaries}\label{sec.preliminary}

For technical convenience, we introduce several intermediate statistics that will be used repeatedly in the proofs. Without loss of generality, we take $\cU=[0,1]$ and partition it into $M$ subintervals $\{B_1,\ldots,B_M\}$ of equal length $M^{-1}$. We then define the midpoint discretizations of $\cU$ as $\cU_M=\{b_1,\ldots,b_M\}$, 
where each $b_m$ is the center of the subinterval $B_m$. Recall that ${T}_{n}^{{\rm G }}=\max_{r\in [Lp^2]}\sup_{(u,v)\in\cU^2}|{\cG}_r(u,v)  |$. We define 
\begin{align}\label{Tn.G}
      \tilde{T}_{n}^{{\rm G }}=\max_{r\in [Lp^2]}\sup_{(u,v)\in\cU^2}|\tilde{\cG}_r(u,v)  |\,,
\end{align}
where $\tilde{\boldsymbol{\cG}}(u,v)=\{\tilde\cG_1(u,v),\ldots,\tilde\cG_{Lp^2}(u,v)\}^{\T}$ is defined analogously to ${\boldsymbol{\cG}}(u,v)$ before \eqref{eq:TnG}, with $\{\hat{\bvep}_t(\cdot)\}$ in $\bxi_n(u,v)$ replaced by $\{{\bvep}_t(\cdot)\}$. Based on this construction, we further define the discretized version of $\tilde{T}_{n}^{{\rm G}}$ as
\begin{align}\label{Tn.dis.G}
      \tilde{T}_{n}^{{\rm dis, G }}=\max_{r\in [Lp^2]}\max_{(u,v)\in\cU_M^2}|\tilde{\cG}_r(u,v)  |\,.
\end{align}
Throughout the proofs, we take $M=M_n\asymp (np)^{C_M}$  for a sufficiently large constant $C_M>0$.
Define the empirical versions of the cross-autocovariance functions $\bSigma_{\varepsilon\delta}^{(\ell)}(u,v)= \cov\{\bvep_{t}(u),\\ \bdelta_{t+\ell}(v)\}  $, $\bSigma_{\delta\vep}^{(\ell)}(u,v) = \cov\{\bdelta_{t}(u), \bvep_{t+\ell}(v)\} $ and $\bSigma_{\delta\delta}^{(\ell)}(u,v)= \cov\{\bdelta_{t}(u), \bdelta_{t+\ell}(v)\} $  as  
\begin{align}
 \widetilde{\bSigma}^{(\ell)}_{\vep\delta}(u,v)
 =&\, \{\widetilde{\Sigma}^{(\ell)}_{\vep\delta,jj'}(u,v)\}_{p\times p} =   \frac{1}{n_\ell} \sum_{t=1}^{n_\ell} \{\bvep_{t}(u)-\bar{\bvep}(u)\} \{\bdelta_{t+\ell}(v)-\bar{\bdelta}(v)\}^{\T}\,, \nonumber \\
  \widetilde{\bSigma}^{(\ell)}_{\delta\vep}(u,v)
 =&\, \{\widetilde{\Sigma}^{(\ell)}_{\delta\vep,jj'}(u,v)\}_{p\times p} =   \frac{1}{n_\ell} \sum_{t=1}^{n_\ell} \{\bdelta_{t}(u)-\bar{\bdelta}(u)\} \{\bvep_{t+\ell}(v)-\bar{\bvep}(v)\}^{\T}\,, \label{eq:def.sigmadv}\\
\widetilde{\bSigma}^{(\ell)}_{\delta\delta}(u,v)
 =&\, \{\widetilde{\Sigma}^{(\ell)}_{\delta\delta,jj'}(u,v)\}_{p\times p} =   \frac{1}{n_\ell} \sum_{t=1}^{n_\ell} \{\bdelta_{t}(u)-\bar{\bdelta}(u)\} \{\bdelta_{t+\ell}(v)-\bar{\bdelta}(v)\}^{\T}\,, \nonumber
\end{align}  
respectively, where $\bar{\bvep}(\cdot)=n^{-1}\sum_{t=1}^n\bvep_t(\cdot)$ and $\bar{\bdelta}(\cdot)=n^{-1}\sum_{t=1}^n\bdelta_t(\cdot)$.  Based on \eqref{eq:def.sigmadv} and the constructions of $T_n$ and $\tilde{T}_n$, we can conclude that
\begin{align}\label{eq:diff.tTn.Tn}
    |T_n-\tilde{T}_n|  
    \leq   n^{1/2}\max_{\ell\in[L]}   \big\{
\|\widetilde\bSigma_{\varepsilon\delta}^{(\ell)}\|_{\infty,\max} + \|\widetilde\bSigma_{\delta\varepsilon}^{(\ell)} \|_{\infty,\max} + \|\widetilde\bSigma_{\delta\delta}^{(\ell)} \|_{\infty,\max}\big\}\,.
\end{align}  
In addition, following the construction of ${\boldsymbol{\cG}}^*(u,v)$ in \eqref{eq:cG}, we define 
\begin{align}
  \tilde{\boldsymbol{\cG}}^*_{\vep\delta}(u,v) =&\, \frac{1}{\sqrt{n_L} }\sum_{t=1}^{n_L} \varrho_t \{ \tilde{\bfeta}^{\vep\delta}_t(u,v)  -\bar{\tilde{\bfeta}} ^{\vep\delta}(u,v)\}  \,,\nonumber\\
  \tilde{\boldsymbol{\cG}}^*_{\delta\vep}(u,v) =&\, \frac{1}{\sqrt{n_L}}\sum_{t=1}^{n_L} \varrho_t \{ \tilde{\bfeta}^{\delta\vep}_t(u,v)  -\bar{\tilde{\bfeta}} ^{\delta\vep}(u,v)\}     \,,\label{eq:def.cgdv}\\
  \tilde{\boldsymbol{\cG}}^*_{\delta\delta}(u,v) =&\, \frac{1}{\sqrt{n_L}}\sum_{t=1}^{n_L} \varrho_t \{ \tilde{\bfeta}^{\delta\delta}_t(u,v)  -\bar{\tilde{\bfeta}} ^{\delta\delta}(u,v)\}   \,,\nonumber
\end{align}
where 
\begin{align*}
    \tilde{\bfeta}_t^{\vep\delta}(u,v)
= &~\{({\rm vec} [\{\bvep_t(u)-\bar{\bvep}(u)\}\{{\bdelta}_{t+1}(v)-\bar{{\bdelta}}(v)\}^{\T}])^{\T},\\
&~~~~~\ldots, 
({\rm vec} [\{\bvep_t(u)-\bar\bvep(u)\}\{\bdelta_{t+L}(v)-\bar\bdelta(v)\}^{\T}] )^{\T}\}^{\T}\,,\\
\tilde{\bfeta}_t^{\delta\vep}(u,v)
=&~ \{({\rm vec} [\{\bdelta_t(u)-\bar{\bdelta}(u)\}\{{\bvep}_{t+1}(v)-\bar{{\bvep}}(v)\}^{\T}])^{\T},\\
&~~~~~ \ldots, 
({\rm vec} [\{\bdelta_t(u)-\bar\bdelta(u)\}\{\bvep_{t+L}(v)-\bar\bvep(v)\}^{\T}] )^{\T}\}^{\T}\,,\\
\tilde{\bfeta}_t^{\delta\delta}(u,v)
=&~ \{({\rm vec} [\{\bdelta_t(u)-\bar{\bdelta}(u)\}\{{\bdelta}_{t+1}(v)-\bar{{\bdelta}}(v)\}^{\T}])^{\T},\\
&~~~~~ \ldots, 
({\rm vec} [\{\bdelta_t(u)-\bar\bdelta(u)\}\{\bdelta_{t+L}(v)-\bar\bdelta(v)\}^{\T}] )^{\T}\}^{\T} \,,
\end{align*}  
 for any $(u,v)\in\cU^2$ and $t\in[n_L]$. Here,  $\bar{\tilde{\bfeta}}^{\vep\delta}(u,v)=n_L^{-1}\sum_{t=1}^{n_L}\tilde{\bfeta}_t^{\vep\delta}(u,v)$, $\bar{\tilde{\bfeta}}^{\delta\vep}(u,v)=n_L^{-1}\sum_{t=1}^{n_L}\tilde{\bfeta}_t^{\delta\vep}(u,v)$, and $\bar{\tilde{\bfeta}}^{\delta\delta}(u,v)= n_L^{-1}\sum_{t=1}^{n_L}\tilde{\bfeta}_t^{\delta\delta}(u,v)$.  Analogously to \eqref{eq:diff.tTn.Tn},  we  have
\begin{align}\label{eq:diff.tTnG*.TnG*}
|{T}_n^{\rm G*} - \tilde{T}_n^{\rm G*}|
\lesssim   \|\tilde{\boldsymbol{\cG}}_{\varepsilon\delta}^*\|_{\infty,\max} + \|\tilde{\boldsymbol{\cG}}_{\delta\varepsilon}^*\|_{\infty,\max}+\|\tilde{\boldsymbol{\cG}}_{\delta\delta}^*\|_{\infty,\max}   \,. 
\end{align}

\subsection{Approximation Results under the Uncontaminated Functional Time Series}\label{sec.oracle}

Recall that $\tilde{T}_n$ and $\tilde{T}_n^{{\rm G}*}$ are defined analogously to ${T}_n$ and ${T}_n^{{\rm G}*}$, respectively, with $\hat{\bvep}_t(\cdot)$ replaced by $\bvep_t(\cdot)$. According to \eqref{eq:diff.tTn.Tn} and \eqref{eq:diff.tTnG*.TnG*}, the difference between the contaminated and uncontaminated statistics is asymptotically negligible under Condition~\ref{c.rate}. Therefore, the proof of Theorem~\ref{thm.Tn.H0} reduces to the uncontaminated case, which is handled by the following lemma and proposition. Specifically, Proposition~\ref{pro.Tn.H0} establishes a Gaussian approximation result under $\{\bvep_t(\cdot)\}$. The proofs of Lemma \ref{lem.dis.Gaussian} and Proposition \ref{pro.Tn.H0} are given in Sections \ref{proof.lem.dis.Gaussian}  and \ref{proof.pro.Tn.H0}, respectively. Recall $M \asymp (np)^{C_M}$  for a sufficiently large constant $C_M>0$.

\begin{lemma}\label{lem.dis.Gaussian}
   Under Conditions {\rm\ref{c.alpha}--\ref{c.subgaussian}}, we have $\sup_{x\geq 0}|\mathbb{P}(\tilde{T}_{n}^{\rm dis,G}\leq x)-\mathbb{P}(\tilde{T}_{n}^{\rm G}\leq x)|=o(1)$.
\end{lemma}

\begin{props}\label{pro.Tn.H0}
Assume $p\geq n^{\upsilon}$ for some  constant $\upsilon>0$. The following results hold. 

{\rm (i)} Under  Conditions {\rm\ref{c.alpha}--\ref{c.subgaussian}} and the null hypothesis $H_0$, we have $\sup_{x\in \mathbb{R}}|\mathbb{P}(\tilde{T}_{n} \leq x)-\mathbb{P}(\tilde{T}_{n}^{\rm G}\leq x)|=o(1)$ provided that $\log p \ll n^{2/21}$.

{\rm (ii)} Assume $b_{n}\asymp n^{\rho}$ for some constant $\rho$ satisfying $0<\rho<(\vartheta-1)/(3\vartheta-2)$. Under  Conditions {\rm\ref{c.alpha}--\ref{c.kernelF}},  we have $\sup_{x\in \mathbb{R}}|\mathbb{P}(\tilde{T}_{n}^{\rm G}\leq x)-\mathbb{P}(\tilde{T}_n^{{\rm G}*} \leq x\,|\,\widetilde{\mathcal{D}}_n)|=o_{\rm p}(1)$, provided that $\log p \ll n^{\iota}$ for some $\iota>0$ depending only on $(\rho,\vartheta)$,  where $\widetilde{\mathcal{D}}_n=\{\bvep_t(\cdot)\}_{t\in [n]}$. 
\end{props}

\subsection{Proof of Theorem \ref{thm.Tn.H0}}

Recall $\mathcal{D}_n=\{\hat\bvep_t(\cdot)\}_{t\in [n]}$. Let $\mathcal{F}_n$ denote the $\sigma$-field generated by $\{\bvep_t(\cdot),\bdelta_t(\cdot)\}_{t\in [n]}$. According to \eqref{err.m}, it holds that $\sigma(\mathcal{D}_n)\subseteq \mathcal{F}_n$. Thus, for any $x\in \mathbb{R}$ and $s_1>0$,
\begin{align}\label{upper.1}
    & \bbP({T}_n^{\rm G*}\leq x\,|\, \cD_n)  =  \bbE\{\bbP({T}_n^{\rm G*}\leq x\,|\, \mathcal{F}_n) \,|\, \cD_n\}\nonumber\\
     &~~~~~ \leq    \bbE\{\bbP(\tilde{T}_n^{\rm G*}\leq x + s_1\,|\, \mathcal{F}_n) \,|\, \cD_n\}+\bbE\{\bbP(|{T}_n^{\rm G*}-\tilde{T}_n^{\rm G*}| > s_1\,|\, \mathcal{F}_n) \,|\, \cD_n\}\nonumber\\
     &~~~~~\leq    \bbP(\tilde{T}_n^{\rm dis,G}\leq x + s_1 ) +\bbE\big\{\bbP(|{T}_n^{\rm G*}-\tilde{T}_n^{\rm G*}| > s_1\,|\, \mathcal{F}_n) \,|\, \cD_n\big\}\nonumber\\
     &~~~~~~~+ \bbE\bigg\{\sup_{y\in \mathbb{R}}\big|\bbP(\tilde{T}_n^{\rm G*}\leq y  \,|\, \mathcal{F}_n) - \bbP(\tilde{T}_n^{\rm dis,G}\leq y  )\big| \,\bigg|\, \cD_n\bigg\}\,.
\end{align}
On the other hand, for any $x\in \mathbb{R}$ and $s_2 >0$,
\begin{align*}
    & \bbP({T}_n\leq x) \geq   \bbP(\tilde{T}_n\leq x - s_2) - \bbP(|{T}_n-\tilde{T}_n|> s_2)\\
    &~~~~\geq  \bbP(\tilde{T}_n^{\rm dis,G}\leq x - s_2 )   - \bbP(|{T}_n-\tilde{T}_n|> s_2) -\sup_{y\in \mathbb{R}}\big|\bbP(\tilde{T}_n\leq y  )-\bbP(\tilde{T}_n^{\rm dis,G}\leq y   )\big|\,.
\end{align*}
This, together with \eqref{upper.1}, yields
\begin{align*}
    & \bbP({T}_n^{\rm G*}\leq x\,|\, \cD_n)  - \bbP({T}_n\leq x) \\
    \leq&~  \sup_{y\in \mathbb{R}}\big|\bbP(\tilde{T}_n\leq y  )-\bbP(\tilde{T}_n^{\rm dis,G}\leq y   )\big|+ \bbE\bigg\{\sup_{y\in \mathbb{R}}\big|\bbP(\tilde{T}_n^{\rm G*}\leq y  \,|\, \mathcal{F}_n) - \bbP(\tilde{T}_n^{\rm dis,G}\leq y  )\big| \,\bigg|\, \cD_n\bigg\}\\
    &~+\bbE\big\{\bbP(|{T}_n^{\rm G*}-\tilde{T}_n^{\rm G*}| > s_1\,|\, \mathcal{F}_n) \,|\, \cD_n\big\}+\bbP(|{T}_n-\tilde{T}_n|> s_2)+\bbP(   |\tilde{T}_n^{\rm dis,G}- x|\leq  s_1 +s_2 ) \,.
\end{align*}
Likewise, we can also obtain the reverse inequality. Then it holds that
\begin{align*}
    & \sup_{x\in \mathbb{R}}\big|\bbP({T}_n^{\rm G*}\leq x\,|\, \cD_n)  - \bbP({T}_n\leq x)\big| \\
    \leq&~  \sup_{x\in \mathbb{R}}\big|\bbP(\tilde{T}_n\leq x  )-\bbP(\tilde{T}_n^{\rm dis,G}\leq x   )\big|+ \bbE\bigg\{\sup_{x\in \mathbb{R}}\big|\bbP(\tilde{T}_n^{\rm G*}\leq x  \,|\, \mathcal{F}_n) - \bbP(\tilde{T}_n^{\rm dis,G}\leq x  )\big| \,\bigg|\, \cD_n\bigg\}\\
    &~ +\bbE\big\{\bbP(|{T}_n^{\rm G*}-\tilde{T}_n^{\rm G*}| > s_1\,|\, \mathcal{F}_n) \,|\, \cD_n\big\}+\bbP(|{T}_n-\tilde{T}_n|> s_2)+\sup_{x\in \mathbb{R}} \bbP(   |\tilde{T}_n^{\rm dis,G}- x|\leq  s_1 +s_2 ) \,.
\end{align*}
In what follows, we bound the five terms on the right-hand side of the above inequality separately. By Lemma \ref{lem.dis.Gaussian} and Proposition \ref{pro.Tn.H0}(i), it holds that
\begin{align}\label{bound.1}
    \sup_{x\in \mathbb{R}}\big|\bbP(\tilde{T}_n\leq x  )-\bbP(\tilde{T}_n^{\rm dis,G}\leq x   )\big| = o(1)\,,
\end{align}
provided that $\log p \ll n^{2/21}$. 
Since $\tilde T_n^{{\rm G}*}$ depends on $\mathcal F_n$ only through the collection of oracle processes $\widetilde{\mathcal D}_n=\{\bvep_t(\cdot)\}_{t\in [n]}$, we have $\mathbb P(\tilde T_n^{{\rm G}*}\le x\,|\, \mathcal F_n)
=
\mathbb P(\tilde T_n^{{\rm G}*}\le x \,|\, \widetilde{\mathcal D}_n)$. Accordingly, by Lemma \ref{lem.dis.Gaussian} and Proposition \ref{pro.Tn.H0}(ii), we have  $A_n :=  \sup_{x\in \mathbb{R}}\big|\bbP(\tilde{T}_n^{\rm G*}\leq x  \,|\, \mathcal{F}_n) - \bbP(\tilde{T}_n^{\rm dis,G}\leq x  )\big|  = o_{\rm p}(1)$, provided that $\log p \ll n^{\iota}$ for some $\iota>0$ depending only on $(\rho,\vartheta)$. Due to $0\leq A_n \leq 1$,  we have $\mathbb{E}(A_n)=o(1)$, which implies that
\begin{align}\label{bound.2}
    \bbE\bigg\{\sup_{x\in \mathbb{R}}\big|\bbP(\tilde{T}_n^{\rm G*}\leq x  \,|\, \mathcal{F}_n) - \bbP(\tilde{T}_n^{\rm dis,G}\leq x  )\big| \,\bigg|\, \cD_n\bigg\} = o_{\rm p}(1)\,.
\end{align}
 Let $s_1=Cn^{-\gamma_3/2}(\log p)^{\gamma_2}$ and  $s_2=Cn^{-\gamma_1/2+1/4}(\log p)^{\gamma_2}$  for some sufficiently large constant $C>0$. Then, using \eqref{eq:diff.tTn.Tn} and \eqref{eq:diff.tTnG*.TnG*}, and following the argument for \eqref{bound.2},  Condition~\ref{c.rate} implies that
\begin{align}\label{bound.3}
    \bbE\big\{\bbP(|{T}_n^{\rm G*}-\tilde{T}_n^{\rm G*}| > s_1\,|\, \mathcal{F}_n) \,|\, \cD_n\big\}+\bbP(|{T}_n-\tilde{T}_n|> s_2) = o_{\rm p}(1)+o(1)=o_{\rm p}(1)\,.
\end{align}
Recall that $\tilde{T}_{n}^{{\rm dis, G }}=\max_{r\in [Lp^2]}\max_{(u,v)\in\cU_M^2}|\tilde{\cG}_r(u,v)  |$, 
where the centered process $\tilde{\cG}_r(u,v)$ has variance
$\var\{\tilde\xi_{n,r}(u,v)\}$. Here, $\tilde\bxi_n(u,v)=\{\tilde\xi_{n,1}(u,v),\ldots,\tilde\xi_{n,Lp^2}(u,v)\}^{\T}$ is defined in the same way as $\bxi_n(u,v)$, except that $\{\hat\bvep_t(\cdot)\}$ is replaced by $\{\bvep_t(\cdot)\}$. Since  $M \asymp (np)^{C_M}$ and $p\geq n^{\upsilon}$, Condition~\ref{c.bvar} and the Nazarov's inequality 
\citep[see, e.g., Lemma~A.1 of][]{Cher2017_supp} yield
\begin{align}\label{bound.4}
    \sup_{x\in \mathbb{R}} \bbP(   |\tilde{T}_n^{\rm dis,G}- x|\leq  s_1 +s_2 ) \lesssim (s_1+s_2)\sqrt{\log(Lp^2M^2)} = o(1)\,,
\end{align}
provided that $\log p \ll  n^{c_1}$ for some $c_1>0$ depending only on $(\gamma_1,\gamma_2,\gamma_3)$.  Combining \eqref{bound.1}--\eqref{bound.4} yields, under $H_0$,
\begin{align}\label{bound.final}
\sup_{x\in \mathbb{R}}|\bbP({T}_n^{\rm G*}\leq x\,|\,\cD_n) -\bbP({T}_n\leq x)| =o_{\rm p}(1)\,,	
\end{align}
provided that $\log p \ll n^{\iota_1}$ for some $\iota_1>0$ depending only on $(\rho,\vartheta,\gamma_1,\gamma_2,\gamma_3)$.

For the choice of $s_2$ specified above, it follows from \eqref{bound.1}, \eqref{bound.3}  and \eqref{bound.4} that 
\begin{align*}
 &~   
\sup_{x\in\mathbb R}
 \big|
\mathbb P(T_n\leq x)-\mathbb P(\tilde T_n^{\rm dis,G}\leq x)
 \big| \\
 \leq &~   \sup_{x\in \mathbb{R}}\big|\bbP(\tilde{T}_n\leq x  )-\bbP(\tilde{T}_n^{\rm dis,G}\leq x   )\big| + \sup_{x\in\mathbb R} \bbP(   |\tilde{T}_n^{\rm dis,G}- x|\leq  s_2 ) +\bbP(|{T}_n-\tilde{T}_n|> s_2)=o(1)\,.
\end{align*}
This, together with \eqref{bound.final}, yields 
\begin{align}\label{eq:Tnstar.Gn.approx}
R_n^*
:=
\sup_{x\in\mathbb R}
\big|
\mathbb P(T_n^{\rm G*}\leq x\,|\,\mathcal D_n)
-
\mathbb P(\tilde{T}_n^{\rm dis,G}\leq x)
\big|
=o_{\rm p}(1)\,.
\end{align}
For a fixed  $\epsilon \in (0, \min\{\alpha,1-\alpha\}/2)$, there exist
two deterministic numbers $q_{n,L}(\epsilon)$ and
$q_{n,U}(\epsilon)$, with $q_{n,L}(\epsilon)<q_{n,U}(\epsilon)$, satisfying 
\begin{align}\label{eq:qLU}
\mathbb P\{\tilde{T}_n^{\rm dis,G} > q_{n,L}(\epsilon)\}
 =\alpha+2\epsilon <1 ~~\mbox{and}~~
\mathbb P\{\tilde{T}_n^{\rm dis,G} > q_{n,U}(\epsilon)\}
 =\alpha-2\epsilon>0\,.
\end{align}
Define an event $\mathcal E_{n,\epsilon}
=
\{R_n^*\leq\epsilon\}$.   \eqref{eq:Tnstar.Gn.approx} yields that $\mathbb P(\mathcal E_{n,\epsilon})\rightarrow 1$. On $\mathcal E_{n,\epsilon}$, by \eqref{eq:qLU},
\begin{align*}
&\mathbb P\{
T_n^{\rm G*}>q_{n,L}(\epsilon)
\,|\,\mathcal D_n
\}
\geq
\mathbb P\{\tilde{T}_n^{\rm dis,G}>q_{n,L}(\epsilon)\}
-\epsilon
=
\alpha+\epsilon
>
\alpha\,,\\
\mbox{and}~~&\mathbb P\{
T_n^{\rm G*}>q_{n,U}(\epsilon)
\mid\mathcal D_n
\}
\leq
\mathbb P\{\tilde{T}_n^{\rm dis,G}>q_{n,U}(\epsilon)\}
+\epsilon
=
\alpha-\epsilon
<
\alpha\,.
\end{align*}
Recall $\hat{\rm cv}_{\alpha} = \inf\{x>0:\mathbb{P}({T}_{n}^{\rm G*}> x\,|\,\cD_n)\leq \alpha\} $. We know that $q_{n,L}(\epsilon)
\leq
\hat{\rm cv}_{\alpha}
\leq
q_{n,U}(\epsilon)$ on $\mathcal E_{n,\epsilon}$. Therefore, it holds that
\begin{align*}
\mathbb P(T_n>\hat{\rm cv}_{\alpha})
 \leq&~
\mathbb P\{
T_n>\hat{\rm cv}_{\alpha},
\mathcal E_{n,\epsilon}
\}
+
\mathbb P(\mathcal E_{n,\epsilon}^{\rm c}) 
 \leq  
\mathbb P\{T_n>q_{n,L}(\epsilon)\}
+
o(1)\\
\leq &~
\mathbb P\{\tilde{T}_n^{\rm dis,G}>q_{n,L}(\epsilon)\} 
+
o(1) =
\alpha+2\epsilon+o(1)\,.
\end{align*} 
It  then follows that $\limsup_{n\to\infty}
\mathbb P(T_n>\hat{\rm cv}_{\alpha})
\leq
\alpha+2\epsilon $. On the other hand,
\begin{align*}
\mathbb P(T_n>\hat{\rm cv}_{\alpha})
 \geq &~
\mathbb P\{
T_n>q_{n,U}(\epsilon),
\mathcal E_{n,\epsilon}
\} \geq
\mathbb P\{T_n>q_{n,U}(\epsilon)\}
-
\mathbb P(\mathcal E_{n,\epsilon}^{\rm c})\\
 \geq &~
\mathbb P\{\tilde{T}_n^{\rm dis,G}>q_{n,U}(\epsilon)\}
-
o(1) =
\alpha-2\epsilon-o(1)\,.
\end{align*}
Thus $\liminf_{n\to\infty}
\mathbb P(T_n>\hat{\rm cv}_{\alpha})
\geq
\alpha-2\epsilon$.  Since $\epsilon>0$ can be chosen arbitrarily small, we have  $\mathbb P(T_n>\hat{\rm cv}_{\alpha})
\rightarrow
\alpha$ by letting $\epsilon\rightarrow 0^{+}$.  This completes the proof of Theorem \ref{thm.Tn.H0}.
\hfill $\Box$

\subsection{Proof  of Theorem \ref{thm.Tn.H1}}

By \eqref{new.def.tndisG} in Section \ref{proof.lemma_tn_dis_GA}, we know that $\tilde{T}^{\rm dis, G}_{n} = |\tilde{\boldsymbol{\cG}}^{\rm dis}|_{\max}$, where $\tilde{\boldsymbol{\cG}}^{\rm dis} \sim \cN(\0, \widetilde\bXi^{\rm dis}_{n})$. Here,     $\widetilde\bXi^{\rm dis}_{n}=(\widetilde{\Xi}^{\rm dis}_{n,k_1k_2})_{L\tilde p^2\times L\tilde p^2}$ is  specified in  Section \ref{proof.lemma_tn_dis_GA} and $\tilde{p}=pM$. Analogously to $\tilde{T}^{\rm dis, G}_{n}$,   define the discretized version of $\tilde{T}_{n}^{{\rm G*}}$ by $\tilde{T}_{n}^{{\rm dis ,G*}} = |\tilde{\boldsymbol{\cG}}^{\rm dis,*}|_{\max}$, where  $\tilde{\boldsymbol{\cG}}^{\rm dis,*} \sim \cN(\0,\widetilde\bXi^{\rm dis,*}_{n})$  conditional on $\widetilde{\mathcal D}_n=\{\bvep_t(\cdot)\}_{t\in [n]}$. Here, $\widetilde\bXi^{\rm dis,*}_{n}=(\widetilde{\Xi}^{\rm dis,*}_{n,k_1k_2})_{L\tilde p^2\times L\tilde p^2}$ is specified in  Section \ref{proof.disGAstar}. Then, by the constructions in Section~\ref{sec.preliminary}, we define the following four events:
\begin{align*} 
\Phi_1(s_1) =  \{|T_n^{\rm G*}-\tilde{T}_n^{\rm G*}|\leq s_1\}~~~~&\mbox{and}~~~~ \Phi_2(s_2) = \{|T_n-\tilde{T}_n|\leq s_2\}\,,\\
\Phi_3(s_3)=\{|\tilde{T}^{\rm G*}_{n}-\tilde{T}^{\rm dis,G*}_{n}|\leq s_3\}~~~~&\mbox{and}~~~~
\Phi_4(s_4)=\bigg\{ \max_{ k \in [L\tilde{p}^2]}\bigg|\frac{\widetilde\Xi^{\rm dis,*}_{n,kk}}{\widetilde{\Xi}_{n,kk}^{\rm dis}}-1\bigg|\leq s_4 \bigg\}\,,
\end{align*}
where $s_1=Cn^{-\gamma_3/2}(\log p)^{\gamma_2}$,  $s_2=Cn^{-\gamma_1/2+1/4}(\log p)^{\gamma_2}$, $s_3=Cn^{-1}(\log  p)^{1/2}$ and $s_4= Cn^{-{c}_1/2}(\log p)^{{c}_2}$ for some sufficiently large constant  $C >0$. Here, \(c_1\) and \(c_2\) depend only on \((\rho,\vartheta)\)
and are specified in \eqref{disSigma} in Section \ref{proof.disGAstar}. Recall   $\mathcal{D}_n=\{\hat\bvep_t(\cdot)\}_{t\in [n]}$. Let $\mathcal{F}_n$ denote the $\sigma$-field generated by $\{\bvep_t(\cdot),\bdelta_t(\cdot)\}_{t\in [n]}$. By \eqref{err.m}, it follows  that $\sigma(\mathcal{D}_n)\subseteq \mathcal{F}_n$. Thus,  
\begin{align}\label{condition.tail}
    \mathbb{P}\big({T}^{\rm G*}_{n}>x \, |\,\cD_n\big) = \bbE\big \{\mathbb{P}\big ({T}^{\rm G*}_{n}>x  \,|\,\mathcal{F}_n\big)\, |\,\cD_n\big\} 
\end{align}
for any $x\geq 0$.  By Conditions \ref{c.subgaussian}--\ref{c.rate} and Lemma \ref{lem:diff.Tn.Gstar}, we have $ \mathbb{P}\{\Phi^{\rm c}_1(s_1) \,|\,  \mathcal{F}_n\}  =o_{\rm p}(1)$ and $ \mathbb{P}\{\Phi^{\rm c}_3(s_3) \,|\,  \mathcal{F}_n \}=o_{\rm p}(1) $.  In addition, by \eqref{disSigma}  in Section \ref{proof.disGAstar}, we have 
\begin{align*} 
\max_{ k \in [L{\tilde{p}}^{2}]} |\widetilde\Xi^{{\rm dis,*}}_{n,kk}- \widetilde\Xi_{n,kk}^{\rm dis}| 
\leq \|\widetilde\bXi^{\rm dis, *}_{n}-\widetilde\bXi^{\rm dis}_{n}\|_{\max}= o_{\rm p}(s_4)   \,,
\end{align*}
provided that $0<\rho<(\vartheta-1)/(3\vartheta-2)$. By Condition \ref{c.bvar}, we know $\min_{ k \in [L\tilde{p}^2]} \widetilde\Xi_{n,kk}^{\rm dis}$ is uniformly bounded away from zero. Therefore, 
\begin{align*}
    \max_{ k \in [L\tilde{p}^2]}\bigg|\frac{ \widetilde\Xi^{{\rm dis,*}}_{n,kk}}{ \widetilde\Xi_{n,kk}^{\rm dis}}-1\bigg| \leq \frac{\max_{ k \in [L{\tilde{p}}^{2}]}| \widetilde\Xi^{{\rm dis, *}}_{n,kk}- \widetilde\Xi_{n,kk}^{\rm dis}|}{\min_{ k \in [L\tilde{p}^2]} \widetilde\Xi_{n,kk}^{\rm dis}}= o_{\rm p}(s_4) \,.
\end{align*}
Since $\Phi_4(s_4)$ is $\mathcal F_n$-measurable, this implies
$\mathbb P\{\Phi_4^{\rm c}(s_4)\,|\,\mathcal F_n\}
=
I\{\Phi_4^{\rm c}(s_4)\}
=o_{\rm p}(1)$. Hence, by the union bound and triangle inequality,
\begin{align}\label{bound.H1.1}
\mathbb P\big(T_n^{\rm G*}>x\,|\,\mathcal F_n\big)
&\leq
\mathbb P\bigg\{
T_n^{\rm G*}>x,\,
\bigcap_{j\in\{1,3,4\}}\Phi_j(s_j)
\,\bigg|\,\mathcal F_n
\bigg\}
+
\mathbb P\bigg\{
\bigcup_{j\in\{1,3,4\}}\Phi_j^{\rm c}(s_j)
\,\bigg|\,\mathcal F_n
\bigg\}\nonumber \\
&=
\mathbb P\bigg\{
T_n^{\rm G*}>x,\,
\bigcap_{j\in\{1,3,4\}}\Phi_j(s_j)
\,\bigg|\,\mathcal F_n
\bigg\}
+o_{\rm p}(1) \nonumber\\
&\leq \mathbb{P}\bigg\{\tilde{T}^{\rm dis,G*}_{n}>x-s_1-s_3\,,\bigcap_{j\in\{1,3,4\}}\Phi_j(s_j)\,\bigg|\,\mathcal F_n\bigg\}+o_{\rm p}(1) 
\end{align}
for any $x >  s_1+s_3$. Recall $\tilde{p}=pM$, $M\asymp(np)^{C_M}$ and $p\geq n^{\upsilon}$. Then $\log (L\tilde{p}^2)\lesssim \log p$.  
Let $\tilde{\varrho}= \max_{k \in [{L\tilde{p}^2}]}\widetilde\Xi_{n,kk}^{\rm dis}$. 
On the event \(\Phi_4(s_4)\), by the standard maximal inequality for Gaussian random vectors, we have 
\begin{align}\label{upper.Ex}
\mathbb{E}(\tilde{T}_{n}^{\rm dis,G*}\,|\,\mathcal F_n)\leq&\,  C_0 (\log  {p} )^{1/2}\max_{k \in [{L\tilde{p}^2}]}(\widetilde\Xi_{n,kk}^{{\rm dis, *}})^{1/2} \leq  C_0 \tilde{\varrho}^{1/2}(1+s_4)^{1/2}(\log  {p} )^{1/2}\,,
\end{align}
where $C_0 > 0$ is a constant. By Borell inequality for Gaussian random vector,
\begin{align*}
\mathbb{P}\{\tilde{T}^{\rm dis,G*}_{n}>\mathbb{E}(\tilde{T}^{\rm dis,G*}_{n}\,|\,\mathcal F_n)+x\,|\,\mathcal F_n\}
\leq 2\exp\bigg( -\frac{x^2}{2\max_{k \in [{L\tilde{p}^2}]} \widetilde\Xi^{{\rm dis,*}}_{n,kk}} \bigg)
\end{align*}
for any $x\geq 0$. Let 
\begin{align*}
    x_0 = s_1+s_3 + \tilde{\varrho}^{1/2}(1+s_4)^{1/2}[C_0(\log  {p} )^{1/2}+\sqrt{2}\{ \log (4/\alpha)\}^{1/2} ]\,.
\end{align*} 
Since the event
$\cap_{j\in\{1,3,4\}}\Phi_j(s_j)$ implies $\Phi_4(s_4)$, by \eqref{upper.Ex}, it holds that
\begin{align*}
    \mathbb{P}\bigg\{\tilde{T}^{\rm dis,G*}_{n}>x_0-s_1-s_3\,,\bigcap_{j\in\{1,3,4\}}\Phi_j(s_j)\,\bigg|\,\mathcal F_n\bigg\} \leq  2\exp\bigg\{ -\frac{2\tilde{\varrho}(1+s_4)\log(4/\alpha)}{2\tilde{\varrho}(1+s_4)} \bigg\}=\frac{\alpha}{2}\,.
\end{align*}
Following from \eqref{bound.H1.1}, it holds that $\mathbb P(T_n^{\rm G*}>x_0\,|\,\mathcal F_n)
\leq {\alpha}/{2}+o_{\rm p}(1)$.  Taking conditional expectation with respect to \(\mathcal D_n\) and using
\eqref{condition.tail}, we obtain $\mathbb P(T_n^{\rm G*}>x_0\,|\,\mathcal D_n)
\leq {\alpha}/{2}+o_{\rm p}(1)$.
Since \(\alpha\in(0,1)\) is fixed, this implies
\begin{align*}
    \mathbb P(T_n^{\rm G*}>x_0\,|\,\mathcal D_n)\leq \alpha
\end{align*}
with probability approaching one. Let $\lambda_1(\tilde{C},p)=\tilde{C}(\log  {p} )^{1/2}  $ for   some sufficiently large constant $\tilde{C}>0$ such that $\lambda_1(\tilde{C},p) \geq 2[C_0(\log  {p} )^{1/2}+ \sqrt{2}\{ \log (4/\alpha)\}^{1/2} ]$, and set $\varrho=\sup_{(u,v)\in\cU^2}\mathrm{\Psi}(u,v)$, where $\mathrm{\Psi}(u,v)$ is the largest diagonal element of $\widetilde\bXi_n(u,v,u,v)$.  Then $\tilde{\varrho} \leq \varrho$. Moreover, since  $s_4=Cn^{-c_1/2}(\log p)^{c_2} $, under
 $\log p\ll n^{\iota} $ for some  $\iota>0 $ depending only on
 $(\rho,\vartheta) $, we have  $(1+s_4)^{1/2}\leq 2 $. Hence, by the definition of  $\hat{\cv}_{\alpha} $, it holds with probability approaching one that
\begin{align}\label{upper.cv}
   \hat{\cv}_{\alpha} 
\leq   s_1+s_3 + {\varrho}^{1/2} \lambda_1(\tilde{C},p )\,,
\end{align}
provided that $\log p \ll n^{\iota}$ for some $\iota>0$ depending only on $(\rho,\vartheta)$.

Recall  $\tilde{T}_{n} = \max_{ \ell \in [L], j,j' \in [p]}\sup_{(u,v)\in\cU^2} n^{1/2} |\widetilde{\Sigma}^{(\ell)}_{jj'}(u,v)|$, where $\widetilde{\Sigma}^{(\ell)}_{jj'}(u,v)=n_{\ell}^{-1}\sum_{t=1}^{n_\ell} \{\vep_{t,j}(u)-\bar{\vep}_j(u)\} \{\vep_{t+\ell,j'}(v)-\bar{\vep}_{j'}(v)\}$. Let $(\ell_0,j_0,{j_0}',u_0,{v_0})=\arg\sup_{ \ell \in [L], j,j' \in [p],(u,v)\in\cU^2} n^{1/2} |\Sigma^{(\ell)}_{jj'}(u,v)| $. Without loss of generality, we assume ${\Sigma}^{(\ell_0)}_{j_0{j_0}'}(u_0,v_0)>0$. Let $\check{C}=2\tilde{C} \geq \tilde{C}(1+\epsilon_n)$ for some $\epsilon_n\rightarrow 0$ satisfying $\epsilon_n(\log {p})^{1/2}\rightarrow\infty$. Then  
$
n^{1/2}{\Sigma}^{(\ell_0)}_{j_0{j_0}'}(u_0,v_0) \geq \varrho^{1/2}(1+\epsilon_n)\lambda_1(\tilde{C},p)
$. Let $\mathcal A_n
=
\Phi_2(s_2)
\cap
\{
\hat{\cv}_{\alpha}
\leq
s_1+s_3+\varrho^{1/2}
\lambda_1(\tilde C,p)
\}$. On the event \(\mathcal A_n\),  it follows that 
\begin{align*}
T_{n}\geq &~   n^{1/2}\Big\{\widetilde{\Sigma}^{(\ell_0)}_{j_0{j_0}'}(u_0,v_0)-{\Sigma}^{(\ell_0)}_{j_0{j_0}'}(u_0,v_0)\Big\}+n^{1/2}{\Sigma}^{(\ell_0)}_{j_0{j_0}'}(u_0,v_0)-s_2\\
\geq &~ n^{1/2}\Big\{\widetilde{\Sigma}^{(\ell_0)}_{j_0{j_0}'}(u_0,v_0)-{\Sigma}^{(\ell_0)}_{j_0{j_0}'}(u_0,v_0)\Big\} + \varrho^{1/2}(1+\epsilon_n)\lambda_1(\tilde{C},p) - s_2\\
\geq &~ n^{1/2}\Big\{\widetilde{\Sigma}^{(\ell_0)}_{j_0{j_0}'}(u_0,v_0)-{\Sigma}^{(\ell_0)}_{j_0{j_0}'}(u_0,v_0)\Big\} + \epsilon_n{\varrho}^{1/2} \lambda_1(\tilde{C},p) + \hat{\cv}_{\alpha}  - (s_1+s_2+s_3)\,.
\end{align*}
Condition~\ref{c.rate} gives $\mathbb P\{\Phi_2^{\rm c}(s_2)\}=o(1)$. Then, by \eqref{upper.cv}, we have $\mathbb P(\mathcal A_n^{\rm c})=o(1)$, provided that $\log p \ll n^{\iota}$ for some $\iota>0$ depending only on $(\rho,\vartheta)$, which implies 
\begin{align*}
&~\mathbb{P}({T}_{n}>\hat{\cv}_{\alpha}) \geq  \mathbb{P}({T}_{n}>\hat{\cv}_{\alpha}\,,\mathcal A_n )  \\ 
\geq&~ \mathbb{P}\bigg[n^{1/2}\Big\{\widetilde{\Sigma}^{(\ell_0)}_{j_0{j_0}'}(u_0,v_0) -{\Sigma}^{(\ell_0)}_{j_0{j_0}'}(u_0,v_0)\Big\}
 > -\epsilon_n{\varrho}^{1/2} \lambda_1(\tilde{C},p)+s_1+s_2+s_3 \bigg] -\mathbb P(\mathcal A_n^{\rm c})\\
\geq&\, 1-\mathbb{P}\bigg[n^{1/2}\Big\{\widetilde{\Sigma}^{(\ell_0)}_{j_0{j_0}'}(u_0,v_0) -{\Sigma}^{(\ell_0)}_{j_0{j_0}'}(u_0,v_0)\Big\}
 \leq -\epsilon_n{\varrho}^{1/2} \lambda_1(\tilde{C},p)+s_1+s_2+s_3\bigg] -o(1)\, .
\end{align*}
Similar to \eqref{eq:sumvep} and \eqref{eq:sumvep2}  in Section \ref{proof.lemma_tn_dis_GA},   Conditions \ref{c.alpha} and \ref{c.subgaussian} yield $n^{1/2}|\widetilde{\Sigma}^{(\ell_0)}_{j_0{j_0}'}(u_0,v_0) -{\Sigma}^{(\ell_0)}_{j_0{j_0}'}(u_0,v_0)| = O_{\rm p}(1)$.  Moreover, under
 $\log p\ll n^{\iota'} $ for some  $\iota'>0 $ depending only on
 $(\gamma_1,\gamma_2,\gamma_3) $, we have  $s_1+s_2+s_3= o(1) $. Since $\epsilon_n\lambda_1(\tilde{C},p)\rightarrow  \infty$, it holds that $-\epsilon_n{\varrho}^{1/2} \lambda_1(\tilde{C},p)+s_1+s_2+s_3 \rightarrow -\infty$. Then, we have $\mathbb{P}({T}_{n}>\hat{\cv}_{\alpha}) \rightarrow 1$, provided that $\log p\ll n^{\iota_1} $ for some  $\iota_1>0 $ depending only on
 $(\rho,\vartheta,\gamma_1,\gamma_2,\gamma_3) $.  We  complete the proof of Theorem \ref{thm.Tn.H1}.
\hfill $\Box$

\section{Proof of Proposition \ref{thm.partial}}\label{sec.proof.pro1}

Similar to Section~\ref{sec.preliminary}, for some large integer $M_1>0$, let $\cU=[0,1]$ and partition it into
$M_1$ subintervals $\{B_1,\ldots,B_{M_1}\}$ of equal length $M_1^{-1}$. Let $b_m$ denote the midpoint of $B_m$.
The proof of Proposition~\ref{thm.partial} relies on following two lemmas, whose proofs are given in Sections~\ref{proof.lemma1.prop1} and~\ref{proof.lemma2.prop1}, respectively.

\begin{lemma}\label{lemma1.prop1}
Under  Condition~{\rm \ref{c.subgaussian}},
it holds that  
\begin{align*} 
&~~~~~~~ \max_{t\in[n]}\max_{j\in[p]}\sup_{u\in\cU}|\vep_{t,j}(u)|  = O_{\rm p}[\{\log(np)\}^{1/2} ]\,,\\
 &\max_{t\in[n]}\max_{j\in[p]}\max_{m\in[M_1]} \sup_{u \in{B}_m}|\vep_{t,j}(u)-\vep_{t,j}(b_m)|  = O_{\rm p}\bigg\{\frac{ (np)^2}{  {M_1}^{ \kappa/2} }   \bigg\}\,.
\end{align*} 
\end{lemma}

\begin{lemma}\label{lemma2.prop1}
Under  Conditions~{\rm \ref{c.subgaussian} and \ref{cond.subgauss}--\ref{cond.kern}},
it holds that 
\begin{align*} 
 \max_{t\in [n]}\max_{j\in[p]}\sup_{u\in\cU}|\delta_{t,j}(u)|   = O_{\rm p}\bigg[ \{\log(np /h )\}^{1/2} \bigg\{ \frac{1}{(Nh)^{1/2}}+ h^\kappa\bigg\}  \bigg]\,, 
\end{align*} 
provided that $\log(np/h)\ll N h$. Moreover, 
\begin{align*}
     \max_{t\in [n]}\max_{j\in[p]}\max_{m\in[{M_1}]}\sup_{u\in B_m}|\delta_{t,j}(u)-\delta_{t,j}(b_m)| = O_{\rm p}\bigg[ \frac{\{\log (np M_1)\}^{1/2}}{h^{2-\kappa} M_1}   + \frac{ (np)^2}{h M_1^{ \kappa/2} }  + \frac{     \{ \log(np)\}^{1/2}   }{h^2 M_1}\bigg]\,,
\end{align*}
provided that $\log(np M_1/h)\ll N h$.
\end{lemma}

For notational simplicity, we define 
\begin{align*} 
   & a_{1n}=\{\log(np)\}^{1/2}\,, ~~ a_{2n}=\frac{ (np)^2}{  M_1^{ \kappa/2} }\,,~~b_{1n} = \{\log(np /h)\}^{1/2}  \bigg\{ 
 \frac{1}{(Nh)^{1/2}}  + h^{\kappa} \bigg\}\,,\nonumber\\
    &~~~~~~~~~~~~\mbox{and}~~~~~~ b_{2n} =     \frac{\{\log (np M_1)\}^{1/2}}{h^{2-\kappa} M_1}   + \frac{ (np)^2}{h M_1^{ \kappa/2} }  + \frac{     \{ \log(np)\}^{1/2}   }{h^2 M_1}\,.
\end{align*}
Since $p\geq n^{\upsilon}$, combining \eqref{eq:def.sigmadv} with Lemmas~\ref{lemma1.prop1} and~\ref{lemma2.prop1}, we have
\begin{align*}
    &\max_{\ell\in [L]} 
\{\|\widetilde{\bSigma}_{\varepsilon\delta}^{(\ell)}\|_{\infty,\max} + \|\widetilde{\bSigma}_{\delta\vep}^{(\ell)}\|_{\infty,\max}  \}\\
 \lesssim  &~ \bigg\{\max_{t\in[n]}\max_{j\in[p]}\sup_{u\in\cU}|\vep_{t,j}(u)|\bigg\}\cdot \bigg\{\max_{t\in [n]}\max_{j\in[p]}\sup_{u\in\cU}|\delta_{t,j}(u)|\bigg\}= O_{\rm p}(a_{1n}b_{1n})\,,
\end{align*}
provided that $\log (p/h)\ll N h$.  In addition,
\begin{align*}
    \max_{\ell\in [L]} 
  \|\widetilde\bSigma_{\delta\delta}^{(\ell)} \|_{\infty,\max} \lesssim  \bigg\{\max_{t\in [n]}\max_{j\in[p]}\sup_{u\in\cU}|\delta_{t,j}(u)|\bigg\}^2 = O_{\rm p}( b_{1n}^2)\,, 
\end{align*}
provided that $\log (p/h)\ll N h$. By Condition \ref{cond.Tij}, it holds that $Nh\rightarrow \infty$ and $h\rightarrow 0$. Then,
\begin{align}\label{Delta.Sigma}
    \Delta_{\widetilde{\bSigma}}^{\varepsilon\delta} =&~ \max_{\ell\in [L]} 
\{\|\widetilde{\bSigma}_{\varepsilon\delta}^{(\ell)}\|_{\infty,\max} + \|\widetilde{\bSigma}_{\delta\vep}^{(\ell)}\|_{\infty,\max} +  \|\widetilde\bSigma_{\delta\delta}^{(\ell)} \|_{\infty,\max} \}\nonumber\\
 =&~O_{\rm p}(a_{1n}b_{1n}+b_{1n}^2) = O_{\rm p}\bigg[  \log (N \vee p)   \bigg\{ 
 \frac{1}{(Nh)^{1/2}}  + h^{\kappa} \bigg\} \bigg]\,,
\end{align}
provided that $\log (N \vee p)\ll N h$. This completes the proof of the first part of Proposition~\ref{thm.partial}.

We next prove the second part of Proposition~\ref{thm.partial}. 
Recall   $\Delta_{\tilde{\boldsymbol{\cG}}}^{\varepsilon\delta} =  \|\tilde{\boldsymbol{\cG}}_{\varepsilon\delta}^{*}\|_{\infty,\max}+\|\tilde{\boldsymbol{\cG}}_{\delta\vep}^{*}\|_{\infty,\max}+\|\tilde{\boldsymbol{\cG}}_{\delta\delta}^{*}\|_{\infty,\max}$. We first establish the convergence rate of
$\|\tilde{\boldsymbol{\cG}}_{\varepsilon\delta}^{*}\|_{\infty,\max}$.
The convergence rates of
$\|\tilde{\boldsymbol{\cG}}_{\delta\vep}^{*}\|_{\infty,\max}$
and
$\|\tilde{\boldsymbol{\cG}}_{\delta\delta}^{*}\|_{\infty,\max}$
can be obtained in the same way. By \eqref{eq:def.cgdv}, we have
\begin{align}\label{div11}
    \|\tilde{\boldsymbol{\cG}}_{\varepsilon\delta}^{*}\|_{\infty,\max}=&~ \max_{r\in[Lp^2]}\sup_{(u,v)\in\cU^2}\bigg|\frac{1}{\sqrt{n_L} }\sum_{t=1}^{n_L} \varrho_t \{ \tilde{\eta}^{\vep\delta}_{t,r}(u,v)  -\bar{\tilde{\eta}}_r^{\vep\delta}(u,v)\} \bigg|\nonumber\\
    \leq &~ \underbrace{\max_{r\in[Lp^2]}\max_{m,m'\in [M_1] }\bigg|\frac{1}{\sqrt{n_L} }\sum_{t=1}^{n_L} \varrho_t \{ \tilde{\eta}^{\vep\delta}_{t,r}(b_m,b_{m'})  -\bar{\tilde{\eta}}_r^{\vep\delta}(b_m,b_{m'})\} \bigg|}_{{\rm R}_1}\nonumber\\
    &~+ \underbrace{\max_{r\in[Lp^2]}\max_{m,m'\in [M_1] }\sup_{(u,v)\in{B}_m\times {B}_{m'}}\bigg|\frac{1}{\sqrt{n_L} }\sum_{t=1}^{n_L} \varrho_t \{ \tilde{\eta}^{\vep\delta}_{t,r}(u,v)  -\tilde{\eta}^{\vep\delta}_{t,r}(b_m,b_{m'})\} \bigg|}_{{\rm R}_2}\nonumber\\
    &~+ \underbrace{\max_{r\in[Lp^2]}\max_{m,m'\in [M_1] }\sup_{(u,v)\in{B}_m\times {B}_{m'}}\bigg|\frac{1}{\sqrt{n_L} }\sum_{t=1}^{n_L} \varrho_t \{\bar{\tilde{\eta}}_r^{\vep\delta}(u,v)  -\bar{\tilde{\eta}}_r^{\vep\delta}(b_m,b_{m'})\} \bigg|}_{{\rm R}_3}\,,
\end{align}
where $\tilde{\eta}^{\vep\delta}_{t,r}(u,v) = \{\vep_{t,j}(u)-\bar{\vep}_j(u)\}\{{\delta}_{t+\ell,j'}(v)-\bar{{\delta}}_{j'}(v)\}$ and $\bar{\tilde{\eta}}_r^{\vep\delta}(u,v) = n_L^{-1}\sum_{t=1}^{n_L}\tilde{\eta}^{\vep\delta}_{t,r}(u,v)$. Recall that $\boldsymbol\varrho =( \varrho_1 ,\dots,\varrho_{{n_L}})^{\T} \sim \cN({\bf 0},\bXi)$, where $\bXi=(\Xi_{ij})_{n_L\times n_L}$ 
with  $\Xi_{ij}= \mathcal{W}\{(i-j)/{b_{n}}\}$.  Let $\mathcal{F}_n$ denote the $\sigma$-field generated by $\{\bvep_t(\cdot),\bdelta_t(\cdot)\}_{t\in [n]}$. Then, conditional on $\mathcal{F}_n$, $n_L^{-1/2}\sum_{t=1}^{n_L}\varrho_t\{\tilde{\eta}_{t,r}^{\vep\delta}(b_m, b_{m'})-\bar{\tilde{\eta}}_r^{\vep\delta}(b_m,b_{m'})\}\sim \mathcal{N}\{0,\Xi^{\vep\delta}_{r}(b_m,b_{m'},b_{m},b_{m'})\}$, where
 \begin{align*}
     \Xi^{\vep\delta}_{r}(u_1,v_1,u_2,{v}_2)=\sum_{i=-n_L+1}^{{n_L}-1} \mathcal{W}\bigg(\frac{i}{b_{n}}\bigg) {H}^{\vep\delta}_{i,r}(u_1,v_1,u_2,v_2) \,,
 \end{align*} 
 where ${H}^{\vep\delta}_{i,r}(u_1,v_1,u_2,v_2) = {n_L}^{-1}\sum_{t=i+1}^{n_L}
\{\tilde{\eta}^{\vep\delta}_{t,r}(u_1,v_1)-\bar{\tilde{\eta}}_r^{\vep\delta}(u_1,v_1)\} \{\tilde{\eta}^{\vep\delta}_{t-i,r}({u}_2,{v}_2)-\bar{\tilde{\eta}}_r^{\vep\delta}( u_2, v_2)\}$ when $i\geq 0$, and ${H}^{\vep\delta}_{i,r}(u_1,v_1,u_2,v_2) = {n_L}^{-1}\sum_{t=-i+1}^{n_L}
 \{\tilde{\eta}^{\vep\delta}_{t+i,r}(u_1,v_1)-\bar{\tilde{\eta}}_r^{\vep\delta}(u_1,v_1)\} \{\tilde{\eta}^{\vep\delta}_{t,r}({u}_2,{v}_2)-\bar{\tilde{\eta}}_r^{\vep\delta}( u_2, v_2)\}$ otherwise. 
Let $Q_n \asymp   n^{\rho}a_{1n}^2b_{1n}^2 \log(pM_1)$, and define the event 
\begin{align*}
    \mathcal E_{\Xi}
=
\bigg\{
\max_{r\in[Lp^2]}\max_{m,m'\in [M_1] }
|\Xi^{\vep\delta}_{r}(b_m,b_{m'},b_m,b_{m'})|
\leq Q_n
\bigg\}\,.
\end{align*}
Therefore, by the conditional Gaussian tail bound and the
union bound,
\begin{align*}
   \mathbb P({\rm R}_1\geq x\,, \mathcal E_{\Xi} ) = \bbE\{ \mathbb P({\rm R}_1\geq x, \mathcal E_{\Xi}\,|\,\mathcal F_n)\}
\lesssim
p^2M_1^2\exp\bigg(-\frac{Cx^2}{Q_n}\bigg)
\end{align*}
for any $x\geq 0$, which further implies that
\begin{align}\label{tail.R11}
    \mathbb P({\rm R}_1\geq x)
\lesssim
p^2M_1^2\exp\bigg(-\frac{Cx^2}{Q_n}\bigg)
+
\mathbb P(\mathcal E_{\Xi}^{\rm c})
\end{align}
for any $x\geq 0$. By Condition \ref{c.kernelF}, we have $ \sum_{i=-n_L+1}^{n_L-1}| \mathcal{W}({i}/{b_{n}})|\lesssim b_n \asymp  n^{\rho}$. Thus, 
\begin{align*}
   &\max_{r\in[Lp^2]}\max_{m,m'\in [M_1] }
|\Xi^{\vep\delta}_{r}(b_m,b_{m'},b_m,b_{m'}) |\\
&~~~~\lesssim   n^{\rho} \max_{-n_L+1\leq i\leq n_L-1}\max_{r\in[Lp^2]}\max_{m,m'\in [M_1] }|{H}^{\vep\delta}_{i,r}(b_m,b_{m'},b_m,b_{m'})|\,.
\end{align*}
By Lemmas \ref{lemma1.prop1} and \ref{lemma2.prop1}, we have
\begin{align*}
    \max_{t\in[n]}\max_{r\in[Lp^2]}\sup_{(u,v)\in\cU^2}| \tilde{\eta}^{\vep\delta}_{t,r}(u,v) -\bar{\tilde{\eta}}^{\vep\delta}_{r}(u,v) | = O_{\rm p}(a_{1n}b_{1n})\,,
\end{align*}
provided that $\log(np/h)\ll N h$, which implies that
\begin{align*}
     \max_{-n_L+1\leq i\leq n_L-1}\max_{r\in[Lp^2]}\max_{m,m'\in [M_1] }|{H}^{\vep\delta}_{i,r}(b_m,b_{m'},b_m,b_{m'})|= O_{\rm p}(a_{1n}^2b_{1n}^2)\,.
\end{align*}  
Then, we have $\mathbb P(\mathcal E_{\Xi}^{\rm c})=o(1)$. Together with \eqref{tail.R11}, we can conclude that
\begin{align}\label{R11.OP}
    {\rm R}_1 = O_{\rm p}\{  n^{\rho/2}a_{1n}b_{1n} \log(pM_1)  \} \,, 
\end{align} 
provided that $\log(np/h)\ll N h$. For ${\rm R}_2$ and ${\rm R}_3$, notice that 
\begin{align*}
     &|\tilde{\eta}^{\vep\delta}_{t,r}(u,v)  -\tilde{\eta}^{\vep\delta}_{t,r}(b_m,b_{m'})| \leq    | \{\vep_{t,j}(u)-\bar{\vep}_j(u)\}-\{\vep_{t,j}(b_m)-\bar{\vep}_j(b_m)\} |\cdot |{\delta}_{t+\ell,j'}(v)-\bar{{\delta}}_{j'}(v)|\\
    &~~~~~~~~~~~~~~~~~~~~~~~~ + | \{{\delta}_{t+\ell,j'}(v)-\bar{{\delta}}_{j'}(v)\}- \{{\delta}_{t+\ell,j'}(b_{m'})-\bar{{\delta}}_{j'}(b_{m'})\}| \cdot|\vep_{t,j}(b_m)-\bar{\vep}_j(b_m)|\,.
\end{align*}
Then, by Lemmas \ref{lemma1.prop1} and \ref{lemma2.prop1},  it holds that
\begin{align*}
    &\max_{t\in[n]}\max_{r\in[Lp^2]}\max_{m,m'\in [M_1] }\sup_{(u,v)\in{B}_m\times {B}_{m'}}\big| \tilde{\eta}^{\vep\delta}_{t,r}(u,v)  -\tilde{\eta}^{\vep\delta}_{t,r}(b_m,b_{m'}) \big|\\
    &~~~~\lesssim \max_{t\in[n]}\max_{j\in[p]}\max_{m\in[M_1]} \sup_{u \in{B}_m}|\vep_{t,j}(u)-\vep_{t,j}(b_m)| \cdot \max_{t\in [n]}\max_{j\in[p]}\sup_{u\in\cU}|\delta_{t,j}(u)|  \\
    &~~~~~~~+\max_{t\in [n]}\max_{j\in[p]}\max_{m\in[{M_1}]}\sup_{u\in B_m}|\delta_{t,j}(u)-\delta_{t,j}(b_m)| \cdot\max_{t\in[n]}\max_{j\in[p]}\sup_{u\in\cU}|\vep_{t,j}(u)| \\
    &~~~~ = O_{\rm p}\{ a_{1n}b_{2n}+a_{2n}b_{1n}  \}\,,
\end{align*}
provided that $\log(np M_1/h)\ll N h$. Since $ n_L^{-1/2} \sum_{t=1}^{n_L} |\varrho_t| = O_{\rm p}(n^{1/2})$, we can conclude that
\begin{align*} 
    {\rm R}_2+{\rm R}_3  = O_{\rm p}\{ n^{1/2} (a_{1n}b_{2n}+a_{2n}b_{1n})   \} \,, 
\end{align*} 
provided that $\log(np M_1/h)\ll N h$. This, together with \eqref{div11} and \eqref{R11.OP}, yields  
\begin{align*}
    \|\tilde{\boldsymbol{\cG}}_{\varepsilon\delta}^{*}\|_{\infty,\max}= O_{\rm p}\{  n^{\rho/2}a_{1n}b_{1n} \log(pM_1) + n^{1/2} (a_{1n}b_{2n}+a_{2n}b_{1n})   \}\,,
\end{align*}
provided that $\log(np M_1/h)\ll N h$. By the same argument,
$\|\tilde{\boldsymbol{\cG}}_{\delta\vep}^{*}\|_{\infty,\max}$ has the same
convergence rate. Similarly, by replacing $(a_{1n},a_{2n})$ with $(b_{1n},b_{2n})$, we obtain
\begin{align*}
    \|\tilde{\boldsymbol{\cG}}_{\delta\delta}^{*}\|_{\infty,\max}= O_{\rm p}\{  n^{\rho/2} b_{1n}^2 \log(pM_1) + n^{1/2}  b_{1n}b_{2n}   \}\,,
\end{align*} 
provided that $\log(np M_1/h)\ll N h$. Choose $M_1\asymp h^{-(2+2\kappa)/\kappa}(np)^C$ for some sufficiently large constant
$C>0$. Since $p\geq n^{\upsilon}$ and $h^{-1}\ll N$, it follows that
\begin{align}\label{Delta.g}
     \Delta_{\tilde{\boldsymbol{\cG}}}^{\varepsilon\delta}  =&~ O_{\rm p}\{  n^{\rho/2} (a_{1n}+b_{1n}) b_{1n} \log(pM_1) + n^{1/2}(a_{1n}b_{2n}+a_{2n}b_{1n}  + b_{1n}b_{2n})   \}\nonumber\\
      = &~O_{\rm p}\bigg[ n^{\rho/2}  \bigg\{ 
 \frac{1}{(Nh)^{1/2}}  + h^{\kappa} \bigg\} \{\log(N \vee p)\}^2  \bigg]\,,
\end{align}
provided that $\log( N \vee p)\ll N h$. We complete the proof of Proposition \ref{thm.partial}.
$\hfill\Box$

\section{Proof of Proposition \ref{prop.facmodel}} \label{sec.proof.pro2}

As in the proof of Proposition~\ref{thm.partial}, we partition $\cU=[0,1]$ into
$M_2$ subintervals $\{B_1,\ldots,B_{M_2}\}$ of equal length $M_2^{-1}$.  Let $b_m$ denote the midpoint of $B_m$. Recall that $\bvep_t(\cdot)=\bX_t(\cdot) - \bA\bZ_t(\cdot)$ and  $\hat{\bvep}_{t}(\cdot) = (\bI_{p}-\widehat{\bA}\widehat{\bA}^{\T})\bX_t(\cdot)$,  
where $\widehat{\bA}=(\hat{\bomega}_1,\ldots,\hat{\bomega}_r)$. Here,
$\hat{\bomega}_1,\ldots,\hat{\bomega}_r$ are the leading $r$ eigenvectors of
$\widehat{\bM}$. It holds that
\begin{align}\label{div.delta}
	\bdelta_{t}(u) = \hat{\bvep}_t(u) - \bvep_t(u) =&\, \bA\bZ_t(u) -  \widehat{\bA}\widehat{\bA}^{\T}\bX_t(u) = \bA\bZ_t(u)  - \widehat{\bA}\widehat{\bA}^{\T}\{\bA\bZ_t(u) + \bvep_t(u)\}\nonumber  \\
	=&\, (\bA -  \widehat{\bA}\widehat{\bA}^{\T}\bA) \bZ_t(u) -  \widehat{\bA}\widehat{\bA}^{\T}\bvep_t(u)
\end{align}
for any $u\in\cU$.   By Condition \ref{cond.facmat}(i), there exists  ${\bf{\Lambda}} = {\diag}(\check\lambda_{1},\ldots,\check\lambda_{r})$ and an orthogonal matrix $\bGamma\in \mathbb{R}^{r\times r}$  such that $\bA^{\T}\bA = \bGamma {\bf{\Lambda}} \bGamma^{\T}$. Here, 
$\check\lambda_{1}\geq \ldots \geq \check\lambda_{r}$ are the eigenvalues of $\bA^{\T}\bA$.  Set $\widetilde{\bA}=  \bA \bGamma {\bf{\Lambda}}^{-1/2} \in \mathbb{R}^{p\times r}$ and $\widetilde{\bZ}_t(\cdot)=  {\bf{\Lambda}}^{1/2} \bGamma^{\T} \bZ_t(\cdot)$. Then $\widetilde{\bA}^{\T}\widetilde{\bA}=\bI_r$, $\widetilde{\bA} \widetilde{\bZ}_t(\cdot)= \bA\bZ_t(\cdot)$, and  $\mathcal{C}(\widetilde{\bA})=\mathcal{C}({\bA})$, where $\mathcal C(\bB)$ denotes the linear space spanned by the columns of $\bB$.
% and $\widetilde{\bZ}_t(\cdot)= p^{-1/2} \bf{\Lambda}^{1/2} \bGamma^{\T} \bZ_t(\cdot)$ such that $\widetilde{\bA} \widetilde{\bZ}_t(\cdot)= \bA\bZ_t(\cdot)$.
The proof of Proposition~\ref{prop.facmodel} relies on following two lemmas, whose proofs are given in Sections~\ref{proof.lemma1.prop2} and~\ref{proof.lemma2.prop2}, respectively.

\begin{lemma}\label{lemma1.prop2}
  Under Conditions  {\rm \ref{cond.facerror}--\ref{cond.facmat}} and model \eqref{eq:ffm}, 
  there exists  an $r\times r$ orthogonal matrix ${\bU} $ such that  {\rm (i)} $\| \widehat{\bA}\bU -  \widetilde{\bA} \|_{\rm op} = O_{\rm p}( n^{-1/2})$; {\rm (ii)} $\| \widehat{\bA}\bU -  \widetilde{\bA} \|_{\max} = O_{\rm p} [ (np)^{-1/2}\{\log (np)\}^{1/2}  ]$, 
provided that 
$\log (np) \ll n^{1/5}$.
\end{lemma}

\begin{lemma}\label{lemma2.prop2}
Under  Conditions  {\rm \ref{cond.facerror}--\ref{cond.facmat}} and model \eqref{eq:ffm},
it holds that 
\begin{align*}
   &~~~~~~~  \max_{t\in[n]}\max_{j\in[p]}\sup_{u\in\cU}|\delta_{t,j}(u)|   
	=  O_{\rm p}[n^{-1/2} \log(np) + p^{ -1/2}\{\log (np)\}^{1/2}] \,,\\
   & \max_{t\in[n]}\max_{j\in[p]}\max_{m\in[M_2]}\sup_{u\in B_m}|\delta_{t,j}(u)-\delta_{t,j}(b_m)| 
 = 	O_{\rm p}\bigg[\frac{ n^{3/2}p^{ 2}\{\log (np)\}^{1/2}+n^2p^{ 3/2}}{  {M_2}^{ \kappa/2} } \bigg] \,,
\end{align*}
provided that 
$\log (np) \ll n^{1/5}$.
\end{lemma}

Recall $ \Delta_{\widetilde{\bSigma}}^{\varepsilon\delta} =  \max_{\ell\in [L]} 
\{\|\widetilde{\bSigma}_{\varepsilon\delta}^{(\ell)}\|_{\infty,\max} + \|\widetilde{\bSigma}_{\delta\vep}^{(\ell)}\|_{\infty,\max} +  \|\widetilde\bSigma_{\delta\delta}^{(\ell)} \|_{\infty,\max} \}$. We first derive the convergence rate of
$\max_{\ell\in [L]}\{\|\widetilde{\bSigma}_{\varepsilon\delta}^{(\ell)}\|_{\infty,\max} + \|\widetilde{\bSigma}_{\delta\vep}^{(\ell)}\|_{\infty,\max}\}$. By \eqref{eq:def.sigmadv}, for each $\ell\in[L]$ and $j,j'\in[p]$, we have
\begin{align*}
	&~\widetilde{\Sigma}_{\vep\delta,jj'}^{(\ell)}(u,v) 
	=  \frac{1}{n_{\ell}}\sum_{t=1}^{n_\ell}\{\vep_{t,j}(u)-\bar{\vep}_j(u)\}\{\delta_{t+\ell,j'}(v)-\bar{\delta}_{j'}(v)\} \\
	=&~ \underbrace{\frac{1}{n_{\ell}}\sum_{t=1}^{n_\ell}\vep_{t,j}(u)\delta_{t+\ell,j'}(v)}_{{\rm I}_{\ell,j,j'}(u,v)} 
	- \underbrace{\bigg[\bigg\{\frac{1}{n_\ell}\sum_{t=1}^{n_\ell}\vep_{t,j}(u)\bigg\}\bar{\delta}_{j'}(v)  
	 + \bar{\vep}_{j}(u)\bigg\{\frac{1}{n_\ell}\sum_{t=1}^{n_\ell}\delta_{t+\ell,j'}(v)\bigg\}-\bar{\vep}_j(u) \bar{\delta}_{j'}(v)\bigg]}_{{\rm II}_{\ell,j,j'}(u,v)}  \,.
\end{align*}
By \eqref{div.delta}, $\widetilde{\bA}^{\T}\widetilde{\bA}=\bI_r$, and $\widetilde{\bA} \widetilde{\bZ}_t(\cdot)= \bA\bZ_t(\cdot)$, it holds that for any $j\in[p]$, 
\begin{align*} 
	\delta_{t,j}(u)  
	=   \underbrace{\be_j^{\T} ( \widetilde{\bA}\widetilde{\bA}^{\T} -  \widehat{\bA}\widehat{\bA}^{\T})}_{\bDelta_j^{\T}}\bA \bZ_t(u) +   \be_j^{\T}(\widetilde{\bA}\widetilde{\bA}^{\T}-\widehat{\bA}\widehat{\bA}^{\T})\bvep_t(u)  - \tilde\ba_j^{\T}\widetilde{\bA}^{\T}\bvep_t(u) \,,
\end{align*}
where  $\widetilde\bA=(\tilde\ba_1,\ldots,\tilde\ba_p)^{\T}$ and $\be_j\in\mathbb R^p$ denotes the $j$-th canonical basis vector, whose $j$-th component is one and all other components are zero. Then, we have
\begin{align*}
	 {\rm I}_{\ell,j,j'}(u,v) = &~  \bDelta_{j'}^{\T}\bigg\{\frac{1}{n_\ell}\sum_{t=1}^{n_\ell}\bA\bZ_{t+\ell}(v)\vep_{t,j}(u)+\frac{1}{n_\ell}\sum_{t=1}^{n_\ell}\bvep_{t+\ell}(v)\vep_{t,j}(u)\bigg\} \\
	&~     - \tilde\ba_{j'}^{\T}\bigg\{\frac{1}{n_\ell}\sum_{t=1}^{n_\ell}\widetilde\bA^{\T}\bvep_{t+\ell}(v)\vep_{t,j}(u)\bigg\} \,.
\end{align*}
By Condition \ref{cond.facmat}, we have $|\tilde\ba_{j'}|_2 \lesssim p^{-1/2}$ for any $j'\in[p]$. Following from  \eqref{eq:diff.AA} in the proof of Lemma  \ref{lemma2.prop2} (see Section \ref{proof.lemma2.prop2}) and $p\geq n^{\upsilon}$, we have, for  any $j'\in[p]$, 
\begin{align}\label{Deltaj}
   \max_{j'\in[p]}|\bDelta_{j'}|_2= \max_{j'\in[p]}| \be_{j'}^{\T}(\widetilde\bA\widetilde\bA^{\T}-\widehat{\bA}\widehat{\bA}^{\T}   )|_2 = O_{\rm p}\big\{  n^{-1/2}p^{ -1/2} (\log  p)^{1/2}  \big\}\,,
\end{align}
provided that 
$\log  p  \ll n^{1/5}$. 
Notice that $\ba_{j'}^{\T}\bZ_{t+\ell}(v)$ is the $j'$-th component of $\bA\bZ_{t+\ell}(v)$, where $j'\in[p]$.   By Conditions \ref{cond.facerror}--\ref{cond.facmat}, we know that $\|\ba_{j'}^{\T}\bZ_{t+\ell}(v)\vep_{t,j}(u)\|_{\psi_1}\leq C$ and $\|\ba_{j'}^{\T}\bZ_{t+\ell}(v_1)\vep_{t,j}(u_1) - \ba_{j'}^{\T}\bZ_{t+\ell}(v_2)\vep_{t,j}(u_2)\|_{\psi_1}\leq C(|u_1-u_2|^{\kappa}+|v_1-v_2|^{\kappa})$. Moreover,  
\begin{align*}
    &\max_{\ell\in[L]}\max_{j,j'\in[p]}\sup_{(u,v)\in\cU^2}\bigg|\frac{1}{n_\ell}\sum_{t=1}^{n_\ell} \ba_{j'}^{\T}\bZ_{t+\ell}(v)\vep_{t,j}(u)\bigg| \leq \underbrace{\max_{\ell\in[L]}\max_{j,j'\in[p]}\max_{m,m'\in[M']}\bigg|\frac{1}{n_\ell}\sum_{t=1}^{n_\ell} \ba_{j'}^{\T}\bZ_{t+\ell}(b_{m'})\vep_{t,j}(b_m)\bigg|}_{{\rm K}_1}\\
    &~~~~~~~~~~~~ + \underbrace{\max_{\ell\in[L]}\max_{j,j'\in[p]}\max_{m,m'\in[M']}  \sup_{(u,v)\in B_{m}\times B_{m'}} \bigg|\frac{1}{n_\ell}\sum_{t=1}^{n_\ell} \{\ba_{j'}^{\T}\bZ_{t+\ell}(v)\vep_{t,j}(u) -\ba_{j'}^{\T}\bZ_{t+\ell}(b_{m'})\vep_{t,j}(b_m) \}\bigg|}_{{\rm K}_2}
\end{align*}
for some $M'>0$. Choose $M' \asymp (np)^C$ for some sufficiently large constant $C>0$. Then, following the proof of \eqref{eq:rate.I1} in Section   \ref{proof.lemma_tn_dis}, we also have  ${\rm K}_2 = O_{\rm p}(n^{-3/2})$. Moreover, by Conditions \ref{cond.facerror}(iii), \ref{cond.facprocess}(iii) and Lemma \ref{tail_chang} with $(B_n,c_n,{r}_1,{r})=(1,\ell,1,1/3)$, it holds that
\begin{align*}
    \bbP({\rm K}_1 > x) \lesssim p^2(M')^2\big\{\exp(-Cn x^2)+ \exp(-Cn^{1/3}x^{1/3})   \big\}
\end{align*}
for any $x\geq 0$, which implies that $ {\rm K}_1 = O_{\rm p}\{ n^{-1/2}(\log p)^{1/2} \} $, provided that $\log  p \ll n^{1/5}$. Then, we can conclude that 
 \begin{align}\label{term1.sup}
     \max_{\ell\in[L]}\max_{j,j'\in[p]}\sup_{(u,v)\in\cU^2}\bigg|\frac{1}{n_\ell}\sum_{t=1}^{n_\ell} \ba_{j'}^{\T}\bZ_{t+\ell}(v)\vep_{t,j}(u)\bigg| = O_{\rm p}\{ n^{-1/2} (\log  p)^{1/2} \}  
 \end{align}
provided that $\log  p \ll n^{1/5}$. Similarly, since $ \mathbb{E}\{ \bvep_t(u) \bvep_{t+\ell}(v)^{\T} \}= \mathbf{0}$ for any $\ell\neq 0$ under model \eqref{eq:ffm}, we have
\begin{align}\label{term2.sup}
     \max_{\ell\in[L]}\max_{j,j'\in[p]}\sup_{(u,v)\in\cU^2}\bigg|\frac{1}{n_\ell}\sum_{t=1}^{n_\ell} \vep_{t+\ell,j'}(v)\vep_{t,j}(u)\bigg| = O_{\rm p}\{ n^{-1/2} (\log  p)^{1/2} \}  
 \end{align}
provided that $\log  p \ll n^{1/5}$. Write $\widetilde\bA=(\tilde\bb_1,\ldots,\tilde\bb_r)$.  Then for any  $i\in[r]$, $\tilde\bb_i^{\T}\bvep_t(u)$ is the $i$-th component of $\widetilde\bA^{\T}\bvep_t(u)$. Notice that $|\tilde\bb_i|_2 = 1$ for any $i\in[r]$. Therefore, by Condition \ref{cond.facerror}, it holds that $ \|\tilde\bb_i^{\T}\bvep_t(u)\|_{\psi_2} \leq C$ and $ \|\tilde\bb_i^{\T}\bvep_t(u) - \tilde\bb_i^{\T}\bvep_t(v)\|_{\psi_2} \leq C|u-v|^{\kappa}$ for any $(t,i,u,v)$. Similarly, under model \eqref{eq:ffm}, we also obtain
\begin{align}\label{bsup}
     \max_{\ell\in[L]}\max_{i\in[r]}\max_{j \in[p]}\sup_{(u,v)\in\cU^2}\bigg|\frac{1}{n_\ell}\sum_{t=1}^{n_\ell} \tilde\bb_i^{\T}\bvep_{t+\ell}(v)\vep_{t,j}(u)\bigg| = O_{\rm p}\{ n^{-1/2} (\log  p)^{1/2} \}  
 \end{align}
provided that $\log  p \ll n^{1/5}$. Then, since $|\bc|_2 \leq q^{1/2}\max_{i\in[q]}|c_i|$ for any $\bc=(c_1,\ldots,c_q)^{\T}\in\mathbb{R}^q$, combining \eqref{Deltaj}--\eqref{bsup} yields 
\begin{align}\label{order.I}
    \max_{\ell\in[L]}\max_{j,j'\in[p]}\sup_{(u,v)\in\cU^2}|{\rm I}_{\ell,j,j'} (u,v)| = &~ O_{\rm p}\bigg\{ \bigg( \frac{  \log  p  }{np}\bigg)^{1/2}\cdot p^{1/2}\bigg( \frac{ \log   p  }{n}  \bigg)^{1/2} + p^{-1/2}\cdot \bigg( \frac{\log  p }{n } \bigg)^{1/2} \bigg\}\nonumber\\
    = &~ O_{\rm p}\bigg\{   \frac{  \log  p  }{n }  +   \bigg( \frac{  \log  p  }{np}\bigg)^{1/2} \bigg\}\,,
\end{align}
provided that $\log  p \ll n^{1/5}$. 
For  ${\rm II}_{\ell,j,j'} (u,v)$, notice that
\begin{align*}
    \bar{\delta}_{j'}(v) = \bDelta_{j'}^{\T}\bigg\{ \frac{1}{n}\sum_{t=1}^n \bA\bZ_{t}(v)+\frac{1}{n}\sum_{t=1}^n\bvep_t(v) \bigg\}  -  \tilde\ba_{j'}^{\T} \bigg\{  \frac{1}{n}\sum_{t=1}^n \widetilde\bA^{\T}\bvep_t(v) \bigg\}
\end{align*}
for any $(j',v)$. Then, analogously to \eqref{order.I}, we have 
\begin{align}\label{delta.mean.order}
    \max_{j'\in[p]}\sup_{v\in \cU}|\bar{\delta}_{j'}(v)| = O_{\rm p}\bigg\{   \frac{  \log  p  }{n }  +   \bigg( \frac{  \log  p  }{np}\bigg)^{1/2} \bigg\}\,,
\end{align}
provided that $\log  p \ll n^{1/5}$. 
In addition, we can also obtain 
\begin{align*}
     \max_{\ell\in[L]_0}\max_{j \in[p]}\sup_{ u \in\cU }\bigg|\frac{1}{n_\ell}\sum_{t=1}^{n_\ell}  \vep_{t,j}(u)\bigg| = O_{\rm p}\{ n^{-1/2} (\log  p)^{1/2} \} 
 \end{align*}
provided that $\log  p \ll n^{1/5}$. Then, we have 
\begin{align*}
    \max_{\ell\in[L]}\max_{j,j'\in[p]}\sup_{(u,v)\in\cU^2}|{\rm II}_{\ell,j,j'} (u,v)| = &~ O_{\rm p}\bigg[ \bigg( \frac{\log  p }{n } \bigg)^{1/2}\cdot\bigg\{   \frac{  \log  p  }{n }  +   \bigg( \frac{  \log  p  }{np}\bigg)^{1/2} \bigg\}\bigg] \,,
\end{align*}
provided that $\log  p \ll n^{1/5}$. 
Therefore, we can conclude that
\begin{align}\label{Sigmavepdelta.order}
\max_{\ell\in[L]}\{\|\widetilde{\bSigma}_{\varepsilon\delta}^{(\ell)}\|_{\infty,\max}+\|\widetilde{\bSigma}_{\delta\vep}^{(\ell)}\|_{\infty,\max}\} = O_{\rm p}\bigg\{   \frac{  \log  p  }{n }  +   \bigg( \frac{  \log  p  }{np}\bigg)^{1/2} \bigg\}\,, 
\end{align}
provided that $\log  p \ll n^{1/5}$.

We next derive the convergence rate of $ \max_{\ell\in[L]}\|\widetilde\bSigma_{\delta\delta}^{(\ell)} \|_{\infty,\max}  $.  Recall
\begin{align*}
	&~\widetilde{\Sigma}_{\delta\delta,jj'}^{(\ell)}(u,v) 
	= \frac{1}{n_{\ell}}\sum_{t=1}^{n_\ell}\{\delta_{t,j}(u)-\bar{\delta}_j(u)\}\{\delta_{t+\ell,j'}(v)-\bar{\delta}_{j'}(v)\} \\
	=&~ \underbrace{\frac{1}{n_{\ell}}\sum_{t=1}^{n_\ell}\delta_{t,j}(u)\delta_{t+\ell,j'}(v)}_{{\rm III}_{\ell,j,j'}(u,v)} 
	- \underbrace{\bigg[\bigg\{\frac{1}{n_\ell}\sum_{t=1}^{n_\ell}\delta_{t,j}(u)\bigg\}\bar{\delta}_{j'}(v)  
	+ \bar{\delta}_{j}(u)\bigg\{\frac{1}{n_\ell}\sum_{t=1}^{n_\ell}\delta_{t+\ell,j'}(v)\bigg\} 
	- \bar{\delta}_j(u) \bar{\delta}_{j'}(v)\bigg]}_{{\rm IV}_{\ell,j,j'}(u,v)}  \,.
\end{align*}
By \eqref{delta.mean.order}, we have 
\begin{align*}
    \max_{\ell\in[L]}\max_{j,j'\in[p]}\sup_{(u,v)\in\cU^2}|{\rm IV}_{\ell,j,j'} (u,v)| = O_{\rm p}\bigg\{  \bigg( \frac{ \log  p   }{n }\bigg)^2  +   \frac{  \log  p  }{n p }   \bigg\}\,,
\end{align*}
provided that $\log  p \ll n^{1/5}$. 
In addition, since $ \delta_{t,j}(u) = \bDelta_j^{\T} \{\bA \bZ_t(u) +  \bvep_t(u)  \}-   \tilde \ba_{j}^{\T} \widetilde\bA^{\T}\bvep_t(u) $, it holds that, for any $(\ell,j,j',u,v)$,
\begin{align*}
   & |{\rm III}_{\ell,j,j'}(u,v)|=  \bigg| \frac{1}{n_{\ell}}\sum_{t=1}^{n_\ell}\delta_{t,j}(u)\delta_{t+\ell,j'}(v) \bigg|\\
   &~~~ \leq  \bigg| \bDelta_j^{\T}\bigg\{ \frac{1}{n_{\ell}}\sum_{t=1}^{n_\ell}\{\bA \bZ_t(u) +  \bvep_t(u)  \} \{\bA \bZ_{t+\ell}(v) +  \bvep_{t+\ell}(v)  \}^{\T} \bigg\} \bDelta_{j'} \bigg|\\
   &~~~~~ + \bigg| \bDelta_j^{\T}\bigg\{ \frac{1}{n_{\ell}}\sum_{t=1}^{n_\ell}\{\bA \bZ_t(u) +  \bvep_t(u)  \}\{ \widetilde\bA^{\T}\bvep_{t+\ell }(v)  \}^{\T} \bigg\}  \tilde \ba_{j'} \bigg|\\
   &~~~~~ + \bigg|  \tilde\ba_{j}^{\T}\bigg\{ \frac{1}{n_{\ell}}\sum_{t=1}^{n_\ell}\{ \widetilde\bA^{\T}\bvep_t(u)   \} \{\bA \bZ_{t+\ell}(v) +  \bvep_{t+\ell}(v)  \}^{\T} \bigg\} \bDelta_{j'} \bigg|\\
    &~~~~~ + \bigg|  \tilde\ba_{j}^{\T}\bigg\{ \frac{1}{n_{\ell}}\sum_{t=1}^{n_\ell}\{ \widetilde\bA^{\T}\bvep_t(u)   \} \{ \widetilde\bA^{\T}\bvep_{t+\ell }(v)  \}^{\T} \bigg\}  \tilde\ba_{j'} \bigg| \,.
\end{align*}
Similar as in the proof of \eqref{bsup}, it holds under model \eqref{eq:ffm} and Condition \ref{cond.facerror}(iii) that  
\begin{align*}
    \max_{\ell\in[L]}\sup_{(u,v)\in\cU^2}\bigg\|  \frac{1}{n_{\ell}}\sum_{t=1}^{n_\ell}\{\bA \bZ_t(u) +  \bvep_t(u)  \}\{ \widetilde\bA^{\T}\bvep_{t+\ell }(v)  \}^{\T}\bigg\|_{\rm F} = &~ O_{\rm p}\{ n^{-1/2}p^{1/2} (\log p)^{1/2}\}\,,\\
    \max_{\ell\in[L]}\sup_{(u,v)\in\cU^2}\bigg\| \frac{1}{n_{\ell}}\sum_{t=1}^{n_\ell}\{ \widetilde\bA^{\T}\bvep_t(u)   \} \{\bA \bZ_{t+\ell}(v) +  \bvep_{t+\ell}(v)  \}^{\T} \bigg\|_{\rm F} = &~ O_{\rm p}\{ n^{-1/2}p^{1/2} (\log p)^{1/2}\}\,,\\
     \max_{\ell\in[L]}\sup_{(u,v)\in\cU^2}\bigg\|   \frac{1}{n_{\ell}}\sum_{t=1}^{n_\ell}\{ \widetilde\bA^{\T} \bvep_t(u)   \} \{ \widetilde\bA^{\T}\bvep_{t+\ell }(v)  \}^{\T}  \bigg\|_{\rm F} =&~ O_{\rm p}\{  n^{-1/2}   (\log  p)^{1/2} \} \,,
\end{align*} 
provided that $\log  p \ll n^{1/5}$. Notice that $|\bbE[\{\ba_j^{\T}\bZ_{t}(u)+   \vep_{t,j}(u)\}\{\ba_{j'}^{\T}\bZ_{t+\ell}(v)+   \vep_{t+\ell,j'}(v)\}  ]|\leq C$ for any $(t,\ell,j,j',u,v)$. Analogously, we also have 
\begin{align*}
     &~\max_{\ell\in[L]}\sup_{(u,v)\in\cU^2}\bigg\|   \frac{1}{n_{\ell}}\sum_{t=1}^{n_\ell}\{\bA \bZ_t(u) +  \bvep_t(u)  \} \{\bA \bZ_{t+\ell}(v) +  \bvep_{t+\ell}(v)  \}^{\T}  \bigg\|_{\rm F}\\
     \leq &~
    \max_{\ell\in[L]}\sup_{(u,v)\in\cU^2}\bigg\|   \frac{1}{n_{\ell}}\sum_{t=1}^{n_\ell} (1-\bbE)\big[ \{\bA \bZ_t(u) +  \bvep_t(u)  \} \{\bA \bZ_{t+\ell}(v) +  \bvep_{t+\ell}(v)  \}^{\T}\big]  \bigg\|_{\rm F}\\
    &~~ +\max_{\ell\in[L]}\sup_{(u,v)\in\cU^2}\bigg\|   \frac{1}{n_{\ell}}\sum_{t=1}^{n_\ell}  \bbE\big[ \{\bA \bZ_t(u) +  \bvep_t(u)  \} \{\bA \bZ_{t+\ell}(v) +  \bvep_{t+\ell}(v)  \}^{\T}\big]  \bigg\|_{\rm F}
    =   O_{\rm p} (p ) \,,
\end{align*}
provided that $\log  p \ll n^{1/5}$. Notice that $|\ba^{\T}\bB\bc|
\le
|\ba|_2 \|\bB\|_{\rm op} |\bc|_2$ 
and $\|\bB\|_{\rm op}
\le
\|\bB\|_{\rm F}$ 
for any $\ba,\bc\in\mathbb{R}^q$ and $\bB\in\mathbb{R}^{q\times q}$. Since $\max_{j\in[p]}| \tilde\ba_{j} |_2 \lesssim p^{-1/2}$,  it follows from  \eqref{Deltaj} that 
\begin{align*}
   &~ \max_{\ell\in[L]}\max_{j,j'\in[p]}\sup_{(u,v)\in\cU^2}|{\rm III}_{\ell,j,j'} (u,v)|\\
    =&~O_{\rm p}\bigg\{ \frac{\log p}{np  } \cdot  p  +  \bigg(\frac{ \log p }{np }\bigg)^{1/2} \cdot p^{-1/2} \cdot \bigg(\frac{p\log p}{n }\bigg)^{1/2} + p^{-1}\cdot  \bigg(\frac{\log p}{n} \bigg)^{1/2} \bigg\} \\
   =&~O_{\rm p}\bigg\{ \frac{ \log p }{n}+ \bigg( \frac{\log p}{np^2}\bigg)^{1/2} \bigg\}\,,
\end{align*}
 provided that $\log  p \ll n^{1/5}$. 
 Then,  we can conclude that 
\begin{align*}
    \max_{\ell\in[L]}\|\widetilde\bSigma_{\delta\delta}^{(\ell)} \|_{\infty,\max} = O_{\rm p}\bigg\{ \frac{ \log p }{n}+ \bigg( \frac{\log p}{np^2}\bigg)^{1/2} \bigg\}\,,
\end{align*}
  provided that $\log  p \ll n^{1/5}$. Combining this with \eqref{Sigmavepdelta.order}, we have  
\begin{align*}
    \Delta_{\widetilde{\bSigma}}^{\varepsilon\delta} = O_{\rm p}\{ (n^{-1} +n^{-1/2}p^{-1/2}) \log p  \}\,,
\end{align*}
 provided that $\log  p \ll n^{1/5}$.  
  This completes the proof of the first part of Proposition~\ref{prop.facmodel}.

For the second part of Proposition~\ref{prop.facmodel}, write 
\begin{align*}
  &a_{1n}^*=(\log p)^{1/2}\,, ~~ a_{2n}^*=\frac{ (np)^2}{  M_2^{ \kappa/2} }\,,~~  b_{1n}^* = \frac{  \log p}{n^{1/2}}+\bigg(\frac{   \log p}{p}\bigg)^{1/2}\,\\
  &~~~~~~~~\mbox{and}~~~~~~b_{2n}^*= \frac{ n^{3/2}p^{ 2}(\log p)^{1/2}+n^2p^{ 3/2}}{  M_2^{ \kappa/2} }\,.
\end{align*}
Notice that Condition \ref{cond.facerror} implies Condition \ref{c.subgaussian}. Then, by Condition \ref{c.kernelF}, Lemmas~\ref{lemma1.prop1} and~\ref{lemma2.prop2}, together with arguments similar to those used in the proof of \eqref{Delta.g}, we have
\begin{align*}
     \Delta_{\tilde{\boldsymbol{\cG}}}^{\varepsilon\delta}  =&~ O_{\rm p}\{  n^{\rho/2} (a_{1n}^*+b_{1n}^*) b_{1n}^* \log(pM_2) + n^{1/2}(a_{1n}^*b_{2n}^*+a_{2n}^*b_{1n}^*  + b_{1n}^*b_{2n}^*)   \}\,,
\end{align*} 
provided that $\log  p \ll n^{1/5}$. 
Choose $M_2\asymp  (np)^C$ for some sufficiently large constant
$C>0$. It follows that 
\begin{align*}
    \Delta_{\tilde{\boldsymbol{\cG}}}^{\varepsilon\delta}  = O_{\rm p}\{ n^{\rho/2}(n^{-1/2} +p^{-1/2})(\log p)^{5/2}  \}\,,
\end{align*}
 provided that $\log  p \ll n^{1/5}$.  
Hence we complete the proof of Proposition \ref{prop.facmodel}.
$\hfill\Box$

\section{Proofs of the Results in Section \ref{sec.oracle}}

\subsection{Auxiliary Lemmas}

We first introduce some useful lemmas.  Lemma \ref{meer}  is Lemma 5.4 of \cite{MWX2013_supp}, which provides exponential tail bounds for local increments of a Gaussian process. Lemma~\ref{tail_chang} follows as a special case of Lemma~L1 in \cite{ChangChenWu2023_supp}, which establishes exponential tail bounds and moment bounds for partial sums of $\alpha$-mixing random variables.
Lemma~\ref{tail_multi} extends classical tail bounds for the supremum of stochastic-process increments in \cite{Ledoux1991_supp}. Working directly with the tail condition in \eqref{psi_multi}, Lemma~\ref{tail_multi} only requires $\gamma>0$, rather than the usual $\gamma\geq1$ condition imposed in Orlicz-norm-based treatments to ensure convexity. The proofs of Lemmas \ref{tail_chang} and  \ref{tail_multi} are given in Sections \ref{sec:lemmaA3} and \ref{sec:lemma_tail_multi}, respectively.

\begin{lemma}\label{meer} 
Let $\{Z(\bt),\bt\in\mathbb{R}^q\}$ be a real-valued centered Gaussian random field satisfying $[\mathbb{E}\{|Z(\bs)-Z(\bt)|^2\}]^{1/2}\leq c_1\rho(\bs,\bt)$ for any $\bs,\bt\in\mathbb{R}^q$ with some finite universal  constant $c_1>0$, where  $\rho(\bs,\bt)=\sum_{j=1}^q|s_j-t_j|^{h_j}$ for a fixed vector $\bh=(h_1,\dots,h_q)^{\T}\in (0,1]^q$. Let $|\bs|=(|s_1|,\ldots,|s_q|)^{\T}$. Then there exist two positive and finite universal constants $x_0$ and $c_2$ such that for  $x\geq x_0$,
\begin{align*}
\sup_{\bt_0\in[0,1]^q}\mathbb{P}\Bigg\{\sup_{|\bs|\leq\ba}|Z(\bt_0+\bs)-Z(\bt_0)|\geq x\sum_{j=1}^q a_j^{h_j}\Bigg\}\leq \exp(-c_2x^2)
\end{align*}
for all $\ba=(a_1,\dots,a_q)^{\T}\in (0,1]^q$ such that $\bt_0-\ba\in[0,1]^q $ and $\bt_0+\ba\in[0,1]^q $.
\end{lemma}

\begin{lemma}
\label{tail_chang}
Let $\{ Z_t \}_{t=1}^{n}$ be an $\alpha$-mixing sequence of centered random variables with $\alpha$-mixing coefficients $\{\alpha(m)\}_{m\geq1}$.
Assume   there exist some universal constants $a_1,a_2,b_1,b_2,r_1 \\ >0$ such that
{\rm (i)} $\max_{t\in[n]}\mathbb{P}(|Z_t|>x)\leq a_1\exp(-a_2x^{ r_1}B_{n}^{-  r_1})$ for any $x\geq 0$, where $B_{n}>0$ may diverge with $n$;
{\rm (ii)}  $\alpha(m)\leq b_1\exp(-b_2|m-c_n|_{+})$ for any integer $m\geq 1$, where $0\leq c_n=o(n)$ may diverge with $n$. Let ${r}=(3+| r_1^{-1}-1|_+)^{-1}$. It then holds that
\begin{align*}
    \mathbb{P}\bigg(\bigg|\sum_{t=1}^{n}Z_t\bigg|\geq x\bigg)\lesssim \exp\bigg\{-\frac{Cx^2}{(1+ c_n)B_n^{2}n} \bigg\}+\exp\bigg\{ -\frac{ Cx^{r} }{ (1+c_n)^{r}B_n^{r} } \bigg\}
\end{align*}
for any $x\geq 0$. Furthermore, if  $c_n$ is a fixed finite integer, we have $\|\sum_{t=1}^{n}Z_t\|_s\lesssim B_n n^{1/2}$ for any finite positive integer $s$.
\end{lemma}

\begin{lemma} \label{tail_multi}
Let $(\mathcal T,d)$ be a totally bounded metric space with
diameter $ D=\sup_{\bs,\bt\in\mathcal T}d(\bs,\bt)\\\in(0,C)$ for some constant $C< \infty$.
Assume that the process $\{Z(\bt):\bt\in\mathcal T\}$ is separable
with respect to $ d$ and satisfies the following increment tail condition: there exist constants $c_{1n}\ge 1$, $c_{2n}>0$ and
$\gamma>0$ such that, for any $\bs,\bt\in\mathcal T$ with
$\bs\neq\bt$ and any $x\ge0$, 
\begin{align}\label{psi_multi}
\mathbb{P}\big\{|Z(\bs)-Z(\bt)|>x\big\}\leq c_{1n} \exp\biggl[-c_{2n} \biggl\{\frac{x}{d(\bs,\bt)} \biggr\}^{\gamma}\biggr] \,.
\end{align}
Let $\Lambda_u= 1+\log_2(1/u) $ for   $u\in (0,1]$, and let   $N_{\epsilon} := N(\mathcal{T}, d; \epsilon)$ denote the covering number of $\mathcal{T}$ with respect to the  metric $d$ and radius $\epsilon$.   The following statements hold.
\begin{itemize}
\item[ ${\rm (i)}$ ] For any $D_1 \in (0, D\wedge 1)$, we have
\begin{align*}
&\mathbb{P}\bigg\{ \sup_{\bs,\bt\in\cT}|Z(\bs)-Z(\bt)| > x\bigg\} \lesssim    \frac{c_{1n}C_{\gamma,1 }}{c_{2n}^{1/\gamma} x} \int_{0}^{D_1}\{\log ( e  N_{\epsilon}    )\}^{1/\gamma} \,{\rm d}\epsilon \\
&~~~~~~~~~~~~~~~~~~~~~~~~~~~~~+ c_{1n} N_{D_1}\Lambda_{D_1}\exp\biggl\{ - c_{2n} C_{\gamma,2}\biggl(\frac{x}{D} \biggr)^{\gamma}\biggr\}
\end{align*}
for any $x > 0$, where $C_{\gamma,1},C_{\gamma,2}$ are positive constants depending only on $\gamma$.
\item[ ${\rm (ii)}$ ] For any $\beta \in (0,1)$ and   $D_1, D_2 \in (0, D\wedge 1)$ with $D_2 < D_1$, we have
\begin{align*}
&\bbP\biggl[ \sup_{\substack{\bs,\bt\in\cT\\ \bs\neq \bt}}\frac{|Z(\bs)-Z(\bt)|}{\{d(\bs,\bt)\}^{\beta}} >x \biggr]  \lesssim \frac{ c_{1n} C_{\gamma,\beta,1 }}{c_{2n}^{1/\gamma}x}   \bigg[\int_{0}^{D_1}\frac{\{\log( e N_{\epsilon})\}^{1/\gamma}}{\epsilon^{\beta}}\, {\rm d}\epsilon + \frac{\int_{0}^{D_2} \{\log(eN_{\epsilon})\}^{1/\gamma}\,{\rm d}\epsilon}{D_1^{\beta}} \bigg]\\
&~~~~~~~~~~~~~~~~~~~~~~~~~~~~~~~ + c_{1n}\big \{N_{D_1}^2\Lambda_{D_1}+N_{D_2} \Lambda_{D_1}\Lambda_{D_2}\big \}\exp\biggl\{ - c_{2n} C_{\gamma,\beta,2}\biggl(\frac{x}{D^{1-\beta}} \biggr)^{\gamma}\biggr\}
\end{align*}
for any $x > 0$, where $ C_{\gamma,\beta,1} ,C_{\gamma,\beta,2 }$ are positive constants depending only on $( \gamma,\beta)$.
\end{itemize} 	
\end{lemma}

\subsection{Proof of Lemma \ref{lem.dis.Gaussian}}\label{proof.lem.dis.Gaussian}

By the definitions of $\tilde{T}_{n}^{\rm G}$ and 
$\tilde{T}_{n}^{\rm dis,G}$ in \eqref{Tn.G} and \eqref{Tn.dis.G}, we have 
\begin{align*}
    |\tilde{T}^{\rm G}_{n}-\tilde{T}^{\rm dis,G}_{n}| \leq \max_{r\in[Lp^2]}\max_{m,m'\in[M]}\sup_{(u,v)\in B_m\times B_{m'}}|\tilde\cG_r(u,v)-\tilde\cG_r(b_m,b_{m'})|\,.
\end{align*}
We claim that, for any $r\in[Lp^2]$ and $(u_1,v_1),(u_2,v_2)\in \cU^2$, 
\begin{align}\label{2nd.moment}
    [\mathbb{E}\{|\tilde\cG_r(u_1,v_1)-\tilde\cG_r(u_2,v_2)|^2\}]^{1/2} \leq  C(|u_1-u_2|^{\kappa}+|v_1-v_2|^{\kappa}) \,.
\end{align}
    Since $\max\{|u-b_m|,|v-b_{m'}|\}\leq (2M)^{-1}$ for any $(u,v)\in B_m\times B_{m'}$, by Lemma \ref{meer} with $q=2$, $(a_1,a_2)^{\T}=((2M)^{-1},(2M)^{-1})^{\T}$ and $\bh=(\kappa ,\kappa )^{\T}$, it holds that
\begin{align}\label{tail.1}
\mathbb{P}(|\tilde{T}^{\rm G}_{n}-\tilde{T}^{\rm dis,G}_{n}|>x)
\leq&\, L p^2M^2\max_{r\in[Lp^2]}\max_{m,m'\in[M]}\mathbb{P}\bigg\{\sup_{(u,v)\in B_m\times B_{m'}}|\tilde\cG_r(u,v)-\tilde\cG_r(b_m,b_{m'})|>x\bigg\}\nonumber\\
\lesssim &\,  p^2M^2\exp(-CM^{2 \kappa}x^2)
\end{align}
for any $x\geq CM^{-\kappa }$. In addition, for any $\delta>0$,
\begin{align*}
      &\mathbb{P}(\tilde{T}^{\rm G}_{n}\leq x)-\mathbb{P}(\tilde{T}^{\rm dis, G}_{n}\leq x)\\
       &~~\leq  \mathbb{P}(\tilde{T}^{\rm G}_{n}\leq x,|\tilde{T}^{\rm G}_{n}-\tilde{T}^{\rm dis,G}_{n}| \leq \delta)+\mathbb{P}( |\tilde{T}^{\rm G}_{n}-\tilde{T}^{\rm dis,G}_{n}|> \delta)-\mathbb{P}(\tilde{T}^{\rm dis,G}_{n}\leq x)  \\
        &~~\leq    \mathbb{P}(x < \tilde{T}^{\rm dis, G}_{n}\leq x+ \delta)+\mathbb{P}( |\tilde{T}^{\rm G}_{n}-\tilde{T}^{\rm dis,G}_{n}|> \delta)\,.
\end{align*}
On the other hand, 
\begin{align*}
&\mathbb{P}(\tilde{T}^{\rm G}_{n}\leq x)
-\mathbb{P}(\tilde{T}^{\rm dis,G}_{n}\leq x) \\
&~~ \geq
\mathbb{P}(\tilde{T}^{\rm dis,G}_{n}\leq x-\delta)
-
\mathbb{P}(
|\tilde{T}^{\rm G}_{n}
-\tilde{T}^{\rm dis,G}_{n}|>\delta
)-\mathbb{P}(\tilde{T}^{\rm dis,G}_{n}\leq x)\\
&~~ =
-\mathbb{P}(
x-\delta < \tilde{T}^{\rm dis,G}_{n}\leq x
)
-
\mathbb{P}(
|\tilde{T}^{\rm G}_{n}
-\tilde{T}^{\rm dis,G}_{n}|>\delta
)\,.
\end{align*}
Then, by \eqref{tail.1}, Condition \ref{c.bvar} and the Nazarov's inequality 
\citep[see, e.g., Lemma~A.1 of][]{Cher2017_supp}, we can conclude that  
\begin{align*}
 \sup_{x\in \mathbb{R}}|\mathbb{P}(\tilde{T}^{\rm G}_{n}\leq x)-\mathbb{P}(\tilde{T}^{\rm dis, G}_{n}\leq x)| \leq &~\sup_{x\in \mathbb{R}}\mathbb{P}( |\tilde{T}^{\rm dis,G}_{n}-x|\leq \delta)+\mathbb{P}(|\tilde{T}^{\rm G}_{n}-\tilde{T}^{\rm dis,G}_{n}|>\delta) \\
\lesssim &~\delta\sqrt{\log ({ p M } )} +p^2M^2\exp(-CM^{2\kappa}\delta^2 )
\end{align*}
for any $\delta\geq CM^{-\kappa}$. Recall $M\asymp (np)^{C_M}$.  Letting $\delta=C' M^{-\kappa}\sqrt{\log ( pM) }$ for some sufficiently large constant $C'>0$, we have
$\sup_{x\in \mathbb{R}}|\mathbb{P}(\tilde{T}^{\rm G}_{n}\leq x)-\mathbb{P}(\tilde{T}^{\rm dis, G}_{n}\leq x)|=o(1)$. To confirm Lemma \ref{lem.dis.Gaussian},
it remains to show that \eqref{2nd.moment} holds.

\noindent\underline{Proof of \eqref{2nd.moment}.} Recall  $\tilde{\boldsymbol\cG}(\cdot,\cdot)=\{\tilde\cG_1(\cdot,\cdot),\ldots,\tilde\cG_{Lp^2}(\cdot,\cdot)\}^{\T}$ is a centered Gaussian process defined on $\cU^2$ with covariance function  $$\cov\{\tilde{\boldsymbol\cG}(u,v),\tilde{\boldsymbol\cG}(\tilde u, \tilde v)\}=\cov\{\tilde{\bxi}_n(u,v),\tilde{\bxi}_n(\tilde u,\tilde v)\}\,,$$
where $(u,v),(\tilde{u},\tilde{v})\in\cU^2$ and $\tilde{\bxi}_n(u,v) =n^{1 / 2} ([{\rm vec}\{\widetilde\bSigma^{(1)}(u,v)\}]^{\T},\ldots, [{\rm vec}\{\widetilde\bSigma^{(L)}(u,v)\}]^{\T} )^{\T}$, with   $\widetilde\bSigma^{(\ell)}(u,v)=\{\widetilde\Sigma_{jj'}^{(\ell)}(u,v)\}_{p\times p} =  n_\ell^{-1}  \sum_{t=1}^{n_\ell} \{\bvep_{t}(u)-\bar{\bvep}(u)\} \{\bvep_{t+\ell}(v)-\bar{\bvep}(v)\}^{\T}$ for each $\ell \in [L]$. For notational simplicity, for any fixed $r\in[Lp^2]$, we write
$(\ell,j,j')=(\ell_r,j_r,j'_r)$ for its associated unique triple. Therefore, for any $r\in [Lp^2]$ and $(u_1,v_1),(u_2,v_2)\in\cU^2$, we have
\begin{align}\label{eq:covg}
&n^{-1}\mathbb{E}\{|\tilde\cG_r(u_1,v_1)-\tilde\cG_r(u_2,v_2)|^2\} =  \var\big\{\widetilde\Sigma_{jj'}^{(\ell)}(u_1,v_1) -\widetilde\Sigma_{jj'}^{(\ell)}(u_2,v_2) \big\}   \,.
\end{align}
For any $\ell\in[L]$, $j,j'\in[p]$ and $u,v\in\cU$, we write 
\begin{equation}\label{eq:bareps.decomp}
\begin{aligned}
\bar{\vep}^{(1)}_{\ell,j,j'}(u,v)
=&~\frac{1}{n_{\ell}} \sum_{t=1}^{n_{\ell}}
[\vep_{t,j}(u)\vep_{t+\ell,j'}(v)
-\mathbb{E}\{ \vep_{t,j}(u)\vep_{t+\ell,j'}(v)\}]\,,\\
\bar{\vep}^{(2)}_{\ell,j,j'}(u,v)
=&~-\frac{1}{n_{\ell}}  \sum_{t=1}^{n_{\ell}}
[\vep_{t,j}(u)\bar{\vep}_{j'}(v)
-\mathbb{E}\{\vep_{t,j}(u)\bar{\vep}_{j'}(v) \}]\,,\\
\bar{\vep}^{(3)}_{\ell,j,j'}(u,v)
=&~-\frac{1}{n_{\ell}}  \sum_{t=1}^{n_{\ell}}
[\vep_{t+\ell,j'}(v)\bar{\vep}_{j}(u)
-\mathbb{E}\{\vep_{t+\ell,j'}(v)\bar{\vep}_{j}(u) \}]\,,\\
\bar{\vep}^{(4)}_{\ell,j,j'}(u,v)
=&~ \bar{\vep}^{(4)}_{j,j'}(u,v)
=\bar{\vep}_{j}(u)\bar{\vep}_{j'}(v)
-\mathbb{E}\{ \bar{\vep}_{j}(u)\bar{\vep}_{j'}(v) \}\,,
\end{aligned}
\end{equation}
where $\bar{\vep}_{j}(u)=n^{-1}\sum_{t=1}^n {\vep}_{t,j}(u)$. Then, for any $(\ell,j,j',u,v)$,
\begin{align*}
\widetilde\Sigma^{(\ell)}_{jj'}(u,v) - \mathbb{E}\{\widetilde\Sigma^{(\ell)}_{jj'}(u,v)\}   = \sum_{k=1}^4\bar{\vep}^{(k)}_{\ell,j,j'}(u,v)\,. 
\end{align*}
By \eqref{eq:covg}, it holds that 
\begin{align}\label{i16}
 n^{-1/2}[\mathbb{E}\{|\tilde\cG_r(u_1,v_1)-\tilde\cG_r(u_2,v_2)|^2\}]^{1/2}
\leq \sum_{k =1}^{4} \|\bar{\vep}^{(k)}_{\ell,j,j'}(u_1,v_1)-\bar{\vep}^{(k)}_{\ell,j,j'}(u_2,v_2)  \|_2  \,.
\end{align} 
Write $Z_{t,\ell,j,j'}(u,v) = \vep_{t,j}(u)\vep_{t+\ell,j'}(v)
-\mathbb{E}\{ \vep_{t,j}(u)\vep_{t+\ell,j'}(v)\}$. Let $d\{(u_1,v_1),(u_2,v_2)\} = |u_1-u_2|^{\kappa} + |v_1-v_2|^{\kappa}$ for any $(u_1,v_1),(u_2,v_2)\in \cU^2$.  By   Condition \ref{c.subgaussian}, it holds that
\begin{align}\label{eq:tailprob.vep2}
    &\|Z_{t,\ell,j,j'}(u_1,v_1)-Z_{t,\ell,j,j'}(u_2,v_2)\|_{\psi_1} \lesssim \|\varepsilon_{t,j}(u_1)\varepsilon_{t+\ell,j'}(v_1) - \varepsilon_{t,j}(u_2)\varepsilon_{t+\ell,j'}(v_2)\|_{\psi_1} \nonumber\\
    &~~~~\lesssim \|\varepsilon_{t,j}(u_1) - \varepsilon_{t,j}(u_2)\|_{\psi_2} \times \|\varepsilon_{t+\ell,j'}(v_1)\|_{\psi_2} +\|\varepsilon_{t+\ell,j'}(v_1) - \varepsilon_{t+\ell,j'}(v_2)\|_{\psi_2} \times  \|\varepsilon_{t,j}(u_2)\|_{\psi_2}\nonumber\\
    &~~~~\leq C( |u_1-u_2|^{\kappa} + |v_1-v_2|^{\kappa}) =  C d\{(u_1,v_1),(u_2,v_2)\}
\end{align}
for any $(u_1,v_1),(u_2,v_2)\in\cU^2$. Following from Condition \ref{c.alpha}, \eqref{eq:tailprob.vep2} and Lemma \ref{tail_chang} with $(B_n,c_n,{r}_1,{r})=(d\{(u_1,v_1),(u_2,v_2)\}, \ell, 1, 1/3)$, we have 
\begin{align}\label{eq:diffvareps2}
 \big\| \bar{\vep}^{(1)}_{\ell,j,j'}(u_1,v_1)-\bar{\vep}^{(1)}_{\ell,j,j'}(u_2,v_2)  \big\|_2 \lesssim n^{-1/2}d\{(u_1,v_1),(u_2,v_2)\}\,.
\end{align}

 For $\|\bar{\vep}^{(2)}_{\ell,j,j'}(u_1,v_1)
-\bar{\vep}^{(2)}_{\ell,j,j'}(u_2,v_2)\|_2$, by  \eqref{eq:bareps.decomp}, we have 
\begin{align*} 
 & \,  \|\bar{\vep}^{(2)}_{\ell,j,j'}(u_1,v_1)
-\bar{\vep}^{(2)}_{\ell,j,j'}(u_2,v_2)\|_2  \notag\\
\lesssim &\,  \bigg\|\bigg\{\frac{1}{n_{\ell}}\sum_{t=1}^{n_\ell}\varepsilon_{t,j}(u_1)\bigg\}\bigg\{\frac{1}{n}\sum_{t=1}^n\varepsilon_{t,j'}(v_1)\bigg\} -  \bigg\{\frac{1}{n_{\ell}}\sum_{t=1}^{n_\ell}\varepsilon_{t,j}(u_2)\bigg\}\bigg\{\frac{1}{n}\sum_{t=1}^n\varepsilon_{t,j'}(v_2)\bigg\}\bigg\|_2  \notag\\
\lesssim&\,   \bigg\| \frac{1}{n_{\ell}}\sum_{t=1}^{n_\ell}\varepsilon_{t,j}(u_1) \bigg\|_4 \bigg\|\frac{1}{n}\sum_{t=1}^n\big\{\varepsilon_{t,j'}(v_1)-\varepsilon_{t,j'}(v_2)\big\}\bigg\|_4  \notag\\
&\, +   \bigg\| \frac{1}{n}\sum_{t=1}^{n}\varepsilon_{t,j'}(v_2) \bigg\|_4 \bigg\|\frac{1}{n_\ell}\sum_{t=1}^{n_\ell}\big\{\varepsilon_{t,j}(u_1)-\varepsilon_{t,j}(u_2)\big\}\bigg\|_4    \,.
\end{align*}
By Conditions \ref{c.alpha} and \ref{c.subgaussian}, and   Lemma \ref{tail_chang} with $(B_n,c_n,{r}_1,{r})=(1, 0, 2, 1/3)$, we obtain $\|n_\ell^{-1}\sum_{t=1}^{n_\ell}\varepsilon_{t,j}(u_1)\|_4\lesssim n^{-1/2}$ and $\|n^{-1}\sum_{t=1}^{n}\varepsilon_{t,j'}(v_2)\|_4\lesssim n^{-1/2}$. Moreover, similar to the proof of \eqref{eq:diffvareps2}, we have
\begin{align*}
    \bigg\|\frac{1}{n}\sum_{t=1}^n\big\{\varepsilon_{t,j'}(v_1)-\varepsilon_{t,j'}(v_2)\big\}\bigg\|_4 \lesssim &~ n^{-1/2}|v_1-v_2|^{\kappa}\,,\\
    \bigg\|\frac{1}{n_\ell}\sum_{t=1}^{n_\ell}\big\{\varepsilon_{t,j}(u_1)-\varepsilon_{t,j}(u_2)\big\}\bigg\|_4\lesssim &~ n^{-1/2}|u_1-u_2|^{\kappa}\,,
\end{align*}
which implies that 
\begin{align}\label{bound.bar22}
    \|\bar{\vep}^{(2)}_{\ell,j,j'}(u_1,v_1)
-\bar{\vep}^{(2)}_{\ell,j,j'}(u_2,v_2)\|_2 \lesssim n^{-1}d\{(u_1,v_1),(u_2,v_2)\}\,.
\end{align} 
By taking similar arguments as in the proof of \eqref{bound.bar22}, we can conclude that
\begin{align}\label{bound.I234}
\sum_{k =2}^{4} \|\bar{\vep}^{(k)}_{\ell,j,j'}(u_1,v_1)-\bar{\vep}^{(k)}_{\ell,j,j'}(u_2,v_2)  \|_2 \lesssim n^{-1}d\{(u_1,v_1),(u_2,v_2)\}\,.
\end{align} 
Combining \eqref{i16}, \eqref{eq:diffvareps2}, and \eqref{bound.I234} yields 
\begin{align*}
    [\mathbb{E}\{|\tilde\cG_r(u_1,v_1)-\tilde\cG_r(u_2,v_2)|^2\}]^{1/2} \lesssim d\{(u_1,v_1),(u_2,v_2)\} 
\end{align*}
for any $r\in [Lp^2]$ and $(u_1,v_1),(u_2,v_2)\in\cU^2$, which yields \eqref{2nd.moment}. Hence, we
complete the proof of Lemma \ref{lem.dis.Gaussian}. 
$\hfill \Box$

\subsection{Proof of Proposition \ref{pro.Tn.H0}}\label{proof.pro.Tn.H0}

\subsubsection{Proof of Proposition \ref{pro.Tn.H0}(i)}

Recall  $\tilde{T}_{n} = \max_{\ell\in[L]}\max_{j,j'\in[p]}\sup_{(u,v)\in\cU^2}n^{1/2} | \widetilde\Sigma^{(\ell)}_{jj'}(u,v)  |$, 
with $\widetilde\Sigma^{(\ell)}_{jj'}(u,v) = n_{\ell}^{-1}\sum_{t=1}^{n_\ell} \{\vep_{t,j}(u)\\-\bar{\vep}_j(u)\} \{\vep_{t+\ell,j'}(v)-\bar{\vep}_{j'}(v)\}$. Analogously to \eqref{Tn.dis.G},  define the discretized version of $\tilde{T}_{n}$ as
\begin{align}\label{Tn.dis}
      \tilde{T}_{n}^{{\rm dis }}=\max_{\ell\in[L]}\max_{j,j'\in[p]}\max_{(u,v)\in\cU_M^2}n^{1/2} | \widetilde\Sigma^{(\ell)}_{jj'}(u,v)  |\,.
\end{align}
To prove   Proposition \ref{pro.Tn.H0}(i), we need the following lemmas. The proofs of Lemmas \ref{lemma_tn_dis} and \ref{lemma_tn_dis_GA}  are given in Sections \ref{proof.lemma_tn_dis} and \ref{proof.lemma_tn_dis_GA}, respectively. Recall $M\asymp (np)^{C_M}$ for a sufficiently large constant $C_M>0$.

\begin{lemma}\label{lemma_tn_dis}
Under Conditions~{\rm\ref{c.alpha}} and~{\rm\ref{c.subgaussian}}, 
and the null hypothesis $H_0$,  we have  $|\tilde{T}_{n}-\tilde{T}^{\rm dis}_{n}| =O_{\rm p}(n^{-1})$ provided that $\log (np)  \ll n^{1/3}$.
\end{lemma}

\begin{lemma}\label{lemma_tn_dis_GA}
Assume $p\geq n^{\upsilon}$ for some  constant $\upsilon>0$. Under Conditions {\rm\ref{c.alpha}}--{\rm \ref{c.subgaussian}} and the null hypothesis $H_0$,   we have    $\sup_{x\in \mathbb{R}} |\mathbb{P}(\tilde{T}^{\rm dis}_{n}\leq x)-\mathbb{P}(\tilde{T}^{\rm dis,G}_{n}\leq x)|=o(1)$, provided that $\log p  \ll n^{2/21}$.
\end{lemma}

For any $s >0$ and $x\in \mathbb{R}$, it holds that
\begin{align*}
    &\bbP(\tilde{T}_n\leq x) 
\leq \bbP(\tilde{T}_n^{\rm dis}\leq x+s) + \mathbb{P}(|\tilde{T}_{n}-\tilde{T}^{\rm dis}_{n}| > s)\\
&~~~\leq \mathbb{P}(\tilde{T}^{\rm dis,G}_{n}\leq x+s) + \sup_{y\in \mathbb{R}} |\mathbb{P}(\tilde{T}^{\rm dis}_{n}\leq y)-\mathbb{P}(\tilde{T}^{\rm dis,G}_{n}\leq y)|+\mathbb{P}(|\tilde{T}_{n}-\tilde{T}^{\rm dis}_{n}| > s)\\
&~~~\leq \bbP(\tilde{T}_n^{\rm G}\leq x) + \sup_{y\in\mathbb{R}}\bbP( |\tilde{T}^{\rm dis,G}_{n}-y|\leq s )+ \sup_{y\in \mathbb{R}} |\mathbb{P}(\tilde{T}^{\rm dis,G}_{n}\leq y)-\mathbb{P}(\tilde{T}^{\rm  G}_{n}\leq y)|\\
&~~~~~~+ \sup_{y\in \mathbb{R}} |\mathbb{P}(\tilde{T}^{\rm dis}_{n}\leq y)-\mathbb{P}(\tilde{T}^{\rm dis,G}_{n}\leq y)|+\mathbb{P}(|\tilde{T}_{n}-\tilde{T}^{\rm dis}_{n}| > s)\,.
\end{align*}
Likewise, we can also obtain the reverse inequality. Then we can conclude that 
\begin{align*}
    \sup_{x\in \mathbb{R}}|\mathbb{P}(\tilde{T}_{n} \leq x)-\mathbb{P}(\tilde{T}_{n}^{\rm G}\leq x)|\leq& \sup_{x\in \mathbb{R}}\bbP( |\tilde{T}^{\rm dis,G}_{n}-x|\leq s )+ \sup_{x\in \mathbb{R}} |\mathbb{P}(\tilde{T}^{\rm dis,G}_{n}\leq x)-\mathbb{P}(\tilde{T}^{\rm  G}_{n}\leq x)|\\
    &+\sup_{x\in \mathbb{R}} |\mathbb{P}(\tilde{T}^{\rm dis}_{n}\leq x)-\mathbb{P}(\tilde{T}^{\rm dis,G}_{n}\leq x)|+\mathbb{P}(|\tilde{T}_{n}-\tilde{T}^{\rm dis}_{n}| > s)\,.
\end{align*}
Taking $s=Cn^{-1}(\log p)^{1/2}$ for a sufficiently large constant $C>0$, 
\eqref{bound.4} yields
\begin{align*}
    \sup_{x\in \mathbb{R}}\bbP( |\tilde{T}^{\rm dis,G}_{n}-x|\leq s )   =o(1)\,,
\end{align*}
provided that $\log p\ll n$. Finally, combining this result with Lemmas~\ref{lem.dis.Gaussian}, 
\ref{lemma_tn_dis}, and~\ref{lemma_tn_dis_GA} proves 
Proposition~\ref{pro.Tn.H0}(i).  $\hfill \Box$

\subsubsection{Proof of Proposition \ref{pro.Tn.H0}(ii)}

Recall $\tilde{T}_{n}^{{\rm G*}}=\max_{r\in [Lp^2]}\sup_{(u,v)\in\cU^2}|\tilde{\cG}^*_r(u,v)|$, where
\begin{align*}
    \tilde{\boldsymbol\cG}^*(u,v) = \{\tilde{\cG}^*_1(u,v),\ldots,\tilde{\cG}^*_{Lp^2}(u,v)\}^{\T}= \frac{1}{n_L^{1/2}}\sum_{t=1}^{{n_L}}\varrho_t\{\tilde{\bfeta}_{t}(u,v)-\bar{\tilde{\bfeta}}(u,v)\}\,.
\end{align*}
Here, $\boldsymbol\varrho=(\varrho_1,\dots,\varrho_{{n_L}})^{\T}\sim \cN({\bf 0},\bXi)$  with $\Xi_{ij}= \mathcal{W}\{(i-j)/{b_{n}}\}$,   $
\tilde{\bfeta}_t(u,v) = \{({\rm vec} [\{\bvep_t(u)-\bar\bvep(u)\}\{\bvep_{t+1}(v)-\bar\bvep(v)\}^{\T}])^{\T},  \ldots,  ({\rm vec}[\{\bvep_t(u)-\bar\bvep(u)\}\{\bvep_{t+L}(v)-\bar\bvep(v)\}^{\T}] )^{\T}\}^{\T} $, and $\bar{\tilde{\bfeta}}(u,v)=n_L^{-1}\sum_{t=1}^{{n_L}}\tilde{\bfeta}_{t}(u,v)$. Analogously to \eqref{Tn.dis.G},  define the discretized version of $\tilde{T}_{n}^{{\rm G*}}$ as
\begin{align}\label{Tn.dis.Gstar}
      \tilde{T}_{n}^{{\rm dis ,G*}}=\max_{r\in [Lp^2]}\max_{(u,v)\in\cU_M^2}|\tilde{\cG}^*_r(u,v)|\,.
\end{align}
To prove Proposition \ref{pro.Tn.H0}(ii), we need the following lemmas.  The proofs of Lemmas \ref{lem:diff.Tn.Gstar} and \ref{disGAstar}  are given in Sections \ref{proof.lem:diff.Tn.Gstar} and \ref{proof.disGAstar}, respectively.   Recall $\widetilde{\mathcal{D}}_n=\{\bvep_t(\cdot)\}_{t\in [n]}$ and $b_n \asymp n^{\rho}$.

\begin{lemma}\label{lem:diff.Tn.Gstar}
Assume $p\geq n^{\upsilon}$ for some   constant $\upsilon>0$. Under Conditions   {\rm \ref{c.subgaussian}} and {\rm \ref{c.kernelF}}, it holds that 	
$	\mathbb{P}\{|\tilde{T}^{\rm G*}_{n}-\tilde{T}_{n}^{\rm dis,G*}|>C n^{-1}(\log {p})^{1/2}\,|\,\widetilde{\mathcal{D}}_n\}=o_{\rm p}(1)$ for some  constant $C>0$.
\end{lemma}

\begin{lemma}\label{disGAstar}
Assume $p\geq n^{\upsilon}$ for some   constant $\upsilon>0$. Under Conditions {\rm \ref{c.alpha}}--{\rm \ref{c.kernelF}}, it holds that $\sup_{x\in \mathbb{R}} |\mathbb{P}(\tilde{T}^{\rm dis,G*}_n\leq x\,|\,\widetilde{\mathcal{D}}_n) -\mathbb{P}(\tilde{T}^{\rm dis,G}_{n}\leq x)|=o_{\rm p}(1)$,  provided that $0<\rho<(\vartheta-1)/(3\vartheta-2)$ and   $\log p \ll n^\iota$ for some $\iota>0$ depending only on $(\rho,\vartheta)$.
\end{lemma}

Recall $\tilde p=pM$, $M \asymp (np)^{C_M}$ and $p\geq n^{\upsilon}$. By Lemma \ref{disGAstar}, Condition \ref{c.bvar} and the Nazarov's inequality 
\citep[see, e.g., Lemma~A.1 of][]{Cher2017_supp}, it holds that
\begin{align*}
\mathbb{P}(\tilde{T}^{\rm dis,G}_{n}\leq x)
=&\, \mathbb{P}(\tilde{T}^{\rm dis,G}_{n}\leq x-s)+ \mathbb{P}(x-s<\tilde{T}^{\rm dis,G}_{n}\leq x)\\
\leq&\, \mathbb{P}(\tilde{T}^{\rm dis,G*}_{n}\leq x-s\,|\,\widetilde{\mathcal{D}}_n)+ Cs(\log   {p} )^{1/2}+o_{\rm p}(1) \\
\leq&\, \mathbb{P}(\tilde{T}^{\rm G*}_{n}\leq x\,|\,\widetilde{\mathcal{D}}_n)+\mathbb{P}(|\tilde{T}^{\rm G*}_{n}-\tilde{T}_{n}^{\rm dis,G*}|>s\,|\,\widetilde{\mathcal{D}}_n)+ Cs(\log   {p} )^{1/2}+o_{\rm p}(1) 
\end{align*}
for any $x\in \mathbb{R}$ and $s>0$, provided that $0<\rho<(\vartheta-1)/(3\vartheta-2)$ and   $\log p \ll n^\iota$ for some $\iota>0$ depending only on $(\rho,\vartheta)$. Likewise, we can also obtain the reverse inequality.
Hence,
\begin{align*}
    \sup_{x\in \mathbb{R}}|\mathbb{P}(\tilde{T}_{n}^{\rm G*}\leq x\,|\,\widetilde{\mathcal{D}}_n)-\mathbb{P}(\tilde{T}^{\rm dis,G}_{n}\leq x)| \leq \mathbb{P}(|\tilde{T}^{\rm G*}_{n}-\tilde{T}_{n}^{\rm dis,G*}|>s\,|\,\widetilde{\mathcal{D}}_n) + Cs(\log   {p} )^{1/2}+o_{\rm p}(1)  \,.
\end{align*}
Choose $s=C'n^{-1}(\log p)^{1/2}$ for some sufficiently large constant
$C'>0$. Then, by Lemma~\ref{lem:diff.Tn.Gstar}, 
\begin{align*}
    \sup_{x\in \mathbb{R}}|\mathbb{P}(\tilde{T}_{n}^{\rm G*}\leq x\,|\,\widetilde{\mathcal{D}}_n)-\mathbb{P}(\tilde{T}^{\rm dis,G}_{n}\leq x)| = o_{\rm p}(1)  \,,
\end{align*}
provided that $0<\rho<(\vartheta-1)/(3\vartheta-2)$ and   $\log p \ll n^\iota$ for some $\iota>0$ depending only on $(\rho,\vartheta)$. This, together with Lemma~\ref{lem.dis.Gaussian}, yields Proposition~\ref{pro.Tn.H0}(ii). $\hfill \Box$

\section{Proofs of Technical Lemmas}

\subsection{Proof of Lemma \ref{tail_chang}} \label{sec:lemmaA3}
By Lemma L1 of \cite{ChangChenWu2023_supp} with $(\tilde n, r_1, r_2, r, \tilde B_{\tilde n}, \tilde{L}_{\tilde n}, \tilde{j}_{\tilde n})=(n,  r_1, 1,  r, B_n, 1, c_n)$,
it holds that $\mathbb{P}(|\sum_{t=1}^{n}Z_t|\geq x)\lesssim \exp\{-C(1+c_n)^{-1}B_n^{-2}n^{-1}x^2\}+\exp\{ -C(1+c_n)^{-{r}}B_n^{-{r}}x^{{r}} \}$ for any $x\geq 0$.
Notice that $\mathbb{E}(|X|^s)=s\int_{0}^{\infty}x^{s-1}\mathbb{P}(|X|>x)\,\md x$ for any finite positive integer $s$. If $c_n$ is fixed, following some basic calculations,    $\mathbb{E}(|\sum_{t=1}^{n}Z_t|^s)\lesssim B_n^sn^{s/2}$, which implies  $\|\sum_{t=1}^{n}Z_t\|_s\lesssim B_nn^{1/2}$. We have completed the proof of Lemma \ref{tail_chang}.
\hfill $\Box$

\subsection{Proof of Lemma \ref{tail_multi}}\label{sec:lemma_tail_multi}   
Throughout the proof, we use $C_{\gamma} \in (0,\infty)$ and $C_{\gamma,\beta} \in (0,\infty)$ to denote two positive constants depending only on $\gamma$ and $(\gamma,\beta)$, respectively, which may differ from line to line.    The proof follows the chaining argument used
in \cite{Ledoux1991_supp} and \cite{Talagrand2014_supp}. By separability, the relevant suprema may be taken over a countable separability set. Hence it is enough to prove the bounds uniformly over
finite subsets of this set and then let an increasing sequence of such
finite subsets exhaust it. We therefore assume throughout the proof that $\mathcal T$ is finite.   Recall $D = \sup_{\bs,\bt\in \cT} d(\bs,\bt) \leq C$. Let $\ell_0 $ be the largest integer $\ell$ such that $2^{-\ell} \geq  2D$, and $\ell_1$
be the smallest integer $\ell$ such that the open balls $B(\bs, 2^{-\ell})$ with center $\bs$ and
radius $2^{-\ell}$ in the metric $d$ are reduced to exactly one point, i.e.,
\begin{align*}
\ell_1 := \min \bigg\{ \ell \in \mathbb{Z}:  2^{-\ell} < \inf_{\bs,\bt\in\cT,\bs\neq \bt } d(\bs,\bt)  \bigg\}\,.
\end{align*}
For each $\ell\in\{\ell_0,\ldots,\ell_1\}$, let $\cT_{\ell}\subset \cT$  of cardinality $N_{2^{-\ell}}=N(\cT, d;2^{-\ell})$ such that the balls $\{B(\bs,2^{-\ell}): \bs\in \cT_{\ell}\}
$	form a covering of $\cT$. By induction, for each $\ell\in\{\ell_0+1,\ldots,\ell_1\}$, define map $h_{\ell}: \cT_{\ell} \rightarrow \cT_{\ell-1} $
such that $\bs\in B(h_{\ell}(\bs),2^{-\ell+1})$. Set then $k_{\ell}:\cT \rightarrow \cT_{\ell}$ by $k_{\ell}=h_{\ell+1}\tilde\circ\cdots\tilde\circ h_{\ell_1}$ for each $\ell\in\{\ell_0,\ldots,\ell_1\}$ ($k_{\ell_1}={\rm identity}$). Here, `$\tilde\circ$' denotes the composition of maps. Since $ 2^{-\ell_0}>D$, one open ball of radius $2^{-\ell_0}$
covers $\mathcal T$. Hence $N_{2^{-\ell_0}}=1$, and we take
$\mathcal T_{\ell_0}=\{\bs_0\}$.   From the above construction, we obtain the following properties:
\begin{align}
& d\{k_{\ell}(\bs),\bs\} \leq  2^{-\ell +1} ~~ \mbox{for all}~~\ell \in \{\ell_0,\ldots,\ell_1 \}~ \mbox{and}~ \bs\in\cT\, ; \label{P2} \\
& d\{ k_{\ell+1}(\bs), k_{\ell}(\bs) \} \leq 2^{-\ell}~~ \mbox{for all}~~\ell \in \{\ell_0,\ldots,\ell_1-1\}~ \mbox{and}~ \bs\in\cT\,;\label{P3}\\
& {\rm Card} \{  ( k_{\ell+1}(\bs), k_{\ell}(\bs) ) :\bs\in\cT  \} \leq N_{2^{-(\ell+1)}}~~\mbox{for all}~~ \ell \in \{\ell_0,\ldots,\ell_1-1\}\,. \label{P4}
\end{align}
In what follows, we provide the proofs of Lemmas \ref{tail_multi}(i) and \ref{tail_multi}(ii) in Sections \ref{proof.a} and \ref{proof.b}, respectively.

\subsubsection{Proof of Lemma \ref{tail_multi}(i)} \label{proof.a}

Notice that, for any $\bs\in\cT$,
\begin{align*}
    Z(\bs) - Z(\bs_0) = \sum_{\ell=\ell_0 }^{\ell_1-1}[ Z\{k_{\ell+1}(\bs)\}-Z\{k_{\ell }(\bs)\}]\,,
\end{align*} 
which implies that
\begin{align}\label{target.1}
    \sup_{\bs,\bt\in\cT}|Z(\bs)-Z(\bt)|\leq 2\sum_{\ell=\ell_0 }^{\ell_1-1}\sup_{\bs\in\cT}|Z\{k_{\ell+1}(\bs)\}-Z\{k_{\ell}(\bs)\}|\,.
\end{align} 
For some  integer $\widetilde{M}\in \{\ell_0 ,\ldots,\ell_1-1\}$  to be determined later, we have
\begin{align}\label{div.1}
&\sum_{\ell=\ell_0 }^{\ell_1-1}\sup_{\bs\in\cT}|Z\{k_{\ell+1}(\bs)\}-Z\{k_{\ell }(\bs)\}|\nonumber \\
&~~~=  \underbrace{\sum_{\ell=\ell_0 }^{\widetilde{M}}\sup_{\bs\in\cT}|Z\{k_{\ell+1}(\bs)\}-Z\{k_{\ell }(\bs)\}|}_{\rm I_1} +\underbrace{\sum_{\ell=\widetilde{M}+1}^{\ell_1-1}\sup_{\bs\in\cT}|Z\{k_{\ell+1}(\bs)\}-Z\{k_{\ell }(\bs)\}|}_{\rm I_2}\,.
\end{align}
Combining the union bound with \eqref{target.1} and \eqref{div.1}, we have, for any $x > 0$,
\begin{align}\label{div.or}
& 	\mathbb{P}\bigg\{  \sup_{\bs,\bt\in\cT}|Z(\bs)-Z(\bt)| > x\bigg\} \leq    \bbP\biggl({\rm I_1} > \frac{x}{4} \biggr)+\bbP\biggl({\rm I_2} > \frac{x}{4} \biggr)\,.
\end{align}
We claim that, for any $x > 0$,
\begin{align}
\bbP\biggl({\rm I_1} > \frac{x}{4} \biggr)  	\lesssim&~ c_{1n}(\widetilde{M}-\ell_0+1)N_{2^{-(\widetilde{M}+1)}} \exp\biggl\{ - c_{2n} C_{\gamma}\biggl(\frac{ x}{D} \biggr)^{\gamma}  \biggr\}\,,\label{Delta1}\\
\bbP\biggl({\rm I_2} > \frac{x}{4} \biggr)  	\lesssim&~ \frac{ c_{1n} C_{\gamma} \int_{0}^{2^{-(\widetilde{M}+1)}}\{ \log (e  N_{\epsilon}    )\}^{1/\gamma} \,\mbox{d}\epsilon}{c_{2n}^{1/\gamma}  x}\,.\label{Delta2}
\end{align}
For any $D_1\in(0,D\wedge 1)$, let 
$m_1=\lfloor-\log_2D_1\rfloor$. Then   
$D_1\leq 2^{-m_1}<2D_1$. Since $  D_1<D$ and $2^{-\ell_0}\ge 2D$, we have
$
2^{-m_1}<2D_1<2D\le 2^{-\ell_0}$, 
and hence $m_1\ge \ell_0+1$.   We choose $\widetilde M= \min\{ m_1-1, \ell_1-1\}$  so that $\widetilde M \in [\ell_0,\ell_1-1]$. Recall $\Lambda_u=1+\log_2(1/u)$. Since $D$ is upper bounded by a constant, it holds that   $\ell_0\geq -c_0$ for some constant $c_0>0$ and $\widetilde{M}-\ell_0+1 \lesssim \Lambda_{D_1}$. Notice that ${\rm I}_2=0$ when $\widetilde M = \ell_1-1$. Combining \eqref{div.or}--\eqref{Delta2} yields
\begin{align}\label{fin.1}
&\mathbb{P}\bigg\{ \sup_{\bs,\bt\in\cT}|Z(\bs)-Z(\bt)| > x\biggr\} \lesssim   \frac{c_{1n} C_{\gamma,1 }}{c_{2n}^{1/\gamma} x} \int_{0}^{D_1}\{\log ( e N_{\epsilon}    )\}^{1/\gamma} \,\mbox{d}\epsilon \nonumber\\
&~~~~~~~~~~~~~~~~~~~~~~~~~~~~~~+	  c_{1n} N_{D_1}\Lambda_{D_1}\exp\biggl\{ - c_{2n} C_{\gamma,2}\biggl(\frac{ x}{D} \biggr)^{\gamma}  \biggr\} 
\end{align}
for any $x > 0$.  Therefore,  the proof of Lemma~\ref{tail_multi}(i) will be complete once \eqref{Delta1} and \eqref{Delta2} are established.

\noindent\underline{Proof of \eqref{Delta1}.}
By \eqref{psi_multi} and \eqref{P3}, we know that for any $\bs\in\cT$ and $x\geq 0$,
\begin{align}\label{tail.psi_multi}
\mathbb{P}\big[ |Z\{k_{\ell+1}(\bs)\}-Z\{k_{\ell}(\bs)\}|>x\big]\leq c_{1n} \exp\biggl\{ -  c_{2n}\biggl(\frac{x}{ 2^{-\ell}} \biggr)^{ \gamma}  \biggr\}
\end{align}
for all $\ell \in \{\ell_0 ,\ldots, \ell_1-1\}$. Let $\bar{c}=\sum_{\ell=\ell_0}^{\widetilde{M}}2^{-\ell+2}$ and $c_{\ell}= 2^{-\ell+2}/\bar{c}$ for each $\ell\in\{ \ell_0,\ldots,\widetilde{M}\}$. Hence,   $\sum_{\ell=\ell_0}^{\widetilde{M}}c_{\ell}= 1$. By the definition of $\ell_0$, we know that $2^{-(\ell_0+2)}<D$. Then
\begin{align}\label{barc}
\bar{c} \leq 2^{2}\times \frac{2^{-  \ell_0 }}{1-2^{-1}}  < 2^5 D \,.
\end{align}
Together with the union bound, \eqref{P4},   \eqref{tail.psi_multi} and \eqref{barc}, we have
\begin{align*}
 \bbP\biggl({\rm I_1} > \frac{x}{4} \biggr)   
=&\,  \mathbb{P}\bigg[  \sum_{\ell=\ell_0}^{\widetilde{M}}\sup_{\bs\in\cT}|Z\{k_{\ell+1}(\bs)\}-Z\{k_{\ell}(\bs)\}| > \sum_{\ell=\ell_0}^{\widetilde{M}} \frac{ c_{\ell} x}{2^2}\bigg] \notag\\
\leq&\,  \sum_{\ell=\ell_0 }^{\widetilde{M}}{\mathbb{P}\bigg[   \sup_{\bs\in\cT}|Z\{k_{\ell+1}(\bs)\}-Z\{k_{\ell}(\bs)\}| >  \frac{c_{\ell} x}{2^2}\bigg] } \notag\\ 
\leq&\,   c_{1n}\sum_{\ell=\ell_0}^{\widetilde{M}}N_{2^{-(\ell+1)}} \exp\biggl\{ - c_{2n} \biggl(\frac{c_{\ell}x}{  2^{-\ell+2}} \biggr)^{ \gamma}  \biggr\}\notag\\
 \leq&\,   c_{1n}\sum_{\ell=\ell_0}^{\widetilde{M}}N_{2^{-(\ell+1)}} \exp\biggl\{ - c_{2n}C_{\gamma}\biggl(\frac{x}{D} \biggr)^{\gamma}  \biggr\}  \notag\\
\lesssim&\, c_{1n}({\widetilde{M}}-\ell_0+1)N_{2^{-(\widetilde{M}+1)}} \exp\biggl\{ - c_{2n}C_{\gamma}\biggl(\frac{ x}{D} \biggr)^{\gamma}  \biggr\}\,,%\label{boundI_multi}
\end{align*}
where   the last inequality follows from the fact that $N_{2^{-\ell}}$ is increasing in $\ell$. We complete the proof of \eqref{Delta1}.

\noindent\underline{Proof of \eqref{Delta2}.} For any $\ell \in \{ \widetilde{M}+1, \ldots, \ell_1-1 \}$, write $W_{\ell}=\sup_{\bs\in\cT}|Z\{k_{\ell+1}(\bs)\}-Z\{k_{\ell }(\bs)\}|$.
For some  $Q_{\ell} \gtrsim c_{2n}^{-1/\gamma}2^{-\ell}$,  \eqref{tail.psi_multi} implies that
\begin{align*}
\mathbb{E}(W_{\ell}) = &~	\mathbb{E}\{W_{\ell } I(W_{\ell}\leq Q_{\ell}) \}+\mathbb{E}\{ W_{\ell} I(W_{\ell}> Q_{\ell}) \}\\
\leq &~ \mathbb{E}\{W_{\ell } I(W_{\ell}\leq Q_{\ell}) \} + Q_{\ell}\mathbb{P}( W_{\ell} > Q_{\ell} )+ \int_{Q_{\ell}}^{\infty}\mathbb{P}( W_{\ell} > s )\, \mbox{d}s\\
\leq &~ Q_{\ell}+ c_{1n} N_{2^{-(\ell+1)}}\int_{Q_{\ell}}^{\infty}\exp\biggl\{ - c_{2n} \biggl(\frac{s}{  2^{-\ell}} \biggr)^{\gamma}  \biggr\}\,  \mbox{d}s\\
\leq &~  Q_{\ell}+ c_{1n} c_{2n}^{-1/\gamma} 2^{-\ell} N_{2^{-(\ell+1)}} \int_{  c_{2n} Q_{\ell}^{\gamma}  2^{ \gamma\ell }}^{\infty} e^{-t}\, \mbox{d}(t^{1/\gamma})\\
\leq &~ Q_{\ell} + C_{\gamma} c_{1n} c_{2n}^{-1} 2^{-\gamma\ell}  N_{2^{-(\ell+1)}}Q_\ell^{1-\gamma} \exp\biggl\{-  c_{2n}\biggl(\frac{Q_{\ell}}{  2^{-\ell}}  \biggr)^{ {\gamma}} \biggr \} \,.
\end{align*}
Selecting $Q_{\ell}^{\gamma} =c_{2n}^{-1} 2^{-\gamma\ell} \log \{ C_{\gamma} N_{2^{-(\ell+1)}}\}$, we have
\begin{align*}
    \mathbb{E}\bigg[\sup_{\bs\in\cT}|Z\{k_{\ell+1}(\bs)\}-Z\{k_{\ell}(\bs)\}|  \bigg] =\mathbb{E}(W_{\ell })  \leq  {C}_{\gamma} c_{1n} c_{2n}^{-1/\gamma}   2^{-\ell} [\log \{ e N_{2^{-(\ell+1)}}\}]^{1/\gamma} \,,
\end{align*} 
which  implies that
\begin{align}\label{integrate}
 &\sum_{\ell=\widetilde{M}+1}^{\ell_1-1}\mathbb{E}\bigg[ \sup_{\bs\in\cT}|Z\{k_{\ell+1}(\bs)\}-Z\{k_{\ell}(\bs)\}|  \bigg]
\leq    {C}_{\gamma}  c_{1n} c_{2n}^{-1/\gamma}   \sum_{\ell=\widetilde{M}+1}^{\ell_1-1} 2^{-\ell} [ \log \{ e N_{2^{-(\ell+1)}}\}]^{1/\gamma} \nonumber\\
 &~~~ \lesssim   {C}_{\gamma}c_{1n} c_{2n}^{-1/\gamma}   \sum_{\ell=\widetilde{M}+1}^{\ell_1-1} \int_{2^{-\ell-2}}^{2^{-\ell-1}}\{ \log ( e N_{\epsilon}    )\}^{1/\gamma} \, \mbox{d}\epsilon  
\lesssim    {C}_{\gamma}  c_{1n} c_{2n}^{-1/\gamma}    \int_{0}^{2^{-(\widetilde{M}+1)}} \{ \log( e N_{\epsilon}    ) \}^{1/\gamma}  \, \mbox{d}\epsilon \,.
\end{align}
This, together with Markov's inequality, yields
\begin{align*} %\label{boundII_multi}
\bbP\biggl({\rm I_2} > \frac{x}{4} \biggr)   \lesssim \frac{1}{x} \sum_{\ell=\widetilde{M}+1}^{\ell_1-1}\mathbb{E}\bigg[ \sup_{\bs\in\cT}|Z\{ k_{\ell+1}(\bs)\}-Z\{k_{\ell}(\bs)\}|  \bigg]
\lesssim \frac{ c_{1n} {C}_{\gamma}  \int_{0}^{2^{-(\widetilde{M}+1)}}\{ \log  ( e N_{\epsilon}    )\}^{1/\gamma}  \,\mbox{d}\epsilon}{ c_{2n}^{1/\gamma} x}
\end{align*}
for any $x >0$. This proves \eqref{Delta2} and hence completes the proof of Lemma~\ref{tail_multi}(i). \hfill $\Box$

\subsubsection{Proof of Lemma \ref{tail_multi}(ii)} \label{proof.b}

Throughout the proof, ratios of the form 
$|Z(\bs)-Z(\bt)|/\{d(\bs,\bt)\}^{\beta}$
are set to be zero when $\bs=\bt$. Hence, for notational simplicity, we write
$\sup_{\bs,\bt\in\cT}$ without explicitly excluding the case $\bs=\bt$.
For each $\ell \in \{ \ell_0,\ldots,\ell_1 -1\}$, define
\begin{align*}
\cS_{\ell} = \{ (\bs,\bt)\in \cT^2: 2^{-(\ell+1)}< d(\bs,\bt)\leq 2^{-\ell} \}\,.
\end{align*}
It holds that
\begin{align}\label{decomposition_pre}
\sup_{\bs,\bt\in\cT}\frac{|Z(\bs)-Z(\bt)|}{\{d(\bs,\bt)\}^{\beta}}  
\leq&\, \max_{\ell\in \{\ell_0,\ldots,\ell_1 -1\}} \sup_{(\bs,\bt)\in \cS_{\ell}}\frac{|Z(\bs)-Z(\bt)|}{\{d(\bs,\bt)\}^{\beta}} \notag\\
\leq&\, \sum_{\ell=\ell_0}^{\ell_1-1}2^{(\ell+1)\beta}\sup_{(\bs,\bt)\in \cS_{\ell}} |Z(\bs)-Z(\bt)|\,.
\end{align}
By triangle inequality and the decomposition $Z(\bs)- Z\{k_{\ell}(\bs)\}=Z\{k_{\ell_1}(\bs)\}- Z\{k_{\ell}(\bs)\} = \sum_{j=\ell+1}^{\ell_1}[Z\{k_{j}(\bs)\}- Z\{k_{j-1}(\bs)\}]$ for any $\bs \in\cT$, we have
\begin{align}\label{Z_uv}
|Z(\bs)-Z(\bt)| \leq&~  |Z\{k_{\ell}(\bs)\}- Z\{k_{\ell}(\bt)\}|+|Z(\bs)- Z\{k_{\ell}(\bs)\}|+|Z(\bt)- Z\{k_{\ell}(\bt)\} \notag|\\
\leq &~  |Z\{k_{\ell}(\bs)\}- Z\{k_{\ell}(\bt)\}|+ 2\sup_{\bs\in\cT}|Z(\bs)- Z\{k_{\ell}(\bs)\}| \notag\\
\leq &~  |Z\{k_{\ell}(\bs)\}- Z\{k_{\ell}(\bt)\}|+ 2\sup_{\bs\in\cT}\bigg\{  \sum_{j=\ell+1}^{\ell_1}|Z\{k_{j}(\bs)\}- Z\{k_{j-1}(\bs)\}|\bigg\} \notag \\ 
\leq &~  |Z\{k_{\ell}(\bs)\}- Z\{k_{\ell}(\bt)\}|+ 2\sum_{j=\ell+1}^{\ell_1} \sup_{\bs\in\cT} |Z\{k_{j}(\bs)\}- Z\{k_{j-1}(\bs)\}|   
\end{align}
for each pair $(\bs,\bt)\in \cT^2$. Then  by \eqref{P2}, for each pair $(\bs,\bt)\in	\cS_{\ell}$,
\begin{align*}
d\{ k_{\ell}(\bs),k_{\ell}(\bt)\}  \leq d\{ k_{\ell}(\bs),\bs\}  +d\{  k_{\ell}(\bt),\bt\}  +d(  \bs ,\bt) \leq 2^{-\ell+3}\,.
\end{align*}
Following from \eqref{Z_uv}, it holds that
\begin{align*}
\sup_{(\bs,\bt)\in \cS_{\ell}} |Z(\bs)-Z(\bt)| \leq \max_{\substack{\bs,\bt\in\cT_{\ell},\\ d( \bs ,\bt) \leq 2^{-\ell+3}}}|Z(\bs)-Z(\bt)|+ 2\sum_{j=\ell+1}^{\ell_1}\sup_{\bs\in\cT}|Z\{k_{j}(\bs)\}- Z\{k_{j-1}(\bs)\}|\,.
\end{align*}
Together with \eqref{decomposition_pre},  we have
\begin{align*}%\label{decomposition}
\sup_{\bs,\bt\in\cT}\frac{|Z(\bs)-Z(\bt)|}{\{d(\bs,\bt)\}^{\beta}}  
\leq &\,  \underbrace{\sum_{\ell=\ell_0}^{\ell_1-1}2^{(\ell+1)\beta}\max_{\substack{\bs,\bt\in\cT_{\ell} ,\\ d( \bs ,\bt) \leq 2^{-\ell+3}}}  |Z(\bs)-Z(\bt)|  }_{\rm II_1} \notag\\
&\,  +\underbrace{\sum_{\ell=\ell_0}^{\ell_1-1}2^{(\ell+1)\beta+1}\sum_{j=\ell+1}^{\ell_1}\sup_{\bs\in\cT}   |Z\{k_{j}(\bs)\}- Z\{k_{j-1}(\bs)\}|  }_{\rm II_2}\,,
\end{align*}
which implies that, for any $x > 0$,
\begin{align}\label{tail.target}
\bbP\biggl[  \sup_{\bs,\bt\in\cT}\frac{|Z(\bs)-Z(\bt)|}{\{d(\bs,\bt)\}^{\beta}} >x  \biggr] \leq \bbP\biggl({\rm II}_{1}>\frac{x}{2} \biggr)  + \bbP\biggl({\rm II}_{2}>\frac{x}{2} \biggr)\,.
\end{align}
 We claim that, for any $x > 0$,
\begin{align}
\bbP\biggl({\rm II}_{1}>\frac{x}{2} \biggr)
\lesssim &\,          \frac{c_{1n} C_{\gamma,\beta,1}}{c_{2n}^{1/\gamma}x} \int_{0}^{D_1}\frac{\{ \log( e  N_{\epsilon}    )\}^{1/\gamma}}{ \epsilon^{\beta}} \,\mbox{d}\epsilon   + c_{1n}  N_{D_1}^2\Lambda_{D_1} \exp\biggl\{ - c_{2n} C_{\gamma,\beta,2}\biggl(\frac{ x}{D^{1-\beta}} \biggr)^{ {\gamma}}  \biggr\}\,,\label{II1}\\
\bbP\biggl({\rm II}_{2}>\frac{x}{2} \biggr)  \lesssim &\,    \frac{ c_{1n}C_{\gamma,\beta,1} }{c_{2n}^{1/\gamma}x} \bigg[\int_{0}^{D_1}\frac{\{\log(  e N_{\epsilon}    )\}^{1/\gamma}}{\epsilon^{\beta}}\,\mbox{d}\epsilon +  \frac{ \int_{0}^{D_2} \{\log( e  N_{\epsilon}    )\}^{1/\gamma}\,\mbox{d}\epsilon}{D_1^{\beta}} \bigg]  \notag\\
&\,  + c_{1n} N_{D_2} \Lambda_{D_1}\Lambda_{D_2}  \exp\biggl\{ - c_{2n} C_{\gamma,\beta,2} \biggl(\frac{ x}{D^{1-\beta}} \biggr)^{ {\gamma}}  \biggr\}\,,\label{II2}
\end{align}
respectively,
where  $D_1,D_2 \in (0, D\wedge 1)$ with $D_2<D_1$. 
By \eqref{tail.target}--\eqref{II2}, we can conclude that, for any $x > 0$,
\begin{align*}
&~\bbP\biggl[  \sup_{\bs,\bt\in\cT}\frac{|Z(\bs)-Z(\bt)|}{\{d(\bs,\bt)\}^{\beta}} >x  \biggr] 
\lesssim  \frac{c_{1n} {C}_{\gamma,\beta,1}}{c_{2n}^{1/\gamma}x}   \bigg[\int_{0}^{D_1 }\frac{\{\log( e  N_{\epsilon}    )\}^{1/\gamma}}{\epsilon^{\beta}}\, \mbox{d}\epsilon + \frac{\int_{0}^{D_2} \{\log( e  N_{\epsilon}    )\}^{1/\gamma}\,\mbox{d}\epsilon }{D_1^{\beta}} \bigg]\\
&~~~~~~~~~~~~~~~~~~~~~~~~~~~~~~~~~+ c_{1n}\big \{N_{D_1}^2\Lambda_{D_1}+N_{D_2} \Lambda_{D_1}\Lambda_{D_2}\big \} \exp\biggl\{ -c_{2n} {C}_{\gamma,\beta,2}\biggl(\frac{ x}{D^{1-\beta}} \biggr)^{ {\gamma}}  \biggr\}\,.
\end{align*}
 Finally, the proof of Lemma~\ref{tail_multi}(ii) will be completed once \eqref{II1} and \eqref{II2} are established.

\noindent\underline{Proof of \eqref{II1}.}  For some integer $M_1\in \{ \ell_0,\ldots,\ell_1-1 \}$, we have
\begin{align*}
{\rm II}_{1} =&\,  \underbrace{\sum_{\ell=\ell_0}^{M_1}2^{(\ell+1)\beta}\max_{\substack{\bs,\bt\in\cT_{\ell},\\ d( \bs ,\bt) \leq 2^{-\ell+3}}}|Z(\bs)-Z(\bt)|}_{\rm II_{11}} 
+  \underbrace{\sum_{\ell=M_1+1}^{\ell_1-1}2^{(\ell+1)\beta}\sup_{\substack{\bs,\bt\in\cT_{\ell},\\ d( \bs ,\bt) \leq 2^{-\ell+3}}}|Z(\bs)-Z(\bt)|}_{\rm II_{12}}\,.
\end{align*}
Notice that $\rm II_{12} =0$ when $M_1 = \ell_1-1$. The union bound yields that
\begin{align}\label{div.I1}
\bbP\biggl({\rm II}_{1}> \frac{x}{2} \biggr) \leq &~\bbP\biggl({\rm II}_{11}> \frac{x}{4} \biggr) + \bbP\biggl({\rm II}_{12}>\frac{x}{4} \biggr) 
\end{align}
 for any $x> 0$. For any $\bs,\bt\in\cT_{\ell}$ satisfying $d( \bs ,\bt) \leq 2^{-\ell+3}$, \eqref{psi_multi} implies that
\begin{align}\label{tail.divz}
\mathbb{P}\{ |Z(\bs)-Z(\bt)|>x\}
\leq c_{1n} \exp\biggl[ - c_{2n}\biggl\{\frac{x}{  d(\bs,\bt)} \biggr\}^{ {\gamma}}  \biggr]\leq c_{1n} \exp\biggl\{ - c_{2n}\biggl(\frac{x}{  2^{-\ell+3}} \biggr)^{ {\gamma}}  \biggr\}
\end{align}
for any $x\geq 0$. Let $\tilde{c}'=\sum_{\ell=\ell_0 }^{M_1}2^{-(1-\beta)\ell+\beta+5} $ and $c_{\ell}'= 2^{-(1-\beta)\ell+\beta+5 }/\tilde{c}'$ for  $\ell \in  \{ \ell_0 ,\ldots, M_1\}$. Then   $\sum_{\ell=\ell_0 }^{M_1}c_{\ell}'= 1$.  By the definition of $\ell_0$, we know that $2^{-(\ell_0+2)}<D$. Hence,
\begin{align}\label{tildec}
\tilde{c}' \leq 2^{\beta+5}\times \frac{2^{-(1-\beta)\ell_0}}{1-2^{-(1-\beta)}} = \frac{2^6}{1-2^{-(1-\beta)}}\times 2^{-(1-\beta)(\ell_0+1)} < \frac{2^7 D^{1-\beta}}{1-2^{-(1-\beta)}} \,.
\end{align}
Notice that ${\rm Card}(\cT_{\ell})= N_{2^{-\ell}}$ for any $\ell \in \{\ell_0,\ldots,\ell_1\}$. Following from the union bound, \eqref{tail.divz} and \eqref{tildec}, it holds that
\begin{align}
\bbP\biggl({\rm II}_{11}> \frac{x}{4} \biggr) = &\,	 {\mathbb{P}\biggl\{\sum_{\ell=\ell_0}^{M_1}2^{(\ell+1)\beta}\max_{\substack{\bs,\bt\in\cT_{\ell},\\ d( \bs ,\bt) \leq 2^{-\ell+3}}}|Z(\bs)-Z(\bt)|> \sum_{\ell=\ell_0 }^{M_1} \frac{c_{\ell}' x}{2^2}\biggr\}} \notag\\
\leq &\, \sum_{\ell=\ell_0 }^{M_1}{\mathbb{P}\biggl\{\max_{\substack{\bs,\bt\in\cT_{\ell},\\ d( \bs ,\bt) \leq 2^{-\ell+3}}}|Z(\bs)-Z(\bt)| >  \frac{c_{\ell}' x}{2^{(\ell+1)\beta+2}}\biggr\}}\notag\\
\leq &\, \sum_{\ell=\ell_0 }^{M_1}N_{2^{-\ell}}^2\max_{\substack{\bs,\bt\in\cT_{\ell},\\ d( \bs ,\bt) \leq 2^{-\ell+3}}}{\mathbb{P}\biggl\{|Z(\bs)-Z(\bt)| >  \frac{c_{\ell}' x}{2^{(\ell+1)\beta+2}}\biggr\}}\notag\\
\leq &\, c_{1n} \sum_{\ell=\ell_0 }^{M_1}N_{2^{-\ell}}^2 \exp\biggl[ - c_{2n} \bigg\{ \frac{c_{\ell}' x}{  2^{-(1-\beta)\ell+\beta+5}} \bigg\}^{ {\gamma}}  \biggr] \notag\\
\lesssim  &\, c_{1n} ({M_1}-\ell_0+1)N_{2^{-M_1}}^2 \exp\biggl\{ - c_{2n} C_{\gamma,\beta}\biggl(\frac{ x}{D^{1-\beta}} \biggr)^{ {\gamma}}  \biggr\}\,. \label{boundI11}
\end{align}
Moreover, by Markov's inequality, we have
\begin{align}\label{tail.I12_pre}
\bbP\biggl({\rm II}_{12}>\frac{x}{4} \biggr) \lesssim  \frac{1}{x}\sum_{\ell=M_1+1}^{\ell_1-1}2^{(\ell+1)\beta }\,\bbE\Bigg\{\max_{\substack{\bs,\bt\in\cT_{\ell},\\ d( \bs ,\bt) \leq 2^{-\ell+3}}}|Z(\bs)-Z(\bt)|\Bigg\} 
\end{align}
for any $x> 0$. For any $\ell \in \{M_1+1,\ldots,\ell_1-1\}$, let $W_{\ell,1} = \max_{ {\bs,\bt\in\cT_{\ell}, d( \bs ,\bt) \leq 2^{-\ell+3}}}|Z(\bs)-Z(\bt)|$.
For some  $Q_{\ell,1}\gtrsim c_{2n}^{-1/\gamma}2^{-\ell+3}$, following from \eqref{tail.divz}, 
\begin{align}\label{moment.disz}
\bbE (W_{\ell,1} )  = &\, \mathbb{E}\{W_{\ell,1} I( W_{\ell,1}\leq Q_{\ell,1})\}+\mathbb{E}\{W_{\ell,1} I(W_{\ell,1}> Q_{\ell,1})\}\notag\\
\leq &\,  Q_{\ell,1} + \int_{Q_{\ell,1}}^{\infty} \mathbb{P} (W_{\ell,1}> s  ) \, \mbox{d}s \notag\\
\leq  &\, Q_{\ell,1}+ c_{1n} N_{2^{-\ell}}^2\int_{Q_{\ell,1}}^{\infty}\exp\biggl\{ - c_{2n} \biggl(\frac{s}{  2^{-\ell+3}} \biggr)^{ {\gamma}}  \biggr\}  \,\mbox{d}s \notag\\
\leq &\,  Q_{\ell,1}+ c_{1n} c_{2n}^{-1/\gamma} 2^{-\ell+3}N_{2^{-\ell}}^2   \int_{  c_{2n} Q_{\ell,1}^{\gamma}  2^{-\gamma(-\ell+3)} }^{\infty} e^{-t}\,\mbox{d} (t^{1/\gamma}) \notag\\
\leq   &~  Q_{\ell,1} + C_{\gamma}c_{1n} c_{2n}^{-1}  2^{\gamma(-\ell+3)} N_{2^{-\ell}}^2 Q_{\ell,1}^{1-\gamma} \exp\biggl\{-  c_{2n} \biggl(\frac{Q_{\ell,1}}{  2^{-\ell+3}}  \biggr)^{ {\gamma}} \bigg \} \,.
\end{align}
Let $Q_{\ell,1}^{\gamma} =c_{2n}^{-1}2^{\gamma(-\ell+3)} \log  (C_{\gamma }N_{2^{-\ell}}^2 )   $. Then we have
\begin{align*}
    \bbE\Bigg\{ \max_{\substack{\bs,\bt\in\cT_{\ell},\\ d( \bs ,\bt) \leq 2^{-\ell+3}}}|Z(\bs)-Z(\bt)|\Bigg\}=\bbE (W_{\ell,1} )  \leq {C}_{\gamma }  c_{1n} c_{2n}^{-1/\gamma}   2^{-\ell} \{ \log(  e N_{2^{-\ell}})\}^{1/\gamma} \,.
\end{align*}
We can conclude that
\begin{align*}
\sum_{\ell=M_1+1}^{\ell_1-1}2^{(\ell+1)\beta }\bbE\Biggl\{ \max_{\substack{\bs,\bt\in\cT_{\ell},\\ d( \bs ,\bt) \leq  2^{-\ell+3}}}|Z(\bs)-Z(\bt)|\Bigg\} 
\leq&\,    c_{1n} c_{2n}^{-1/\gamma}  {C}_{\gamma ,\beta}  \sum_{\ell=M_1+1}^{\ell_1} 2^{-(1-\beta)\ell}\{ \log( e  N_{2^{-\ell}}  )\}^{1/\gamma} \\
\overset{({\rm i})}{\lesssim}&\,  c_{1n} c_{2n}^{-1/\gamma}  {C}_{\gamma ,\beta}    \sum_{\ell>M_1} \int_{2^{-\ell-1}}^{2^{-\ell}}\epsilon^{-\beta}\{ \log( e  N_{\epsilon}     ) \}^{1/\gamma} \,\mbox{d}\epsilon \\
\lesssim&\,    c_{1n} c_{2n}^{-1/\gamma}  {C}_{\gamma ,\beta}   \int_{0}^{2^{-M_1}}\epsilon^{-\beta} \{\log(  e N_{\epsilon}     )\}^{1/\gamma} \,\mbox{d}\epsilon\,,
\end{align*}
where  ${\rm (i)}$ follows from
\begin{align*}
\int_{2^{-\ell-1}}^{2^{-\ell}}\epsilon^{-\beta}\{ \log( e  N_{\epsilon}     )\}^{1/\gamma}  \,\mbox{d}\epsilon   \geq 2^{\ell \beta} \{\log (  e N_{2^{-\ell}})\}^{1/\gamma}     \times 	\int_{2^{-\ell-1}}^{2^{-\ell}} 1\, \mbox{d}\epsilon = 2^{-(1-\beta)\ell-1 } \{ \log( e  N_{2^{-\ell}}     )\}^{1/\gamma} \,.
\end{align*}
Together with \eqref{tail.I12_pre}, we have
\begin{align}\label{boundI12}
\bbP\biggl({\rm II}_{12}>\frac{x}{4} \biggr) \lesssim  \frac{c_{1n}{C}_{\gamma,\beta} }{c_{2n}^{1/\gamma}x}   \int_{0}^{2^{-M_1}}\frac{\{\log ( e  N_{\epsilon}    )\}^{1/\gamma}}{\epsilon^{\beta}} \,\mbox{d}\epsilon \,.
\end{align}
Combining with \eqref{div.I1}, \eqref{boundI11} and \eqref{boundI12}, we have, for any $x>0$, 
\begin{align*}
\bbP\biggl({\rm II}_{1}> \frac{x}{2} \biggr)
\lesssim&~     \frac{c_{1n}{C}_{\gamma,\beta} }{c_{2n}^{1/\gamma}x}   \int_{0}^{2^{-M_1}}\frac{\{\log (  e N_{\epsilon}    )\}^{1/\gamma}}{\epsilon^{\beta}} \,\mbox{d}\epsilon \\
&~+  c_{1n} ({M_1}-\ell_0+1)N_{2^{-M_1}}^2 \exp\biggl\{ - c_{2n} C_{\gamma,\beta}\biggl(\frac{ x}{D^{1-\beta}} \biggr)^{ {\gamma}}  \biggr\}\,.
\end{align*}
Similar to the proof of \eqref{fin.1}, by taking $M_1=\min \{ \lfloor -\log_2  D_1  \rfloor ,\ell_1-1\}$ for some $ 0< D_1< D\wedge 1$, we can conclude that
\begin{align*}
\bbP\biggl({\rm II}_{1}> \frac{x}{2} \biggr)
\lesssim &  \frac{c_{1n}{C}_{\gamma,\beta,1}  }{c_{2n}^{1/\gamma}x}   \int_{0}^{D_1}\frac{\{ \log( e  N_{\epsilon}  )\}^{1/\gamma}}{\epsilon^{\beta}} \,\mbox{d}\epsilon + c_{1n} N_{D_1}^2\Lambda_{D_1} \exp\biggl\{ - c_{2n} C_{\gamma,\beta,2}\biggl(\frac{ x}{D^{1-\beta}} \biggr)^{ {\gamma}}  \biggr\} 
\end{align*}
 for any $x> 0$, which completes the proof of \eqref{II1}.

\noindent\underline{Proof of \eqref{II2}.} 
Let $M_2,M_3$ be integers satisfying $\ell_0 \le M_2  \le \ell_1-1$ and $ M_2  \le M_3\le \ell_1$. We adopt the convention that any sum with the lower limit exceeding the upper
limit is zero. It holds that
\begin{align*}
{\rm II_2}=&\, \underbrace{\sum_{\ell=\ell_0}^{M_2}2^{(\ell+1)\beta+1}\sum_{j=\ell+1}^{M_3}\sup_{\bs\in\cT}|Z\{k_{j}(\bs)\}- Z\{k_{j-1}(\bs)\}|}_{\rm II_{21}} \notag\\
&\, +\underbrace{\sum_{\ell=M_2+1}^{\ell_1-1}2^{(\ell+1)\beta+1}\sum_{j=\ell+1}^{\ell_1}\sup_{\bs\in\cT}|Z\{k_{j}(\bs)\}- Z\{k_{j-1}(\bs)\}|}_{\rm II_{22}} \notag\\
&\, +\underbrace{\sum_{\ell=\ell_0}^{M_2}2^{(\ell+1)\beta+1}\sum_{j=M_3+1}^{\ell_1}\sup_{\bs\in\cT}|Z\{k_{j}(\bs)\}- Z\{k_{j-1}(\bs)\}|}_{\rm II_{23}}\,.
\end{align*}
Notice that $\rm II_{23} =0$ when $M_3 = \ell_1$. By the union bound, we have  for any $x > 0$,
\begin{align}\label{div.I2}
\bbP\biggl({\rm II}_{2}> \frac{x}{2} \biggr)   \leq &~\bbP\biggl({\rm II}_{21}> \frac{x}{4} \biggr) + \bbP\biggl({\rm II}_{22}>\frac{x}{8} \biggr) + \bbP\biggl({\rm II}_{23}>\frac{x}{8} \biggr) \,.
\end{align}
Following from  \eqref{psi_multi} and \eqref{P3}, for  $j\in\{ \ell_0+1,\ldots,\ell_1\}$, we have
\begin{align}\label{tail.divz.new}
\mathbb{P}\big[ |Z\{k_{j}(\bs)\}- Z\{k_{j-1}(\bs)\}|>x\big]
\leq &\, c_{1n} \exp\biggl( - c_{2n} \biggl[\frac{x}{  d\{ k_{j}(\bs),k_{j-1}(\bs)\}} \biggr]^{ {\gamma}}  \biggr) \notag\\
\leq &\, c_{1n} \exp\biggl\{ - c_{2n} \biggl( \frac{x}{  2^{-j+1}} \biggr)^{ {\gamma}}  \biggr\}  
\end{align}
for any  $x\geq 0$. Let $\check{c}=\sum_{\ell=\ell_0}^{M_2}\sum_{j=\ell+1}^{M_3}2^{-j+\beta\ell+\beta+4}$ and $c_{\ell,j}= 2^{-j+\beta\ell+\beta+4}/\check{c}$ for each $\ell\in\{\ell_0,\ldots, M_2\}$ and $j\in\{ \ell+1,\ldots,M_3\}$ such that   $\sum_{\ell=\ell_0}^{M_2}\sum_{j=\ell+1}^{M_3}c_{\ell,j}= 1$.   By the definition of $\ell_0$ and \eqref{tildec}, we have
\begin{align*}
\check{c} = \sum_{\ell=\ell_0}^{M_2}\bigg\{2^{ \beta\ell+\beta+4} \bigg(\sum_{j=\ell+1}^{M_3}2^{-j }\bigg)\bigg\} 
\leq \sum_{\ell=\ell_0}^{M_2}2^{ \beta\ell+\beta+4-\ell }  < \frac{2^6D^{1-\beta}}{1-2^{-(1-\beta)}} \,.
\end{align*}
By  \eqref{P4}  and \eqref{tail.divz.new}, we have
\begin{align}
\bbP\biggl({\rm II}_{21}> \frac{x}{4} \biggr) 
= &\,	 {\mathbb{P}\biggl[\sum_{\ell=\ell_0}^{M_2}2^{(\ell+1)\beta+1}\sum_{j=\ell+1}^{M_3}\sup_{\bs\in\cT}|Z\{k_{j}(\bs)\}- Z\{k_{j-1}(\bs)\}|> \sum_{\ell=\ell_0}^{M_2} \sum_{j=\ell+1}^{M_3}\frac{c_{\ell,j} x}{2^2}\biggr]} \notag\\
\leq &\, \sum_{\ell=\ell_0}^{M_2}{\mathbb{P}\bigg[\sum_{j=\ell+1}^{M_3}\sup_{\bs\in\cT}|Z\{k_{j}(\bs)\}- Z\{k_{j-1}(\bs)\}| >  \sum_{j=\ell+1}^{M_3}\frac{c_{\ell,j} x}{2^{(\ell+1)\beta+3}}\bigg]}\notag\\
\leq &\, \sum_{\ell=\ell_0}^{M_2}\sum_{j=\ell+1}^{M_3}N_{2^{-j}} \sup_{\bs\in\cT}{\mathbb{P}\bigg[|Z\{k_{j}(\bs)\}- Z\{k_{j-1}(\bs)\}| >  \frac{c_{\ell,j} x}{2^{(\ell+1)\beta+3}}\bigg]}\notag\\
\leq &\,  c_{1n}\sum_{\ell=\ell_0}^{M_2}\sum_{j=\ell+1}^{M_3}N_{2^{-j}} \exp\biggl\{ - c_{2n}\biggl(\frac{c_{\ell,j}x}{  2^{-j+\beta\ell+\beta+4}} \biggr)^{ {\gamma}}  \biggr\} \notag\\
\lesssim  &~ c_{1n} ({M_2}-\ell_0+1)(M_3-\ell_0+1)N_{2^{-M_3}}  \exp\biggl\{ - c_{2n} C_{\gamma,\beta} \biggl(\frac{ x}{D^{1-\beta}} \biggr)^{ {\gamma}}  \biggr\}\,.\label{boundI21}
\end{align}
For  $j\in \{\ell_0+1,\ldots,\ell_1\}$, denote $W_{j,2} =\sup_{\bs\in\cT}|Z\{k_{j}(\bs)\}- Z\{k_{j-1}(\bs)\}|$. Given some  $Q_{j,2} \gtrsim c_{2n}^{-1/\gamma}2^{-j+1}$,  \eqref{P4}  and \eqref{tail.divz.new}  imply that 
\begin{align*}
\bbE(W_{j,2})= &\,	\mathbb{E}\{W_{j,2} I(W_{j,2} \leq Q_{j,2})\}  +\mathbb{E}\{W_{j,2} I(W_{j,2} > Q_{j,2})\}\\
\leq &\,  Q_{j,2} + \int_{Q_{j,2}}^{\infty}\mathbb{P} (W_{j,2}> s  )\, \mbox{d}s\\
\leq &\,  Q_{j,2}+ c_{1n} N_{2^{-j}} \int_{Q_{j,2}}^{\infty}\exp\biggl\{ - c_{2n} \biggl(\frac{s}{  2^{-j+1}} \biggr)^{ {\gamma}}  \biggr\} \, \mbox{d}s\\
\leq &\,  Q_{j,2}+ c_{1n} c_{2n}^{-1/\gamma} 2^{-j+1} N_{2^{-j}}\int_{ c_{2n}Q_{j,2}^{ {\gamma}}   2^{-\gamma(-j+1)}  }^{\infty} e^{-t} \,  \mbox{d}(t^{1/\gamma})\\
\leq  &\, Q_{j,2} + C_{\gamma}c_{1n} c_{2n}^{-1}  N_{2^{-j}}2^{\gamma(-j+1)} Q_{j,2}^{1-\gamma} \exp\biggl\{-  c_{2n} \biggl(\frac{Q_{j,2}}{  2^{-j+1}}  \biggr)^{ {\gamma}} \bigg \}  \,.
\end{align*}
Let $Q_{j,2}^{\gamma} =c_{2n}^{-1}2^{\gamma(-j+1)}  \log  (C_{\gamma} N_{2^{-j}} )  $. Then
\begin{align*}
    \bbE\bigg[\sup_{\bs\in\cT}|Z\{k_{j}(\bs)\}- Z\{k_{j-1}(\bs)\}|\bigg] =\bbE(W_{j,2})\leq {C}_{\gamma} c_{1n} c_{2n}^{-1/\gamma}    2^{-j} \{\log( e   N_{2^{-j}} )\}^{1/\gamma} \,,
\end{align*}
 which implies that
\begin{align*}
\bbE({\rm II}_{22})=&\, \sum_{\ell=M_2+1}^{\ell_1-1}2^{(\ell+1)\beta+1}\sum_{j=\ell+1}^{\ell_1}\bbE\bigg[ \sup_{\bs\in\cT}|Z\{k_{j}(\bs)\}- Z\{k_{j-1}(\bs)\}|\bigg] \\
\leq &\,   c_{1n} c_{2n}^{-1/\gamma}  {C}_{\gamma}  \sum_{\ell=M_2+1}^{\ell_1-1}2^{(\ell+1)\beta+1}
\sum_{j=\ell+1}^{\ell_1} 2^{-j} \{\log( e  N_{2^{-j}}  )\}^{1/\gamma} \\
=  &\,  c_{1n} c_{2n}^{-1/\gamma}  {C}_{\gamma}  \sum_{j=M_2+2}^{\ell_1}2^{-j} \{\log(  e N_{2^{-j}}  )\}^{1/\gamma} \sum_{\ell=M_2+1}^{j-1}2^{(\ell+1)\beta+1}\\
\lesssim &\, c_{1n} c_{2n}^{-1/\gamma} {C}_{\gamma,\beta} \sum_{j=M_2+2}^{\ell_1}2^{-(1-\beta)j} \{\log(  e N_{2^{-j}} )\}^{1/\gamma} \\
\lesssim &\, c_{1n} c_{2n}^{-1/\gamma} {C}_{\gamma,\beta}  \int_{0}^{2^{-M_2}}\epsilon^{-\beta} \{ \log( e N_{\epsilon}  )\}^{1/\gamma} \,\mbox{d}\epsilon\,.
\end{align*}
By Markov's inequality, for any $x> 0$,
\begin{align}\label{boundI22}
\bbP\biggl({\rm II}_{22}>\frac{x}{8} \biggr) \lesssim   \frac{ c_{1n} {C}_{\gamma,\beta} }{c_{2n}^{1/\gamma}x} \int_{0}^{2^{-M_2}} \frac{\{ \log( e  N_{\epsilon}    )\}^{1/\gamma}}{\epsilon^{\beta}}\,\mbox{d}\epsilon \,.
\end{align}
Similar to \eqref{integrate}, we can also obtain
\begin{align*}
\bbE({\rm II}_{23})=&\, \sum_{\ell=\ell_0}^{M_2}2^{(\ell+1)\beta+1}\sum_{j=M_3+1}^{\ell_1}\bbE\biggl[ \sup_{\bs\in\cT}|Z\{k_{j}(\bs)\}- Z\{k_{j-1}(\bs)\}|\biggr] \\
\leq &\, c_{1n} c_{2n}^{-1/\gamma}  {C}_{\gamma} \sum_{\ell=\ell_0}^{M_2}2^{(\ell+1)\beta+1}\sum_{j=M_3+1}^{\ell_1}2^{-j} \{\log( e  N_{2^{-j}}  )\}^{1/\gamma}\\
\lesssim &\, c_{1n} c_{2n}^{-1/\gamma}  {C}_{\gamma,\beta}  2^{\beta M_2} \int_{0}^{2^{-M_3}} \{\log( e  N_{\epsilon})\}^{1/\gamma} \,\mbox{d}\epsilon\,.
\end{align*}
Together with Markov's inequality, for any $x > 0$,
\begin{align}\label{boundI23}
\bbP\biggl({\rm II}_{23}>\frac{x}{8} \biggr) \lesssim   \frac{ c_{1n} {C}_{\gamma,\beta} }{c_{2n}^{1/\gamma}x} 2^{\beta M_2} \int_{0}^{2^{-M_3}} \{ \log( e   N_{\epsilon} )\}^{1/\gamma} \,\mbox{d}\epsilon \,.
\end{align}
Combining \eqref{div.I2} with \eqref{boundI21}--\eqref{boundI23}, we can conclude that, for any $x > 0$,
\begin{align}\label{I2.tail.pre}
&\bbP\biggl({\rm II}_{2}> \frac{x}{2} \biggr)  \lesssim \frac{ c_{1n} {C}_{\gamma,\beta,1} }{c_{2n}^{1/\gamma}x} \biggl[ \int_{0}^{2^{-M_2}} \frac{\{\log( e   N_{\epsilon}    )\}^{1/\gamma}}{\epsilon^{\beta}}\,\mbox{d}\epsilon +  2^{\beta M_2} \int_{0}^{2^{-M_3}} \{\log(  e N_{\epsilon}    )\}^{1/\gamma}\,\mbox{d}\epsilon\biggr] \notag\\
&~~~~~~~~~~~~~~~~ + c_{1n}({M_2}-\ell_0+1)(M_3-\ell_0+1)N_{2^{-M_3}}  \exp\biggl\{ - c_{2n} C_{\gamma,\beta,2} \biggl(\frac{ x}{D^{1-\beta}} \biggr)^{ {\gamma}}  \biggr\}\,.
\end{align}
Taking $M_2=\min\{\lfloor -\log_2 D_1\rfloor ,\ell_1-1\}$ and $M_3=\min\{\lfloor -\log_2 D_2\rfloor ,\ell_1 \}$ with  $0<D_2<D_1<D\wedge 1$. By \eqref{I2.tail.pre}, it holds that for any $x > 0$,
\begin{align*}
\bbP\biggl({\rm II}_{2}> \frac{x}{2} \biggr)   \lesssim &\,     \frac{ c_{1n} {C}_{\gamma,\beta,1} }{c_{2n}^{1/\gamma}x} \biggl\{\int_{0}^{D_1}\frac{ \{\log( e  N_{\epsilon}    )\}^{1/\gamma}}{\epsilon^{\beta}}\,\mbox{d}\epsilon +  \frac{ \int_{0}^{D_2} \{\log(  e N_{\epsilon}    )\}^{1/\gamma}\,\mbox{d}\epsilon}{D_1^{\beta}} \biggr\}  \notag\\
&\,~~~~~ + c_{1n} N_{D_2}\Lambda_{D_1}\Lambda_{D_2}  \exp\biggl\{ - c_{2n} C_{\gamma,\beta,2} \biggl(\frac{ x}{D^{1-\beta}} \biggr)^{ {\gamma}}  \biggr\}\,.
\end{align*}
This proves \eqref{II2} and hence completes the proof of Lemma~\ref{tail_multi}(ii). \hfill $\Box$

\subsection{Proof of Lemma \ref{lemma_tn_dis}}\label{proof.lemma_tn_dis}

Notice that
\begin{align*}%\label{dis}
|\tilde{T}_{n}-\tilde{T}^{\rm dis}_{n}|\leq \max_{ \ell \in [L] }\max_{ {j,j'} \in [p]}\max_{ {m,m'} \in [M]} \sup_{(u,v)\in B_m\times B_{m'}} n^{1/2}|\widetilde\Sigma^{(\ell)}_{jj'}(u,v)-\widetilde\Sigma^{(\ell)}_{jj'}(b_{m},b_{m'})| \,.
\end{align*}
Following from triangle inequality, we have    
\begin{align*} %\label{dis_I4}
& |\widetilde \Sigma^{(\ell)}_{jj'}(u,v)-\widetilde \Sigma^{(\ell)}_{jj'}(b_m,b_{m'})|  \notag\\
 =&\bigg|\frac{1}{n_{\ell}}\sum_{t=1}^{n_{\ell}}\big[\{\varepsilon_{t,j}(u) -\bar{\varepsilon}_{j}(u)\}\{\varepsilon_{t+\ell,j'}(v)-\bar{\varepsilon}_{j'}(v)\}
-\{\varepsilon_{t,j}(b_m)-\bar{\varepsilon}_{j}(b_m)\} \{\varepsilon_{t+\ell,j'}(b_{m'})-\bar{\varepsilon}_{j'}(b_{m'})\}\big]\bigg| \notag\\
  \leq &\underbrace{\bigg|\frac{1}{n_{\ell}}\sum_{t=1}^{n_{\ell}} \{\varepsilon_{t,j}(u)\varepsilon_{t+\ell,j'}(v) -\varepsilon_{t,j}(b_m)\varepsilon_{t+\ell,j'}(b_{m'})\}\bigg|}_{\textrm{I}_1(\ell,j,j',u,v)} +\underbrace{\bigg|\frac{1}{n_{\ell}}\sum_{t=1}^{n_{\ell}}\{\varepsilon_{t,j}(u)\bar{\varepsilon}_{j'}(v) -\varepsilon_{t,j}(b_m)\bar{\varepsilon}_{j'}(b_{m'})\}\bigg|}_{\textrm{I}_2(\ell,j,j',u,v)} \notag\\
&~~~~ ~+\underbrace{\bigg|\frac{1}{n_{\ell}}\sum_{t=1}^{n_{\ell}}\{\varepsilon_{t+\ell,j'}(v) \bar{\varepsilon}_{j}(u)-\varepsilon_{t+\ell,j'}(b_{m'})\bar{\varepsilon}_{j}(b_m)\}\bigg|}_{\textrm{I}_3(\ell,j,j',u,v)} +\underbrace{|\bar{\varepsilon}_{j}(u) \bar{\varepsilon}_{j'}(v) -\bar{\varepsilon}_{j}(b_m)\bar{\varepsilon}_{j'}(b_{m'})|}_{\textrm{I}_4(\ell,j,j',u,v)}
\end{align*}
for any   $\ell\in[L]$, $j,j'\in[p]$, $m,m'\in[M]$ and $(u,v)\in B_m\times B_{m'}$, which implies that 
\begin{align}\label{dis}
|\tilde{T}_{n}-\tilde{T}^{\rm dis}_{n}|\leq \sum_{k=1}^4 \max_{ \ell \in [L] }\max_{ {j,j'} \in [p]}\max_{ {m,m'} \in [M]} \sup_{(u,v)\in B_m\times B_{m'}} n^{1/2}\textrm{I}_k(\ell,j,j',u,v) \,.
\end{align}
In what follows, we will bound the four terms on the right-hand side of \eqref{dis}.

Let $d\{(u_1,v_1),(u_2,v_2)\} = |u_1-u_2|^{\kappa} + |v_1-v_2|^{\kappa}$ for any $(u_1,v_1),(u_2,v_2)\in \cU^2$. For the first term, by Condition \ref{c.subgaussian} and  \eqref{eq:tailprob.vep2} in Section \ref{proof.lem.dis.Gaussian}, it holds that
\begin{align*} 
     \|\varepsilon_{t,j}(u_1)\varepsilon_{t+\ell,j'}(v_1) - \varepsilon_{t,j}(u_2)\varepsilon_{t+\ell,j'}(v_2)\|_{\psi_1}  \leq   C(  |u_1-u_2|^{\kappa} + |v_1-v_2|^{\kappa}) \lesssim d\{(u_1,v_1),(u_2,v_2)\}  
\end{align*}
for any $(t,\ell,j,j')$ and   $(u_1,v_1),(u_2,v_2)\in \cU^2$. For any given $(\ell,j,j',u_1,v_1,u_2,v_2)$, Condition \ref{c.alpha}  implies that $\{\varepsilon_{t,j}(u_1)\varepsilon_{t+\ell,j'}(v_1) - \varepsilon_{t,j}(u_2)\varepsilon_{t+\ell,j'}(v_2)\}$ is an $\alpha$-mixing sequence
with  $\alpha$-mixing coefficients $\tilde{\alpha}(m)\leq \alpha(|m-\ell|_{+})\lesssim \exp(-C|m-\ell|_{+})$. Then, by  Lemma \ref{tail_chang} with $(B_n,c_n,{r}_1,{r})=(d\{(u_1,v_1),(u_2,v_2)\},\ell,1,1/3)$,  we have under $H_0$ that 
\begin{align}\label{psi.tail}
&  \bbP\bigg[\bigg| \frac{n^{1/2}}{n_\ell}\sum_{t=1}^{n_\ell}\{\varepsilon_{t,j}(u_1)\varepsilon_{t+\ell,j'}(v_1) - \varepsilon_{t,j}(u_2)\varepsilon_{t+\ell,j'}(v_2) \}  \bigg|>x \bigg]\nonumber \\
&~~~~ \lesssim \exp\bigg(-C\bigg[\frac{x}{d\{(u_1,v_1),(u_2,v_2)\}}\bigg]^2\bigg)
+ \exp\bigg(-C\bigg[\frac{n^{1/2}x}{d\{(u_1,v_1),(u_2,v_2)\}}\bigg]^{1/3}\bigg) \nonumber\\
&~~~~ \lesssim \exp\bigg(-C\bigg[\frac{x}{d\{(u_1,v_1),(u_2,v_2)\}}\bigg]^{1/3}\bigg) 
\end{align}
for any $x\geq 0$ and $(\ell,j,j')$. Since $|u-b_m|\leq M^{-1}$ and $|v-b_{m'}|\leq M^{-1}$ for any $m,m'\in[M]$ and $(u,v)\in B_m\times B_{m'}$, Then for any given $m,m'\in[M]$, $D= \sup\{d\{(u_1,v_1),(u_2,v_2)\}: (u_1,v_1),(u_2,v_2) \in B_m\times B_{m'}\}\leq CM^{-\kappa}$. Notice that  $N_\epsilon = N(B_m\times B_{m'}, d;\epsilon) \leq CM^{-2}\epsilon^{-2/\kappa}$ for any $\epsilon\leq D'$ and $m,m'\in[M]$. Then, by \eqref{psi.tail} 
and Lemma \ref{tail_multi}(i) with $(c_{1n},c_{2n},\gamma)=(C,C,1/3)$, for any given $(\ell,j,j',m,m')$, we have for $0<D'\ll M^{-\kappa}\asymp D$ and any $x > 0$,
\begin{align}\label{I1.sup}
& \bbP\bigg\{\sup_{(u,v)\in B_m\times B_{m'}}n^{1/2}\textrm{I}_1(\ell,j,j',u,v) > x\bigg\} \nonumber \\
&~~~~~~~~~~~~ \lesssim N_{D'}\log_2\bigg(\frac{1}{D'}\bigg)\exp(-CM^{\kappa/3} x^{1/3}) + \frac{\int_0^{D'} \{\log(eN_\epsilon)\}^3 \,{\rm d}\epsilon}{ x} \nonumber\\
&~~~~~~~~~~~~ \lesssim \bigg(\frac{1}{M^{\kappa}D'} \bigg)^{2/\kappa}\log_2\bigg(\frac{1}{D'}\bigg)\exp(-CM^{\kappa/3} x^{1/3}) + \frac{\int_0^{D'}\{ \log(CM^{-2}\epsilon^{-2/\kappa})\}^3\,{\rm d}\epsilon}{ x} \nonumber \\
&~~~~~~~~~~~~ \lesssim_{\rm (a)}  \bigg(\frac{1}{M^{\kappa}D'} \bigg)^{2/\kappa}\log_2\bigg(\frac{1}{D'}\bigg)\exp(-CM^{\kappa/3} x^{1/3}) +  \frac{D' [\log\{1/(M^{\kappa}D')\}]^3}{ x} \nonumber\\
&~~~~~~~~~~~~ \lesssim_{\rm (b)}  \bigg(\frac{1}{M^{\kappa}D'} \bigg)^{2/\kappa}\log_2\bigg(\frac{1}{D'}\bigg)\exp(-CM^{\kappa/3} x^{1/3})  +  \frac{M^{-\kappa\delta}(  D')^{1-\delta}}{ x}
\end{align}
for some sufficiently small constant $\delta >0$, where ${\rm (a)}$ follows from $M^{\kappa}D'\ll 1$ and 
\begin{align*}
    \int_0^{D'}\{\log(CM^{-2}\epsilon^{-2/\kappa})\}^3\,{\rm d}\epsilon 
  =&\,  \int_{+\infty}^{\log\{CM^{-2}(D')^{-2/\kappa}\}} y^3 C^{\kappa/2} M^{-\kappa}(-\kappa/2)\exp(-\kappa y /2)\,{\rm d}y  \\
  \lesssim&\, M^{-\kappa}\int_{\log\{CM^{-2}(D')^{-2/\kappa}\}}^{+\infty} y^3\exp(-\kappa y/2)\,{\rm d}y \\
  \lesssim&\, M^{-\kappa} \int_{\log\{CM^{-\kappa}(D')^{-1}\}}^{+\infty} s^3\exp(-s)\,{\rm d}s \\
  \lesssim&\, M^{-\kappa}\bigg\{\log\bigg(\frac{1}{M^{\kappa}D'}\bigg)\bigg\}^{3} M^{\kappa}D' \\
  \lesssim&\, D' \bigg\{\log\bigg(\frac{1}{M^{\kappa}D'}\bigg)\bigg\}^{3} \,,
\end{align*}
and ${\rm (b)}$ follows from
\begin{align*}
    D' \bigg\{\log\bigg(\frac{1}{M^{\kappa}D'}\bigg)\bigg\}^{3}  
  \lesssim  D' (M^{\kappa}D')^{-\delta} 
  = M^{-\kappa \delta} (D')^{1-\delta}\,.
\end{align*}
Combining \eqref{I1.sup} with Bonferroni's inequality, it holds that
\begin{align*}
    &\bbP\bigg\{ \max_{\ell\in[L]}\max_{j,j'\in[p]}\max_{m,m'\in[M]}\sup_{(u,v)\in B_m\times B_{m'}}n^{1/2}\textrm{I}_1(\ell,j,j',u,v) > x\bigg\} \notag\\
&~~~~~~~~~~ \leq Lp^2M^2 \max_{\ell\in[L]}\max_{j,j'\in[p]}\max_{m,m'\in[M]} \bbP\bigg\{\sup_{(u,v)\in B_m\times B_{m'}}n^{1/2}\textrm{I}_1(\ell,j,j',u,v) > x\bigg\} \notag\\
&~~~~~~~~~~ \lesssim p^2\bigg(\frac{1}{D'}\bigg)^{2/\kappa}\log_2\bigg(\frac{1}{D'}\bigg) \exp(-CM^{\kappa/3} x^{1/3}) + \frac{p^2M^{2-\kappa\delta}(D')^{1-\delta}}{ x}
\end{align*}
for any $x\geq0$. Recall $M\asymp (np)^{C_M}$ for some sufficiently large constant $C_M>0$. Since $\kappa \leq 1$, we can select $ D' =C (n^2p^2M^{2-\kappa\delta})^{1/(\delta-1)} $ such that $D'\ll M^{-\kappa}$ and 
\begin{align}\label{eq:rate.I1}
\max_{\ell\in[L]}\max_{j,j'\in[p]}\max_{m,m'\in[M]}\sup_{(u,v)\in B_m\times B_{m'}} n^{1/2}\textrm{I}_1(\ell,j,j',u,v) = O_{\rm p}(n^{-1})\,.
\end{align}

For the second term, triangle inequality yields
\begin{align*}
&\textrm{I}_2(\ell,j,j',u,v)
= \bigg|\frac{1}{n_{\ell}}\sum_{t=1}^{n_{\ell}}\{\varepsilon_{t,j}(u)
\bar{\varepsilon}_{j'}(v)-\varepsilon_{t,j}(b_m)\bar{\varepsilon}_{j'}(b_{m'})\}\bigg|\\
&~~~=   \bigg|\frac{1}{n_{\ell}}\sum_{t=1}^{n_{\ell}}[\{\varepsilon_{t,j}(u)-\varepsilon_{t,j}(b_m)\}
\bar{\varepsilon}_{j'}(v)+\{\bar{\varepsilon}_{j'}(v)-\bar{\varepsilon}_{j'}(b_{m'})\}
\varepsilon_{t,j}(b_m)]\bigg|\\
&~~~\leq   \underbrace{\bigg|\frac{1}{n_{\ell}}\sum_{t=1}^{n_{\ell}}\{\varepsilon_{t,j}(u)
	-\varepsilon_{t,j}(b_m)\}\bigg|\bigg|\frac{1}{n}\sum_{t=1}^n\varepsilon_{t,j'}(v)\bigg|}_{\textrm{I}_2^{(1)}(\ell,j,j',u,v)}
+\underbrace{\bigg|\frac{1}{n} \sum_{t=1}^{n}\{\varepsilon_{t,j'}(v)-\varepsilon_{t,j'} (b_{m'})\}\bigg|\bigg|\frac{1}{n_{\ell}}\sum_{t=1}^{n_{\ell}} \varepsilon_{t,j}(b_m)\bigg|}_{\textrm{I}_2^{(2)}(\ell,j,j',u,v)}
\end{align*}
for any $(\ell,j,j',m,m')$ and $(u,v)\in B_m\times B_{m'}$. Following from the triangle inequality again, 
\begin{align}\label{eq:I21}
\textrm{I}_2^{(1)}(\ell,j,j',u,v)
\leq&\, \bigg|\frac{1}{n_{\ell}}\sum_{t=1}^{n_{\ell}}\{\varepsilon_{t,j}(u)
-\varepsilon_{t,j}(b_m)\}\bigg| \bigg|\frac{1}{n}\sum_{t=1}^n \{\varepsilon_{t,j'}(v)- \varepsilon_{t,j'}(b_{m'})\}\bigg| \notag\\
&\, + \bigg|\frac{1}{n_{\ell}}\sum_{t=1}^{n_{\ell}}\{\varepsilon_{t,j}(u)-\varepsilon_{t,j}(b_m)\}\bigg|
\bigg|\frac{1}{n}\sum_{t=1}^n  \varepsilon_{t,j'}(b_{m'})\bigg|\,.
\end{align}
Let $d(u,v) = |u-v|^{\kappa}$ for any $u,v\in \cU$. According to Condition \ref{c.subgaussian},  we have
\begin{align*} 
     \|\varepsilon_{t,j}(u) - \varepsilon_{t,j}(b_m) \|_{\psi_2}  \leq   C   |u-b_m|^{\kappa} \lesssim d(u,b_m)   
\end{align*}
for any $(t,j,m)$ and   $u\in B_m$.
Since $\mathbb{E}\{\varepsilon_{t,j}(u)\}=0$ for any $(t,j,u)$,   Condition \ref{c.alpha} and Lemma \ref{tail_chang} with $(B_n,c_n,{r}_1,{r})=(d(u,b_m),0,2,1/3)$ yield 
\begin{align}\label{psi.tail2}
&  \bbP\bigg[\bigg| \frac{n^{1/2}}{n_\ell}\sum_{t=1}^{n_\ell}\{\varepsilon_{t,j}(u) - \varepsilon_{t,j}(b_m) \}  \bigg|>x \bigg]\nonumber \\
 \lesssim & \exp\bigg[-C\bigg\{\frac{x}{d(u,b_m)}\bigg\}^2\bigg]
+ \exp\bigg[-C\bigg\{\frac{n^{1/2}x}{d(u,b_m)}\bigg\}^{1/3}\bigg]  \lesssim \exp\bigg[-C\bigg\{\frac{x}{d(u,b_m)}\bigg\}^{1/3}\bigg]
\end{align}
for any $x\geq 0$, $(\ell,j,m)$ and $u\in B_m$. Similar to the proof of \eqref{eq:rate.I1}, we can also have
\begin{align}\label{eq:diff.eps}
\max_{\ell\in[L]_0}\max_{j\in[p]}\max_{m\in [M]}\sup_{u\in B_m}\bigg|\frac{n^{1/2}}{n_{\ell}}\sum_{t=1}^{n_{\ell}}\{\varepsilon_{t,j}(u)
-\varepsilon_{t,j}(b_m)\}\bigg| = O_{\rm p}(n^{-1}) \,.
\end{align}
Condition \ref{c.subgaussian} yields that, for any $x\geq 0$, 
\begin{align}\label{tail.vep}
\max_{t\in[n]}\max_{j\in[p]}\sup_{u\in\cU}\bbP\{|\varepsilon_{t,j}(u)|>x\}\lesssim \exp(-Cx^2)\,.
\end{align}
By Bonferroni's inequality, Condition \ref{c.alpha} and Lemma \ref{tail_chang} with $(B_n,c_n,{r}_1,{r})=(1,0,2,1/3)$,
\begin{align*}
\bbP\bigg\{\max_{\ell\in[L]_0}\max_{j\in[p]}\max_{m\in[M]}\bigg|\frac{1}{n_{\ell}}\sum_{t=1}^{n_\ell}  \varepsilon_{t,j}(b_{m})\bigg| \geq x\bigg\} \lesssim pM\exp(-Cnx^2) + pM\exp(-Cn^{1/3}x^{1/3})
\end{align*}
for any $x\geq 0 $. Recall that  $M\asymp (np)^{C_M}$. Then  we have
\begin{align*}
\max_{\ell\in[L]_0}\max_{j\in[p]}\max_{m\in[M]} \bigg|\frac{1}{n_{\ell}}\sum_{t=1}^{n_{\ell}}  \varepsilon_{t,j}(b_{m})\bigg| 
=O_{\rm p}\big[n^{-1/2}\{\log(n p)\}^{1/2} \vee n^{-1}\{\log(np)\}^3 \big] \,.
\end{align*}
This, together with \eqref{eq:I21} and \eqref{eq:diff.eps}, yields
\begin{align*}
\max_{\ell\in[L]}\max_{ {j,j'} \in [p]}\max_{m,m'\in[M]} \sup_{(u,v)\in B_m\times B_{m'}}n^{1/2} {\rm I}_2^{(1)}(\ell,j,j',u,v) = O_{\rm p}(n^{-1})
\end{align*}
provided that $\log (np)  \ll n^{1/3}$. Analogously, we can also obtain
\begin{align*}
\max_{\ell\in[L]}\max_{ {j,j'} \in [p]}\max_{m,m'\in[M]} \sup_{(u,v)\in B_m\times B_{m'}}n^{1/2} {\rm I}_2^{(2)}(\ell,j,j',u,v) = O_{\rm p}(n^{-1})
\end{align*}
provided that $\log (np)  \ll n^{1/3}$, which further implies 
\begin{align}\label{eq:I2}
\max_{\ell\in[L]}\max_{j,j'\in[p]}\max_{m,m'\in[M]}\sup_{(u,v)\in B_m\times B_{m'}} n^{1/2} \textrm{I}_2(\ell,j,j',u,v) = O_{\rm p}(n^{-1})
\end{align}
provided that $\log (np)  \ll n^{1/3}$. Using  similar arguments as in the proof of \eqref{eq:I2},  we have 
\begin{align*}
& \max_{\ell\in[L]}\max_{ {j,j'} \in [p]}\max_{m,m'\in[M]}\sup_{(u,v)\in B_m\times B_{m'}} n^{1/2}\textrm{I}_3(\ell,j,j',u,v)=  O_{\rm p}(n^{-1}) \,, \\
& \max_{\ell\in[L]}\max_{ {j,j'} \in [p]}\max_{m,m'\in[M]}\sup_{(u,v)\in B_m\times B_{m'}} n^{1/2}\textrm{I}_4(\ell,j,j',u,v)=  O_{\rm p}(n^{-1})\,,
\end{align*}
provided that $\log (np)  \ll n^{1/3}$. This, together with \eqref{dis}, \eqref{eq:rate.I1} and \eqref{eq:I2}, yields $|\tilde{T}_{n}-\tilde{T}^{\rm dis}_{n}|=O_{\rm p}(n^{-1})$ provided that  $\log (np)  \ll n^{1/3}$. We   complete  the proof of Lemma \ref{lemma_tn_dis}.
\hfill $\Box$

\subsection{Proof of Lemma \ref{lemma_tn_dis_GA}}\label{proof.lemma_tn_dis_GA}

Let  ${\bvep}^{\rm dis}_t=({\vep}^{\rm dis}_{t,1},\ldots,{\vep}^{\rm dis}_{t,\tilde{p}})^{\T} =\{\vep_{t,1}(b_1),\ldots,\vep_{t,p}(b_1),\vep_{t,1}(b_2), \ldots,\vep_{t,p}(b_2),\ldots, \vep_{t,p}(b_M)\}^{\T}$, where $\tilde{p}=pM$. Moreover,   we denote $\tilde{\bfeta}_{t}^{\rm dis}=(\tilde{\eta}_{t,1}^{\rm dis},\ldots,\tilde{\eta}_{t,L\tilde{p}^2}^{\rm dis})^{\T}$ with
$\tilde{\eta}_{t,k}^{\rm dis}= (\varepsilon^{\rm dis}_{t,j}-\bar{{\varepsilon}}^{\rm dis}_{ j}) (\varepsilon^{\rm dis}_{t+\ell,j'}-\bar{\varepsilon}^{\rm dis}_{j'})$, where $\bar{\vep}^{\rm dis}_j=n^{-1}\sum_{t=1}^{n} \vep^{\rm dis}_{t,j}$ for any $j\in[\tilde{p}]$. Here, for any fixed $k \in[L\tilde{p}^2]$, we write
$(\ell,j,j')=(\ell_k,j_k,j'_k)$ for its associated unique triple  for notational simplicity. In addition, we define $\bUpsilon =(\Upsilon_1,\dots,\Upsilon_{L\tilde p^2})^{\T}= ([{\rm{vec}}\{{\widetilde{\bSigma}}^{(1), {\rm dis}}\}]^{\T}, \ldots, [{\rm{vec}}\{\widetilde{\bSigma}^{(L), {\rm dis}}\}]^{\T})^{\T}$, where $\widetilde\bSigma^{(\ell), {\rm dis}}=\{\widetilde\Sigma^{(\ell), {\rm dis}}_{jj'}\}_{\tilde{p}\times\tilde p}$ with $\widetilde\Sigma^{(\ell),{\rm dis}}_{j{j}'} = n_\ell^{-1}\sum_{t=1}^{n_\ell} \tilde{\eta}_{t,k}^{\rm dis}$. Then, the   statistic  $\tilde{T}_n^{\rm dis}$ defined in \eqref{Tn.dis} can be rewritten  as
\begin{align*}
\tilde{T}^{\rm dis}_{n}=\max_{ \ell \in [L]}\max_{ j,j' \in [\tilde{p}]} \bigg| \frac{n^{1/2}}{n_\ell}\sum_{t=1}^{n_\ell} (\varepsilon^{\rm dis}_{t,j}-\bar{{\varepsilon}}^{\rm dis}_{ j}) (\varepsilon^{\rm dis}_{t+\ell,j'}-\bar{\varepsilon}^{\rm dis}_{j'})  \bigg|   = n^{1/2}\max_{k\in[L\tilde{p}^2]}|\Upsilon_k|  \,.
\end{align*}
At the same time, by \eqref{Tn.dis.G}, we know that
\begin{align}\label{new.def.tndisG}
    \tilde{T}^{\rm dis, G}_{n}= |\tilde{\boldsymbol{\cG}}^{\rm dis}|_{\max} = \max_{k\in[L\tilde{p}^2]}|\tilde\cG^{\rm dis}_k| \,,
\end{align}
where  $\tilde{\boldsymbol{\cG}}^{\rm dis}=(\tilde\cG^{\rm dis}_1,\dots,\tilde\cG^{\rm dis}_{L\tilde p^2})^{\T}\sim \cN(\0, \widetilde\bXi^{\rm dis}_{n})$ with $\widetilde\bXi^{\rm dis}_{n}
%=(\widetilde{\Xi}^{\rm dis}_{n,k_1k_2})_{L\tilde p^2\times L\tilde p^2} 
= \var(n^{1 / 2}\bUpsilon)$.   Therefore,
\begin{align*}%\label{eq:diff.Tndis}
\sup_{x\in \mathbb{R}}|\mathbb{P}(\tilde{T}_n^{\rm dis}\leq x) - \mathbb{P}(\tilde{T}_n^{\rm dis,G}\leq x)|= \sup_{x\in \mathbb{R}}\bigg|\mathbb{P}\bigg(n^{1/2}\max_{k\in[L\tilde{p}^2]}|\Upsilon_k| \leq x\bigg) - \mathbb{P}\bigg(\max_{k\in[L\tilde{p}^2]}|\tilde\cG^{\rm dis}_k|\leq x\bigg)\bigg| \,.
\end{align*}
Let $Z= n^{1 / 2}\max_{ k \in [L\tilde p^2]} \Upsilon_{k}$, $ V=\max_{k \in [L\tilde p^2]} \tilde\cG^{\rm dis}_{k}$ and $\check{Z}=\max_{ k \in [L\tilde p^2] }n_L^{-1/2}\sum_{t=1}^{n_L} {\tilde\eta}_{t,k}^{\rm dis}$. We claim that
\begin{align}\label{lem.claim1}
    d_{1n}:=\sup_{x\in \mathbb{R}}|\mathbb{P}(\check{Z}\leq x)-\mathbb{P}(V\leq x)|=o(1)\,,
\end{align}
provided that $\log \tilde{p} \ll n^{2/21}$, and that with $\epsilon_1=C_0n^{-1/2} \log \tilde{p} $ for some sufficiently large constant $C_0>0$, 
\begin{align}\label{lem.claim2}
    \mathbb{P}(|Z-\check{Z}|> \epsilon_1)+\sup_{x\in \mathbb{R}}\mathbb{P}(| V-x| \leq \epsilon_1) =o(1)\,,
\end{align} 
provided that $\log\tilde{p} \ll n^{3/10}$.  Then, by \eqref{lem.claim1} and \eqref{lem.claim2}, we can conclude that
\begin{align}\label{dn}
    \sup_{x\in \mathbb{R}}|\mathbb{P}({Z}\leq x)-\mathbb{P}(V\leq x)| \leq d_{1n}+ \mathbb{P}(|Z-\check{Z}|> \epsilon_1)+\sup_{x\in \mathbb{R}}\mathbb{P}(| V-x| \leq \epsilon_1) = o(1)\,,
\end{align}
provided that $\log \tilde{p} \ll n^{2/21}$. Denote $\tilde{\boldsymbol{\cG}}^{\rm dis, ext}=\{(\tilde{\boldsymbol{\cG}}^{\rm dis})^{\T},(-\tilde{\boldsymbol{\cG}}^{\rm dis})^{\T}\}^{\T} =(\tilde\cG^{\rm dis, ext}_{1},\dots,\tilde\cG^{\rm dis, ext}_{2{L\tilde{p}^2}})^{\T}$ and
${\boldsymbol{\Upsilon}}^{\rm{ext}}=\{{\boldsymbol{\Upsilon}}^{\T},(-{\boldsymbol{\Upsilon}})^{\T}\}^{\T}=(\Upsilon_{1}^{\rm{ext}},\dots,\Upsilon_{2{L\tilde{p}^2}}^{\rm{ext}})^{\T}$. Then
\begin{align*}
\sup_{x\in \mathbb{R}}|\mathbb{P}(\tilde{T}_{n}^{\rm dis}\leq x)-\mathbb{P}(\tilde{T}_{n}^{\rm dis,G}\leq x)|
=& \sup_{x\in \mathbb{R}}\bigg|\mathbb{P}\bigg(\max_{k \in [2{L\tilde{p}^2}]} n^{1/2} {\Upsilon}_{k}^{\rm{ext}}\leq x\bigg) -\mathbb{P}\bigg(\max_{k \in [2{L\tilde{p}^2}]} \tilde\cG_{k}^{\rm{dis,ext}}\leq x\bigg)\bigg|\,.
%=:& \tilde{d}_n \,.
\end{align*} 
Recall $M\asymp (np)^{C_M}$. Following the analogous arguments as in the proof of \eqref{dn}, we also have $\sup_{x\geq 0}|\mathbb{P}(\tilde{T}_{n}^{\rm dis}\leq x)-\mathbb{P}(\tilde{T}_{n}^{\rm dis,G}\leq x)|=o(1)$ provided that $\log   p \ll n^{2/21}$. Hence, the proof of Lemma \ref{lemma_tn_dis_GA} will be completed once \eqref{lem.claim1} and \eqref{lem.claim2} are established.

\noindent\underline{Proof of \eqref{lem.claim1}.} Define $\check{{\boldsymbol{\cG}}}^{\rm dis}=(\check{\cG}_1^{\rm dis},\ldots, \check{\cG}_{{L\tilde p^2}}^{\rm dis})^{\T} \sim \cN({\bf 0}, \check{\bXi}_{n}^{\rm dis})$, where 
\begin{align*}
    \check{\bXi}_{n}^{\rm dis}=(\check{\Xi}_{n,k_1k_2}^{\rm dis})_{L\tilde p^2\times L\tilde p^2}= \var\bigg( \frac{1}{n_L^{1/2}} \sum_{t=1}^{n_L}  \tilde\bfeta_t^{\rm dis} \bigg) \,.
\end{align*}
Let $\check{V}=\max _{k \in [L\tilde p^2]} \check{\cG}_{k}^{\rm dis}$ and ${d}_{2n}:=\sup_{x\in \mathbb{R}}|\mathbb{P}(\check{Z}\leq x)-\mathbb{P}(\check{V}\leq x)|$. It holds that
\begin{align}\label{d1ntod2n}
d_{1n}  
\leq \sup_{x\in \mathbb{R}} |\mathbb{P}(\check{V}\leq x)-\mathbb{P}(V\leq x) | +d_{2n}\, .
\end{align}
Recall $ V=\max_{k \in [L\tilde p^2]} \tilde\cG^{\rm dis}_{k}$, where $\tilde{\boldsymbol{\cG}}^{\rm dis}=(\tilde\cG^{\rm dis}_1,\dots,\tilde\cG^{\rm dis}_{L\tilde p^2})^{\T}\sim \cN(\0, \widetilde\bXi^{\rm dis}_{n})$. Condition \ref{c.bvar} implies $\min_{k\in [L\tilde p^2]}\widetilde\Xi^{\rm dis}_{n,kk}\geq C$. By Proposition 2.1 of \cite{Cher2022_supp}, we have
\begin{align}
\label{comparison}
\sup_{x\in \mathbb{R}} |\mathbb{P}(\check{V}\leq x)-\mathbb{P}(V\leq x) |
\lesssim  \Delta_{0}^{1 / 2} \log  (L\tilde p^2 )  \, ,
\end{align}
where $\Delta_{0}=\|\widetilde\bXi^{\rm dis}_{n}-\check{\bXi}_{n}^{\rm dis}\|_{\max}$.  Notice that $\widetilde\bXi^{\rm dis}_{n}=(\widetilde{\Xi}^{\rm dis}_{n,k_1k_2})_{L\tilde p^2\times L\tilde p^2} =\var(n^{1/2}\bUpsilon)$, where  $\bUpsilon=(\Upsilon_1,\ldots,\Upsilon_{L\tilde p^2})^{\T}$ with $\Upsilon_k=n_{\ell_k}^{-1}\sum_{t=1}^{n_{\ell_k}}\tilde\eta_{t,k}^{\rm dis}$ for any $k\in[L\tilde p^2]$. Let $\mathring{\bfeta}_t^{\rm dis} =(\mathring{\eta}_{t,1}^{\rm dis},\dots, \mathring{\eta}_{t,L\tilde p^2}^{\rm dis})^{\T}=\tilde\bfeta_t^{\rm dis}-\mathbb{E}(\tilde\bfeta_t^{\rm dis})$ for any $t\in[n]$. Then, the $(k_1,k_2)$-th component of $|\widetilde\bXi^{\rm dis}_{n}-\check{\bXi}_{n}^{\rm dis}|$ is
\begin{align} \label{eq:deltasplit}
& \bigg|\frac{n}{n_{\ell_1}n_{\ell_2}}\mathbb{E}\bigg\{\bigg(\sum_{t=1}^{n_{\ell_1}}\mathring{\eta}_{t,k_1}^{\rm dis}\bigg)
\bigg(\sum_{t=1}^{n_{\ell_2}}\mathring{\eta}_{t,k_2}^{\rm dis}\bigg)\bigg\}
-\frac{1}{{n_L}}\mathbb{E}\bigg\{\bigg(\sum_{t=1}^{{n_L}}\mathring{\eta}_{t,k_1}^{\rm dis}\bigg)
\bigg(\sum_{t=1}^{{n_L}}\mathring{\eta}_{t,k_2}^{\rm dis}\bigg)\bigg\} \bigg|  \\
&~ \leq \bigg|\frac{n_{\ell_1}n_{\ell_2}-n{n_L}}{{n_L}n_{\ell_1}n_{\ell_2}}\mathbb{E}\bigg\{\bigg(\sum_{t=1}^{{n_L}}
\mathring{\eta}_{t,k_1}^{\rm dis}\bigg)
\bigg(\sum_{t=1}^{{n_L}}\mathring{\eta}_{t,k_2}^{\rm dis}\bigg)\bigg\}\bigg|
+\bigg|\frac{n}{n_{\ell_1}n_{\ell_2}}\mathbb{E}\bigg\{\bigg(\sum_{t=1}^{{n_L}}\mathring{\eta}_{t,k_1}^{\rm dis}\bigg)
\bigg(\sum_{t={n_L}+1}^{n_{\ell_2}}\mathring{\eta}_{t,k_2}^{\rm dis}\bigg)\bigg\}\bigg|
\notag\\
&~~~~~ +\bigg|\frac{n}{n_{\ell_1}n_{\ell_2}}\mathbb{E}\bigg\{\bigg(\sum_{t={n_L}+1}^{n_{\ell_1}}\mathring{\eta}_{t,k_1}^{\rm dis}\bigg)
\bigg(\sum_{t=1}^{{n_L}}\mathring{\eta}_{t,k_2}^{\rm dis}\bigg)\bigg\}\bigg|+\bigg|\frac{n}{n_{\ell_1}n_{\ell_2}}
\mathbb{E}\bigg\{\bigg(\sum_{t={n_L}+1}^{n_{\ell_1}}\mathring{\eta}_{t,k_1}^{\rm dis}\bigg)
\bigg(\sum_{t={n_L}+1}^{n_{\ell_2}}\mathring{\eta}_{t,k_2}^{\rm dis}\bigg)\bigg\} \bigg| \, ,   \notag
\end{align}
where $ \ell_i =\ell_{k_i}$ for $i\in\{1,2\}$. Recall $\tilde{\eta}_{t,k}^{\rm dis}= (\varepsilon^{\rm dis}_{t,j}-\bar{{\varepsilon}}^{\rm dis}_{ j}) (\varepsilon^{\rm dis}_{t+\ell,j'}-\bar{\varepsilon}^{\rm dis}_{j'})$. 
For any $k\in[L\tilde p^2]$,
\begin{align*}%\label{eq:sumetac}
\sum_{t=1}^{{n_L}}\mathring\eta_{t,k}^{\rm dis}
=&\,\sum_{t=1}^{{n_L}}\{ \vep_{t,j}^{\rm dis}
\vep_{t+\ell,j'}^{\rm dis}-\mathbb{E}(\vep_{t,j}^{\rm dis}
\vep_{t+\ell,j'}^{\rm dis}) \}\\
&\,- \frac{1}{n}\bigg(\sum_{t=1}^{n_L}\vep_{t,j}^{\rm dis}\bigg)
\bigg(\sum_{t=1}^n\vep_{t,j'}^{\rm dis}\bigg)
+\frac{1}{n}\mathbb{E}\bigg\{\bigg(\sum_{t=1}^{n_L}\vep_{t,j}^{\rm dis}\bigg)
\bigg(\sum_{t=1}^n\vep_{t,j'}^{\rm dis}\bigg)\bigg\} \notag\\
&\,
-\frac{1}{n}\bigg(\sum_{t=1}^{n_L}\vep_{t+\ell,j'}^{\rm dis}\bigg)\bigg(\sum_{t=1}^n \vep_{t,j}^{\rm dis}\bigg)
+\frac{1}{n}\mathbb{E}\bigg\{\bigg(\sum_{t=1}^{n_L}\vep_{t+\ell,j'}^{\rm dis}\bigg)
\bigg(\sum_{t=1}^n\vep_{t,j}^{\rm dis}\bigg)\bigg\} \notag\\
&\, +\frac{n_L}{n^2}\bigg(\sum_{t=1}^{n}\vep_{t,j}^{\rm dis}\bigg)
\bigg(\sum_{t=1}^n\vep_{t,j'}^{\rm dis}\bigg) -\frac{n_L}{n^2}\mathbb{E}\bigg\{\bigg(\sum_{t=1}^{n}\vep_{t,j}^{\rm dis}\bigg)
\bigg(\sum_{t=1}^n\vep_{t,j'}^{\rm dis}\bigg)\bigg\}\,.
\end{align*}
Cauchy-Schwarz inequality and Jensen's inequality yield that
\begin{align*}%\label{eq:esumetac}
\mathbb{E}\bigg\{\bigg(\sum_{t=1}^{{n_L}}\mathring{{\eta}}_{t,k}^{\rm dis}\bigg)^2\bigg\}
\lesssim&\, \mathbb{E}\bigg(\bigg[\sum_{t=1}^{{n_L}}\{ \vep_{t,j}^{\rm dis}\vep_{t+\ell,j'}^{\rm dis}-\mathbb{E}(\vep_{t,j}^{\rm dis}\vep_{t+\ell,j'}^{\rm dis}) \}\bigg]^2\bigg)
+\frac{1}{n^2}\mathbb{E}\bigg\{\bigg(\sum_{t=1}^{n_L}\vep_{t,j}^{\rm dis}\bigg)^2
\bigg(\sum_{t=1}^n\vep_{t,j'}^{\rm dis}\bigg)^2\bigg\} \notag\\
&\, +\frac{1}{n^2}\mathbb{E}\bigg\{\bigg(\sum_{t=1}^{n_L}\vep_{t+\ell,j'}^{\rm dis}\bigg)^2
\bigg(\sum_{t=1}^n\vep_{t,j}^{\rm dis} \bigg)^2\bigg\}+\frac{1}{n^2}\mathbb{E}\bigg\{\bigg(\sum_{t=1}^{n}\vep_{t,j}^{\rm dis}\bigg)^2
\bigg(\sum_{t=1}^n\vep_{t,j'}^{\rm dis} \bigg)^2\bigg\}\,.
\end{align*}
By Condition \ref{c.subgaussian} and \eqref{tail.vep}, we have
\begin{align}\label{eq:tail2}
    \mathbb{P}\{|\vep_{t,j}^{\rm dis}\vep_{t+\ell,j'}^{\rm dis} -\mathbb{E}(\vep_{t,j}^{\rm dis} \vep_{t+\ell,j'}^{\rm dis} )|>x\}\lesssim \exp(-Cx )
\end{align}  
for any $x\geq 0$. Then, Condition \ref{c.alpha} and Lemma \ref{tail_chang} with $( B_n, c_n, r_1, {r}) =(1, \ell,1 , 1/3)$ yield
\begin{align}\label{eq:vep2}
\max_{\ell\in[L]}\max_{j,j'\in[\tilde p]} \bigg\|\sum_{t=1}^{{n_L}}\{ \vep_{t,j}^{\rm dis} \vep_{t+\ell,j'}^{\rm dis} -\mathbb{E}(\vep_{t,j}^{\rm dis}\vep_{t+\ell,j'}^{\rm dis}) \} \bigg\|_s\lesssim n^{1/2}
\end{align}
for any integer $s>0$. Similarly, for any integer $s>0$,
\begin{align}\label{eq:vep1}
    \max_{\ell\in[L]_0}\max_{j\in[\tilde p]} \bigg\|\sum_{t=1}^{{n_\ell}}\vep_{t+\ell,j}^{\rm dis}\bigg\|_s \lesssim n^{1/2}\,.
\end{align}
 Then, by Cauchy-Schwarz inequality again, it holds that
$  n^{-2} [\mathbb{E}\{(\sum_{t=1}^{n_L}\vep_{t,j}^{\rm dis})^2(\sum_{t=1}^n\vep_{t,j'}^{\rm dis})^2\}
+\mathbb{E}\{(\sum_{t=1}^{n_L}\vep_{t+\ell,j'}^{\rm dis})^2(\sum_{t=1}^n\vep_{t,j}^{\rm dis})^2\}
+\mathbb{E}\{(\sum_{t=1}^{n}\vep_{t,j}^{\rm dis})^2(\sum_{t=1}^n\vep_{t,j'}^{\rm dis})^2\}] \leq C$. We can conclude  that
\begin{align}\label{moment}
\max_{k\in[L\tilde p^2]}\bigg\|\sum_{t=1}^{{n_L}}\mathring{\eta}_{t,k}^{\rm dis}\bigg\|_2\lesssim n^{1/2}  \,.
\end{align}
Since $L$ is a fixed integer, it follows by Condition \ref{c.subgaussian} that for any $x\geq 0$,
\begin{align} 
&\, \max_{\ell\in[L]_0}\max_{j'\in[\tilde p]} \mathbb{P}\bigg(\sum_{t=n_L+1}^{n_{\ell}}|{\vep}^{\rm dis}_{t+\ell,j'}|>x\bigg)\lesssim \exp(-Cx^2)  \,, \label{eq:tailsum1} \\
&\max_{\ell\in[L]}\max_{j,j'\in[\tilde p]} \mathbb{P}\bigg(\sum_{t=n_L+1}^{n_{\ell}}|{\vep}^{\rm dis}_{t,j}{\vep}^{\rm dis}_{t+\ell,j'}|>x\bigg) \lesssim \exp(-Cx) \,. \label{eq:tailsum2}
\end{align}
By   \eqref{eq:tailsum1} and \eqref{eq:tailsum2}, we have  $\max_{\ell\in[L]}\max_{j,j'\in[\tilde p]}\|\sum_{t=n_L+1}^{{n_\ell}}\{ \vep_{t,j}^{\rm dis} \vep_{t+\ell,j'}^{\rm dis}-\mathbb{E}(\vep_{t,j}^{\rm dis} \vep_{t+\ell,j'}^{\rm dis}) \}\|_s\leq C$ and $\max_{\ell\in[L]_0}\max_{ j'\in[\tilde p]} \|\sum_{t=n_L+1}^{{n_\ell}}\vep_{t+\ell,j'}^{\rm dis}\|_s\leq C$      for any integer $s>0$.  
Following the similar arguments as in the proof of \eqref{moment}, 
it holds that
\begin{align}\label{moment2} 
\max_{k\in[L\tilde p^2]} \bigg\|\sum_{t=n_L+1}^{{n_\ell}}\mathring{\eta}_{t,k}^{\rm dis}\bigg\|_2\lesssim 1   \,.
\end{align}
Combining \eqref{eq:deltasplit}, \eqref{moment} and \eqref{moment2}, we can conclude that
\begin{align*} 
\bigg|\frac{n}{n_{\ell_1}n_{\ell_2}}\mathbb{E}\bigg\{\bigg(\sum_{t=1}^{n_{\ell_1}}\mathring{\eta}_{t,k_1}^{\rm dis}\bigg)
\bigg(\sum_{t=1}^{n_{\ell_2}}\mathring{\eta}_{t,k_2}^{\rm dis}\bigg)\bigg\}
-\frac{1}{{n_L}}\mathbb{E}\bigg\{\bigg(\sum_{t=1}^{{n_L}}\mathring{\eta}_{t,k_1}^{\rm dis}\bigg)
\bigg(\sum_{t=1}^{{n_L}}\mathring{\eta}_{t,k_2}^{\rm dis}\bigg)\bigg\} \bigg|\lesssim n^{-1/2}
\end{align*}
for any  $k_1,k_2\in[L\tilde p^2]$, which further implies that $|\widetilde{\Xi}^{\rm dis}_{n,k_1k_2}-\check{\Xi}_{n,k_1k_2}^{\rm dis}|\lesssim n^{-1/2}$ for any $k_1,k_2\in[L\tilde p^2]$.  Note that this upper bound  does not depend on $(k_1,k_2)$. Thus,
\begin{align}\label{var_dis}
\Delta_{0}=\|\widetilde{\bXi}_{n}^{\rm dis}-\check{\bXi}_{n}^{\rm dis}\|_{\max}\lesssim n^{-1/2}\,.
\end{align}
Together with \eqref{comparison}, we have
\begin{align}\label{eq:diff.V.Vc}
\sup_{x\in \mathbb{R}}|\mathbb{P}(\check{V}\leq x)-\mathbb{P}(V\leq x)|\lesssim n^{-1 / 4} \log  (L\tilde p^2 )=o(1)\,,
\end{align}
provided that $\log \tilde{p} \ll n^{1/4}$. This, together with \eqref{d1ntod2n}, yields
\begin{align}\label{d1ntod2n.2}
    d_{1n}  
\leq   d_{2n} + o(1)\,,
\end{align}
provided that $\log \tilde{p} \ll n^{1/4}$.

Next, we will show that  $d_{2n}=\sup_{x\in \mathbb{R}}|\mathbb{P}(\check{Z}\leq x)-\mathbb{P}(\check{V}\leq x)|=o(1)$, where $\check{Z}=\max_{ k \in [L\tilde p^2] }n_L^{-1/2}\sum_{t=1}^{n_L} {\tilde\eta}_{t,k}^{\rm dis}$ and $\check{V}=\max _{k \in [L\tilde p^2]} \check{\cG}_{k}^{\rm dis}$, where $\check{{\boldsymbol{\cG}}}^{\rm dis}=(\check{\cG}_1^{\rm dis},\ldots, \check{\cG}_{{L\tilde p^2}}^{\rm dis})^{\T} \sim \cN({\bf 0}, \check{\bXi}_{n}^{\rm dis})$ with $\check{\bXi}_{n}^{\rm dis} = \var( {n_L^{-1/2}} \sum_{t=1}^{n_L}  \tilde\bfeta_t^{\rm dis})$.  For any $t\in[n_L]$, let
\begin{align}\label{dot.eta}
    \dot\bfeta_t^{\rm dis}=(\dot\eta_{t,1}^{\rm dis},\dots,\dot\eta_{t,L\tilde p^2}^{\rm dis})^{\T} =([\rm{vec}\{\bvep_t^{\rm dis} (\bvep_{t+1}^{{\rm dis}})^{\T}\}]^{\T}, \ldots, [\rm{vec}\{\bvep_t^{\rm dis} (\bvep_{t+L}^{{\rm dis}})^{\T}\}]^{\T})^{\T} \,.
\end{align}
 Define $\dot{Z}=\max_{k\in[L\tilde p^2]}n_L^{-1/2}\sum_{t=1}^{{n_L}}\dot\eta_{t,k}^{\rm dis}$ and $\dot{V}=\max_{k \in [L\tilde p^2]}\dot{\cG}_k^{\rm dis}$, where $\dot{\boldsymbol{\cG}}^{\rm dis}=(\dot{\cG}_1^{\rm dis},\dots,\dot{\cG}_{L\tilde p^2}^{\rm dis})^{\T}\sim \cN({\bf{0}},\dot{\bXi}_{n}^{\rm dis})$, with $\dot{\bXi}_{n}^{\rm dis} = \var( {n_L^{-1/2}} \sum_{t=1}^{n_L}  \dot\bfeta_t^{\rm dis})$.
 Denote $d_{3n}:=\sup_{x \in \mathbb{R}}|\mathbb{P}(\dot{Z}\leq x)-\mathbb{P}(\dot{V}\leq x)|$. Analogous to \eqref{dn}, it holds that, for any $\epsilon_2 > 0$,
\begin{align}\label{pre_gau}
d_{2n} \leq&\,\mathbb{P}(|\check{Z}-\dot{Z}|>\epsilon_2)
+\sup_{x\in \mathbb{R}} \mathbb{P}( |\check{V}-x|\leq  \epsilon_2) +\sup_{x\in \mathbb{R}}|\mathbb{P}(\check{V}\leq x)-\mathbb{P}(\dot{V}\leq x)|+ d_{3n}\,.
\end{align}
We next bound the four terms on the right-hand side of the above inequality separately.  
Let ${\bf W}_t^{\rm dis}=(W_{t,1}^{\rm dis},\dots,W_{t,L\tilde p^2}^{\rm dis})^{\T}=\tilde{\bfeta}_t^{\rm dis}-\dot\bfeta_t^{\rm dis}$ for any $t\in[n_L]$, where $W_{t,k}^{\rm dis}=-\vep_{t,j}^{\rm dis} \bar{\vep}^{\rm dis}_{j'}-\bar{\vep}^{\rm dis}_{j}\vep_{t+\ell,j'}^{\rm dis}+\bar{\vep}^{\rm dis}_{j} \bar{\vep}^{\rm dis}_{j'}$. Then, it holds that
\begin{align*}
    |\check{Z}-\dot{Z}|\leq \max_{k \in [L\tilde p^2] } \bigg|\frac{1}{n_L^{1/2}}\sum_{t=1}^{{n_L}}(\tilde\eta_{t,k}^{\rm dis}-\dot\eta_{t,k}^{\rm dis})\bigg|
= \max_{k \in [L\tilde p^2] } \bigg|\frac{1}{n_L^{1/2}}\sum_{t=1}^{{n_L}}W_{t,k}^{\rm dis}\bigg|\,.
\end{align*}
By Conditions \ref{c.alpha} and \ref{c.subgaussian}, and Lemma \ref{tail_chang} with $( B_n,c_n,  r_1, {r}) =( 1,0, 2, 1/3)$, we have
\begin{align}\label{eq:sumvep}
\max_{\ell\in[L]_0}\max_{j\in[\tilde p]}\mathbb{P}\bigg( \bigg| \sum_{t=1}^{n_\ell}{\vep}^{\rm dis}_{t,j} \bigg|>x   \bigg) \lesssim \exp(-Cn^{-1}x^2)+\exp(-Cx^{1/3})
\end{align}
for any $x\geq 0$. Since 
\begin{align}\label{Wt}
\bigg|\frac{1}{n_L^{1/2}}\sum_{t=1}^{{n_L}}W_{t,k}^{\rm dis}\bigg|\lesssim & \frac{1}{n^{3/2}}\bigg(\bigg|\sum_{t=1}^{n_L}{\vep}^{\rm dis}_{t,j}\bigg|\bigg|\sum_{t=1}^{n}{\vep}^{\rm dis}_{t,j'}\bigg| + \bigg|\sum_{t=1}^{n_L}{\vep}^{\rm dis}_{t+\ell,j'}\bigg|\bigg|\sum_{t=1}^{n}{\vep}^{\rm dis}_{t,j}\bigg|+ \bigg|\sum_{t=1}^{n}{\vep}^{\rm dis}_{t,j}\bigg|\bigg|\sum_{t=1}^{n}{\vep}^{\rm dis}_{t,j'}\bigg|\bigg)\,,
\end{align}
by Bonferroni's inequality and \eqref{eq:sumvep},
\begin{align*}
\mathbb{P}(|\check{Z}-\dot{Z}|>\epsilon_2)
\lesssim&\, \tilde{p}^2\max_{k \in [L\tilde{p}^2] } \mathbb{P}\bigg(\bigg|\frac{1}{{n_L}^{1/2}}\sum_{t=1}^{{n_L}}W_{t,k}^{\rm dis}\bigg|>\epsilon_2\bigg) \\
\lesssim&\,  \tilde{p}^2\exp(-C n^{1/2}\epsilon_2) + \tilde{p}^2 \exp(-Cn^{1/4}\epsilon_2^{1/6})\,.
\end{align*}
By Condition \ref{c.bvar} and \eqref{var_dis}, we know that $\min_{k\in [L\tilde p^2]}\check\Xi^{\rm dis}_{n,kk}\geq  C$. 
In addition, by the Nazarov's inequality 
\citep[see, e.g., Lemma~A.1 of][]{Cher2017_supp}, we have
\begin{align*}
    \sup_{x\in \mathbb{R}} \mathbb{P}( |\check{V}-x|\leq  \epsilon_2) \lesssim \epsilon_2 \sqrt{\log(L\tilde{p}^2)}\,.
\end{align*}
Selecting $\epsilon_2=Cn^{-1/2}\log \tilde{p}$ for some sufficiently large constant $C>0$ yields
\begin{align}\label{first.2}
    \mathbb{P}(|\check{Z}-\dot{Z}|>\epsilon_2)
+\sup_{x\in \mathbb{R}} \mathbb{P}( |\check{V}-x|\leq  \epsilon_2) =o(1)\,,
\end{align}
provided that $\log \tilde{p} \ll n^{1/5}$.

Following from  Proposition 2.1 of \cite{Cher2022_supp},  we have
\begin{align}
\label{comp_2}
\sup_{x\in\mathbb{R}} |\mathbb{P}(\check{V}\leq x)-\mathbb{P}(\dot{V}\leq x) |
\lesssim \check{\Delta}_{0}^{1 / 2}  \log  ({L\tilde{p}^2} ) \, ,
\end{align}
where $\check{\Delta}_{0}=\|\check{\bXi}_{n}^{\rm dis}-\dot{\bXi}_{n}^{\rm dis}\|_{\max}$. Notice that, the $(k_1,k_2)$-th component of $|\check{\bXi}_{n}^{\rm dis}-\dot{\bXi}_{n}^{\rm dis}|$ is
\begin{align*}
|\check\Xi_{n,k_1k_2}^{\rm dis}-\dot\Xi_{n,k_1k_2}^{\rm dis}|
\leq&\,  \underbrace{\bigg| \frac{1}{{n_L}}\sum_{t_1,t_2=1}^{n_L}\mathbb{E}[\{\dot\eta_{t_1,k_1}^{\rm dis}-\mathbb{E}(\dot\eta_{t_1,k_1}^{\rm dis})\}\{W_{t_2,k_2}^{\rm dis}-\mathbb{E}(W_{t_2,k_2}^{\rm dis})\}]
	\bigg|}_{{\rm I}_{1,k_1,k_2}}\\
&\, +\underbrace{\bigg| \frac{1}{{n_L}}\sum_{t_1,t_2=1}^{n_L}\mathbb{E}[
	\{W_{t_1,k_1}^{\rm dis}-\mathbb{E}(W_{t_1,k_1}^{\rm dis})\}\{\dot\eta_{t_2,k_2}^{\rm dis}-\mathbb{E}(\dot\eta_{t_2,k_2}^{\rm dis})\}]
	\bigg|}_{{\rm I}_{2,k_1,k_2}}\\
&\, +\underbrace{\bigg| \frac{1}{{n_L}}\sum_{t_1,t_2=1}^{n_L}\mathbb{E}[
	\{W_{t_1,k_1}^{\rm dis}-\mathbb{E}(W_{t_1,k_1}^{\rm dis})\}\{W_{t_2,k_2}^{\rm dis}-\mathbb{E}(W_{t_2,k_2}^{\rm dis})\}]
	\bigg|}_{{\rm I}_{3,k_1,k_2}} \,.
\end{align*}
Following from \eqref{eq:vep1}, \eqref{Wt} and Cauchy-Schwarz inequality, we have
\begin{align}\label{eq:Wnorm2}
& \mathbb{E}\bigg(\bigg[ \sum_{t=1}^{n_L} \{W_{t,k}^{\rm dis}-\mathbb{E}(W_{t,k}^{\rm dis})\}\bigg]^2\bigg) 
\lesssim  \mathbb{E}\bigg\{\bigg(\sum_{t=1}^{n_L} W_{t,k}^{\rm dis}\bigg)^2\bigg\} \notag\\
&~~~~ \lesssim \frac{1}{n^2}\mathbb{E}\bigg\{\bigg(\sum_{t=1}^{n_L}\vep_{t,j}^{\rm dis}\bigg)^2
\bigg(\sum_{t=1}^n\vep_{t,j'}^{\rm dis}\bigg)^2\bigg\}
+\frac{1}{n^2}\mathbb{E}\bigg\{\bigg(\sum_{t=1}^{n_L}\vep_{t+\ell,j'}^{\rm dis}\bigg)^2
\bigg(\sum_{t=1}^n\vep_{t,j}^{\rm dis}\bigg)^2\bigg\}  \notag\\
&~~~~~~~~  +\frac{1}{n^2}\mathbb{E}\bigg\{\bigg(\sum_{t=1}^{n}\vep_{t,j}^{\rm dis}\bigg)^2
\bigg(\sum_{t=1}^n\vep_{t,j'}^{\rm dis}\bigg)^2\bigg\}  \lesssim 1
\end{align}
for any  $k\in[L\tilde{p}^2]$. In addition, following from \eqref{eq:vep2}, it holds that
\begin{align}\label{eq:cetanorm2}
\bigg\| \sum_{t=1}^{n_L}\{ \dot\eta_{t,k}^{\rm dis}-\mathbb{E}(\dot\eta_{t,k}^{\rm dis}) \} \bigg\|_2
=  \bigg\|\sum_{t=1}^{{n_L}}\{ \vep_{t,j}^{\rm dis}\vep_{t+\ell,j'}^{\rm dis}-\mathbb{E}(\vep_{t,j}^{\rm dis} \vep_{t+\ell,j'}^{\rm dis}) \} \bigg\|_2\lesssim n^{1/2} 
\end{align}
for any  $k\in[L\tilde{p}^2]$. By combining \eqref{eq:Wnorm2}, \eqref{eq:cetanorm2} and Cauchy-Schwarz inequality, we have ${\rm I}_{1,k_1,k_2}\lesssim n^{-1/2}$, ${\rm I}_{2,k_1,k_2}\lesssim n^{-1/2}$ and ${\rm I}_{3,k_1,k_2}\lesssim n^{-1}$ for any $k_1,k_2\in[L\tilde{p}^2]$, which implies $|\check{\Xi}_{n,k_1k_2}^{\rm dis}-\dot{\Xi}_{n,k_1k_2}^{\rm dis}|\lesssim n^{-1/2}$ for any $k_1,k_2\in[L\tilde{p}^2]$. Since this upper bound does not depend on $(k_1,k_2)$, we can conclude that
\begin{align}\label{var_dis2}
\check{\Delta}_{0}=\|\check{\bXi}_{n}^{\rm dis}-\dot{\bXi}_{n}^{\rm dis}\|_{\max}\lesssim n^{-1/2} \,.
\end{align}
This, together with \eqref{comp_2}, yields
\begin{align}\label{comp_2_bound}
    \sup_{x\in \mathbb{R}} |\mathbb{P}(\check{V}\leq x)-\mathbb{P}(\dot{V}\leq x) |
\lesssim n^{-1 / 4}  \log  ({L\tilde{p}^2} ) =o(1)\,,
\end{align}
provided that $\log \tilde{p} \ll n^{1/4}$.

Finally, recall  $d_{3n}=\sup_{x \in  \mathbb{R}}|\mathbb{P}(\dot{Z}\leq x)-\mathbb{P}(\dot{V}\leq x)|$, where $\dot{Z}=\max_{k\in[L\tilde{p}^2]}n_L^{-1/2}\sum_{t=1}^{{n_L}}\dot\eta_{t,k}^{\rm dis}$ and $\dot{V}=\max_{k \in [L\tilde p^2]}\dot{\cG}_k^{\rm dis}$, with $ \dot{\cG}_{k}^{\rm dis} \sim \cN(0,  \var( {n_L^{-1/2}} \sum_{t=1}^{n_L}  \dot\eta_{t,k}^{\rm dis}))$. Let $A_{1,k}=n_L^{-1/2}\sum_{t=1}^{n_L}\dot\eta_{t,k}^{\rm dis}$, $A_{2,k}=n_L^{-1/2}\sum_{t=1}^{n_L}W_{t,k}^{\rm dis}$ and $B_k=n_L^{-1/2}\sum_{t=n_L+1}^{n_{\ell_k}}\tilde\eta_{t,k}^{\rm dis}$ for any $k\in[L\tilde p^2]$. Note that 
\begin{align}\label{low.1}
    \var( A_{1,k}+A_{2,k}+B_k ) = \frac{n_{\ell}^2}{n n_L}\var\bigg(  \frac{n^{1/2}}{n_{\ell}}\sum_{t=1}^{n_\ell} \tilde\eta_{t,k}^{\rm dis}\bigg) \geq C
\end{align}
for any $k\in[L\tilde p^2]$, where the last inequality follows from Condition \ref{c.bvar} and  $ n_{\ell}^2 (n n_L)^{-1} \asymp 1 $. On the other hand, by  \eqref{moment2} and \eqref{eq:Wnorm2}, it holds that $\var(B_{k}) \lesssim n^{-1}$ and $\var(A_{2,k})\lesssim n^{-1}$  for any $k\in[L\tilde p^2]$. Then, it holds that
\begin{align*}
     \var( A_{1,k}+A_{2,k}+B_k ) \leq 3\var( A_{1,k}) +3\var(A_{2,k}) + 3\var(B_{k}) \leq 3\var( A_{1,k}) + Cn^{-1}  
\end{align*}
for any $k\in[L\tilde p^2]$. This, together with \eqref{low.1}, yields
\begin{align}\label{bvar}
 \var(A_{1,k}) =  \var\bigg( \frac{1}{n_L^{1/2}} \sum_{t=1}^{n_L}  \dot\eta_{t,k}^{\rm dis}\bigg)\geq C 
\end{align} 
for any $k\in[L\tilde p^2]$ and sufficiently large $n$. By \eqref{eq:tail2}, we have $\|\dot\eta_{t,k}^{\rm dis}\|_{\psi_{1 }}\leq C$ for any $k\in[L\tilde p^2]$. Since $p\geq n^{\upsilon}$ for some constant $\upsilon>0$ and $L$ is fixed,  by \eqref{bvar} and Condition \ref{c.alpha}, and following similar arguments as the proof of Theorem 1 in \cite{ChangChenWu2023_supp}, it holds  that
\begin{align*}%\label{GA}
d_{3n}=\sup_{x\in \mathbb{R}}|\mathbb{P}(\dot{Z}\leq x)-\mathbb{P}(\dot{V}\leq x)|=o(1)\,,
\end{align*}
provided that $\log \tilde{p} \ll n^{2/21}$. This, together with \eqref{d1ntod2n.2}, \eqref{pre_gau}, \eqref{first.2} and \eqref{comp_2_bound}, yields 
\begin{align*} 
    d_{1n} =\sup_{x\in \mathbb{R}}|\mathbb{P}(\check{Z}\leq x)-\mathbb{P}(V\leq x)|=o(1)\,,
\end{align*}
provided that $\log \tilde{p} \ll n^{2/21}$. Hence, we complete the proof of \eqref{lem.claim1}.

\noindent\underline{  Proof of \eqref{lem.claim2}.} Let ${\bf R}_n =(R_{n,1},\ldots,R_{n,L\tilde p^2})^{\T}=\bUpsilon-n_L^{-1}\sum_{t=1}^{{n_L}}\tilde\bfeta_t^{\rm dis}$. Recall that $Z= n^{1 / 2}\max_{ k \in [L\tilde p^2]} \Upsilon_{k}$ and $\check{Z}=\max_{ k \in [L\tilde p^2] }n_L^{-1/2}\sum_{t=1}^{n_L} {\tilde\eta}_{t,k}^{\rm dis}$. Then, by triangle inequality, we have
\begin{align*}
|Z-\check{Z}| 
=&\, \bigg|(n^{1/2}-n_L^{1/2}) \max _{ k \in [L\tilde p^2]} \Upsilon_{k}+n_L^{1/2} \max _{ k \in [L\tilde p^2]} \Upsilon_{k}-n_L^{1/2}\max_{ k \in [L\tilde p^2]}\bar{\tilde\eta}_k^{\rm dis}\bigg|   \notag\\
\leq&\, (n^{1/2}-n_L^{1/2})\max _{ k \in [L\tilde p^2]} |\Upsilon_{k}|
+n_L^{1/2} \, \max _{ k \in [L\tilde p^2]} |\Upsilon_{k}-\bar{\tilde\eta}_k^{\rm dis}|   \notag\\
\leq &\,  Cn^{-1/2} \max _{ k \in [L\tilde p^2]} |\Upsilon_{k}|
+Cn^{1/2}\max _{ k \in [L\tilde p^2]}|R_{n,k}|   \,,
\end{align*}
where $\bar{\tilde\eta}_k^{\rm dis} = n_L^{-1}\sum_{t=1}^{{n_L}}\tilde\eta_{t,k}^{\rm dis}$. Following from Bonferroni's inequality, it holds that
\begin{align} \label{eq:diff.Z.Zc}
\mathbb{P}(|Z-\check{Z}|> \epsilon_1) 
\leq \mathbb{P}\bigg( n^{-1/2} \max _{ k \in [L\tilde p^2]} |\Upsilon_{k}|> \frac{\epsilon_1}{2C} \bigg)
+\mathbb{P}\bigg(n^{1/2}\max _{ k \in [L\tilde p^2]}|R_{n,k}|> \frac{\epsilon_1}{2C} \bigg)\,.
\end{align} 
Recall    $\Upsilon_{k} = n_\ell^{-1}\sum_{t=1}^{n_\ell} (\varepsilon^{\rm dis}_{t,j}-\bar{{\varepsilon}}^{\rm dis}_{ j}) (\varepsilon^{\rm dis}_{t+\ell,j'}-\bar{\varepsilon}^{\rm dis}_{j'})$.
We have
$|n^{-1/2}\Upsilon_{k}| \lesssim {n^{-3/2}}| \sum_{t=1}^{n_\ell}{\vep}^{\rm dis}_{t,j}{\vep}^{\rm dis}_{t+\ell,j'} |+ {n^{-5/2}}( |\sum_{t=1}^{n_\ell}{\vep}^{\rm dis}_{t,j} | |\sum_{t=1}^{n}{\vep}^{\rm dis}_{t,j'} |+| \sum_{t=1}^{n_\ell}{\vep}^{\rm dis}_{t+\ell,j'} || \sum_{t=1}^{n}{\vep}^{\rm dis}_{t,j}|+ | \sum_{t=1}^{n}{\vep}^{\rm dis}_{t,j} || \sum_{t=1}^{n}{\vep}^{\rm dis}_{t,j'} |)$. By \eqref{eq:tail2}, Condition \ref{c.alpha} and Lemma \ref{tail_chang} with $( B_n,c_n,  r_1, {r} ) =( 1, \ell,1, 1/3)$,    
it holds under $H_0$ that 
\begin{align}\label{eq:sumvep2}
\max_{\ell\in[L]}\max_{j,j'\in[\tilde p]}\mathbb{P}\bigg( \bigg| \sum_{t=1}^{n_L}{\vep}^{\rm dis}_{t,j}{\vep}^{\rm dis}_{t+\ell,j'} \bigg|>x   \bigg) \lesssim \exp(-Cn^{-1}x^2)+\exp(-Cx^{1/3})  
\end{align}
for any $x\geq 0$. Moreover, by \eqref{tail.vep}, Condition \ref{c.alpha} and Lemma \ref{tail_chang} with $( B_n,c_n,  r_1, {r} ) =( 1, 0,2, 1/3)$, we have, for any $x\geq 0$, 
\begin{align}\label{eq:sumvep1}
\max_{\ell\in[L]_0}\max_{j \in[\tilde p]}\mathbb{P}\bigg( \bigg| \sum_{t=1}^{n_{\ell}}{\vep}^{\rm dis}_{t+\ell,j}  \bigg|>x   \bigg) \lesssim \exp(-Cn^{-1}x^2)+\exp(-Cx^{1/3})  \,.
\end{align}
Then, by \eqref{eq:sumvep2}, \eqref{eq:sumvep1} and Bonferroni's inequality, it holds that
\begin{align}\label{tail.R1}
     &\mathbb{P}\bigg( n^{-1/2} \max _{ k \in [L\tilde p^2]} |\Upsilon_{k}|> \frac{\epsilon_1}{2C} \bigg) \lesssim \tilde p^2\max _{ k \in [L\tilde p^2]}\mathbb{P}\bigg(    |n^{-1/2}\Upsilon_{k}|> \frac{\epsilon_1}{2C} \bigg)\nonumber\\
     &~~~\lesssim \tilde p^2\big\{ \exp(-Cn^2\epsilon_1^2) + \exp(-Cn^{1/2}\epsilon_1^{1/3}) +   \exp(-Cn^{3/2}\epsilon_1) +\exp(-Cn^{5/12}\epsilon_1^{1/6})  \big\}\,.
\end{align}
Recall $\tilde{\eta}_{t,k}^{\rm dis}= (\varepsilon^{\rm dis}_{t,j}-\bar{{\varepsilon}}^{\rm dis}_{ j}) (\varepsilon^{\rm dis}_{t+\ell,j'}-\bar{\varepsilon}^{\rm dis}_{j'})$  for any $k\in [L\tilde{p}^2]$.  We have
\begin{align}\label{eq:Rnsplit}
    &n^{1/2}|R_{n,k}|  =  n^{1/2}\bigg| \frac{1}{n_{\ell}}\sum_{t=1}^{n_{\ell}} \tilde\eta_{t,k}^{\rm dis} - \frac{1}{n_L}\sum_{t=1}^{n_L} \tilde\eta_{t,k}^{\rm dis}\bigg |\lesssim \frac{1}{n^{3/2}}\bigg| \sum_{t=1}^{n_L}\tilde\eta_{t,k}^{\rm dis} \bigg|+\frac{1}{n^{1/2}}\bigg| \sum_{t=n_L+1}^{n_{\ell}}\tilde\eta_{t,k}^{\rm dis} \bigg| \nonumber\\
    \lesssim &\, \frac{1}{n^{3/2}}\bigg| \sum_{t=1}^{n_L}{\vep}^{\rm dis}_{t,j}{\vep}^{\rm dis}_{t+\ell,j'} \bigg|
+ \frac{1}{n^{5/2}}\bigg(\bigg| \sum_{t=1}^{n_L}{\vep}^{\rm dis}_{t,j} \bigg|\bigg| \sum_{t=1}^{n}{\vep}^{\rm dis}_{t,j'} \bigg|+\bigg| \sum_{t=1}^{n_L}{\vep}^{\rm dis}_{t+\ell,j'} \bigg|\bigg| \sum_{t=1}^{n}{\vep}^{\rm dis}_{t,j}\bigg|+ \bigg| \sum_{t=1}^{n}{\vep}^{\rm dis}_{t,j} \bigg|\bigg| \sum_{t=1}^{n}{\vep}^{\rm dis}_{t,j'} \bigg|\bigg)  \nonumber\\
  &\, +\frac{1}{n^{1/2}}\bigg| \sum_{t=n_L+1}^{n_{\ell}}{\vep}^{\rm dis}_{t,j}{\vep}^{\rm dis}_{t+\ell,j'} \bigg|+\frac{1}{n^{3/2}}\bigg(\bigg| \sum_{t=n_L+1}^{n_{\ell}}{\vep}^{\rm dis}_{t,j} \bigg|\bigg| \sum_{t=1}^{n}{\vep}^{\rm dis}_{t,j'} \bigg|+\bigg| \sum_{t=n_L+1}^{n_{\ell}}{\vep}^{\rm dis}_{t+\ell,j'} \bigg|\bigg| \sum_{t=1}^{n}{\vep}^{\rm dis}_{t,j}\bigg|\bigg)  
\end{align}
for any $k\in [L\tilde{p}^2]$. By \eqref{eq:tailsum1}, \eqref{eq:sumvep} and Bonferroni's inequality,  it holds that
\begin{align*}
& \max_{j,j'\in[\tilde p]}\mathbb{P}\bigg( \bigg| \sum_{t=1}^{n_L}{\vep}^{\rm dis}_{t,j} \bigg|\bigg| \sum_{t=1}^{n}{\vep}^{\rm dis}_{t,j'} \bigg|>x   \bigg) \lesssim \exp(-Cn^{-1}x)+\exp(-Cx^{1/6}) \,, \\
&  \max_{\ell\in[L]_0}\max_{j,j'\in[\tilde p]} \mathbb{P}\bigg( \bigg| \sum_{t=n_L+1}^{n_{\ell}}{\vep}^{\rm dis}_{t,j} \bigg|\bigg| \sum_{t=1}^{n}{\vep}^{\rm dis}_{t,j'} \bigg|>x   \bigg) \lesssim \exp(-Cn^{-1/2}x )+\exp(-Cn^{1/12}x^{1/6})
\end{align*}
for any $x\geq 0$. This, together with \eqref{eq:tailsum2} and \eqref{eq:sumvep2}, yields 
\begin{align*}
    \bbP\bigg(\frac{1}{n^{3/2}}\bigg| \sum_{t=1}^{n_L}{\vep}^{\rm dis}_{t,j}{\vep}^{\rm dis}_{t+\ell,j'} \bigg| > C\epsilon_1\bigg) \lesssim&~ \exp(-Cn^2\epsilon_1^2) + \exp(-Cn^{1/2}\epsilon_1^{1/3})\,,\\
    \bbP\bigg(\frac{1}{n^{5/2}} \bigg| \sum_{t=1}^{n_L}{\vep}^{\rm dis}_{t,j} \bigg|\bigg| \sum_{t=1}^{n}{\vep}^{\rm dis}_{t,j'} \bigg| > C\epsilon_1\bigg)
	\lesssim&~ \exp(-Cn^{3/2}\epsilon_1) + \exp(-Cn^{5/12}\epsilon_1^{1/6})\,,\\
    \bbP\bigg(\frac{1}{n^{1/2}}\bigg| \sum_{t=n_L+1}^{n_{\ell}}{\vep}^{\rm dis}_{t,j}{\vep}^{\rm dis}_{t+\ell,j'} \bigg| >C\epsilon_1 \bigg)
	\lesssim&~ \exp(-Cn^{1/2}\epsilon_1)\,,\\
    \bbP\bigg(\frac{1}{n^{3/2}}\bigg| \sum_{t=n_L+1}^{n_{\ell}}{\vep}^{\rm dis}_{t,j} \bigg|\bigg| \sum_{t=1}^{n}{\vep}^{\rm dis}_{t,j'} \bigg| > C\epsilon_1\bigg)
	\lesssim& ~\exp(-Cn\epsilon_1) + \exp(-Cn^{1/3}\epsilon_1^{1/6}) \,.
\end{align*}
Following from \eqref{eq:Rnsplit} and Bonferroni's inequality, we can conclude that 
\begin{align}\label{tail.R2}
     &\mathbb{P}\bigg(n^{1/2}\max _{ k \in [L\tilde p^2]}|R_{n,k}|> \frac{\epsilon_1}{2C} \bigg) \lesssim \tilde p^2\max _{ k \in [L\tilde p^2]}\mathbb{P}\bigg(   n^{1/2} |R_{n,k}|> \frac{\epsilon_1}{2C} \bigg)\nonumber\\
     &~~~\lesssim \tilde p^2\big\{ \exp(-Cn^{2}\epsilon_1^{2})+\exp(-Cn^{1/2}\epsilon_1^{1/3})+\exp(-Cn^{1/2}\epsilon_1 )+\exp(-Cn^{1/3}\epsilon_1^{1/6})    \big\}\,.
\end{align}
Combining \eqref{eq:diff.Z.Zc}, \eqref{tail.R1}, and \eqref{tail.R2}, and recalling that 
$\epsilon_1 = C_0n^{-1/2}\log \tilde{p}$, we obtain
\begin{align} \label{Zmins} 
\mathbb{P}(|Z-\check{Z}|> \epsilon_1)  =o(1)\,,
\end{align} 
provided that $\log \tilde{p} \ll n^{3/10}$. By   Condition \ref{c.bvar} and the Nazarov's inequality 
\citep[see, e.g., Lemma~A.1 of][]{Cher2017_supp}, we have
\begin{align*}
    \sup_{x\in \mathbb{R}}\mathbb{P}(| V-x| \leq \epsilon_1) \lesssim \epsilon_1 \sqrt{\log(L\tilde{p}^2)} =o(1)\,,
\end{align*}
provided that $\log \tilde{p} \ll n^{1/3}$. Combining this with \eqref{Zmins} proves \eqref{lem.claim2}, 
and hence completes the proof of Lemma~\ref{lemma_tn_dis_GA}.   $\hfill \Box$

\subsection{Proof of Lemma \ref{lem:diff.Tn.Gstar}}\label{proof.lem:diff.Tn.Gstar}

By the definitions of $\tilde{T}_{n}^{{\rm G*}}$ and 
$\tilde{T}^{\rm dis,G*}_{n}$, it holds that
\begin{align*}
    |\tilde{T}^{\rm G*}_{n}-\tilde{T}^{\rm dis,G*}_{n}| \leq \max_{r\in[Lp^2]}\max_{m,m'\in[M]}\sup_{(u,v)\in B_m\times B_{m'}}|\tilde\cG_r^*(u,v)-\tilde\cG_r^*(b_m,b_{m'})|\,.
\end{align*}
   For some sufficiently large constant $\tilde{C}>0$, define an event
\begin{align*}
\cE(\widetilde{\mathcal{D}}_n) := \left\{\max_{r\in[Lp^2]}\max_{m,m'\in[M]}\sup_{\substack{(u_1,v_1),(u_2,v_2)\in B_m\times B_{m'}\\ (u_1,v_1)\neq (u_2,v_2)}} \frac{\mathbb{E}[\{\tilde{\cG}^*_r(u_1,v_1)-\tilde{\cG}^*_r(u_2,v_2) \}^2\,|\,\widetilde{\mathcal{D}}_n]}{|u_1-u_2|^{\kappa }+|v_1-v_2|^{\kappa }} \leq \tilde{C}n^\rho  \right\}\,.
\end{align*}
 We   claim that 
\begin{align}\label{event}
    \bbP\{  \cE^{\rm c}(\widetilde{\mathcal{D}}_n) \} = o(1)\,.
\end{align}   
For any $r\in[Lp^2]$, $m,m'\in[M]$ and $(u_1,v_1),(u_2,v_2)\in B_m\times B_{m'}$, it holds under $\cE(\widetilde{\mathcal{D}}_n)$ that 
\begin{align*}
    \frac{(\mathbb{E}[\{\tilde{\cG}^*_r(u_1,v_1)- \tilde{\cG}^*_r(u_2,v_2)\}^2\,|\,\widetilde{\mathcal{D}}_n])^{1/2}}{n^{\rho/2}}
\lesssim |u_1-u_2|^{\kappa/2 } + |v_1-v_2|^{\kappa/2 } \,.
\end{align*}
By \eqref{event}, we have $\bbP\{  \cE^{\rm c}(\widetilde{\mathcal{D}}_n)\,|\,\widetilde{\mathcal{D}}_n \} =o_{\rm p}(1) $. Since  $\max\{|u-b_m|,|v-b_{m'}|\}\leq (2M)^{-1}$ for any $(u,v)\in B_m\times B_{m'}$, by Lemma \ref{meer} with $q=2$, $(a_1,a_2)^{\T}=((2M)^{-1},(2M)^{-1})^{\T}$ and $\bh=(\kappa/2,\kappa/2)^{\T}$,  it holds that
\begin{align*}
& \mathbb{P}(|\tilde{T}^{\rm G*}_{n}-\tilde{T}^{\rm dis,G*}_{n}|>\delta\,|\,\widetilde{\mathcal{D}}_n)  \notag\\
&~~~~ \leq \mathbb{P}\bigg\{\max_{r\in[Lp^2]}\max_{m,m'\in[M]}\sup_{(u,v)\in B_m\times B_{m'}}|\tilde{\cG}^*_r(u,v)-\tilde{\cG}^*_r(b_m,b_{m'})|>\delta\,\bigg|\,\widetilde{\mathcal{D}}_n\bigg\} \notag\\
&~~~~\leq Lp^2M^2\max_{r\in[Lp^2]}\max_{m,m'\in[M]}\mathbb{P}\bigg\{\sup_{(u,v)\in B_m\times B_{m'}}\bigg|\frac{\tilde{\cG}^*_r(u,v)}{n^{\rho/2}}-\frac{\tilde{\cG}^*_r(b_m,b_{m'})}{n^{\rho/2}}\bigg|> \frac{\delta}{n^{\rho/2}},\,\cE(\widetilde{\mathcal{D}}_n)\,\bigg|\,\widetilde{\mathcal{D}}_n\bigg\}  \\
&~~~~~~~~ +\mathbb{P}\{\cE^{\rm c}(\widetilde{\mathcal{D}}_n)\,|\,\widetilde{\mathcal{D}}_n\} \nonumber\\
&~~~~ \lesssim p^2M^2\exp(-CM^{\kappa}n^{-\rho}\delta^2)+o_{\rm p}(1) 
\end{align*}
for any $\delta\geq CM^{-\kappa/2} n^{\rho/2}$. Recall $M\asymp (np)^{C_M}$ and $p\geq n^{\upsilon}$.  Letting $\delta=C_0 M^{-\kappa/2}n^{\rho/2}(\log   {p})^{1/2} \\\lesssim n^{-1}(\log  {p})^{1/2}$ for some sufficiently large constant $C_0>0$, we have
\begin{align*}
    \mathbb{P}\{|\tilde{T}^{\rm G*}_{n}-\tilde{T}^{\rm dis,G*}_{n}|>Cn^{-1}(\log   {p})^{1/2}\,|\,\widetilde{\mathcal{D}}_n\}  = o_{\rm p}(1)\,.
\end{align*}
  Therefore, to confirm Lemma \ref{lem:diff.Tn.Gstar},
it remains to show that \eqref{event} holds.

\noindent\underline{Proof of \eqref{event}.}  Recall that, for any $r\in[Lp^2]$ and $(u,v)\in\cU^2$,  
\begin{align*}
     \tilde{\cG}^*_r(u,v) = \frac{1}{n_L^{1/2}}\sum_{t=1}^{{n_L}}\varrho_t\{\tilde{\eta}_{t,r}(u,v)-\bar{\tilde{\eta}}_r(u,v)\}\,,
\end{align*}
where $\bar{\tilde{\eta}}_r(u,v)=n_L^{-1}\sum_{t=1}^{{n_L}}\tilde{\eta}_{t,r}(u,v)$ with $\tilde{\eta}_{t,r}(u,v)= \{\vep_{t,j}(u)-\bar\vep_j(u)\}\{\vep_{t+\ell,j'}(v)-\bar\vep_{j'}(v)\}$, and $\boldsymbol\varrho=(\varrho_1,\dots,\varrho_{{n_L}})^{\T}\sim \cN({\bf 0},\bXi)$  with $\Xi_{ij}= \mathcal{W}\{(i-j)/{b_{n}}\}$. Let $\dot{\tilde{\eta}}_{t,r}(u,v) = \tilde{\eta}_{t,r}(u,v)-\bar{\tilde{\eta}}_r(u,v)$ for any $t\in[n_L]$, $r\in[Lp^2]$ and $(u,v)\in \cU^2$. Then
\begin{align*}
     \mathbb{E}[\{\tilde{\cG}^*_r(u_1,v_1)-\tilde{\cG}^*_r(u_2,v_2) \}^2\,|\,\widetilde{\mathcal{D}}_n]= \sum_{i=-n_L+1}^{n_L-1} \mathcal{W}\left(\frac{i}{b_{n}}\right) \widetilde{H}^{*}_{i,r}(u_1,v_1,u_2,v_2) 
\end{align*} 
where $\widetilde{H}^{*}_{i,r}(u_1,v_1,u_2,v_2)= {n_L}^{-1}\sum_{t=i+1}^{n_L}
\{\dot{\tilde{\eta}}_{t,r}(u_1,v_1)-\dot{\tilde{\eta}}_{t,r}(u_2,v_2)\}
\{\dot{\tilde{\eta}}_{t-i,r}(u_1,v_1)-\dot{\tilde{\eta}}_{t-i,r}(u_2,v_2)\}$ for $i\geq 0$, and $\widetilde{H}^{*}_{i,r}(u_1,v_1,u_2,v_2)= {n_L}^{-1}\sum_{t=-i+1}^{n_L}
\{\dot{\tilde{\eta}}_{t+i,r}(u_1,v_1)-\dot{\tilde{\eta}}_{t+i,r}(u_2,v_2)\}
\{\dot{\tilde{\eta}}_{t,r}(u_1,v_1)-\dot{\tilde{\eta}}_{t,r}(u_2,v_2)\}$ otherwise. 
Condition \ref{c.kernelF} yields $ \sum_{i=-n_L+1}^{n_L-1}| \mathcal{W}({i}/{b_{n}})|\lesssim b_n \asymp n^{\rho}$. Then, for any $(r,u_1,v_1,u_2,v_2)$, it holds that 
\begin{align*}
    \mathbb{E}[\{\tilde{\cG}^*_r(u_1,v_1)-\tilde{\cG}^*_r(u_2,v_2) \}^2\,|\,\widetilde{\mathcal{D}}_n] \lesssim  n^{\rho} \max_{t\in[n_L]} |\dot{\tilde{\eta}}_{t,r}(u_1,v_1)-\dot{\tilde{\eta}}_{t,r}(u_2,v_2)  |^2\,.
\end{align*}
Moreover, for any $(u_1,v_1),(u_2,v_2)\in \cU^2$, we define $d\{(u_1,v_1),(u_2,v_2)\} =   |u_1-u_2|^{\kappa}+|v_1-v_2|^{\kappa}$. Therefore, it suffices, for proving \eqref{event}, to establish that, for some sufficiently large constant $C>0$,
 \begin{align}\label{Delta.1}
     \bbP\left\{  \max_{t\in[n_L]}\max_{r\in[Lp^2]}\max_{m,m'\in[M]}\sup_{\substack{(u_1,v_1),(u_2,v_2)\in B_m\times B_{m'}\\(u_1,v_1)\neq (u_2,v_2) }}  \frac{|\dot{\tilde{\eta}}_{t,r}(u_1,v_1)-\dot{\tilde{\eta}}_{t,r}(u_2,v_2)  | }{ [d\{(u_1,v_1),(u_2,v_2)\}]^{1/2} } > C \right\} = o(1)\,.
 \end{align}
Notice that $\tilde{\eta}_{t ,r}(u_1,v_1)-\tilde{\eta}_{t ,r}(u_2,v_2)=\{ \vep_{t ,j}(u_1)\vep_{t +\ell,j'}(v_1)-\vep_{t ,j}(u_2)\vep_{t +\ell,j'}(v_2) \}-\{\vep_{t ,j}(u_1)\bar\vep_{j'}(v_1)-\vep_{t ,j}(u_2)\bar\vep_{j'}(v_2)  \}-\{ \bar\vep_{j}(u_1)\vep_{t +\ell,j'}(v_1)-\bar\vep_{j}(u_2)\vep_{t +\ell,j'}(v_2) \}+\{\bar\vep_{j}(u_1)\bar\vep_{j'}(v_1)-\bar\vep_{j}(u_2)\bar\vep_{j'}(v_2)  \}$. By Condition \ref{c.subgaussian} and  \eqref{eq:tailprob.vep2} in Section \ref{proof.lem.dis.Gaussian}, we have $\|\tilde{\eta}_{t ,r}(u_1,v_1)-\tilde{\eta}_{t ,r}(u_2,v_2)\|_{\psi_{1}}\lesssim d\{(u_1,v_1),(u_2,v_2)\}$ and  $\| \bar{\tilde{\eta}}_r(u_1,v_1)-\bar{\tilde{\eta}}_r(u_2,v_2)\|_{\psi_{1}}\lesssim d\{(u_1,v_1),(u_2,v_2)\}$, which implies 
\begin{align}\label{tail.eta}
    \|  \dot{\tilde{\eta}}_{t ,r}(u_1,v_1)- \dot{\tilde{\eta}}_{t ,r}(u_2,v_2) \|_{\psi_{1}}  \lesssim d\{(u_1,v_1),(u_2,v_2)\}
\end{align} 
for any  $t\in[n_L]$, $r\in[Lp^2]$ and $(u_1,v_1),(u_2,v_2)\in \cU^2$. This further implies $|\bbE\{\dot{\tilde{\eta}}_{t ,r}(u_1,v_1)- \dot{\tilde{\eta}}_{t ,r}(u_2,v_2) \}| \lesssim  d\{(u_1,v_1),(u_2,v_2)\} $ for any $t\in[n_L]$, $r\in[Lp^2]$ and $(u_1,v_1),(u_2,v_2)\in \cU^2$. Since $\sup_{(u_1,v_1),(u_2,v_2)\in B_m\times B_{m'}} d\{(u_1,v_1),(u_2,v_2)\} \lesssim 1$ for any $(m,m')$, by \eqref{Delta.1}, it suffices to prove that, for some sufficiently large constant $C>0$,
 \begin{align}\label{Delta.1.new}
     \bbP\left\{  \max_{t\in[n_L]}\max_{r\in[Lp^2]}\max_{m,m'\in[M]}\sup_{\substack{(u_1,v_1),(u_2,v_2)\in B_m\times B_{m'}\\(u_1,v_1)\neq (u_2,v_2) }}  \frac{|\mathring{\tilde{\eta}}_{t,r}(u_1,v_1)-\mathring{\tilde{\eta}}_{t,r}(u_2,v_2)  | }{ [d\{(u_1,v_1),(u_2,v_2)\}]^{1/2} } > C \right\}  = o(1) \,,
 \end{align} 
where $\mathring{\tilde{\eta}}_{t,r}(u,v)= \dot{\tilde{\eta}}_{t,r}(u,v)- \bbE\{\dot{\tilde{\eta}}_{t,r}(u,v)\}$.  By \eqref{tail.eta}, we have for any $x\geq 0$,  
 \begin{align*} 
      \bbP\big\{ |\mathring{\tilde{\eta}}_{t,r}(u_1,v_1)-\mathring{\tilde{\eta}}_{t,r}(u_2,v_2)  |   > x \big\} \leq C\exp\bigg[- \frac{Cx}{d\{(u_1,v_1),(u_2,v_2)\} }\bigg] \,.
 \end{align*}
 Write $D = \max_{m,m'\in[M]}\sup_{(u_1,v_1),(u_2,v_2)\in B_m\times B_{m'}} d\{(u_1,v_1),(u_2,v_2)\} $. By Condition \ref{c.subgaussian}, the process
$\{\mathring{\tilde\eta}_{t,r}(u,v):(u,v)\in\mathcal U^2\}$
is separable with respect to $d$. Therefore, by Lemma \ref{tail_multi}(ii)   with $(c_{1n},c_{2n},\gamma,\beta)=(C, C , 1, 1/2)$ and metric $d$, it holds for any $x >0$ that 
\begin{align*}
&\bbP\left(\sup_{\substack{(u_1,v_1),(u_2,v_2)\in B_m\times B_{m'}\\(u_1,v_1)\neq (u_2,v_2) }}  \frac{|\mathring{\tilde{\eta}}_{t,r}(u_1,v_1)-\mathring{\tilde{\eta}}_{t,r}(u_2,v_2)  | }{[d  \{(u_1,v_1),(u_2,v_2)\}]^{1/2}} >x \right)  \\
&~~~~~~~~  	\lesssim \frac{1}{x}\bigg[\int_0^{D_1}\frac{ \log(eN_\epsilon) }{\epsilon^{1/2}}\, {\rm d}\epsilon + \frac{\int_0^{D_2}  \log(eN_\epsilon) \,{\rm d}\epsilon}{D_1^{1/2}}  \bigg] \\
&~~~~~~~~~~~  +  \bigg\{N_{D_1}^2 \log_2 \bigg(\frac{1}{D_1}\bigg) + N_{D_2}\log_2\bigg(\frac{1}{D_1}\bigg) \log_2\bigg(\frac{1}{D_2}\bigg) \bigg\}	 \exp\bigg(-\frac{Cx }{ D^{1/2}}\bigg)  \\
&~~~~~~~~  	\lesssim  \frac{1}{x} \bigg(\frac{ D_1^{1/4}}{M^{\kappa/4}} + \frac{ D_2^{3/4}}{M^{\kappa/4}D_1^{1/2}} \bigg) 
+  \frac{1}{M^{4}D_1^{4/\kappa}} \log_2\bigg(\frac{1}{D_1}\bigg) \exp\bigg(-\frac{Cx }{ D^{1/2}}\bigg)
 \\
&~~~~~~~~~~~ +  \frac{1}{M^{2} D_2^{2/\kappa}}\log_2\bigg(\frac{1}{D_1}\bigg) \log_2\bigg( \frac{1}{D_2}\bigg)   \exp\bigg(-\frac{Cx }{ D^{1/2}}\bigg)
\end{align*}
for some $D_1,D_2\in (0, D\wedge 1)$ with $D_2<D_1$, where   the last inequality follows from $N_{\epsilon} \leq \max_{m,m'\in[M]} N(B_m\times B_{m'}, d ;\epsilon) \lesssim M^{-2}\epsilon^{-2/\kappa}$,  
\begin{align*}
   &\int_0^{D_1} \frac{ \log(eN_{\epsilon}) }{\epsilon^{1/2}}\,{\rm d}\epsilon
   \leq   \int_0^{D_1}\frac{ \log(CM^{-2}\epsilon^{-2/\kappa}) }{\epsilon^{1/2}}\,{\rm d}\epsilon \\ 
  &~~~~~~~~\lesssim   M^{-\kappa/2 }\int^{+\infty}_{\log(CM^{-2}D_1^{-2/\kappa})} y e^{-\kappa y/4}\,{\rm d}y  
   \lesssim    \frac{ D_1^{1/4}}{M^{\kappa/4}} \,,
\end{align*} 
 and $\int_0^{D_2}  \log(eN_\epsilon) \,{\rm d}\epsilon \lesssim D_2^{3/4}M^{-\kappa/4}$.
Notice that $D\asymp M^{- \kappa}$ and $M\asymp (np)^{C_M}$. Choose
$D_2=D_1/2$  and set $D_1\asymp M^{-a}$ for some
constant $a> \kappa$ large enough such that $np^2M^2 \ll M^{(a+\kappa)/4}$.
Then $D_1\ll D$. 
Consequently,
 \begin{align*} 
     &\bbP\left( \max_{t\in[n_L]} \max_{r\in[Lp^2]}\max_{m,m'\in[M]}\sup_{\substack{(u_1,v_1),(u_2,v_2)\in B_m\times B_{m'}\\(u_1,v_1)\neq (u_2,v_2) }}  \frac{|\mathring{\tilde{\eta}}_{t,r}(u_1,v_1)-\mathring{\tilde{\eta}}_{t,r}(u_2,v_2)  | }{[d  \{(u_1,v_1),(u_2,v_2)\}]^{1/2}} > C \right) \\
     &~~\lesssim  np^2M^2 \max_{t\in[n_L]}\max_{r\in[Lp^2]}\max_{m,m'\in[M]} \bbP\left(\sup_{\substack{(u_1,v_1),(u_2,v_2)\in B_m\times B_{m'}\\(u_1,v_1)\neq (u_2,v_2) }} \frac{|\mathring{\tilde{\eta}}_{t,r}(u_1,v_1)-\mathring{\tilde{\eta}}_{t,r}(u_2,v_2)  | }{[d  \{(u_1,v_1),(u_2,v_2)\}]^{1/2}}   > C \right) \\
     &~~\lesssim  np^2M^2 \frac{ D_1^{1/4}}{M^{\kappa/4}}
+  \bigg[\frac{np^2}{M^{2}D_1^{4/\kappa}} \log_2\bigg(\frac{1}{D_1}\bigg) +\frac{np^2}{ D_1^{2/\kappa}}\bigg\{\log_2\bigg(\frac{1}{D_1}\bigg)  \bigg\}^2  \bigg]\exp\bigg(-\frac{C }{ D^{1/2}}\bigg) = o(1)\,.
 \end{align*} 
This proves \eqref{Delta.1.new}, hence implies \eqref{event} and completes the proof of Lemma~\ref{lem:diff.Tn.Gstar}.  
% We first show that $ \|\sum_{i=1}^s  X_i\|_{\psi_{1/2}}\lesssim \sum_{i=1}^s\|  X_i\|_{\psi_{1/2}}$ for any positive integer $s$. Write $K_i  = \|  X_i\|_{\psi_{1/2}}$ and $K = \sum_{i=1}^s K_i$. Notice that $K_i = \||X_i|^{1/2}\|_{\psi_1}^2$. Then, by the equivalence in Proposition 2.7.1 of  \cite{Vershynin2018_supp}, it holds that, for any $q\geq 1$,
% \begin{align*}
%     \|X_i\|_q = \||X_i|^{1/2}\|_{2q}^2 \lesssim (2q K_i^{1/2})^2  \lesssim q^2 K_i\,.
% \end{align*}
% Therefore, by Minkowski inequality, we have $\|  \sum_{i=1}^s  X_i \|_{q} \leq \sum_{i=1}^s   \|    X_i \|_{q} \lesssim q^2 K$ for any $q\geq 1$, which implies  
% \begin{align*}
%     \bigg\| \bigg| \sum_{i=1}^s  X_i \bigg|^{1/2}   \bigg\|_{q} = \bigg\|  \sum_{i=1}^s  X_i \bigg\|_{q/2}^{1/2} \lesssim q K^{1/2}  
% \end{align*}
%  for any $q \geq 2$. In addition, for $1\leq  q <2$, we also have $\|  | \sum_{i=1}^s  X_i |^{1/2}    \|_{q} \leq \|  | \sum_{i=1}^s  X_i |^{1/2}    \|_{2} \lesssim q K^{1/2}$. Applying the same equivalence in Proposition 2.7.1 of \cite{Vershynin2018_supp} yields $\| | \sum_{i=1}^s  X_i |^{1/2} \|_{\psi_{1}} \lesssim K^{1/2}$, which implies 
%  \begin{align}\label{psi_triangle}
%      \bigg\|\sum_{i=1}^s  X_i\bigg\|_{\psi_{1/2}} = \bigg\|\bigg| \sum_{i=1}^s  X_i \bigg|^{1/2} \bigg\|_{\psi_{1}}^2 \lesssim K = \sum_{i=1}^s\|  X_i\|_{\psi_{1/2}}\,.
%  \end{align}
$\hfill\Box$

\subsection{Proof of Lemma \ref{disGAstar}}\label{proof.disGAstar}

Recall $\tilde{T}^{\rm dis, G}_{n} = \max_{k\in[L\tilde{p}^2]}|\tilde\cG^{\rm dis}_k|$ and $\tilde{T}_{n}^{\rm dis,G*}=\max_{ k \in [L\tilde p^2]} |\tilde\cG^{\rm dis,*}_{k}|$, 
where $\tilde{p}=pM$, $\tilde{\boldsymbol{\cG}}^{\rm dis}=(\tilde\cG^{\rm dis}_1,\dots,\tilde\cG^{\rm dis}_{L\tilde p^2})^{\T}\sim \cN(\0, \widetilde\bXi^{\rm dis}_{n})$ and $\tilde{\boldsymbol{\cG}}^{\rm dis,*} =(\tilde\cG^{\rm dis,*}_1,\dots,\tilde\cG^{\rm dis,*}_{L\tilde p^2})^{\T} \sim \cN(\0,\widetilde\bXi^{\rm dis,*}_{n})$ conditional on $\widetilde{\mathcal{D}}_n$, respectively. Here, $\widetilde\bXi^{\rm dis}_{n}$ is specified in the proof of Lemma \ref{lemma_tn_dis_GA} (Section \ref{proof.lemma_tn_dis_GA}) and $\widetilde\bXi^{\rm dis, *}_{n}=\sum_{i=-{n_L}+1}^{{n_L}-1} \mathcal{W}(i/b_{n}) \widetilde{\bH}^{\rm dis,*}_{i}$, where
\begin{align*} 
\widetilde{\bH}^{\rm dis,*}_{i}
=
\begin{cases}
\displaystyle
\frac{1}{n_L}\sum_{t=i+1}^{n_L}
(\tilde\bfeta_{t}^{\rm dis}-\bar{\tilde\bfeta}^{\rm dis})( \tilde\bfeta_{t-i}^{\rm dis}-\bar{\tilde\bfeta}^{\rm dis})^{\T}\,,
  &i\geq 0\,, \\ 
\displaystyle
\frac{1}{n_L}\sum_{t=-i+1}^{n_L}
(\tilde\bfeta_{t+i}^{\rm dis}-\bar{\tilde\bfeta}^{\rm dis})( \tilde\bfeta_{t}^{\rm dis}-\bar{\tilde\bfeta}^{\rm dis})^{\T}\,,
 & i<0 \,,
\end{cases} 
\end{align*}
where $\{\tilde{\bfeta}_{t}^{\rm dis}\}$ is  specified in the proof of   
Lemma \ref{lemma_tn_dis_GA} and $\bar{\tilde\bfeta}^{\rm dis} = n_L^{-1}\sum_{t=1}^{n_L} \tilde{\bfeta}_{t}^{\rm dis}$. We claim  that there exist two positive constants ${c}_1$ and ${c}_2$ depending only on $(\rho,\vartheta)$ such that
\begin{align}\label{disSigma}
\hat{\Delta}_{0}:=\|\widetilde\bXi^{\rm dis}_{n}-\widetilde{\bXi}^{\rm dis,*}_{n}\|_{\max }= O_{\rm p}\{n^{-
	{c}_1}(\log  {p})^{{c}_2}\}\,,
\end{align}
provided that $0<\rho< (\vartheta-1)/(3\vartheta-2)$. Recall $\tilde{p}=pM$, $M\asymp (np)^{C_M}$ and $p\geq n^{\upsilon}$. By Condition \ref{c.bvar} and  Proposition 2.1 of \cite{Cher2022_supp}, it holds that
\begin{align*} 
\sup_{x\in \mathbb{R}} |\mathbb{P}(\tilde{T}_n^{\rm dis,G*}\leq x\,|\,\widetilde{\mathcal{D}}_n)-\mathbb{P}(\tilde{T}^{\rm dis,G}_{n}\leq x) | 
\lesssim  \hat{\Delta}_{0}^{1 / 2} \log  (L\tilde p^2 ) =o_{\rm p}(1)  \, ,
\end{align*}
provided that $0<\rho< (\vartheta-1)/(3\vartheta-2)$
and $\log p \ll n^\iota$ for some $\iota>0$ depending only on $(\rho,\vartheta)$. Therefore, to confirm Lemma \ref{disGAstar},
it remains to show that \eqref{disSigma} holds.

\noindent\underline{Proof of \eqref{disSigma}.} 
Define $\dot{\bXi}_{n}^{\rm dis} = \var( {n_L^{-1/2}} \sum_{t=1}^{n_L}  \dot\bfeta_t^{\rm dis})$, where  $\{\dot\bfeta_t^{\rm dis}\}$ is specified in \eqref{dot.eta}. By \eqref{var_dis} and \eqref{var_dis2}, it holds that
\begin{align}\label{dis.theta1}
    \hat{\Delta}_{0} \lesssim \|\dot\bXi^{\rm dis}_{n}-\widetilde{\bXi}^{\rm dis,*}_{n}\|_{\max } + n^{-1/2}\,.
\end{align}
Analogously, define $\dot\bXi^{\rm dis, *}_{n}=\sum_{i=-{n_L}+1}^{{n_L}-1} \mathcal{W}(i/b_{n}) \dot{\bH}^{\rm dis,*}_{i}$, where
\begin{align*} 
\dot{\bH}^{\rm dis, *}_{i}
=
\begin{cases}
\displaystyle
\frac{1}{n_L}\sum_{t=i+1}^{n_L}
(\dot\bfeta_{t}^{\rm dis}-\bar{\dot\bfeta}^{\rm dis})( \dot\bfeta_{t-i}^{\rm dis}-\bar{\dot\bfeta}^{\rm dis})^{\T}\,,
  &i\geq 0\,, \\ 
\displaystyle
\frac{1}{n_L}\sum_{t=-i+1}^{n_L}
(\dot\bfeta_{t+i}^{\rm dis}-\bar{\dot\bfeta}^{\rm dis})( \dot\bfeta_{t}^{\rm dis}-\bar{\dot\bfeta}^{\rm dis})^{\T}\,,
 & i<0 \,,
\end{cases} 
\end{align*}
where $\bar{\dot\bfeta}^{\rm dis} = n_L^{-1}\sum_{t=1}^{n_L} \dot{\bfeta}_{t}^{\rm dis}$. Since $\{\bvep_t(\cdot)\}_{t=1}^{n}$ is a weakly stationary functional time series, it holds that $\dot\bfeta_{t}^{\rm dis}-\bar{\dot\bfeta}^{\rm dis} = \dot\bfeta_{t}^{\rm dis}-  \mathbb{E} (\dot\bfeta_{t}^{\rm dis})  - \{\bar{\dot\bfeta}^{\rm dis} -\mathbb{E}(\bar{\dot\bfeta}^{\rm dis})\} $ for any $t\in[n_L]$. Therefore, by Conditions \ref{c.alpha}, \ref{c.subgaussian}, \ref{c.kernelF} and following similar arguments as the proof of Theorem 11 in \cite{ChangChenWu2023_supp} with $B_n=1$, it also holds that, if $0<\rho<(\vartheta-1)/(3\vartheta-2)$, there exist two constants $c_1',c_2'>0$ depending only on $(\rho, \vartheta)$ such that
\begin{align*} 
\|\dot{\bXi}_{n}^{\rm dis}-\dot\bXi^{\rm dis, *}_{n}\|_{\max}=O_{\rm p}\{ n^{-c_1'}(\log p)^{c_2'} \}+O(n^{-\rho}) \,.
\end{align*}
This, together with \eqref{dis.theta1}, yields
\begin{align}\label{dis.theta2}
    \hat{\Delta}_{0} \lesssim \|\dot{\bXi}^{\rm dis,*}_{n}-\widetilde{\bXi}^{\rm dis,*}_{n}\|_{\max } + O_{\rm p}\{ n^{-c_1'}(\log p)^{c_2'} \}+O(n^{-\rho})\,.
\end{align}
For notational simplicity, for any fixed $k\in[L\tilde{p}^2]$, we write
$(\ell,j,j')=(\ell_k,j_k,j'_k)$ for its associated unique triple. For any $t\in[n_L]$ and $k\in[L\tilde{p}^2]$, denote
\begin{align*}
    &~Q_{t,k}^{(1)}=\vep_{t,j}^{\rm dis}\vep_{t+\ell,j'}^{\rm dis}\,,~~Q_{t,k}^{(2)}=-\vep_{t,j}^{\rm dis}\bar{\vep}^{\rm dis}_{j'}\,,~~Q_{t,k}^{(3)}=-\bar{\vep}^{\rm dis}_{j}\vep_{t+\ell,j'}^{\rm dis}\,,\\
    &~Q_{t,k}^{(4)}=-
\frac{1}{n_L}\sum_{t'=1}^{n_L}\vep_{t',j}^{\rm dis}\vep_{t'+\ell,j'}^{\rm dis}\,,~~Q_{t,k}^{(5)}=\bigg(\frac{1}{n_L}\sum_{t'=1}^{n_L}\vep_{t',j}^{\rm dis}\bigg)\bar{\vep}^{\rm dis}_{j'} \,,~~Q_{t,k}^{(6)}=\bar{\vep}^{\rm dis}_{j}\bigg(\frac{1}{n_L}\sum_{t'=1}^{n_L}\vep_{t'+\ell,j'}^{\rm dis}\bigg)\,,
\end{align*}
where $\bar{\vep}^{\rm dis}_{j} = n^{-1}\sum_{t=1}^n {\vep}^{\rm dis}_{t,j}$.
 Then it holds that $\tilde\eta_{t,k}^{\rm dis} -\bar{\tilde\eta}_{k}^{\rm dis}=\sum_{a=1}^6 Q_{t,k}^{(a)}$ and $\dot{\eta}_{t,k}-\bar{\dot{\eta}}_{k}=Q_{t,k}^{(1)}+Q_{t,k}^{(4)}$ for any $t\in[n_L]$ and $k\in[L\tilde{p}^2]$. Let 
\begin{align*}
    \mathcal A
=
\bigl\{(a_1,a_2)\in[6]^2:
(a_1,a_2)\notin\{(1,1),(1,4),(4,1),(4,4)\}\bigr\}\,.
\end{align*}
Then, we have $(\tilde\eta_{t_1,k_1}^{\rm dis} -\bar{\tilde\eta}_{k_1}^{\rm dis})(\tilde\eta_{t_2,k_2}^{\rm dis} -\bar{\tilde\eta}_{k_2}^{\rm dis})-(\dot\eta_{t_1,k_1}^{\rm dis} -\bar{\dot\eta}_{k_1}^{\rm dis})(\dot\eta_{t_2,k_2}^{\rm dis} -\bar{\dot\eta}_{k_2}^{\rm dis}) = \sum_{(a_1,a_2)\in \mathcal{A}}Q_{t_1,k_1}^{(a_1)}Q_{t_2,k_2}^{(a_2)}$, which implies that
\begin{align}\label{dis.theta3.pre}
    \|\dot{\bXi}^{\rm dis,*}_{n}-\widetilde{\bXi}^{\rm dis,*}_{n}\|_{\max } 
    \leq &~ \sum_{(a_1,a_2)\in \mathcal{A}}\max_{k_1,k_2\in[L\tilde{p}^2]}\bigg|\frac{1}{n_L}\sum_{t_1,t_2=1}^{n_L} \mathcal{W}\bigg(\frac{t_1-t_2}{b_{n}}\bigg)Q_{t_1,k_1}^{(a_1)}Q_{t_2,k_2}^{(a_2)}\bigg|\nonumber\\
     	= &~  \underbrace{\sum_{(a_1,a_2)\in \mathcal{A}}\max_{k_1,k_2\in[L\tilde{p}^2]} \bigg|\sum_{i=0}^{n_L-1}\mathcal{W}\bigg(\frac{i}{b_{n}}\bigg) \bigg\{\frac{1}{n_L}\sum_{t=i+1}^{n_L}Q_{t,k_1}^{(a_1)}Q_{t-i,k_2}^{(a_2)}\bigg\}\bigg|}_{{\rm I}} \\
&~  +\underbrace{\sum_{(a_1,a_2)\in \mathcal{A}}\max_{k_1,k_2\in[L\tilde{p}^2]}\bigg|\sum_{i=-n_L+1}^{-1}\mathcal{W}\bigg(\frac{i}{b_{n}}\bigg) \bigg\{\frac{1}{n_L}\sum_{t=-i+1}^{n_L}Q_{t+i,k_1}^{(a_1)}Q_{t,k_2}^{(a_2)}\bigg\}\bigg|}_{{\rm II}}\,.\nonumber
\end{align}
Notice that $Q_{t,k}^{(2)}$ and $Q_{t,k}^{(3)}$ are of the same type, and so are
$Q_{t,k}^{(5)}$ and $Q_{t,k}^{(6)}$. Moreover, using the symmetry of
$\mathcal W$, the terms corresponding to $(a_1,a_2)$ and $(a_2,a_1)$ can be
handled in the same way by interchanging $(t_1,k_1)$ and $(t_2,k_2)$. Therefore,
it suffices to establish the desired convergence rate for the following
representative pairs:
\begin{align*}
    \mathcal{A}^*
=
\{(1,2),(1,5),(2,2),(2,4),(2,5),(4,5),(5,5)\}\,.
\end{align*} 
By the preceding symmetry argument, it suffices to establish the convergence
rate for the following representative part of ${\rm I}$:
\begin{align}\label{Istar}
{\rm I}^*
=
\sum_{(a_1,a_2)\in \mathcal A^*}
\underbrace{\max_{k_1,k_2\in[L\tilde p^2]}
\bigg|
\sum_{i=0}^{n_L-1}
\mathcal W\bigg(\frac{i}{b_n}\bigg)
\bigg\{
\frac{1}{n_L}
\sum_{t=i+1}^{n_L}
Q_{t,k_1}^{(a_1)}
Q_{t-i,k_2}^{(a_2)}
\bigg\}
\bigg|}_{H_{a_1,a_2}} = \sum_{(a_1,a_2)\in \mathcal A^*}H_{a_1,a_2} \,.
\end{align}

To handle $H_{1,2}$, choose $M_{1n}\in\mathbb N$ such that
$b_n\ll M_{1n}\ll n$. By the triangle inequality and Bonferroni's inequality,
we have 
\begin{align} \label{diff.cov.I1}
\mathbb{P}(H_{1,2}>x)\leq&\, \mathbb{P}\bigg[\max_{k_1,k_2\in[L\tilde{p}^2]}\bigg|\sum_{i=0}^{M_{1n}} \mathcal{W}\bigg(\frac{i}{b_{n}}\bigg)\bigg\{\frac{1}{n_L}\sum_{t=i+1}^{n_L}Q_{t,k_1}^{(1)} Q_{t-i,k_2}^{(2)}\bigg\}\bigg|>\frac{x}{2}\bigg]\notag\\
&\, +\mathbb{P}\bigg[\max_{k_1,k_2\in[L\tilde{p}^2]}\bigg|\sum_{i=M_{1n}+1}^{n_L-1} \mathcal{W}\bigg(\frac{i}{b_{n}}\bigg)\bigg\{\frac{1}{n_L}\sum_{t=i+1}^{n_L}Q_{t,k_1}^{(1)} Q_{t-i,k_2}^{(2)}\bigg\}\bigg|>\frac{x}{2}\bigg]
\end{align}
for any $x\geq0$. Since $\sum_{i=0}^{M_{1n}}|\mathcal{W}({i}/{b_{n}})|\lesssim b_{n}\asymp n^{\rho}$, Bonferroni's inequality yields
\begin{align} \label{trun.M1n}
&\mathbb{P}\bigg[\max_{k_1,k_2\in[L\tilde{p}^2]}\bigg|\sum_{i=0}^{M_{1n}} \mathcal{W}\bigg(\frac{i}{b_{n}}\bigg)\bigg\{\frac{1}{n_L}\sum_{t=i+1}^{n_L} Q_{t,k_1}^{(1)}Q_{t-i,k_2}^{(2)}\bigg\}\bigg|>\frac{x}{2}\bigg] \notag\\
&~~~~	\lesssim  \tilde{p}^4M_{1n}\max_{k_1,k_2\in[L\tilde{p}^2]}\max_{i\in [M_{1n}]_0}\mathbb{P}\bigg\{\bigg| \frac{1}{n_L}\sum_{t=i+1}^{n_L}Q_{t,k_1}^{(1)}Q_{t-i,k_2}^{(2)} \bigg|>\frac{Cx}{n^{\rho}}\bigg\} \\
&~~~~		\lesssim   \tilde{p}^4M_{1n}\max_{\ell_1\in[L]}\max_{j_1,j'_1,j_2,j'_2\in[\tilde{p}]}\max_{i\in [M_{1n}]_0}\mathbb{P}\bigg(\bigg| \frac{1}{nn_L}\sum_{t_1=i+1}^{n_L}\sum_{t_2=1}^n\vep_{t_1,j_1}^{\rm dis} \vep_{t_1+\ell_1,j'_1}^{\rm dis} \vep_{t_1-i,j_2}^{\rm dis}\vep_{t_2,j'_2}^{\rm dis} \bigg|>\frac{Cx}{n^{\rho}}\bigg)  \notag
\end{align}
for any $x\geq0$. For any given $(i,\ell_1,j_1,j'_1,j_2)$, by Condition \ref{c.alpha}, $\{\vep_{t,j_1}^{\rm dis} \vep_{t+\ell_1,j'_1}^{\rm dis} \vep_{t-i,j_2}^{\rm dis}\}$ is an $\alpha$-mixing sequence with $\alpha$-mixing coefficients $\check{\alpha}(m)\leq \alpha(|m-i-\ell_1|_{+})\lesssim \exp(-C|m-i-\ell_1|_+)$. In addition, it holds by Condition \ref{c.subgaussian} that $|\sum_{t=i+1}^{n_L}\bbE(\vep_{t,j_1}^{\rm dis}\vep_{t+\ell_1,j'_1}^{\rm dis} \vep_{t-i,j_2}^{\rm dis})|\lesssim n$ and 
\begin{align*}
 \max_{i+1\leq t\leq n_L}\mathbb{P}\big\{\big|\vep_{t,j_1}^{\rm dis}\vep_{t+\ell_1,j'_1}^{\rm dis}\vep_{t-i,j_2}^{\rm dis} -\mathbb{E}(\vep_{t,j_1}^{\rm dis} \vep_{t+\ell_1,j_1'}^{\rm dis} \vep_{t-i,j_2}^{\rm dis})\big|>x\big\}
\lesssim \exp(-Cx^{2/3})
\end{align*}
for any $x\geq 0$ and $(i, \ell_1,j_1,j'_1,j_2)$. By Lemma \ref{tail_chang} with $( B_n,c_n,  r_1, {r} ) =( 1, i+\ell_1,2/3, 2/7)$,  
\begin{align*}
  \mathbb{P}\bigg(\bigg|\sum_{t=i+1}^{n_L}\vep_{t,j_1}^{\rm dis}\vep_{t+\ell_1,j'_1}^{\rm dis} \vep_{t-i,j_2}^{\rm dis}\bigg|\geq x\bigg) \leq &~ \mathbb{P}\bigg(\bigg|\sum_{t=i+1}^{n_L}(1-\bbE)(\vep_{t,j_1}^{\rm dis}\vep_{t+\ell_1,j'_1}^{\rm dis} \vep_{t-i,j_2}^{\rm dis})\bigg|\geq x - Cn\bigg) \\
 \lesssim &~ \exp\bigg( -\frac{Cx^2}{M_{1n}n} \bigg)  + \exp\bigg( -\frac{Cx^{2/7}}{M_{1n}^{2/7}}\bigg) 
\end{align*}
for any $x>Cn$, $i\in [M_{1n}]_0$ and $( \ell_1,j_1,j'_1,j_2)$. By Conditions \ref{c.alpha} and \ref{c.subgaussian}, it follows from Lemma \ref{tail_chang} with $( B_n,c_n,  r_1, {r} ) =( 1, 0,2, 1/3)$   that 
\begin{align*}
    \mathbb{P}\bigg(\bigg|\sum_{t=1}^{n}\vep_{t,j_2'}^{\rm dis}\bigg| \geq x\bigg)\lesssim \exp\bigg( -\frac{Cx^2}{ n} \bigg)+\exp(-Cx^{1/3})
\end{align*}
for any $x\geq 0$ and $j_2'\in[\tilde{p}]$. Then, we can conclude that
\begin{align*}
&\mathbb{P}\bigg(\bigg| \frac{1}{nn_L}\sum_{t_1=i+1}^{n_L}\sum_{t_2=1}^n \vep_{t_1,j_1}^{\rm dis} \vep_{t_1+\ell_1,j_1'}^{\rm dis} \vep_{t_1-i,j_2}^{\rm dis} \vep_{t_2,j_2'}^{\rm dis} \bigg|>\frac{Cx}{n^{\rho}}\bigg)\\
&~~~~		\leq  \mathbb{P}\bigg(\bigg| \sum_{t=i+1}^{n_L} \vep_{t,j_1}^{\rm dis} \vep_{t+\ell_1,j_1'}^{\rm dis} \vep_{t-i,j_2}^{\rm dis} \bigg| >\frac{C nx^{1/2}M_{1n}^{1/4}}{ {n}^{\rho/2}}\bigg)+\mathbb{P}\bigg(\bigg| \sum_{t=1}^{n}\vep_{t,j_2'}^{\rm dis} \bigg|>\frac{Cnx^{1/2}}{ {n}^{\rho/2}M_{1n}^{1/4}}\bigg)\\
&~~~~		\lesssim \exp\bigg(-\frac{Cn^{1- \rho }x}{M_{1n}^{1/2}}\bigg)  +\exp\bigg\{-\frac{Cn^{(2-\rho)/7}x^{1/7}}{M_{1n}^{3/14}}\bigg\} +\exp\bigg\{-\frac{Cn^{(2- \rho)/6}x^{1/6}}{M_{1n}^{1/12}}\bigg\} 
\end{align*}
for $x>Cn^{\rho}M_{1n}^{-1/2}$, $i\in [M_{1n}]_0$  and $(i, \ell_1,j_1,j'_1,j_2,j_2')$. Together with \eqref{trun.M1n}, we have 
\begin{align*}
&\mathbb{P}\bigg[\max_{k_1,k_2\in[L\tilde{p}^2]}\bigg|\sum_{i=0}^{M_{1n}}\mathcal{W} \bigg(\frac{i}{b_{n}}\bigg)\bigg\{\frac{1}{n_L}\sum_{t=i+1}^{n_L}Q_{t,k_1}^{(1)}Q_{t-i,k_2}^{(2)}\bigg\}\bigg|>\frac{x}{2}\bigg]\notag\\
&~~~~ 		\lesssim  \tilde{p}^4M_{1n}\bigg[\exp\bigg(-\frac{Cn^{1- \rho }x}{M_{1n}^{1/2}}\bigg) +\exp\bigg\{-\frac{Cn^{(2-\rho)/7}x^{1/7}}{M_{1n}^{3/14}}\bigg\} +\exp\bigg\{-\frac{Cn^{(2- \rho)/6}x^{1/6}}{M_{1n}^{1/12}}\bigg\}\bigg]
\end{align*}
for $x>Cn^{\rho }M_{1n}^{-1/2}$. Since $b_n\asymp n^{\rho}$, by Condition \ref{c.kernelF}, we have $\sum_{i={M_{1n}}+1}^{n_L-1}|\mathcal{W}({i}/{b_{n}})|\lesssim n^{\vartheta\rho} M_{1n}^{ 1-\vartheta }$. 
Then, Condition \ref{c.subgaussian}, triangle inequality and Bonferroni's inequality yield
\begin{align*}
&\mathbb{P}\bigg[\max_{k_1,k_2\in[L\tilde{p}^2]}\bigg|\sum_{i=M_{1n}+1}^{n_L-1}\mathcal{W}\bigg(\frac{i}{b_{n}}\bigg)\bigg\{\frac{1}{n_L}\sum_{t=i+1}^{n_L}Q_{t,k_1}^{(1)}Q_{t-i,k_2}^{(2)}\bigg\}\bigg|>\frac{x}{2}\bigg]\notag\\ 
&~~~~		\lesssim   \tilde{p}^4\max_{k_1,k_2\in[L\tilde{p}^2]}\sum_{i=M_{1n}+1}^{n_L-1} \mathbb{P}\bigg\{\bigg|\frac{1}{n_L}\sum_{t=i+1}^{n_L}Q_{t,k_1}^{(1)}Q_{t-i,k_2}^{(2)}\bigg|>\frac{CxM_{1n}^{\vartheta-1}}{n^{\vartheta\rho}}\bigg\}\\
&~~~~		\lesssim  \tilde{p}^4\max_{\ell_1\in[L]}\max_{j_1,j_1',j_2,j_2'\in[\tilde{p}]}\sum_{i=M_{1n}+1}^{n_L-1} \mathbb{P}\bigg(\bigg|\sum_{t_1=i+1}^{n_L}\sum_{t_2=1}^n \vep_{t_1,j_1}^{\rm dis} \vep_{t_1+\ell_1,j_1'}^{\rm dis} \vep_{t_1-i,j_2}^{\rm dis} \vep_{t_2,j_2'}^{\rm dis}\bigg|> \frac{Cn^2xM_{1n}^{\vartheta-1}}{n^{\vartheta\rho}}\bigg)\\
&~~~~		\lesssim  \tilde{p}^4\max_{\ell_1\in[L]}\max_{j_1,j_1',j_2\in[\tilde{p}]}\sum_{i=M_{1n}+1}^{n-1} \sum_{t=i+1}^{n_L}\mathbb{P}\bigg\{| \vep_{t,j_1}^{\rm dis} \vep_{t+\ell_1,j_1'}^{\rm dis} \vep_{t-i,j_2}^{\rm dis} |>\frac{Cx^{3/4}M_{1n}^{(3\vartheta-3)/4}}{n^{3\vartheta\rho/4}}\bigg\}\\
&~~~~~~~ + \tilde{p}^4\max_{j_2'\in[\tilde{p}]}\sum_{i=M_{1n}+1}^{n-1}\sum_{t=1}^{n} \mathbb{P}\bigg\{|\vep_{t,j_2'}^{\rm dis}|>\frac{Cx^{1/4}M_{1n}^{(\vartheta-1)/4}} {n^{\vartheta\rho/4}}\bigg\} \\	 &~~~~	\lesssim \tilde{p}^4n^2\exp\bigg\{-\frac{Cx^{1/2}M_{1n}^{(\vartheta-1)/2}}{n^{\vartheta\rho/2}}\bigg\}
\end{align*}
for any $x\geq 0$. Due to $M_{1n}\ll n$, by \eqref{diff.cov.I1}, it holds that
\begin{align*}
\mathbb{P}(H_{1,2}>x)\lesssim&~  \tilde{p}^4n^2 \bigg[\exp\bigg(-\frac{Cn^{1- \rho }x}{M_{1n}^{1/2}}\bigg) +\exp\bigg\{-\frac{Cn^{(2-\rho)/7}x^{1/7}}{M_{1n}^{3/14}}\bigg\} \\
&~~~~~~~~~~+\exp\bigg\{-\frac{Cn^{(2- \rho)/6}x^{1/6}}{M_{1n}^{1/12}}\bigg\}
 + \exp\bigg\{-\frac{CM_{1n}^{(\vartheta-1)/2}x^{1/2}}{n^{\vartheta\rho/2}}\bigg\}\bigg] 
\end{align*}
for any $x>Cn^{\rho }M_{1n}^{-1/2}$. If $0<\rho<\min\{1/2 ,(\vartheta-1)/\vartheta\}$, we have $\max\{2\rho,\rho\vartheta/(\vartheta-1)\} <\min\{2-2\rho,(4-2\rho)/3, 1\}$. Therefore, we can choose $a_1$ such that $\max\{2\rho,\rho\vartheta/(\vartheta-1)\}<a_1<\min\{2-2\rho,(4-2\rho)/3, 1\}$. Recall $\tilde{p}=pM$ and $M\asymp (np)^{C_M}$. Taking $M_{1n}\asymp n^{a_1}$ yields
\begin{align}\label{eq:rateH12}
H_{1,2}=O_{\rm p}\{n^{-c_1''}(\log p)^{c_2''}\}\,,
\end{align}
provided that $0<\rho<\min\{1/2 ,(\vartheta-1)/\vartheta\}$, where $c_1'',c_2''$ are two positive constants depending only on $(\rho,\vartheta)$.

For $H_{1,5}$, recall  $Q_{t,k}^{(1)}=\vep_{t,j}^{\rm dis}\vep_{t+\ell,j'}^{\rm dis}$ and  $Q_{t,k}^{(5)}=Q_{k}^{(5)}=(n_L^{-1}\sum_{t'=1}^{n_L}\vep_{t',j}^{\rm dis})\bar{\vep}^{\rm dis}_{j'} $.  Then, we have
\begin{align}\label{eq:II_1}
\mathbb{P}(H_{1,5}>x) &=  \bbP\bigg[\max_{k_1,k_2\in[L\tilde{p}^2]}\bigg|\sum_{i=0}^{n_L-1}\mathcal{W}\bigg(\frac{i}{b_n}\bigg)\bigg\{\frac{1}{n_L} \sum_{t=i+1}^{n_L}Q_{t,k_1}^{(1)}Q_{k_2}^{(5)}\bigg\} \bigg| > x\bigg]  \notag\\
&\lesssim    \tilde{p}^4\max_{k_1,k_2\in[L\tilde{p}^2]}\bbP\bigg\{\sum_{i=0}^{n_L-1}\bigg|\mathcal{W}\bigg(\frac{i}{b_n}\bigg)\bigg|  \bigg|\frac{1}{n_L}\sum_{t=i+1}^{n_L} (1-\bbE)( \vep_{t,j_1}^{\rm dis}\vep_{t+\ell_1,j_1'}^{\rm dis})  \bigg| |Q_{k_2}^{(5)}| > \frac{x}{2}\bigg\}  \notag\\
&~~  +  \tilde{p}^4\max_{k_1,k_2\in[L\tilde{p}^2]}\bbP\bigg\{ \sum_{i=0}^{n_L-1}\bigg|\mathcal{W}\bigg(\frac{i}{b_n}\bigg)\bigg| \bigg|\frac{1}{n_L}\sum_{t=i+1}^{n_L}\bbE(\vep_{t,j_1}^{\rm dis}\vep_{t+\ell_1,j_1'}^{\rm dis})\bigg| |Q_{k_2}^{(5)}| > \frac{x}{2} \bigg\} 
\end{align}
for any $x\geq0$. Since $\sum_{i=0}^{n_L-1}|\mathcal{W}(i/b_n)|\lesssim b_n\asymp n^{\rho}$, it follows from Conditions~\ref{c.alpha} and \ref{c.subgaussian}, together with
\eqref{eq:sumvep} and \eqref{eq:sumvep2}, that
\begin{align*}
    & \bbP\bigg\{\sum_{i=0}^{n_L-1}\bigg|\mathcal{W}\bigg(\frac{i}{b_n}\bigg)\bigg|  \bigg|\frac{1}{n_L}\sum_{t=i+1}^{n_L} (1-\bbE)( \vep_{t,j_1}^{\rm dis}\vep_{t+\ell_1,j_1'}^{\rm dis})  \bigg| |Q_{k_2}^{(5)}| > \frac{x}{2}\bigg\}\\
   &~~~~  \leq  \sum_{i=0}^{n_L-1}\bbP\bigg\{\bigg|\frac{1}{n_L}\sum_{t=i+1}^{n_L} (1-\bbE)( \vep_{t,j_1}^{\rm dis}\vep_{t+\ell_1,j_1'}^{\rm dis}) \bigg| > \frac{Cx^{1/3}}{n^{\rho/3}}  \bigg\}\\
    &~~~~~~~+  n\bbP\bigg( \bigg|\frac{1}{n_L}\sum_{t=1}^{n_L}\vep_{t,j_2}^{\rm dis}\bigg|  >  \frac{Cx^{1/3}}{n^{\rho/3}} \bigg) + n\bbP\bigg(\bigg|\frac{1}{n}\sum_{t=1}^n\vep_{t,j_2'}^{\rm dis} \bigg|> \frac{Cx^{1/3}}{n^{\rho/3}} \bigg) \\
    &~~~~  \lesssim n [ \exp\{-Cn^{(3-2\rho)/3}x^{2/3}\}
+\exp\{ -Cn^{(3-\rho)/9}x^{1/9} \} ]
\end{align*}
for any $x\geq 0$. Notice that $\bbE(|\vep_{t,j}^{\rm dis}\vep_{t+\ell,j'}^{\rm dis}|)\leq C$ for any $(t,\ell,j,j')$. It holds by \eqref{eq:sumvep} that
\begin{align*}
    &\bbP\bigg\{ \sum_{i=0}^{n_L-1}\bigg|\mathcal{W}\bigg(\frac{i}{b_n}\bigg)\bigg| \bigg|\frac{1}{n_L}\sum_{t=i+1}^{n_L}\bbE(\vep_{t,j_1}^{\rm dis}\vep_{t+\ell_1,j_1'}^{\rm dis})\bigg| |Q_{k_2}^{(5)}| > \frac{x}{2} \bigg\} \\
    &~~~~  \leq n\bbP\bigg( \bigg|\frac{1}{n_L}\sum_{t=1}^{n_L}\vep_{t,j_2}^{\rm dis}\bigg|  >  \frac{Cx^{1/2}}{n^{\rho/2}} \bigg) + n\bbP\bigg(\bigg|\frac{1}{n}\sum_{t=1}^n\vep_{t,j_2'}^{\rm dis} \bigg|> \frac{Cx^{1/2}}{n^{\rho/2}} \bigg)\\
    &~~~~  \lesssim n[\exp(-Cn^{1-\rho}x) + \exp\{-Cn^{(2-\rho)/6}x^{1/6}\} ]
\end{align*}
for any $x\geq 0$. This, together with \eqref{eq:II_1}, yields
\begin{align*}
    \mathbb{P}(H_{1,5}>x)  \lesssim &~  \tilde{p}^4n\big[ \exp\{-Cn^{(3-2\rho)/3}x^{2/3}\}
+\exp\{ -Cn^{(3-\rho)/9}x^{1/9} \}\\
&~~~~~~~~~~~~~+ \exp(-Cn^{1-\rho}x) + \exp\{-Cn^{(2-\rho)/6}x^{1/6}\}\big]
\end{align*}
for any $x\geq 0$.  Therefore, we can conclude that
\begin{align}\label{eq:rateH15}
H_{1,5}=O_{\rm p}\{n^{-c_1'''}(\log p)^{c_2'''}\} \,,
\end{align} 
where $c_1''',c_2'''$ are two positive constants depending only on $\rho$.

For $H_{2,2}$, recall $Q_{t,k}^{(2)}=-\vep_{t,j}^{\rm dis}\bar{\vep}^{\rm dis}_{j'}$. It holds that
\begin{align}\label{eq:III_1}
    \mathbb{P}(H_{2,2}>x) \lesssim &~\tilde{p}^4\max_{k_1,k_2\in[L\tilde{p}^2]}\bbP\bigg[\bigg|\sum_{i=0}^{n_L-1}\mathcal{W}\bigg(\frac{i}{b_n}\bigg)\bigg\{\frac{1}{n_L} \sum_{t=i+1}^{n_L}Q_{t,k_1}^{(2)}Q_{t-i,k_2}^{(2)}\bigg\} \bigg| > x\bigg]\nonumber\\
    \lesssim &~  \tilde{p}^4n\max_{k_1,k_2\in[L\tilde{p}^2]}\max_{i\in[n_L-1]_0} \bbP\bigg(  \bigg|\frac{1}{n_L}\sum_{t=i+1}^{n_L}  \vep_{t,j_1}^{\rm dis}\vep_{t-i,j_2}^{\rm dis}   \bigg|
	|\bar{\vep}_{j_1'}^{\rm dis}| |\bar{\vep}_{j_2'}^{\rm dis}| > \frac{Cx}{n^{\rho}}\bigg) \nonumber\\
    \lesssim &~  \tilde{p}^4n^2\max_{k_1,k_2\in[L\tilde{p}^2]}\max_{i\in[n_L-1]_0} \max_{ i+1\leq t \leq n_L}\bbP\bigg(   | \vep_{t,j_1}^{\rm dis}\vep_{t-i,j_2}^{\rm dis}    |
	|\bar{\vep}_{j_1'}^{\rm dis}| |\bar{\vep}_{j_2'}^{\rm dis}| > \frac{Cx}{n^{\rho}}\bigg)
\end{align}
for any $x\geq 0$. Since $\bar{\vep}_{j}^{\rm dis} = n^{-1}\sum_{t=1}^n {\vep}_{t,j}^{\rm dis}$, it holds by Condition \ref{c.subgaussian} and \eqref{eq:sumvep} that
\begin{align*}
    &~\bbP\bigg(   | \vep_{t,j_1}^{\rm dis}\vep_{t-i,j_2}^{\rm dis}    |
	|\bar{\vep}_{j_1'}^{\rm dis}| |\bar{\vep}_{j_2'}^{\rm dis}| > \frac{Cx}{n^{\rho}}\bigg) \\
    \leq &~ \bbP\big \{   | \vep_{t,j_1}^{\rm dis}\vep_{t-i,j_2}^{\rm dis}    |
	 >  Cn^{(1-\rho)/2} x^{1/2} \big\} + \bbP\bigg(    
	|\bar{\vep}_{j_1'}^{\rm dis}|   > \frac{Cx^{1/4}}{n^{(1+\rho)/4}}\bigg)+ \bbP\bigg(    
	 |\bar{\vep}_{j_2'}^{\rm dis}|  > \frac{Cx^{1/4}}{n^{(1+\rho)/4}}\bigg)\\
     \lesssim &~ \exp\{ -Cn^{(1-\rho)/2} x^{1/2}  \}+\exp\{ -Cn^{(3-\rho)/12} x^{1/12}  \}
\end{align*}
for any $x\geq 0$. This, together with \eqref{eq:III_1}, yields
\begin{align*}
    \mathbb{P}(H_{2,2}>x) \lesssim \tilde{p}^4n^2\big[ \exp\{ -Cn^{(1-\rho)/2} x^{1/2}  \}+\exp\{ -Cn^{(3-\rho)/12} x^{1/12}  \} \big]
\end{align*}
for any $x\geq 0$. Therefore, we can conclude that
\begin{align}\label{eq:rateH22}
H_{2,2}=O_{\rm p}\{n^{-c_1''''}(\log p)^{c_2''''}\} \,,
\end{align} 
where $c_1'''',c_2''''$ are two positive constants depending only on $\rho$.

For $H_{2,4}$, $H_{2,5}$, $H_{4,5}$ and $H_{5,5}$, recall $Q_{t,k}^{(2)}=-\vep_{t,j}^{\rm dis}\bar{\vep}^{\rm dis}_{j'}$, $Q_{t,k}^{(4)}=Q_{k}^{(4)}=-
 {n_L^{-1}}\sum_{t'=1}^{n_L}\vep_{t',j}^{\rm dis}\vep_{t'+\ell,j'}^{\rm dis}$ and $Q_{t,k}^{(5)}=Q_{k}^{(5)}=(n_L^{-1}\sum_{t'=1}^{n_L}\vep_{t',j}^{\rm dis})\bar{\vep}^{\rm dis}_{j'} $. Then, it holds that
 \begin{align*}
    H_{2,4}=&~ \max_{k_1,k_2\in[L\tilde p^2]}
\bigg|
\sum_{i=0}^{n_L-1}
\mathcal W\bigg(\frac{i}{b_n}\bigg)
\bigg\{
\frac{1}{n_L}
\sum_{t=i+1}^{n_L}
Q_{t,k_1}^{(2)}
Q_{t-i,k_2}^{(4)}
\bigg\}
\bigg|\\
\leq &~ \max_{k_1,k_2\in[L\tilde p^2]} \max_{t'\in [n_L]}\sum_{i=0}^{n_L-1}\bigg|\mathcal{W}\bigg(\frac{i}{b_n}\bigg)\bigg| \bigg|\frac{1}{n_L}\sum_{t=i+1}^{n_L}\vep_{t,j_1}^{\rm dis}\bigg||\bar{\vep}_{j_1'}^{\rm dis}|  | \vep_{t',j_2}^{\rm dis}\vep_{t'+\ell_2,j_2'}^{\rm dis}   |\,,\\
  H_{2,5}=&~ \max_{k_1,k_2\in[L\tilde p^2]}
\bigg|
\sum_{i=0}^{n_L-1}
\mathcal W\bigg(\frac{i}{b_n}\bigg)
\bigg\{
\frac{1}{n_L}
\sum_{t=i+1}^{n_L}
Q_{t,k_1}^{(2)}
Q_{t-i,k_2}^{(5)}
\bigg\}
\bigg|\\
\leq &~ \max_{k_1,k_2\in[L\tilde p^2]} \max_{t'\in [n_L]}\max_{t''\in [n]}\sum_{i=0}^{n_L-1}\bigg|\mathcal{W}\bigg(\frac{i}{b_n}\bigg)\bigg| \bigg|\frac{1}{n_L}\sum_{t=i+1}^{n_L}\vep_{t,j_1}^{\rm dis}\bigg||\bar{\vep}_{j_1'}^{\rm dis}|  | \vep_{t',j_2}^{\rm dis}\vep_{t'',j_2'}^{\rm dis}   |\,,\\
   H_{4,5}=&~ \max_{k_1,k_2\in[L\tilde p^2]}
\bigg|
\sum_{i=0}^{n_L-1}
\mathcal W\bigg(\frac{i}{b_n}\bigg)
\bigg\{
\frac{1}{n_L}
\sum_{t=i+1}^{n_L}
Q_{t,k_1}^{(4)}
Q_{t-i,k_2}^{(5)}
\bigg\}
\bigg|\\
\leq &~ \max_{k_1,k_2\in[L\tilde p^2]} \max_{t\in [n_L]}\sum_{i=0}^{n_L-1}\bigg|\mathcal{W}\bigg(\frac{i}{b_n}\bigg)\bigg| |\vep_{t,j_1}^{\rm dis}\vep_{t+\ell_1,j_1'}^{\rm dis} |\bigg|\frac{1}{n_L}\sum_{t'=1}^{n_L} \vep_{t',j_2}^{\rm dis}    \bigg||\bar{\vep}_{j_2'}^{\rm dis}|  \,,\\
 H_{5,5}=&~ \max_{k_1,k_2\in[L\tilde p^2]}
\bigg|
\sum_{i=0}^{n_L-1}
\mathcal W\bigg(\frac{i}{b_n}\bigg)
\bigg\{
\frac{1}{n_L}
\sum_{t=i+1}^{n_L}
Q_{t,k_1}^{(5)}
Q_{t-i,k_2}^{(5)}
\bigg\}
\bigg|\\
\leq &~ \max_{k_1,k_2\in[L\tilde p^2]} \max_{t'\in [n_L]}\max_{t''\in [n]}\sum_{i=0}^{n_L-1}\bigg|\mathcal{W}\bigg(\frac{i}{b_n}\bigg)\bigg| \bigg|\frac{1}{n_L}\sum_{t=1}^{n_L}\vep_{t,j_1}^{\rm dis}\bigg||\bar{\vep}_{j_1'}^{\rm dis}|  | \vep_{t',j_2}^{\rm dis}\vep_{t'',j_2'}^{\rm dis}   |\,.
 \end{align*}  
By taking similar arguments as in the proof for $H_{2,2}$, we can also have 
 \begin{align}\label{eq:rateHrest}
H_{2,4}+H_{2,5}+H_{4,5}+H_{5,5}=O_{\rm p}\{n^{-c_1''''}(\log p)^{c_2''''}\} \,.
\end{align}  
Combining \eqref{Istar} with  \eqref{eq:rateH12}, \eqref{eq:rateH15}, \eqref{eq:rateH22} and \eqref{eq:rateHrest}, we can conclude that
\begin{align*}
    {\rm I}^*
= O_{\rm p}\{n^{-c_1^*}(\log p)^{c_2^*}\}\,,
\end{align*}
provided that $0<\rho<\min\{1/2 ,(\vartheta-1)/\vartheta\}$, where $c_1^*,c_2^*$ are two positive constants depending only on $(\rho,\vartheta)$.  Notice that, for every
$(a_1,a_2)\in\mathcal A\setminus\mathcal A^*$, the corresponding term can be
controlled by the same argument as one of the representative terms in
$\mathcal A^*$.   Hence, it follows that $ {\rm I}=O_{\rm p}\{n^{-c_1^*}(\log p)^{c_2^*}\}$. The same convergence rate holds for ${\rm II}$. Since   $(\vartheta-1)/(3\vartheta-2) < \min\{1/2 ,(\vartheta-1)/\vartheta\}$, \eqref{dis.theta2} and \eqref{dis.theta3.pre} imply \eqref{disSigma}, which
completes the proof of Lemma~\ref{disGAstar}. $\hfill\Box$

\subsection{Proof of Lemma \ref{lemma1.prop1}}\label{proof.lemma1.prop1}

As in Section~\ref{sec.preliminary}, let $\cU=[0,1]$ and partition it into
$M$ subintervals $\{B_1,\ldots,B_{M}\}$ of equal length $M^{-1}$. Let $b_m$ denote the midpoint of $B_m$. Then, it holds that
\begin{align*}
    \max_{t\in[n]}\max_{j\in[p]}\sup_{u\in\cU}|\vep_{t,j}(u)| \leq \max_{t\in[n]}\max_{j\in[p]}\max_{m\in[M]} | \vep_{t,j}(b_m)| + \max_{t\in[n]}\max_{j\in[p]}\max_{m\in[M]}\sup_{u\in{B}_m}|\vep_{t,j}(u)-\vep_{t,j}(b_m)|\,.
\end{align*}
By Condition \ref{c.subgaussian} and Bonferroni's inequality, we have 
\begin{align*}
    \bbP\bigg\{\max_{t\in[n]}\max_{j\in[p]}\max_{m\in[M]} | \vep_{t,j}(b_m)|  >x  \bigg\} \lesssim  npM\exp(-C x^2)
\end{align*}
for any $x\geq 0$, which implies
\begin{align}\label{order.dis}
    \max_{t\in[n]}\max_{j\in[p]}\max_{m\in[M]} | \vep_{t,j}(b_m)| = O_{\rm p}[\{ \log(npM) \}^{1/2}]\,.
\end{align}
Moreover, by Condition \ref{c.subgaussian} again, it holds that
\begin{align*}
    \max_{t\in[n]}\max_{j\in[p]}\bbP\{ |\vep_{t,j}(u)-\vep_{t,j}(v)| > x \}\lesssim \exp\bigg[ -C\bigg\{ \frac{x}{d(u,v)} \bigg\}^2 \bigg]
\end{align*}
for any $x \geq 0$ and $u,v\in \cU$, where $d(u,v) =  |u-v|^{\kappa}$. Then, by Lemma \ref{tail_multi}(i) with $(c_{1n},c_{2n},\gamma)=(C,C,2)$,  we have
 \begin{align*}
	&~\bbP\bigg\{ \sup_{u,v\in{B}_m}|\vep_{t,j}(u)-\vep_{t,j}(v)|>x \bigg\}  \\
     \lesssim&~ \frac{1}{x}\int_{0}^{D_1}\{\log(eN_\epsilon)\}^{1/2}\,{\rm d}\epsilon + N_{D_1}\log_2\biggl(\frac{1}{D_1}\biggr)\exp\bigg\{-C\bigg( \frac{x}{D} \bigg)^2  \bigg\}  
\end{align*}
for any $x >0$ and $D_1\in(0,D\wedge 1)$, where $D= \sup\{d(u ,v): u,v \in B_m\}\leq M^{-\kappa}$. Notice that $N_{\epsilon}\lesssim M^{-1}\epsilon^{-1/\kappa}$ for any $\epsilon\leq D_1$. Similar to \eqref{I1.sup} with $\delta=1/2$, we also have 
\begin{align*}
    \int_{0}^{D_1}\{\log(eN_\epsilon)\}^{1/2}\,{\rm d}\epsilon \lesssim \frac{D_1^{1/2}}{M^{\kappa/2}}\,.
\end{align*}
Taking $D_1\asymp M^{-2}$ yields
\begin{align*}
    &~\bbP\bigg\{\max_{t\in[n]}\max_{j\in[p]}\max_{m\in[M]} \sup_{u \in{B}_m}|\vep_{t,j}(u)-\vep_{t,j}(b_m)|>x \bigg\} \nonumber \\
   &~~~~~~~~~~ \lesssim   \frac{np}{M^{\kappa/2} x} + np M^{2/\kappa}\log_2(M)\exp(-CM^{2\kappa}x^2  )
\end{align*}
 for any $x> 0$, which implies that 
  \begin{align}\label{div.vep}
     \max_{t\in[n]}\max_{j\in[p]}\max_{m\in[M]} \sup_{u \in{B}_m}|\vep_{t,j}(u)-\vep_{t,j}(b_m)|  = O_{\rm p}\bigg\{\frac{ (np)^2}{  {M}^{ \kappa/2} }   \bigg\}  \,,
 \end{align}
 for any large $M>0$. Take $M\asymp (np)^{C_0}$ for some sufficiently large constant $C_0>0$. It holds that $\max_{t\in[n]}\max_{j\in[p]}\max_{m\in[M]} \sup_{u \in{B}_m}|\vep_{t,j}(u)-\vep_{t,j}(b_m)| = O_{\rm p}\{(np)^{-1}\} $.
 This, together with \eqref{order.dis},   yields
 \begin{align*}
     \max_{t\in[n]}\max_{j\in[p]}\sup_{u\in\cU}|\vep_{t,j}(u)| = O_{\rm p}[\{ \log(np ) \}^{1/2}]\,. 
 \end{align*}
Taking $ M =M_1$ in  \eqref{div.vep} yields the second assertion. We complete the proof of Lemma \ref{lemma1.prop1}. $\hfill\Box$

\subsection{Proof of Lemma \ref{lemma2.prop1}}\label{proof.lemma2.prop1}

The local linear estimator $\hat{\vep}_{t,j}(u)$ admits the following closed-form expression:
 \begin{align*}
     \hat{\vep}_{t,j}(u) = \frac{A_0(u)B_2(u)-A_1(u)B_1(u)}{B_0(u)B_2(u)-B_1^2(u)}\,,
 \end{align*}
 where for $r\in \{0,1,2\}$,
 \begin{align*}
     A_r(u)=&~  \frac{1}{N_{t,j}   }\sum_{k=1}^{N_{t,j}}\cK_{h_{t,j}}(u_{t,j,k}-u)\bigg(\frac{u_{t,j,k}-u}{h_{t,j}}\bigg)^r\vep^*_{t,j,k}\,,\\
     B_r(u)= &~  \frac{1}{N_{t,j}   }\sum_{k=1}^{N_{t,j}}\cK_{h_{t,j}}(u_{t,j,k}-u)\bigg(\frac{u_{t,j,k}-u}{h_{t,j}}\bigg)^r\,.
 \end{align*}
 Denote $\be_0=(1,0)^{\T}$, $\widetilde{\bu}_{t,j,k}=\{1,(u_{t,j,k}-u)/h_{t,j}\}^{\T}$ and
\begin{align}\label{def.SR}
\widehat{\bf S}_{t,j}(u)=&\,\frac{1}{N_{t,j}}\sum_{k=1}^{N_{t,j}} \widetilde{\bu}_{t,j,k}\widetilde{\bu}_{t,j,k}^{\T}\cK_{h_{t,j}}(u_{t,j,k}-u)\,,\nonumber \\
\widehat{\bR}_{t,j}(u)=&\, \frac{1}{N_{t,j}}\sum_{k=1}^{N_{t,j}} \widetilde{\bu}_{t,j,k}\{\vep^*_{t,j,k} - \vep_{t,j}(u)\}\cK_{h_{t,j}}(u_{t,j,k}-u) \,.
\end{align} 
Since $|\be_0|_2=1$, a simple calculation yields, for any $u,u_1,u_2\in \cU$, 
\begin{align}\label{div.delta1}
    &|\delta_{t,j}(u)|=|\hat{\vep}_{t,j}(u) - \vep_{t,j}(u)|= |\be_0^{\T}\{\widehat{\bf S}_{t,j}(u)\}^{-1}\widehat{\bR}_{t,j}(u)| \leq \| \widehat{\bf S}_{t,j}(u) \|_{\min}^{-1} |\widehat \bR_{t,j}(u) |_2  \,, \nonumber\\
    &|\delta_{t,j}(u_1) - \delta_{t,j}(u_2)| =    |\be_0^{\T}[\{\widehat{\bS}_{t,j}(u_1)\}^{-1}\widehat{\bR}_{t,j}(u_1) - \{\widehat{\bS}_{t,j}(u_2)\}^{-1}\widehat{\bR}_{t,j}(u_2)] | \nonumber\\
&~~~~~~~~~~~~~~~~~~~~~~~~\leq    \|\widehat{\bS}_{t,j}(u_1)\|_{\min}^{-1}|\widehat{\bR}_{t,j}(u_1)-\widehat{\bR}_{t,j}(u_2)|_2\nonumber\\
&~~~~~~~~~~~~~~~~~~~~~~~~~~~
+ \|\widehat{\bS}_{t,j}(u_1)\|_{\min}^{-1}\|\widehat{\bS}_{t,j}(u_1)-\widehat{\bS}_{t,j}(u_2)\|_{\rm F} \|\widehat{\bS}_{t,j}(u_2)\|_{\min}^{-1} |\widehat{\bR}_{t,j}(u_2)|_2 \,.
\end{align}
Here we have used the identity $\bA^{-1}-\bB^{-1}=\bA^{-1}(\bB-\bA)\bB^{-1}$, 
together with $|\ba^{\T}\bB\bc|
\le
|\ba|_2 \|\bB\|_{\rm op} |\bc|_2$, $\|\bB^{-1}\|_{\rm op}
=
\|\bB\|_{\min}^{-1}$ 
and $\|\bB\|_{\rm op}
\le
\|\bB\|_{\rm F}$ 
for any $\ba,\bc\in\mathbb{R}^q$ and $\bB\in\mathbb{R}^{q\times q}$. 
Here,  $\|\bB\|_{\min}
=
\{\lambda_{\min}(\bB^{\T}\bB)\}^{1/2}$.  By Conditions~\ref{cond.fu} and~\ref{cond.kern},  for any
$\ba=(a_0,a_1)^{\T}$ with $|\ba|_2=1$, there exists a constant $c_S>0$,
independent of $(t,j,u)$, such that
\begin{align*}
\ba^{\T}\bbE\{\widehat{\bf S}_{t,j}(u)\}\ba
=&~
\int 
\bigg(a_0+a_1\frac{v-u}{h_{t,j}}\bigg)^2
\cK_{h_{t,j}}(v-u)f_U(v)\,{\rm d}v  
\geq  c_S \,.
\end{align*}
Taking the infimum over $|\ba|_2=1$ gives  $\|\bbE\{\widehat{\bf S}_{t,j}(u)\}\|_{\min}
    \geq c_S$  for any $(t,j,u)$. Define
\begin{align*}
    \mathcal{E}_1 = \bigg\{ \max_{t\in [n]}\max_{j\in[p]}\sup_{u\in\cU} \|\widehat{\bf S}_{t,j}(u) - \bbE\{\widehat{\bf S}_{t,j}(u)\}\|_{\rm F} \leq \frac{c_S}{2} \bigg\}\,.
\end{align*}    
On the event $\mathcal E_1$,  we have  $ \|\widehat{\bS}_{t,j}(u)\|_{\min}  \geq \|\bbE\{\widehat{\bf S}_{t,j}(u)\}\|_{\min}-\|\widehat{\bf S}_{t,j}(u) - \bbE\{\widehat{\bf S}_{t,j}(u)\}\|_{\rm F}  \geq c_S/2 $ for all $(t,j,u)$. Combining this with \eqref{div.delta1}, we obtain, on $\mathcal E_1$, that for any $(u,u_1,u_2)\in\cU^3$, 
\begin{align}\label{div.delta2}
    |\delta_{t,j}(u)| \leq  &~  2c_S^{-1}|\widehat \bR_{t,j}(u) |_2 \lesssim |\widehat \bR_{t,j}(u) |_2 \,, \nonumber\\
    |\delta_{t,j}(u_1) - \delta_{t,j}(u_2)| \lesssim &~   |\widehat{\bR}_{t,j}(u_1)-\widehat{\bR}_{t,j}(u_2)|_2 
+  \|\widehat{\bS}_{t,j}(u_1)-\widehat{\bS}_{t,j}(u_2)\|_{\rm F}   |\widehat{\bR}_{t,j}(u_2)|_2 \,.
\end{align}
Similar to Section~\ref{sec.preliminary}, let $\cU=[0,1]$ and partition it into
$\widetilde{M}$ subintervals $\{B_1,\ldots,B_{\widetilde{M}}\}$ of equal length $\widetilde{M}^{-1}$. Let $b_m$ denote the midpoint of $B_m$. Then, by \eqref{div.delta2},  it holds that
\begin{align}\label{delta.sup}
    &\max_{t\in [n]}\max_{j\in[p]}\sup_{u\in\cU}|\delta_{t,j}(u)|\leq  \max_{t\in [n]}\max_{j\in[p]}\max_{m\in[\widetilde{M}]}|\delta_{t,j}(b_m)|+ \max_{t\in [n]}\max_{j\in[p]}\max_{m\in[\widetilde{M}]}\sup_{u\in B_m}|\delta_{t,j}(u)-\delta_{t,j}(b_m)|\nonumber\\
   &~~~~~~~~~~~~~~~~~~~ \lesssim  \underbrace{\max_{t\in [n]}\max_{j\in[p]}\max_{m\in[\widetilde{M}]}|\widehat \bR_{t,j}(b_m) |_2}_{{\rm I}_1} + \underbrace{\max_{t\in [n]}\max_{j\in[p]}\max_{m\in[\widetilde{M}]}\sup_{u\in B_m}|\widehat{\bR}_{t,j}(u)-\widehat{\bR}_{t,j}(b_m)|_2}_{{\rm I}_2}\nonumber \\
   &~~~~~~~~~~~~~~~~~~~~~~+ \underbrace{\max_{t\in [n]}\max_{j\in[p]}\max_{m\in[\widetilde{M}]}\sup_{u\in B_m}\|\widehat{\bS}_{t,j}(u)-\widehat{\bS}_{t,j}(b_m)\|_{\rm F}   |\widehat{\bR}_{t,j}(b_m)|_2}_{{\rm I}_3} \,.
\end{align}
In what follows, we specify the convergence rates of ${\rm I}_1$, ${\rm I}_2$ and ${\rm I}_3$, respectively.

By Conditions~\ref{cond.Tij}--\ref{cond.kern}, we have that for $q = 2,3,\dots$ and $s = 0,1,2$,
\begin{align}\label{moment.bound1}
   &~ \bbE\bigg\{\bigg | \cK_{h_{t,j}}(u_{t,j,k}-u) \bigg( \frac{u_{t,j,k} - u}{h_{t,j}}\bigg)^s\bigg|^q\bigg\} 
= \int h_{t,j}^{-q} \cK^q\bigg(\frac{v-u}{h_{t,j}} \Big) \bigg|\frac{v-u}{h_{t,j}}\bigg|^{sq} f_{U}(v) \,{\rm d}v\nonumber \\
&~~~~~~~~~~~~~~~~~~~~~~~~~~~~~~~~~~~~= h_{t,j}^{1-q }\int  \cK^q(w)  |w|^{sq} f_{U}(h_{t,j} w+u) \,{\rm d}w \leq C_0^q h^{1-q }
\end{align}
for some positive constant $C_0$. For $l\in\{1,2\}$, let $\widehat R_{t,j,l}(u)$ be the $l$-th element of $\widehat{\bR}_{t,j}(u)$. Recall that $\vep^*_{t,j,k} = \vep_{t,j}(u_{t,j,k})+ \varsigma_{t,j,k}$.  Then $ \widehat R_{t,j,2}(b_m) = \widehat R_{t,j,2}^{(1)}(b_m)+\widehat R_{t,j,2}^{(2)}(b_m) $ with 
\begin{align*} 
\widehat R_{t,j,2}^{(1)}(b_m)&= \frac{1}{N_{t,j}}\sum_{k=1}^{N_{t,j}} \cK_{h_{t,j}}(u_{t,j,k}-b_m)\bigg(\frac{u_{t,j,k}-b_m}{h_{t,j}}\bigg)\varsigma_{t,j,k}\,,
\\
\widehat R_{t,j,2}^{(2)}(b_m)&=\frac{1}{N_{t,j}}\sum_{k=1}^{N_{t,j}} \cK_{h_{t,j}}(u_{t,j,k}-b_m)\bigg(\frac{u_{t,j,k}-b_m}{h_{t,j}}\bigg)\{\vep_{t,j}(u_{t,j,k}) - \vep_{t,j}(b_m)\}\,.
\end{align*}
By Condition \ref{cond.subgauss} and  Proposition 2.5.2 of \cite{Vershynin2018_supp}, for any positive integer $q$, $\bbE(|\varsigma_{t,j,k}|^q)\leq  {C}^q  q^{q/2}$. Together with the elementary inequality $q^{q/2}\leq q!$ for any integer
$q\geq1$, \eqref{moment.bound1} implies    
\begin{align*}
\bbE\bigg\{\bigg | \cK_{h_{t,j}}(u_{t,j,k}-b_m)\bigg(\frac{u_{t,j,k}-b_m}{h_{t,j}}\bigg)\varsigma_{t,j,k}\bigg|^q\bigg\} \leq \check{C}_0^q h^{1-q}q!   
\end{align*}
for some positive constant $\check{C}_0$. Notice that $\bbE[\cK_{h_{t,j}}(u_{t,j,k}-b_m)\{({u_{t,j,k}-b_m})/{h_{t,j}}\}\varsigma_{t,j,k}]=0$ for any $k\in[N_{t,j}]$. Then, by the i.i.d.   assumptions on $\{  u_{t,j,k}\}$ and 
$\{\varsigma_{t,j,k}\}$ in \eqref{eq:dismodel}, together with
Conditions~\ref{cond.subgauss} and \ref{cond.fu}(ii),
Bernstein's inequality \citep[see, e.g., Theorem 2.10 and Corollary 2.11 of][]{Boucheron2013_supp} gives 
\begin{align*}
\bbP\bigg\{ \max_{t\in [n]}\max_{j\in[p]}\max_{m\in[\widetilde{M}]}|\widehat R_{t,j,2}^{(1)}(b_m)|\geq x \bigg\}\lesssim np\widetilde{M} \exp\bigg(-\frac{  C N h x^2}{1+x}\bigg) 
\end{align*}
 for any $x \geq 0$, which implies  
\begin{align}\label{eq.R2}
    \max_{t\in [n]}\max_{j\in[p]}\max_{m\in[\widetilde{M}]}|\widehat R_{t,j,2}^{(1)}(b_m)| = O_{\rm p}\bigg[ \bigg\{ \frac{\log(np \widetilde{M})}{N h } \bigg\}^{1/2} \vee  \frac{\log(np \widetilde{M})}{N h }\bigg]\,.
\end{align}
For $\widehat R_{t,j,2}^{(2)}(b_m)$,  define  an event
\begin{align*}
    \mathcal E_U
=
\bigg\{
\max_{t\in[n]}\max_{j\in[p]}\max_{m\in[\widetilde M]}
\frac{1}{N_{t,j}h_{t,j}}
\sum_{k=1}^{N_{t,j}}
I(|u_{t,j,k}-b_m|\leq h_{t,j})
\leq C_U
\bigg\}\,,
\end{align*}
where $C_U>0$ is a sufficiently large constant. By
Conditions~\ref{c.subgaussian} and \ref{cond.Tij}--\ref{cond.kern}, it holds on $\mathcal E_U$ that
\begin{align*}
&\big\|
\widehat R_{t,j,2}^{(2)}(b_m)
\,\big|\,
\{u_{t,j,k}\}_{k=1}^{N_{t,j}}
\big\|_{\psi_2} \\
&~~~\leq
\frac{1}{N_{t,j}}
\sum_{k=1}^{N_{t,j}}
\bigg|
\cK_{h_{t,j}}(u_{t,j,k}-b_m)
\bigg(\frac{u_{t,j,k}-b_m}{h_{t,j}}\bigg)
\bigg|
\big\|
\{\vep_{t,j}(u_{t,j,k})-\vep_{t,j}(b_m)\}\,|\,
u_{t,j,k} 
\big\|_{\psi_2}  \\
&~~~\lesssim
\frac{1}{N_{t,j}h_{t,j}}
\sum_{k=1}^{N_{t,j} } 
I(|u_{t,j,k}-b_m|\leq h_{t,j})
|u_{t,j,k}-b_m|^\kappa   
\lesssim h^\kappa  \,,
\end{align*}
where $\|X\,|\, \{u_{t,j,k}\}_{k=1}^{N_{t,j}}\|_{\psi_2}$ denotes the $\psi_2$-norm of $X$ given $\{u_{t,j,k}\}_{k=1}^{N_{t,j}}$.
%$\|\cdot\,|\, \{u_{t,j,k}\}_{k=1}^{N_{t,j,k}}\|_{\psi_2}:= \inf\{\lambda>0: \bbE[\psi_2(|\cdot|/\lambda)\,|\, \{u_{t,j,k}\}_{k=1}^{N_{t,j,k}}] \leq 1\}$. 
Consequently, on $\mathcal E_U$, it holds that
\begin{align*}
    \bbP\big[
|\widehat R_{t,j,2}^{(2)}(b_m)|>x
\,\big|\,
\{u_{t,j,k}\}_{k=1}^{N_{t,j}}
\big]
\lesssim \exp\bigg(-\frac{Cx^2}{h^{2\kappa}}\bigg) 
\end{align*}
for any $x\geq 0$.   
By the union bound, we can conclude that
\begin{align*}
&\bbP\bigg\{
\max_{t\in[n]}\max_{j\in[p]}\max_{m\in[\widetilde M]}
|\widehat R_{t,j,2}^{(2)}(b_m)|>x
\bigg\} \lesssim 
\bbP(\mathcal E_U^{\rm c})
+
 np\widetilde M
\exp\bigg(-\frac{Cx^2}{h^{2\kappa}}\bigg) 
\end{align*}
for any $x\geq 0$.  To bound $\bbP(\mathcal E_U^{\rm c})$, Condition~\ref{cond.fu} implies that
\begin{align*}
    \mathbb E\{I(|u_{t,j,k}-b_m|\leq h_{t,j})\} = \mathbb P(|U-b_m|\leq h_{t,j})= \int_{b_m-h_{t,j}}^{b_m+h_{t,j}}f_U(v)\,\mbox{d}v \leq 2C h_{t,j}\,.
\end{align*}
Then, Bernstein's inequality and the union bound give
\begin{align*}
    \mathbb P(\mathcal E_U^{\rm c})
\leq
np\widetilde M\exp(-CNh)=o(1)
\end{align*}
provided that
$\log(np\widetilde M)\ll Nh$. Therefore, it follows that
\begin{align*} 
\max_{t\in[n]}\max_{j\in[p]}\max_{m\in[\widetilde M]}
|\widehat R_{t,j,2}^{(2)}(b_m)|
=
O_{\rm p}[
\{\log (np\widetilde M)\}^{1/2} h^\kappa ]
\end{align*}
provided that
$\log(np\widetilde M)\ll Nh$. This, together with \eqref{eq.R2}, yields 
\begin{align*} 
\max_{t\in[n]}\max_{j\in[p]}\max_{m\in[\widetilde M]}
|\widehat R_{t,j,2}(b_m)|
=
O_{\rm p}\bigg[
\{\log (np\widetilde M)\}^{1/2} \bigg\{ \frac{1}{(Nh)^{1/2}}+ h^\kappa\bigg\}\bigg ] 
\end{align*}
provided that
$\log(np\widetilde M)\ll Nh$. The same rate holds for $\max_{t,j,m}|\widehat R_{t,j,1}(b_m)|$. Hence,
\begin{align}\label{I1.OP}
    {\rm I}_1 = O_{\rm p}\bigg[
\{\log (np\widetilde M)\}^{1/2} \bigg\{ \frac{1}{(Nh)^{1/2}}+ h^\kappa\bigg\}\bigg ]
\end{align}
provided that
$\log(np\widetilde M)\ll Nh$.

We next consider ${\rm I}_2$. By  \eqref{def.SR}, it holds that
\begin{align*}
    \widehat{R}_{t,j,1}(u)- \widehat{R}_{t,j,1}(b_m)=&~ \frac{1}{N_{t,j}}\sum_{k=1}^{N_{t,j}}   \vep^*_{t,j,k} \{ \cK_{h_{t,j}}(u_{t,j,k}-u)-\cK_{h_{t,j}}(u_{t,j,k}-b_m)\} \\
    &~ - \frac{1}{N_{t,j}}\sum_{k=1}^{N_{t,j}}   \vep_{t,j}(u) \{\cK_{h_{t,j}}(u_{t,j,k}-u)-  \cK_{h_{t,j}}(u_{t,j,k}-b_m)  \}\\
    &~ - \frac{1}{N_{t,j}}\sum_{k=1}^{N_{t,j}}  \{ \vep_{t,j}(u)  -\vep_{t,j}(b_m) \}\cK_{h_{t,j}}(u_{t,j,k}-b_m)   \,.
\end{align*} 
Since $|u-b_m|\leq \widetilde{M}^{-1}$ for any $u\in B_m$, by Conditions \ref{cond.Tij} and  \ref{cond.kern}, we have  $|\cK_{h_{t,j}}(u_{t,j,k}-u)-\cK_{h_{t,j}}(u_{t,j,k}-b_m)| \lesssim (h^2 \widetilde{M})^{-1}$ for any $u\in B_m$. Thus, for any   $(t,j)$ and  $u\in B_m$,
\begin{align*}
    |\widehat{R}_{t,j,1}(u)- \widehat{R}_{t,j,1}(b_m)| \lesssim &~\frac{1}{h} \max_{t\in [n]}\max_{j\in[p]}\max_{m\in [\widetilde{M}]} \sup_{u\in B_m} | \vep_{t,j}(u)  -\vep_{t,j}(b_m)|  \\
    &~ +\frac{1}{h^2 \widetilde{M}}\bigg( \max_{t\in [n]}\max_{j\in[p]}\frac{1}{N_{t,j}}\sum_{k=1}^{N_{t,j}}  |\vep^*_{t,j,k}| + \max_{t\in [n]}\max_{j\in[p]}\sup_{u\in \cU}  |\vep_{t,j}(u) |\bigg) \,,
\end{align*} 
where the inequality follows from $|\cK_{h_{t,j}}(u_{t,j,k}-b_m)|\lesssim h^{-1}$, which implies that 
\begin{align}\label{div.R1}
     &~\max_{t\in [n]}\max_{j\in[p]}\max_{m\in [\widetilde{M}]} \sup_{u\in B_m} |\widehat{R}_{t,j,1}(u)- \widehat{R}_{t,j,1}(b_m)|\lesssim\frac{1}{h} \max_{t\in [n]}\max_{j\in[p]}\max_{m\in [\widetilde{M}]} \sup_{u\in B_m} | \vep_{t,j}(u)  -\vep_{t,j}(b_m)|\nonumber\\
     &~~~~~~~~~~~~~~~~~~~~~~~~~~~~ +\frac{1}{h^2 \widetilde{M}}\bigg( \max_{t\in [n]}\max_{j\in[p]}\frac{1}{N_{t,j}}\sum_{k=1}^{N_{t,j}} |\vep^*_{t,j,k}| + \max_{t\in [n]}\max_{j\in[p]}\sup_{u\in \cU}  |\vep_{t,j}(u) |\bigg) \,.
\end{align}
By \eqref{div.vep}, it holds that
\begin{align*}
    \frac{1}{h} \max_{t\in [n]}\max_{j\in[p]}\max_{m\in [\widetilde{M}]} \sup_{u\in B_m} | \vep_{t,j}(u)  -\vep_{t,j}(b_m)| = O_{\rm p}\bigg\{\frac{ (np)^2}{h \widetilde{M}^{ \kappa/2} }   \bigg\}\,.
\end{align*}
 By Conditions \ref{c.subgaussian} and \ref{cond.subgauss}, we have $\| \vep^*_{t,j,k} \|_{\psi_2}\leq C$ for any $(t,j,k)$.
 %$k\in [N_{t,j}]$, 
 which implies that $\|N_{t,j}^{-1}\sum_{k=1}^{N_{t,j}} |\vep^*_{t,j,k}|\|_{\psi_2}\leq C$.   Then,  we can conclude that
\begin{align*}
     \max_{t\in [n]}\max_{j\in[p]}\frac{1}{N_{t,j}}\sum_{k=1}^{N_{t,j}} |\vep^*_{t,j,k}| =  O_{\rm p} [ \{ \log(np) \}^{1/2}  ]\,.
\end{align*}
 This, together with Lemma \ref{lemma1.prop1} and \eqref{div.R1}, yields 
 \begin{align}\label{R1.OP}
     \max_{t\in [n]}\max_{j\in[p]}\max_{m\in [\widetilde{M}]} \sup_{u\in B_m} |\widehat{R}_{t,j,1}(u)- \widehat{R}_{t,j,1}(b_m)| =  O_{\rm p}\bigg[\frac{ (np)^2}{h \widetilde{M}^{ \kappa/2} }  + \frac{   \{ \log(np) \}^{1/2}   }{h^2 \widetilde{M}}  \bigg]\,. 
 \end{align}
Moreover, write $\widetilde{\cK}_{h_{t,j}}(u_{t,j,k}-u)=\{(u_{t,j,k}-u)/h_{t,j}\}{\cK}_{h_{t,j}}(u_{t,j,k}-u) $ for any $u\in\cU$. We have
\begin{align*}
    \widehat{R}_{t,j,2}(u)- \widehat{R}_{t,j,2}(b_m)=&~ \frac{1}{N_{t,j}}\sum_{k=1}^{N_{t,j}}   \vep^*_{t,j,k} \{ \widetilde{\cK}_{h_{t,j}}(u_{t,j,k}-u)-\widetilde{\cK}_{h_{t,j}}(u_{t,j,k}-b_m)\} \\
    &~ - \frac{1}{N_{t,j}}\sum_{k=1}^{N_{t,j}}   \vep_{t,j}(u) \{\widetilde{\cK}_{h_{t,j}}(u_{t,j,k}-u)-  \widetilde{\cK}_{h_{t,j}}(u_{t,j,k}-b_m)  \}\\
    &~ - \frac{1}{N_{t,j}}\sum_{k=1}^{N_{t,j}}  \{ \vep_{t,j}(u)  -\vep_{t,j}(b_m) \}\widetilde{\cK}_{h_{t,j}}(u_{t,j,k}-b_m)   \,.
\end{align*} 
It follows from Conditions \ref{cond.Tij} and  \ref{cond.kern} that, for any $u\in B_m$, $|\widetilde{\cK}_{h_{t,j}}(u_{t,j,k}-u)|\lesssim h^{-1} $ and 
 \begin{align*}
    |\widetilde{\cK}_{h_{t,j}}(u_{t,j,k}-u)-\widetilde{\cK}_{h_{t,j}}(u_{t,j,k}-b_m)|\lesssim  \frac{1}{h^2 \widetilde{M}}\,.
\end{align*}
Therefore, by the same arguments as those used to prove \eqref{R1.OP}, we obtain the
same convergence rate for $\max_{t\in [n]}\max_{j\in[p]}\max_{m\in [\widetilde{M}]} \sup_{u\in B_m} |\widehat{R}_{t,j,2}(u)- \widehat{R}_{t,j,2}(b_m)|$. Consequently,
\begin{align}\label{I2.OP}
    {\rm I}_2  =  O_{\rm p}\bigg[\frac{ (np)^2}{h \widetilde{M}^{ \kappa/2} }  + \frac{   \{ \log(np) \}^{1/2}   }{h^2 \widetilde{M}}  \bigg]\,. 
\end{align}

For ${\rm I}_3$,   it follows from Conditions \ref{cond.Tij} and  \ref{cond.kern} that, for any $s\in\{0,1,2\}$ and $u\in B_m$,
\begin{align*}
    \bigg|\frac{1}{N_{t,j}}\sum_{k=1}^{N_{t,j}}\bigg\{\bigg(\frac{u_{t,j,k}-u}{h_{t,j}}\bigg)^s\cK_{h_{t,j}}(u_{t,j,k}-u) - \bigg(\frac{u_{t,j,k}-b_m}{h_{t,j}}\bigg)^s\cK_{h_{t,j}}(u_{t,j,k}-b_m)\bigg\}\bigg|\lesssim  \frac{1}{h^2 \widetilde{M}}\,,
\end{align*}
which implies that 
\begin{align}\label{S.div}
    \max_{t\in [n]}\max_{j\in[p]}\max_{m\in[\widetilde{M}]}\sup_{u\in B_m}\|\widehat{\bS}_{t,j}(u)-\widehat{\bS}_{t,j}(b_m)\|_{\rm F} \lesssim  \frac{1}{h^2 \widetilde{M}}\,.
\end{align}
Therefore, we have 
\begin{align}\label{I3.OP}
    {\rm I}_3    = O_{\rm p}\bigg[
\frac{\{\log (np\widetilde M)\}^{1/2}}{h^2 \widetilde{M}} \bigg\{ \frac{1}{(Nh)^{1/2}}+ h^\kappa\bigg\}\bigg ] 
\end{align}
provided that
$\log(np\widetilde M)\ll Nh$. On the event $\mathcal E_1$, combining \eqref{delta.sup}, \eqref{I1.OP}, \eqref{I2.OP} and \eqref{I3.OP} gives
\begin{align}\label{dis.delta.sup}
    &~\max_{t\in [n]}\max_{j\in[p]}\sup_{u\in\cU}|\delta_{t,j}(u)| \lesssim {\rm I}_1+{\rm I}_2+{\rm I}_3\nonumber\\
    &~~~~~~= O_{\rm p}\bigg[\frac{\{\log (np\widetilde M)\}^{1/2}}{(h^2 \widetilde{M}) \wedge 1} \bigg\{ \frac{1}{(Nh)^{1/2}}+ h^\kappa\bigg\} + \frac{ (np)^2}{h \widetilde{M}^{ \kappa/2} }  + \frac{  \{ \log(np)\}^{1/2}   }{h^2 \widetilde{M}}\bigg]\,,\nonumber\\
    &~\max_{t\in [n]}\max_{j\in[p]}\max_{m\in[\widetilde{M}]}\sup_{u\in B_m}|\delta_{t,j}(u)-\delta_{t,j}(b_m)| \lesssim {\rm I}_2+{\rm I}_3\nonumber\\
    &~~~~~~= O_{\rm p}\bigg[ \frac{\{\log (np\widetilde M)\}^{1/2}}{h^{2-\kappa} \widetilde{M}}   + \frac{ (np)^2}{h \widetilde{M}^{ \kappa/2} }  + \frac{     \{ \log(np)\}^{1/2}   }{h^2 \widetilde{M}}\bigg]\,,
\end{align}
provided that
$\log(np\widetilde M)\ll Nh$. Choose $\widetilde M\asymp h^{-(2+2\kappa)/\kappa}(np)^C$ for some sufficiently large constant
$C>0$. Then, on the event $\mathcal E_1$, 
\begin{align*}
    \max_{t\in [n]}\max_{j\in[p]}\sup_{u\in\cU}|\delta_{t,j}(u)|  = O_{\rm p}\bigg[ \{\log(np /h )\}^{1/2} \bigg\{ \frac{1}{(Nh)^{1/2}}+ h^\kappa\bigg\}  \bigg]\,,
\end{align*}
provided that $\log(np/h)\ll N h$. Taking $\widetilde M =M_1$ in the second bound of \eqref{dis.delta.sup} yields the second assertion.  Therefore, to confirm Lemma \ref{lemma2.prop1}, it remains to show that $\bbP(\mathcal E_1^{\rm c})=o(1)$.

For some integer $M'>0$, it follows from \eqref{S.div} that
\begin{align*}
     \max_{t\in [n]}\max_{j\in[p]}\sup_{u\in\cU} \|\widehat{\bf S}_{t,j}(u) - \bbE\{\widehat{\bf S}_{t,j}(u)\}\|_{\rm F} \lesssim \max_{t\in [n]}\max_{j\in[p]}\max_{m\in[M']} \|\widehat{\bf S}_{t,j}(b_m) - \bbE\{\widehat{\bf S}_{t,j}(b_m)\}\|_{\rm F} +  \frac{1}{h^2  {M'}}\,.
\end{align*}
By  \eqref{moment.bound1}, we have,  for any $q\ge 2$, $s\in\{0,1,2\}$ and  $u\in\cU$, 
\begin{align*} 
    \bbE\bigg\{\bigg |  \cK_{h_{t,j}}(u_{t,j,k}-u) \bigg( \frac{u_{t,j,k} - u}{h_{t,j}}\bigg)^s\bigg|^q\bigg\} \leq &~  \tilde{C}     h^{1-q}  q! 
\end{align*}
 for some sufficiently large constant $\tilde{C} >0$.  For $l_1, l_2 \in \{1,2\}$, let $\widehat S_{t,j,l_1,l_2}(u)$ be the $(l_1,l_2)$-th element of $\widehat{\bf S}_{t,j}(u)$. By the Bernstein inequality \citep[see, e.g., Theorem 2.10 and Corollary 2.11 of][]{Boucheron2013_supp},   it holds that
 \begin{align*}
     \bbP\big[ |\widehat{S}_{t,j,l_1,l_2}(u) - \bbE\{\widehat{S}_{t,j,l_1,l_2}(u)\}| >x   \big]\lesssim  \exp\bigg( -C\frac{  N h x^2}{1+x} \bigg)
 \end{align*}
for any $x\geq 0$ and $l_1,l_2 \in \{1,2\}$, which implies that, for any $x\geq 0$ and $u \in \cU$,
\begin{align*} 
 \max_{t\in[n]}\max_{j\in[p]}\bbP\big[ \|\widehat{\bf S}_{t,j}(u) - \bbE\{\widehat{\bf S}_{t,j}(u)\}\|_{\rm F} \geq x \big]\lesssim  \exp\bigg( -C\frac{  N h x^2}{1+x} \bigg) \,.
\end{align*}
Then it holds that
\begin{align*}
   \bbP(\mathcal E_1^{\rm c})\leq &~  \bbP\bigg[ \max_{t\in [n]}\max_{j\in[p]}\max_{m\in[M']} \|\widehat{\bf S}_{t,j}(b_m) - \bbE\{\widehat{\bf S}_{t,j}(b_m)\}\|_{\rm F} \geq \frac{c_S}{2} - \frac{C}{h^2  {M'}} \bigg] \\ 
   \lesssim &~ npM'\exp( -C  N h ) =o(1) \,,
\end{align*}
provided that $(M')^{-1}\ll h^2$ and $\log(np M')\ll Nh$. Taking $M'\asymp h^{-4}$, we obtain $\bbP(\mathcal E_1^{\rm c})=o(1)$ provided that
$\log(np/h)\ll Nh$.  We complete the proof of Lemma \ref{lemma2.prop1}.  $\hfill\Box$

\subsection{Proof of Lemma \ref{lemma1.prop2}}\label{proof.lemma1.prop2}

Recall  $\widetilde{\bSigma}_{xx}^{(\ell)}(u,v) = n_{\ell}^{-1}\sum_{t=1}^{n_\ell}\{\bX_{t}(u)-\widebar{\bX}(u)\}\{\bX_{t+\ell}(v) - \widebar{\bX}(v)\}^{\T}$  for any $\ell\in[\ell_0]$ and $(u,v)\in\cU^2$, where $\widebar{\bX}(\cdot) = n^{-1}\sum_{t=1}^n\bX_t(\cdot)$ and $\bX_t(\cdot)=\{X_{t,1}(\cdot),\ldots,X_{t,p}(\cdot)\}^{\T}$. Moreover, by Conditions \ref{cond.facerror} and \ref{cond.facprocess}, we know that  ${\bSigma}_{xx}^{(\ell)}(u,v) =\cov\{\bX_{t}(u),\bX_{t+\ell}(v)\} = \bbE[\bX_{t}(u) \{\bX_{t+\ell}(v)\}^{\T}]$ and ${\bSigma}_{zz}^{(\ell)}(u,v) =\cov\{\bZ_{t}(u),\bZ_{t+\ell}(v)\} = \bbE[\bZ_{t}(u) \{\bZ_{t+\ell}(v)\}^{\T}]$.
The proof of Lemma~\ref{lemma1.prop2} relies on the following lemma, whose proof is given in Section~\ref{proof.lem.fac.Sigma}.

\begin{lemma}\label{lem.fac.Sigma}
Under Conditions {\rm \ref{cond.facerror}--\ref{cond.facmat}}, we have  {\rm (i)}  $\max_{\ell\in[\ell_0]}\max_{j,k\in[p]}\bbE\{\|\widetilde{\Sigma}^{(\ell)}_{xx,jk}-\Sigma^{(\ell)}_{xx,jk}\|_{\cS}^2\} \lesssim n^{-1}$;  and {\rm (ii)} $\max_{\ell\in[\ell_0]}\|\widetilde{\bSigma}_{xx}^{(\ell)} - \bSigma_{xx}^{(\ell)}\|_{\cS,\max}=O_{\rm p}[ n^{-1/2}\{\log (np)\}^{1/2} ] $, provided that $\log (np) \ll n^{1/5}$.
\end{lemma}

Recall    $\bX_t(u) = \bA\bZ_t(u)+\bvep_t(u)$ and $ \mathbb{E}\{ \bvep_t(u) \bvep_{t+\ell}(v)^{\T} \}= \mathbf{0}$ for any $\ell\neq 0$ under model \eqref{eq:ffm}. Then, by Conditions \ref{cond.facerror}(iii), it holds  that $ \bSigma_{xx}^{(\ell)}(u,v) =  \bbE[\{\bA\bZ_t(u)+\bvep_t(u)\} \{\bA\bZ_{t+\ell}(v)+\bvep_{t+\ell}(v)\}^{\T}] =  \bA\bSigma_{zz}^{(\ell)}(u,v)\bA^{\T}$. Moreover, by Condition \ref{cond.facmat}(i), we have ${\rm rank}(\bA)=r$ and
\begin{align*}
    {\bM} := &~\sum_{\ell=1}^{\ell_0} \iint {\bSigma}_{xx}^{(\ell)}(u,v) \{{\bSigma}_{xx}^{(\ell)}(u,v)\}^{\T}\,{\rm d}u{\rm d}v \\
    = &~ \bA  \underbrace{\bigg[\sum_{\ell=1}^{\ell_0} \iint \bSigma_{zz}^{(\ell)}(u,v)\bA^{\T}\bA\{\bSigma_{zz}^{(\ell)}(u,v)\}^{\T}  \,{\rm d}u{\rm d}v \bigg]}_{\bH} \bA^{\T} \,.
\end{align*}
By Conditions \ref{cond.facprocess}(iv) and \ref{cond.facmat}(i), we know that
\begin{align*}
    \lambda_{\min}(\bH) \geq \lambda_{\min}\bigg[ \sum_{\ell=1}^{\ell_0} \iint \bSigma_{zz}^{(\ell)}(u,v)\{\bSigma_{zz}^{(\ell)}(u,v)\}^{\T}  \,{\rm d}u{\rm d}v  \bigg] \cdot \lambda_{\min}(\bA^{\T}\bA) \gtrsim p 
    \,,
\end{align*}
where $\lambda_{\min}(\bB)$ denotes the smallest eigenvalue of symmetric matrix $\bB$. Then
${\rm rank}(\bH)=r$, which implies ${\rm rank}(\bM)=r$  and $\mathcal{C}(\bA)=\mathcal{C}({\bM})$, where $\mathcal C(\bB)$ denotes the linear space spanned by the columns of $\bB$. By Condition \ref{cond.facprocess}  again, we have $\|{\bSigma}_{zz}^{(\ell)}(u,v)\|_{\max}\leq C$ for any $\ell \in [\ell_0]$ and $(u,v)\in\cU^2$. Since $r$ is fixed, it holds that $\|\bSigma_{zz}^{(\ell)}(u,v)\|_{\rm op} \leq \|\bSigma_{zz}^{(\ell)}(u,v)\|_{\rm F} \leq C$ for any $\ell \in [\ell_0]$ and $(u,v)\in\cU^2$. Let $\nu_1\geq\cdots\geq\nu_r$ denote the $r$ nonzero eigenvalues of $\bM$. Then, it holds that
\begin{align}\label{eq:Mop}
    \nu_1 
    =&\, \|{\bM}\|_{\rm op} 
    \leq \sum_{\ell=1}^{\ell_0} \iint \|{\bSigma}_{xx}^{(\ell)}(u,v) \|_{\rm op}^2\,{\rm d}u{\rm d}v  
    = \sum_{\ell=1}^{\ell_0} \iint \|\bA\bSigma_{zz}^{(\ell)}(u,v)\bA^{\T}\|_{\rm op}^2\,{\rm d}u{\rm d}v \notag\\
    \leq &\, \|\bA\|_{\rm op}^4 \sum_{\ell=1}^{\ell_0} \iint \|\bSigma_{zz}^{(\ell)}(u,v)\|_{\rm op}^2 \,{\rm d}u{\rm d}v   \lesssim  p^2\,,
\end{align}
where the last inequality follows from Condition \ref{cond.facmat}(i). 
% In addition, by the Poincaré separation theorem, since $(p^{-1/2}\bA)^{\T}(p^{-1/2}\bA)=  \bI_r$ and $\bM = \bA\bH\bA^{\T}$, it holds that 
% \begin{align*}
%     \nu_r =  \lambda_r(\bM) \geq  \lambda_r\{(p^{-1/2}\bA)^{\T}\bM (p^{-1/2}\bA) \} = p  \lambda_r( \bH ) \gtrsim p^2\,,
% \end{align*}
% where the last inequality follows from Condition~\ref{cond.facprocess}(iv). 
In addition, 
\begin{align*}
    \nu_r =  \lambda_r(\bM) = \lambda_r(\bA\bH\bA^{\T}) \geq \lambda_r( \bH )\lambda_r(\bA \bA^{\T}) \gtrsim p^{2 }\,.
\end{align*}
Here, $\lambda_r(\bB)$ denotes the $r$-th largest eigenvalue of a symmetric matrix $\bB$. Then, we can conclude that 
\begin{align}\label{eq:eigval}
    \nu_i\asymp  p^{2 }~~~~\mbox{for}~~ i\in[r]\,.
\end{align}

Recall $\widehat{\bM} =  \sum_{\ell=1}^{\ell_0} \iint \widetilde{\bSigma}_{xx}^{(\ell)}(u,v) \{\widetilde{\bSigma}_{xx}^{(\ell)}(u,v)\}^{\T}\,{\rm d}u{\rm d}v$. It holds that
\begin{align*}
    \|\widehat{\bM}-\bM\|_{\rm op} 
    =&\, \bigg\|\sum_{\ell=1}^{\ell_0} \iint \big[\widetilde{\bSigma}_{xx}^{(\ell)}(u,v) \{\widetilde{\bSigma}_{xx}^{(\ell)}(u,v)\}^{\T} - {\bSigma}_{xx}^{(\ell)}(u,v) \{{\bSigma}_{xx}^{(\ell)}(u,v)\}^{\T} \big] \,{\rm d}u{\rm d}v\bigg\|_{\rm op}  \\
    \leq &\, \sum_{\ell=1}^{\ell_0}\iint \|\widetilde{\bSigma}_{xx}^{(\ell)}(u,v) - {\bSigma}_{xx}^{(\ell)}(u,v) \|_{\rm op}^2 \,{\rm d}u{\rm d}v  \\
    &\, +2 \sum_{\ell=1}^{\ell_0} \iint \|\widetilde{\bSigma}_{xx}^{(\ell)}(u,v) - {\bSigma}_{xx}^{(\ell)}(u,v)\|_{\rm op} \|{\bSigma}_{xx}^{(\ell)}(u,v)\|_{\rm op} \, {\rm d}u{\rm d}v \,.
\end{align*}
By Lemma \ref{lem.fac.Sigma}(i), for each $\ell \in [\ell_0]$ and $j,k\in[p]$,  we have $\|\widetilde{\Sigma}_{xx,jk}^{(\ell)} - \Sigma_{xx,jk}^{(\ell)}\|_{\cS}^2=O_{\rm p}(n^{-1})$. Then
\begin{align*}
    \iint \|\widetilde{\bSigma}_{xx}^{(\ell)}(u,v) - {\bSigma}_{xx}^{(\ell)}(u,v) \|_{\rm op}^2 \,{\rm d}u{\rm d}v 
  \leq&~ \sum_{j,k=1}^p   \|\widetilde{\Sigma}_{xx,jk}^{(\ell)} - \Sigma_{xx,jk}^{(\ell)}\|_{\cS}^2 = O_{\rm p}(p^2 n^{-1}) \,.
\end{align*}
Moreover, by  Cauchy-Schwarz inequality and \eqref{eq:Mop}, we have 
\begin{align*}
    &~\iint \|\widetilde{\bSigma}_{xx}^{(\ell)}(u,v) - {\bSigma}_{xx}^{(\ell)}(u,v)\|_{\rm op} \|{\bSigma}_{xx}^{(\ell)}(u,v)\|_{\rm op} \, {\rm d}u{\rm d}v\\
     \leq &~ \bigg\{\iint \|\widetilde{\bSigma}_{xx}^{(\ell)}(u,v) - {\bSigma}_{xx}^{(\ell)}(u,v) \|_{\rm op}^2 \,{\rm d}u{\rm d}v \bigg\}^{1/2}\bigg\{  \iint \|{\bSigma}_{xx}^{(\ell)}(u,v) \|_{\rm op}^2\,{\rm d}u{\rm d}v \bigg\}^{1/2} = O_{\rm p}(p^2 n^{-1/2})
\end{align*}
for any $\ell \in [\ell_0]$, which implies 
\begin{align}\label{eq:diffM}
     \|\widehat{\bM}-\bM\|_{\rm op} 
    = O_{\rm p}(p^2 n^{-1/2})\,.
\end{align}

Recall $\widetilde{\bA}=  \bA \bGamma {\bf{\Lambda}}^{-1/2}$ such that $\widetilde{\bA}^{\T}\widetilde{\bA}=\bI_r$. Let $\bQ_1=\widetilde{\bA} \in \mathbb{R}^{p\times r}$, and let $\bQ_2\in \mathbb{R}^{p\times (p-r)}$ be its orthogonal complement such that
$\bQ=(\bQ_1,\bQ_2) $ is an orthogonal matrix.  Since
$\mathcal C(\bM)=\mathcal C(\bA)=\mathcal C(\widetilde{\bA})$ and ${\rm rank}(\bM)=r$, we have
\begin{align*}
   \bQ^{\T}\bM\bQ
=
\biggl(
\begin{array}{cc}
\bD_1 & \0\\
\0 & \0
\end{array}
\biggr)  \,, 
\end{align*} 
where the eigenvalues of $\bD_1$ are $\nu_1,\ldots,\nu_r$. Hence, by
\eqref{eq:eigval}, it holds that 
\begin{align}\label{eigengap}
    {\rm sep}(\bD_1,\0)=\nu_r \asymp  p^2\,,
\end{align}
where ${\rm sep}(\bD_1,\bD_2) = \min_{\mu_1\in\lambda(\bD_1),\mu_2\in\lambda(\bD_2)}|\mu_1-\mu_2|$, and $\lambda(\bD)$ denotes the set of eigenvalues of the matrix $\bD$.
Moreover, let $\bE_M=\widehat{\bM}-\bM$, and partition 
\begin{align*}
    \bQ^{\T}\bE_M\bQ
= \biggl(
\begin{array}{cc}
\bE_{11} & \bE_{21}^{\T}\\
\bE_{21} & \bE_{22}
\end{array}
\biggr) \,.
\end{align*}
By \eqref{eq:diffM}, we have $\|\bE_M\|_{\rm op}=O_{\rm p}(n^{-1/2}p^2)$. This, together with \eqref{eigengap}, yields $\|\bE_M\|_{\rm op} \leq {\rm sep}(\bD_1,\0)/5 $ with probability approaching one. Then, by Lemma 3 of \cite{Lam2011_supp},  there exists a matrix
$\bP\in\mathbb R^{(p-r)\times r}$ such that 
\begin{align*}
    \|\bP\|_{\rm op}
\leq
\frac{4}{ {\rm sep}(\bD_1,0)}\|\bE_{21}\|_{\rm op}
\leq
\frac{4}{\nu_r}\|\bE_M\|_{\rm op}
=
O_{\rm p}(  n^{-1/2}) \,.
\end{align*}
Moreover, the columns of $\bar{\bA}
=
(\bQ_1+\bQ_2\bP)(\bI_r+\bP^{\T}\bP)^{-1/2}$ form an  orthonormal basis for a subspace that is invariant for $\widehat{\bM}$. Therefore, due to $\bQ_1=\widetilde\bA$, it holds that
\begin{align*}
\|\bar{\bA}-\widetilde\bA\|_{\rm op}=&~\|\bar{\bA}-\bQ_1\|_{\rm op}
 =
\left\|
[\bQ_1\{\bI_r-(\bI_r+\bP^{\T}\bP)^{1/2}\}
+\bQ_2\bP](\bI_r+\bP^{\T}\bP)^{-1/2}
\right\|_{\rm op}  \\
 \leq &~ 
\left\|\bI_r-(\bI_r+\bP^{\T}\bP)^{1/2}\right\|_{\rm op}
+
\|\bP\|_{\rm op}  \leq 
2\|\bP\|_{\rm op}
=
O_{\rm p}( n^{-1/2})\,.
\end{align*}
Since $\widehat{\bA}=(\hat{\bomega}_1,\ldots,\hat{\bomega}_r)$, where
$\hat{\bomega}_1,\ldots,\hat{\bomega}_r$ are the leading $r$ eigenvectors of
$\widehat{\bM}$, the columns of $\widehat{\bA}$ and $\bar{\bA}$ form two
orthonormal bases for the same leading invariant subspace of $\widehat{\bM}$.
Therefore, there exists an $r\times r$ orthogonal matrix $\bU_0$ such that
$\bar{\bA}=\widehat{\bA}\bU_0$. Then, we complete the proof of Lemma \ref{lemma1.prop2}(i).

For  Lemma \ref{lemma1.prop2}(ii), since $\|{\bSigma}_{zz}^{(\ell)}(u,v)\|_{\max}\leq C$ for any $(u,v)\in\cU^2$, it holds by Condition \ref{cond.facmat}(ii) that $\|\bSigma_{xx}^{(\ell)}\|_{\cS,1} \lesssim p$ and $\|\bSigma_{xx}^{(\ell)}\|_{\cS,\infty} \lesssim p$. Then, by Lemma A.6(ii) of \cite{LQW2025_supp} and Lemma \ref{lem.fac.Sigma}(ii), we have 
\begin{align*}
    \|\widehat{\bM}-\bM\|_{\infty} 
    \leq &~ \sum_{\ell=1}^{\ell_0} \bigg\| \iint \big[\widetilde{\bSigma}_{xx}^{(\ell)}(u,v) \{\widetilde{\bSigma}_{xx}^{(\ell)}(u,v)\}^{\T} - {\bSigma}_{xx}^{(\ell)}(u,v) \{{\bSigma}_{xx}^{(\ell)}(u,v)\}^{\T} \big] \,{\rm d}u{\rm d}v\bigg\|_{\infty} \\
    \leq&~  \sum_{\ell=1}^{\ell_0}\bigg\|\iint \big\{\widetilde{\bSigma}_{xx}^{(\ell)}(u,v) - \bSigma_{xx}^{(\ell)}(u,v)\big\} \big\{\widetilde{\bSigma}_{xx}^{(\ell)}(u,v) - \bSigma_{xx}^{(\ell)}(u,v)\big\}^{\T} \,{\rm d}u{\rm d}v \bigg\|_{\infty} \\
 &~ +   \sum_{\ell=1}^{\ell_0}\bigg\|\iint  \big\{\widetilde{\bSigma}_{xx}^{(\ell)}(u,v) - \bSigma_{xx}^{(\ell)}(u,v)\big\}\big\{\bSigma_{xx}^{(\ell)}(u,v)\big\}^{\T}\,{\rm d}u{\rm d}v  \bigg\|_{\infty}  \\
 &~ +   \sum_{\ell=1}^{\ell_0}\bigg\|\iint  \bSigma_{xx}^{(\ell)}(u,v)\big\{\widetilde{\bSigma}_{xx}^{(\ell)}(u,v) - \bSigma_{xx}^{(\ell)}(u,v)\big\}^{\T}\,{\rm d}u{\rm d}v  \bigg\|_{\infty}  \\
 \leq&~ \sum_{\ell=1}^{\ell_0}\big\{ \|\widetilde{\bSigma}_{xx}^{(\ell)} - \bSigma_{xx}^{(\ell)}\|_{\cS,\infty} \|\widetilde{\bSigma}_{xx}^{(\ell)} - \bSigma_{xx}^{(\ell)}\|_{\cS,1}
 \\
 &~~~~~~ ~~~+   \|\widetilde{\bSigma}_{xx}^{(\ell)} - \bSigma_{xx}^{(\ell)}\|_{\cS,\infty} \|\bSigma_{xx}^{(\ell)}\|_{\cS,1}+\|\bSigma_{xx}^{(\ell)}\|_{\cS,\infty}\|\widetilde{\bSigma}_{xx}^{(\ell)} - \bSigma_{xx}^{(\ell)}\|_{\cS,1} \big\}\\
 \lesssim &~ p^2\sum_{\ell=1}^{\ell_0} \|\widetilde{\bSigma}_{xx}^{(\ell)} - \bSigma_{xx}^{(\ell)}\|_{\cS,\max} =O_{\rm p}[n^{-1/2}p^2\{\log (np)\}^{1/2}] \,,
\end{align*}
provided that $\log (np) \ll n^{1/5}$.  
% In addition, since $v_i\asymp p^{2}$ for $i\in[r]$, by \eqref{eq:diffM} and Condition \ref{cond.facprocess}(iv), we can find $\iota \approx\|\widehat{\bM}-\bM\|_{\rm op}=o_{\rm p}(p^{2})$ and $\iota >  \|\widehat{\bM}-\bM\|_{\rm op}$ such that  for any $i\in[r]$, the interval $(\nu_i-\iota,\nu_i+\iota)$ does not
% contain any eigenvalues of $\bM$ other than $\nu_i$ with probability approaching one. 
% Since ${\rm rank}(\bM)=r$, it holds by Condition \ref{cond.facmat}  and  Lemma B.13 of  \cite{LQW2025_supp}   that  $\mu= (p/r)\max_{j\in[p]}\sum_{k=1}^r {\color{red}\widetilde{A}_{jk}^2} \asymp    \max_{j\in[p]}|\ba_j|_2^2\leq C  $ and 
Recall $\nu_i\asymp p^{2}$ for $i\in[r]$. Since ${\rm rank}(\bM)=r$, it holds by Condition \ref{cond.facmat} and Theorem 2 of \cite{Fan2018_supp} that $\mu=(p/r)\max_{j\in[p]}\sum_{k=1}^r \widetilde{A}_{jk}^2\asymp\max_{j\in[p]}|\ba_j|_2^2\leq C$ and there exists an $r\times r$ orthogonal matrix $\bU_1$ such that, with $\widebar{\bA}_1=\widehat{\bA}\bU_1$,
\begin{align*}
    \|\widebar{\bA}_1- \widetilde\bA  \|_{\max} =   O_{\rm p} \bigg(\frac{r^{5/2}\mu^2 \|\widehat{\bM}-\bM\|_{\infty}}{\sqrt{p}|\nu_r| }\bigg)
    =   O_{\rm p}\big[ {n^{-1/2}p^{ -1/2}\{\log (np)\}^{1/2}}  \big]\,,
\end{align*}
provided that $\log (np) \ll n^{1/5}$.

It remains to verify that the same orthogonal alignment can be used in parts
{\rm (i)} and {\rm (ii)}.  Since $r$ is fixed,
\begin{align*}
\|\widehat{\bA}\bU_1-\widetilde{\bA}\|_{\rm op}
&\leq \sqrt{pr}\,
\|\widehat{\bA}\bU_1-\widetilde{\bA}\|_{\max}
=O_{\rm p}\big[n^{-1/2}\{\log(np)\}^{1/2}\big]\,.
\end{align*}
Since $\widehat{\bA}^{\T}\widehat{\bA}=\bI_r$, the preceding bound and
part {\rm (i)} imply that
\begin{align*}
\|\bU_0-\bU_1\|_{\rm op}
&=\|\widehat{\bA}(\bU_0-\bU_1)\|_{\rm op}\\
&\leq
\|\widehat{\bA}\bU_0-\widetilde{\bA}\|_{\rm op}
+\|\widehat{\bA}\bU_1-\widetilde{\bA}\|_{\rm op}
=O_{\rm p}\big[n^{-1/2}\{\log(np)\}^{1/2}\big]\,.
\end{align*}
Moreover, Condition~\ref{cond.facmat} and the max-norm bound for
$\widehat{\bA}\bU_1$ yield
\begin{align*}
\max_{j\in[p]}\|\be_j^{\T}\widehat{\bA}\|_2
 =\max_{j\in[p]}\|\be_j^{\T}\widehat{\bA}\bU_1\|_2 \leq \max_{j\in[p]}\|\be_j^{\T}\widetilde{\bA}\|_2
+\sqrt r\,\|\widehat{\bA}\bU_1-\widetilde{\bA}\|_{\max}
=O_{\rm p}(p^{-1/2})\,,
\end{align*}
provided that $\log (np) \ll n$.
Consequently,
\begin{align*}
&~\|\widehat{\bA}\bU_0-\widetilde{\bA}\|_{\max}
 \leq
\|\widehat{\bA}\bU_1-\widetilde{\bA}\|_{\max}
+\|\widehat{\bA}(\bU_0-\bU_1)\|_{\max}\\
 \leq &~
\|\widehat{\bA}\bU_1-\widetilde{\bA}\|_{\max}
+\max_{j\in[p]}\|\be_j^{\T}\widehat{\bA}\|_2
\|\bU_0-\bU_1\|_{\rm op}
=O_{\rm p}\big[ {n^{-1/2}p^{ -1/2}\{\log (np)\}^{1/2}}  \big]\,,
\end{align*}
provided that $\log (np) \ll n$.
Thus, $\bU_0$ satisfies both assertions {\rm (i)} and {\rm (ii)}. Taking
$\bU=\bU_0$ and $\widebar{\bA}=\widehat{\bA}\bU$ proves the two assertions
with the same orthogonal matrix.
We complete the proof of Lemma \ref{lemma1.prop2}. $\hfill\Box$

\subsection{Proof of Lemma \ref{lemma2.prop2}}\label{proof.lemma2.prop2}

Recall $\bdelta_{t}(u) = (\bA -  \widehat{\bA}\widehat{\bA}^{\T}\bA) \bZ_t(u) -  \widehat{\bA}\widehat{\bA}^{\T}\bvep_t(u)$ for any $u\in\cU$,  $\widetilde{\bA}^{\T}\widetilde{\bA}=\bI_r$, and $\widetilde{\bA} \widetilde{\bZ}_t(\cdot)= \bA\bZ_t(\cdot)$, where $\widetilde{\bZ}_t(\cdot)=  \bf{\Lambda}^{1/2} \bGamma^{\T} \bZ_t(\cdot)$ and $\widetilde\bA=(\tilde\ba_1,\ldots,\tilde\ba_p)^{\T}\in \mathbb{R}^{p\times r}$.
Let $\be_j\in\mathbb R^p$ denote the $j$-th canonical basis vector, whose $j$-th component is one and all other components are zero.  Then, for any $u\in\cU$, 
\begin{align}\label{eq:fac.delta} 
	&\delta_{t,j}(u)  
	=   \be_j^{\T} ( \widetilde\bA -  \widehat{\bA}\widehat{\bA}^{\T}\widetilde\bA) \widetilde\bZ_t(u)  -   \be_j^{\T}(\widehat{\bA}\widehat{\bA}^{\T}-  \widetilde\bA\widetilde\bA^{\T})\bvep_t(u)    -   \be_j^{\T}\widetilde\bA\widetilde\bA^{\T}\bvep_t(u)  \notag\\
	&~~~=     \be_j^{\T} (\widetilde\bA\widetilde\bA^{\T}-  \widehat{\bA}\widehat{\bA}^{\T}) \bA  \bZ_t(u)   -   \be_j^{\T}(\widehat{\bA}\widehat{\bA}^{\T}-\widetilde\bA\widetilde\bA^{\T})\bvep_t(u)  -    \tilde\ba_{j}^{\T}\widetilde\bA^{\T}\bvep_t(u) \,.
\end{align}
Following from  Lemma \ref{lemma1.prop2}, there exists  an orthogonal matrix ${\bf U}$ such that $\widebar{\bA} = \widehat{\bA}{\bf U}$, $\| \widebar{\bA} - \widetilde\bA \|_{\rm op} = O_{\rm p}( n^{-1/2})$ and $\max_{j\in [p]} | \bar{\ba}_j - \tilde{\ba}_j |_{2} = O_{\rm p} [ {n^{-1/2}p^{ -1/2}\{\log (np)\}^{1/2}}  ]$, 
provided that $\log (np) \ll n^{1/5}$, where $\widebar{\bA}=(\bar{\ba}_1,\ldots,\bar{\ba}_p)^{\T}$.  By Condition \ref{cond.facmat}, it holds that $\max_{j\in[p]}|\tilde\ba_j|_2 \leq \max_{j\in[p]}| \ba_j|_2 \|\bGamma {\bf{\Lambda}}^{-1/2}\|_{\rm op} \lesssim p^{ -1/2}$, then
\begin{align}\label{eq:diff.AA}
	&\max_{j\in[p]} |\be_j^{\T}(\widehat{\bA}\widehat{\bA}^{\T} - \widetilde\bA \widetilde\bA^{\T} ) |_2 
   =  \max_{j\in[p]}|\be_j^{\T}(\widebar{\bA}\widebar{\bA}^{\T}-\widetilde\bA \widetilde\bA^{\T} )|_2 \notag\\
   =&~\max_{j\in[p]}|  (\bar{\ba}_j -   \tilde\ba_j)^{\T} (\widebar{\bA} -   \widetilde\bA)^{\T} + (\bar{\ba}_j -   \tilde\ba_j)^{\T} \widetilde\bA^{\T}  +    \tilde\ba_j^{\T} (\widebar{\bA} -  \widetilde\bA )^{\T}|_2  \notag\\
  \lesssim&\, \max_{j\in[p]}|\bar{\ba}_{j}-\tilde\ba_j|_2 \|\widebar{\bA}-\tilde\bA\|_{\rm op} 
  + \max_{j\in[p]}|\bar{\ba}_{j}-\tilde\ba_j|_2 \|\tilde\bA\|_{\rm op}\notag\\
  &\, + \max_{j\in[p]} |\tilde\ba_j|_2 \|\widebar{\bA}-\tilde\bA\|_{\rm op}   
  =  O_{\rm p}\big[  n^{-1/2}p^{ -1/2}\{\log (np)\}^{1/2}  \big]\,, 
\end{align}
provided that $\log (np) \ll n^{1/5}$.  By Conditions \ref{cond.facprocess} and \ref{cond.facmat}, \eqref{eq:tail.x} and \eqref{incer.Z} in Section \ref{proof.lem.fac.Sigma},  we have $\|\ba_{j}^{\T}\bZ_t(u)\|_{\psi_2} \leq C$ and $\|\ba_{j}^{\T}\bZ_t(u) -\ba_{j}^{\T}\bZ_t(v) \|_{\psi_2} \leq C|u-v|^{\kappa}$ for any $(t,j,u,v)$. Then,  following similar arguments as the proof of Lemma \ref{lemma1.prop1}, we also have
\begin{align*}
&~~~~~~~ \max_{t\in[n]}\max_{j\in[p]}\sup_{u\in\cU}|\ba_j^{\T}\bZ_{t}(u)|  = O_{\rm p}[\{\log(np)\}^{1/2} ]\,,\\
 &\max_{t\in[n]}\max_{j\in[p]}\max_{m\in[M_2]} \sup_{u \in{B}_m}|\ba_j^{\T}\bZ_{t}(u)-\ba_j^{\T}\bZ_{t}(b_m)|  = O_{\rm p}\bigg\{\frac{ (np)^2}{  M_2^{ \kappa/2} }   \bigg\}\,.
\end{align*}
Notice that $|\bA\bZ_{t}(u)|_2^2  = \sum_{j=1}^p |\ba_j^{\T}\bZ_{t}(u)|^2 \leq p \max_{j\in [p]} |\ba_j^{\T}\bZ_{t}(u)|^2$ for any $(t,u)$. We have 
\begin{align}\label{term1.norm}
&	\max_{t\in[n]}\sup_{u\in\cU}|\bA\bZ_{t}(u)|_2 \leq p^{1/2}\max_{t\in[n]}\max_{j\in[p]}\sup_{u\in\cU}|\ba_j^{\T}\bZ_{t}(u)|  = O_{\rm p}[p^{1/2}\{\log(np)\}^{1/2} ]  \,, \nonumber\\
&~~~~~~~~~ \max_{t\in[n]} \max_{m\in[M_2]} \sup_{u \in{B}_m}|\bA\bZ_{t}(u)-\bA\bZ_{t}(b_m)|_2  = O_{\rm p}\bigg(\frac{ n^2p^{5/2}}{  M_2^{ \kappa/2} }   \bigg)\,.
\end{align}
Analogously, by Condition \ref{cond.facerror}, we also have
\begin{align}\label{term2.norm}
    &~~~~~~\max_{t\in[n]}\sup_{u\in\cU}|\bvep_t(u)|_2 =  O_{\rm p}[p^{1/2}\{\log(np)\}^{1/2} ]\,,\nonumber\\
    &\max_{t\in[n]}\max_{m\in[M_2]}\sup_{u\in B_m}|\bvep_t(u)-\bvep_t(b_m)|_2 =   O_{\rm p}\bigg(\frac{ n^2p^{5/2}}{  M_2^{ \kappa/2} }   \bigg)\,.
\end{align}
Write $\widetilde\bA=(\tilde\bb_1,\ldots,\tilde\bb_r)$.  For any  $i\in[r]$, $\tilde\bb_i^{\T}\bvep_t(u)$ is the $i$-th component of $\widetilde\bA^{\T}\bvep_t(u)$. 
%By Condition \ref{cond.facmat}(i), we know that $|p^{-1/2}\bb_i|_2 = 1$ for any $i\in[r]$. 
By the fact that $\widetilde{\bA}^{\T}\widetilde{\bA}=\bI_r$, we have $|\tilde\bb_i|_2 = 1$ for any $i\in[r]$.
Therefore, by Condition \ref{cond.facerror}, it holds that $ \|\tilde\bb_i^{\T}\bvep_t(u)\|_{\psi_2} \leq C$ and $ \|\tilde\bb_i^{\T}\bvep_t(u) - \tilde\bb_i^{\T}\bvep_t(v)\|_{\psi_2} \leq C|u-v|^{\kappa}$ for any $(t,i,u,v)$. Following similar arguments as in the proof of Lemma \ref{lemma1.prop1},  we have
\begin{align*}
&~~~~~~~ \max_{t\in[n]}\max_{i\in[r]}\sup_{u\in\cU}|\tilde\bb_i^{\T}\bvep_t(u)|  = O_{\rm p}[\{\log(np)\}^{1/2} ]\,,\\
 &\max_{t\in[n]}\max_{i\in[r]}\max_{m\in[M_2]} \sup_{u \in{B}_m}|\tilde\bb_i^{\T}\bvep_t(u)-\tilde\bb_i^{\T}\bvep_t(b_m)|  = O_{\rm p}\bigg\{\frac{ (np)^2}{  M_2^{ \kappa/2} }   \bigg\}\,,
\end{align*}
which yields that
\begin{align}\label{term3.norm}
&	\max_{t\in[n]}\sup_{u\in\cU}|\widetilde\bA^{\T}\bvep_t(u)|_2 \leq  r^{1/2} \max_{t\in[n]}\max_{i\in[r]}\sup_{u\in\cU}|\tilde\bb_i^{\T}\bvep_t(u)|  = O_{\rm p}[\{\log(np)\}^{1/2} ]  \,,\nonumber \\
&~~~~~~~~~~~~~ \max_{t\in[n]} \max_{m\in[M_2]} \sup_{u \in{B}_m}|\widetilde\bA^{\T}\bvep_t(u)-\widetilde\bA^{\T}\bvep_t(b_m)|_2  = O_{\rm p}\bigg\{\frac{ (np)^2}{  M_2^{ \kappa/2} } \bigg\} \,.
\end{align}
Notice that $\max_{j\in[p]}|\tilde\ba_j|_2   \lesssim p^{-1/2}$ and $|\ba^{\T}\bb|\leq |\ba|_2|\bb|_2$ for any $\ba,\bb\in\mathbb{R}^q$. Then, following from \eqref{eq:fac.delta}--\eqref{term3.norm}, it holds  that
\begin{align*}
   &~~~~~~~  \max_{t\in[n]}\max_{j\in[p]}\sup_{u\in\cU}|\delta_{t,j}(u)|   
	=  O_{\rm p}[n^{-1/2} \log(np) + p^{ -1/2}\{\log (np)\}^{1/2}] \,,\\
   & \max_{t\in[n]}\max_{j\in[p]}\max_{m\in[M_2]}\sup_{u\in B_m}|\delta_{t,j}(u)-\delta_{t,j}(b_m)| 
 = 	O_{\rm p}\bigg[\frac{ n^{3/2}p^{ 2}\{\log (np)\}^{1/2}+n^2p^{ 3/2}}{  M_2^{ \kappa/2} } \bigg] \,,
\end{align*}
provided that $\log (np) \ll n^{1/5}$. We complete the proof of Lemma \ref{lemma2.prop2}. $\hfill\Box$

\subsection{Proof of Lemma \ref{lem.fac.Sigma}}\label{proof.lem.fac.Sigma}

Recall  $\bX_t(u) = \bA\bZ_t(u)+\bvep_t(u)$ and $X_{t,j}(u) = \ba_{j}^{\T}\bZ_t(u)+\vep_{t,j}(u)$ for any $t\in[n]$, $j\in[p]$ and $u\in\cU$, where $\bA= (\ba_1,\ldots,\ba_p)^{\T}\in\mathbb{R}^{p\times r}$ and $\bZ_t(\cdot)=\{Z_{t,1}(\cdot),\ldots,Z_{t,r}(\cdot)\}^{\T}$.  By Conditions \ref{cond.facerror} and \ref{cond.facprocess}, it holds that $\mathbb{E}\{X_{t,j}(u)\}=0$, $ \|\vep_{t,j}(u)\|_{\psi_2}\leq C$ and $\max_{i\in[r]}\| Z_{t,i}(u) \|_{\psi_2}\leq C$ for any $(t,j,u)$. Since $r$ is a fixed integer, it follows from Condition \ref{cond.facmat} that $\|\ba_{j}^{\T}\bZ_t(u)\|_{\psi_2} \leq \sum_{i=1}^r |\ba_{j}|_{\max}\|  Z_{t,i}(u)  \|_{\psi_2}\leq C$ for any $(t,j,u)$. Then,  
\begin{align}\label{eq:tail.x}
    \bbP\{|X_{t,j}(u)|>x\} \leq&\, \bbP\bigg\{|\ba_{j}^{\T}\bZ_t(u)|>\frac{x}{2}\bigg\} + \bbP\bigg\{|\vep_{t,j}(u)|>\frac{x}{2}\bigg\}  \lesssim  \exp(-Cx^2)
\end{align}
for any $(t,j,u)$ and $x\geq 0$. Notice that
 \begin{align*}
 	& \widetilde{\Sigma}_{xx,jk}^{(\ell)}(u,v) - \Sigma_{xx,jk}^{(\ell)}(u,v) 
 	= \underbrace{\frac{1}{n_\ell}\sum_{t=1}^{n_\ell} [X_{t,j}(u)X_{t+\ell,k}(v) - \bbE\{X_{t ,j}(u)X_{t+\ell,k}(v)\}]}_{{\rm I}_{\ell,jk}(u,v)} \\
    &~~~~~~~~~~~~~~~~~~~~~~~~+ \underbrace{ \bigg[\widebar{X}_{j}(u)\widebar{X}_k(v) -  \bigg\{\frac{1}{n_\ell}\sum_{t=1}^{n_\ell}X_{t ,j}(u)\bigg\}\widebar{X}_k(v) 
 	- \widebar{X}_j(u)\bigg\{\frac{1}{n_\ell}\sum_{t=1}^{n_\ell} X_{t+\ell,k}(v)\bigg\} \bigg]}_{{\rm II}_{\ell,jk}(u,v)}  
 \end{align*} 
for any $j,k\in[p]$, $\ell\in[\ell_0]$ and $(u,v)\in\cU^2$. By Condition  \ref{cond.facprocess},   $\{\bX_t(u)\}_{t\in[n]}$ is also an $\alpha$-mixing sequence with $\alpha$-mixing coefficients ${\alpha}_x(m) \lesssim \exp(-Cm)$. 
Then, by \eqref{eq:tail.x} and Lemma \ref{tail_chang} with $(B_n,c_n,{r}_1,{r})=(1,\ell,1,1/3)$, it holds that $\|{\rm I}_{\ell,jk}(u,v)\|_2\lesssim n^{-1/2}$ and 
\begin{align}
\label{eq:fac.tail.I}
\bbP\big\{|{\rm I}_{\ell,jk}(u,v)|>x\big\}
\lesssim \exp(-C n x^2)+\exp(-C n^{1/3}x^{1/3}) 
\end{align} 
 for any $(\ell,j,k,u,v)$ and $x\geq 0$.  Analogously, by Lemma \ref{tail_chang} again with $(B_n,c_n,{r}_1,{r})=(1,0,2,1/3)$,  we have  $\|{\rm II}_{\ell,jk}(u,v)\|_2\lesssim n^{-1 }$ and 
\begin{align}
\label{eq:fac.tail.II}
 \bbP\big\{|{\rm II}_{\ell,jk}(u,v)| >x\big\} 
\lesssim \exp(-Cn x)+\exp(-C n^{1/3}x^{1/6}) 
\end{align}
for any $(\ell,j,k,u,v)$ and $x\geq 0$, which  implies $\max_{\ell,j,k }\|\widetilde{\Sigma}_{xx,jk}^{(\ell)}(u,v)-\Sigma_{xx,jk}^{(\ell)}(u,v)\|_2 \lesssim n^{-1/2}$ for any  $(u,v)\in\cU^2$. Then,  by 
%Jensen's inequality and 
Fubini's theorem, it holds that
\begin{align*}
\max_{\ell\in[\ell_0]}\max_{j,k\in[p]}\bbE\big\{\|\widetilde{\Sigma}_{xx,jk}^{(\ell)}-\Sigma_{xx,jk}^{(\ell)}\|_{\cS}^2\big\}
\lesssim n^{-1}\,.
\end{align*}
This completes the proof of the first part of Lemma~\ref{lem.fac.Sigma}.

By \eqref{eq:fac.tail.I} and \eqref{eq:fac.tail.II},  we obtain, for any $(\ell,j,k,u,v)$ and $x\geq 0$, 
\begin{align}\label{tail.Xtilde}
 &~\bbP\big ( |\widetilde{\Sigma}_{xx,jk}^{(\ell)}(u,v)-\Sigma_{xx,jk}^{(\ell)}(u,v)|>x\big )\nonumber\\
&~~~~\lesssim    \exp(-Cn x^2)+ \exp(-Cn^{1/3}x^{1/3}) +\exp(-Cn x)+ \exp(-C n^{1/3}x^{1/6}) \,.
\end{align}
As in Section~\ref{sec.preliminary}, let $\cU=[0,1]$ and partition it into
$M$ subintervals $\{B_1,\ldots,B_{M}\}$ of equal length $M^{-1}$. Let $b_m$ denote the midpoint of $B_m$. Write $\Delta_{\ell,jk}(u,v)= \widetilde{\Sigma}_{xx,jk}^{(\ell)}(u,v)-\Sigma_{xx,jk}^{(\ell)}(u,v)$. Then, it holds that
\begin{align*}
    &\max_{\ell\in[\ell_0]}\|\widetilde{\bSigma}_{xx}^{(\ell)}-\bSigma_{xx}^{(\ell)}\|_{\cS,\max} = \max_{\ell\in[\ell_0]}\max_{j,k\in[p]}\|\widetilde{\Sigma}_{xx,jk}^{(\ell)}-\Sigma_{xx,jk}^{(\ell)}\|_{\cS}  \lesssim \max_{\ell\in[\ell_0]}\max_{j,k\in[p]}\sup_{(u,v)\in \cU^2}|\Delta_{\ell,jk}(u,v)|\\
    &~~~~~~\lesssim \underbrace{\max_{\ell\in[\ell_0]}\max_{j,k\in[p]}\max_{m,m'\in[M]} \sup_{(u,v)\in B_m\times B_{m'}}|\Delta_{\ell,jk}(u,v)-\Delta_{\ell,jk} (b_m,b_{m'})|}_{Q_1}\\
     &~~~~~~~~~+ \underbrace{\max_{\ell\in[\ell_0]}\max_{j,k\in[p]}\max_{m,m'\in[M]} |\Delta_{\ell,jk}(b_m,b_{m'})|}_{Q_2}\,.
\end{align*}
By Conditions \ref{cond.facprocess} and \ref{cond.facmat}, we know that 
\begin{align}\label{incer.Z}
    \|\ba_{j}^{\T}\bZ_t(u) -\ba_{j}^{\T}\bZ_t(v) \|_{\psi_2} \leq \sum_{i=1}^r |\ba_{j}|_{\max}\|  Z_{t,i}(u) -Z_{t,i}(v)  \|_{\psi_2}\leq C|u-v|^{\kappa}
\end{align}
 for any $(t,j,u,v)$. This, together with Condition \ref{cond.facerror}, yields 
\begin{align*}
    &~~~~~~~~~~~~~~~~\|X_{t,j}(u_1)-X_{t,j}(u_2)\|_{\psi_2} \leq C|u_1-u_2|^{\kappa} \,,~~\mbox{and}\\
    &\|X_{t,j}(u_1)X_{t+\ell,k}(v_1)-X_{t,j}(u_2)X_{t+\ell,k}(v_2)\|_{\psi_1} \leq C(|u_1-u_2|^{\kappa}+|v_1-v_2|^{\kappa})
\end{align*}
for any $(t,\ell,j,k,u_1,v_1,u_2,v_2)$. Then, following the proof of Lemma \ref{lemma_tn_dis}, we also have  $Q_1 = O_{\rm p}(n^{-3/2})$, provided that $\log (np)\ll n^{1/3}$ and $M\asymp (np)^C$ for some sufficiently large constant $C>0$. Moreover, by \eqref{tail.Xtilde}, we have 
\begin{align*}
    \bbP(Q_2 > x) \lesssim p^2M^2\big\{\exp(-Cn x^2)+ \exp(-Cn^{1/3}x^{1/3}) +\exp(-Cn x)+ \exp(-C n^{1/3}x^{1/6})  \big\}
\end{align*}
 for any $x\geq 0$, which implies that $ Q_2 = O_{\rm p}[ n^{-1/2}\{\log (np)\}^{1/2} ] $, provided that $\log (np)\ll n^{1/5}$. Then, we can conclude that 
 \begin{align*}
     \max_{\ell\in[\ell_0]}\|\widetilde{\bSigma}_{xx}^{(\ell)}-\bSigma_{xx}^{(\ell)}\|_{\cS,\max} = O_{\rm p}[ n^{-1/2}\{\log (np)\}^{1/2} ] \,,
 \end{align*}
provided that $\log (np)\ll n^{1/5}$. We complete the proof of Lemma  \ref{lem.fac.Sigma}. $\hfill\Box$

\clearpage
\section{Additional simulation results and details for real data analysis}

\subsection{Additional simulation results}
Tables \ref{tab:p1_supp}--\ref{tab:dis.HA_supp} present the simulation results of the empirical sizes and  powers for our proposed white noise test using three kernel types for Models 1--6 at the 0.05 nominal level.
Table \ref{tab:factor_supp} reports the rejection rates of our proposed goodness-of-fit test using three kernel functions for Model 7 at the 0.05 nominal level.

\begin{table}[!h]
\caption{Empirical sizes (Models 1--2) and powers (Models 3--6) of our test $T_n$ based on the three kernels for $p=1$ at the 0.05 nominal level. All numbers are multiplied by 100.} 
\footnotesize
\vspace{-0.2cm}
\renewcommand{\arraystretch}{0.6}
\centering
\begin{tabular}{cccccccccccc}
\toprule
& \multirow{2}{*}{$n$} & \multirow{2}{*}{$N$}
& \multicolumn{3}{c}{$L=2$}
& \multicolumn{3}{c}{$L=4$}
& \multicolumn{3}{c}{$L=6$} \\
\cmidrule(lr){4-6}
\cmidrule(lr){7-9}
\cmidrule(lr){10-12}
& & &
${\rm QS}$ & ${\rm Pz}$ & ${\rm Bt}$
& ${\rm QS}$ & ${\rm Pz}$ & ${\rm Bt}$
& ${\rm QS}$ & ${\rm Pz}$ & ${\rm Bt}$ \\
\midrule
Model 1 & 200 & 25 & 7.1 & 7.0 & 7.0 & 5.5 & 6.0 & 6.2 & 6.0 & 5.5 & 6.2 \\
& & 51 & 6.5 & 6.3 & 6.3 & 5.6 & 5.4 & 5.5 & 5.5 & 6.0 & 5.7 \\
& 400 & 25 & 5.2 & 5.5 & 5.5 & 5.2 & 4.9 & 5.3 & 5.1 & 5.3 & 5.3 \\
& & 51 & 5.0 & 5.4 & 5.5 & 5.1 & 4.9 & 5.2 & 5.2 & 5.2 & 5.3 \\ 
\midrule
Model 2 & 200 & 25 & 4.3 & 4.5 & 4.5 & 3.0 & 3.4 & 3.5 & 3.2 & 3.4 & 3.4 \\
& & 51 & 3.8 & 4.3 & 4.2 & 3.4 & 3.6 & 3.6 & 3.5 & 3.7 & 3.8 \\
& 400 & 25 & 3.6 & 4.2 & 4.0 & 4.8 & 4.9 & 4.6 & 4.5 & 4.7 & 4.8 \\
& & 51 & 4.8 & 5.2 & 4.9 & 4.6 & 4.9 & 5.2 & 4.9 & 4.8 & 4.6 \\ 
\midrule
Model 3 & 200 & 25 & 97.1 & 97.3 & 97.4 & 95.5 & 95.4 & 95.9 & 93.9 & 94.0 & 94.0 \\
& & 51 & 97.2 & 97.6 & 97.6 & 95.8 & 95.6 & 95.9 & 93.7 & 94.0 & 94.7 \\
& 400 & 25 & 100 & 100 & 100 & 100 & 100 & 100 & 100 & 100 & 100 \\
& & 51 & 100 & 100 & 100 & 100 & 100 & 100 & 100 & 100 & 100  \\ 
\midrule
Model 4 & 200 & 25 & 90.7 & 90.4 & 91.2 & 80.9 & 80.8 & 83.2 & 71.6 & 72.7 & 74.5 \\
& & 51 & 90.1 & 89.8 & 91.9 & 79.7 & 79.3 & 81.6 & 69.9 & 69.3 & 72.6 \\
& 400 & 25 & 100 & 100 & 100 & 100 & 100 & 100 & 100 & 99.9 & 100 \\
& & 51 & 100 & 100 & 100 & 100 & 100 & 100 & 100 & 100 & 100 \\ 
\midrule
Model 5 & 200 & 25 & 86.8 & 84.6 & 88.8 & 82.4 & 81.2 & 85.2 & 79.8 & 79.5 & 82.9 \\
& & 51 & 87.4 & 85.1 & 89.4 & 83.6 & 81.7 & 86.6 & 80.6 & 79.3 & 83.4 \\
& 400 & 25 & 99.7 & 99.5 & 99.7 & 99.7 & 99.5 & 99.7 & 99.4 & 99.4 & 99.5 \\
& & 51 & 99.9 & 99.6 & 99.9 & 99.9 & 99.8 & 99.9 & 99.8 & 99.6 & 99.9 \\ 
\midrule
Model 6 & 200 & 25 & 86.7 & 84.9 & 88.6 & 82.7 & 81.8 & 85.5 & 80.3 & 79.8 & 83.4 \\
& & 51 & 87.5 & 85.3 & 89.1 & 84.1 & 82.7 & 86.8 & 80.0 & 79.4 & 84.3 \\
& 400 & 25 & 99.7 & 99.5 & 99.7 & 99.7 & 99.5 & 99.6 & 99.5 & 99.4 & 99.5 \\
& & 51 & 99.9 & 99.5 & 99.9 & 99.7 & 99.8 & 99.8 & 99.7 & 99.5 & 99.8 \\ 
\bottomrule
\end{tabular}
\label{tab:p1_supp}
\end{table}

\begin{table}[ht]
\caption{Empirical sizes of the proposed test $T_n$ based on the three kernels under Models 1 and 2 for $p\in\{10,50,100\}$ at the 0.05 nominal level. All numbers are multiplied by 100.} 
%\vspace{-0.5em}
\footnotesize
\centering
\vspace{-0.2cm}
\renewcommand\arraystretch{0.6}
\begin{tabular}{ccccccccccccc}
\toprule
& \multirow{2}{*}{$n$} & \multirow{2}{*}{$p$}
& \multirow{2}{*}{$N$}
& \multicolumn{3}{c}{$L=2$}
& \multicolumn{3}{c}{$L=4$}
& \multicolumn{3}{c}{$L=6$} \\
\cmidrule(lr){5-7}
\cmidrule(lr){8-10}
\cmidrule(lr){11-13}
& & & &
${\rm QS}$ & ${\rm Pz}$ & ${\rm Bt}$
& ${\rm QS}$ & ${\rm Pz}$ & ${\rm Bt}$
& ${\rm QS}$ & ${\rm Pz}$ & ${\rm Bt}$ \\
\midrule
Model 1 & 200 & 10  & 25 & 4.3 & 4.3 & 4.7 & 4.1 & 4.6 & 4.3 & 4.6 & 4.6 & 4.6 \\
&&     & 51 & 5.1 & 5.1 & 5.0 & 4.1 & 4.4 & 3.9 & 4.0 & 4.4 & 4.2 \\
&& 50  & 25 & 3.9 & 4.1 & 3.8 & 3.5 & 3.7 & 3.7 & 4.0 & 3.5 & 4.1 \\
&&     & 51 & 3.9 & 3.4 & 4.0 & 3.4 & 3.5 & 3.5 & 3.5 & 3.4 & 3.6 \\
&& 100 & 25 & 2.7 & 2.3 & 2.4 & 2.3 & 2.0 & 2.4 & 2.5 & 2.3 & 2.9 \\
&&     & 51 & 3.1 & 3.4 & 3.4 & 3.2 & 3.2 & 3.6 & 2.8 & 3.2 & 3.5 \\
\cmidrule(lr){2-13}
& 400 & 10  & 25 & 4.6 & 5.0 & 4.9 & 5.0 & 5.3 & 5.5 & 4.7 & 4.7 & 4.8 \\
&&     & 51 & 4.7 & 4.8 & 4.7 & 4.1 & 4.9 & 4.6 & 3.9 & 4.1 & 4.3 \\
&& 50  & 25 & 5.1 & 5.1 & 5.2 & 4.9 & 4.5 & 5.2 & 4.4 & 4.8 & 4.9 \\
&&     & 51 & 4.0 & 4.0 & 4.5 & 3.8 & 4.5 & 4.2 & 3.9 & 4.0 & 4.3 \\
&& 100 & 25 & 3.7 & 3.8 & 4.3 & 3.9 & 3.5 & 4.0 & 3.2 & 3.4 & 3.2 \\
&&     & 51 & 3.6 & 4.0 & 4.0 & 4.7 & 4.2 & 4.4 & 4.0 & 3.8 & 3.8 \\
\midrule
Model 2 & 200 & 10  & 25 & 3.5 & 3.2 & 3.7 & 2.5 & 2.4 & 2.8 & 2.4 & 2.9 & 2.9 \\
&&     & 51 & 2.1 & 1.8 & 2.5 & 1.7 & 1.7 & 2.0 & 1.9 & 1.8 & 2.0 \\
&& 50  & 25 & 2.2 & 2.2 & 2.5 & 1.7 & 1.6 & 1.8 & 1.4 & 1.5 & 1.8 \\
&&     & 51 & 2.0 & 2.3 & 2.7 & 1.4 & 1.7 & 1.6 & 1.4 & 1.9 & 1.9 \\
&& 100 & 25 & 1.5 & 1.7 & 2.1 & 1.4 & 1.4 & 1.6 & 1.2 & 1.2 & 1.4 \\
&&     & 51 & 2.0 & 1.8 & 1.8 & 1.9 & 1.4 & 1.9 & 1.2 & 1.3 & 1.5 \\
\cmidrule(lr){2-13}
& 400 & 10  & 25 & 3.8 & 3.8 & 3.9 & 3.1 & 3.4 & 3.5 & 3.1 & 3.2 & 3.5 \\
&&     & 51 & 3.9 & 3.9 & 4.2 & 2.8 & 2.6 & 3.1 & 2.1 & 1.9 & 2.3 \\
&& 50  & 25 & 2.4 & 2.4 & 2.4 & 2.3 & 2.3 & 2.4 & 2.2 & 2.1 & 2.5 \\
&&     & 51 & 2.5 & 2.8 & 3.3 & 2.9 & 3.0 & 3.1 & 2.9 & 2.7 & 2.8 \\
&& 100 & 25 & 2.8 & 2.6 & 2.9 & 1.6 & 1.7 & 2.0 & 2.4 & 2.2 & 2.4 \\
&&     & 51 & 3.5 & 3.3 & 3.6 & 3.2 & 3.4 & 3.5 & 3.0 & 3.2 & 3.5 \\
\bottomrule
\end{tabular}
\label{tab:dis.H0_supp}
\end{table}

\clearpage

\begin{table}[ht]
\caption{Empirical powers of the proposed test $T_n$ based on the three kernels under Models 3--6 for $p\in\{10,50,100\}$ at the 0.05 nominal level. All numbers are multiplied by 100.} 
%\vspace{-0.1em}
\footnotesize
\centering
\vspace{-0.2cm}
\renewcommand\arraystretch{0.56}
\begin{tabular}{ccccccccccccc}
\toprule
& \multirow{2}{*}{$n$} & \multirow{2}{*}{$p$}
& \multirow{2}{*}{$N$}
& \multicolumn{3}{c}{$L=2$}
& \multicolumn{3}{c}{$L=4$}
& \multicolumn{3}{c}{$L=6$} \\
\cmidrule(lr){5-7}
\cmidrule(lr){8-10}
\cmidrule(lr){11-13}
& & & &
${\rm QS}$ & ${\rm Pz}$ & ${\rm Bt}$
& ${\rm QS}$ & ${\rm Pz}$ & ${\rm Bt}$
& ${\rm QS}$ & ${\rm Pz}$ & ${\rm Bt}$ \\
\midrule
Model 3 & 200 & 10  & 25 & 100 & 100 & 100 & 100 & 100 & 100 & 100 & 100 & 100 \\
&&     & 51 & 100 & 100 & 100 & 100 & 100 & 100 & 100 & 100 & 100 \\
&& 50  & 25 & 99.7 & 99.7 & 99.8 & 99.6 & 99.6 & 99.6 & 99.6 & 99.6 & 99.6 \\
&&     & 51 & 100 & 100 & 100 & 100 & 100 & 100 & 100 & 100 & 100 \\
&& 100 & 25 & 99.9 & 100 & 99.9 & 99.7 & 99.7 & 99.8 & 99.5 & 99.5 & 99.5 \\
&&     & 51 & 100 & 100 & 100 & 100 & 100 & 100 & 100 & 100 & 100 \\
\cmidrule(lr){2-13}
& 400 & 10  & 25 & 100 & 100 & 100 & 100 & 100 & 100 & 100 & 100 & 100 \\
&&     & 51 & 100 & 100 & 100 & 100 & 100 & 100 & 100 & 100 & 100 \\
&& 50  & 25 & 100 & 100 & 100 & 99.9 & 99.9 & 99.9 & 99.9 & 99.9 & 99.9 \\
&&     & 51 & 100 & 100 & 100 & 100 & 100 & 100 & 100 & 100 & 100 \\
&& 100 & 25 & 100 & 100 & 100 & 100 & 100 & 100 & 100 & 100 & 100 \\
&&     & 51 & 100 & 100 & 100 & 100 & 100 & 100 & 100 & 100 & 100 \\
\midrule
Model 4 & 200 & 10  & 25 & 99.3 & 99.4 & 99.7 & 99.3 & 99.4 & 99.7 & 99.5 & 99.5 & 99.9 \\
&&     & 51 & 99.0 & 98.8 & 99.7 & 99.1 & 99.1 & 99.6 & 99.1 & 99.1 & 99.4 \\
&& 50  & 25 & 99.0 & 99.0 & 99.3 & 99.1 & 99.0 & 99.3 & 98.6 & 98.9 & 99.3 \\
&&     & 51 & 100 & 100 & 100 & 100 & 100 & 100 & 100 & 100 & 100 \\
&& 100 & 25 & 98.7 & 98.9 & 98.9 & 98.3 & 98.2 & 98.3 & 98.0 & 98.1 & 98.1 \\
&&     & 51 & 100 & 100 & 100 & 100 & 100 & 100 & 100 & 100 & 100 \\
\cmidrule(lr){2-13}
& 400 & 10  & 25 & 99.7 & 99.7 & 99.9 & 99.7 & 99.7 & 100 & 99.8 & 99.9 & 100 \\
&&     & 51 & 99.8 & 99.9 & 99.9 & 99.9 & 99.8 & 99.9 & 99.8 & 99.8 & 99.9 \\
&& 50  & 25 & 99.5 & 99.4 & 99.5 & 98.8 & 99.0 & 99.1 & 98.7 & 98.6 & 98.8 \\
&&     & 51 & 100 & 100 & 100 & 100 & 100 & 100 & 100 & 100 & 100 \\
&& 100 & 25 & 99.6 & 99.8 & 99.7 & 99.5 & 99.5 & 99.5 & 99.5 & 99.5 & 99.5 \\
&&     & 51 & 100 & 100 & 100 & 100 & 100 & 100 & 100 & 100 & 100 \\
\midrule
Model 5 & 200 & 10  & 25 & 83.6 & 82.7 & 87.7 & 78.6 & 77.5 & 82.7 & 75.8 & 74.3 & 79.4 \\
&&     & 51 & 86.0 & 84.6 & 90.4 & 80.9 & 79.2 & 85.0 & 77.9 & 75.7 & 82.1 \\
&& 50  & 25 & 72.2 & 71.6 & 77.7 & 65.3 & 65.4 & 71.6 & 61.3 & 60.0 & 66.8 \\
&&     & 51 & 73.8 & 72.7 & 78.0 & 67.0 & 66.3 & 71.8 & 61.5 & 61.9 & 67.7 \\
&& 100 & 25 & 62.1 & 61.4 & 68.7 & 52.9 & 52.3 & 58.0 & 48.3 & 47.5 & 54.4 \\
&&     & 51 & 67.7 & 66.4 & 72.7 & 57.2 & 55.9 & 63.4 & 50.1 & 49.0 & 56.6 \\
\cmidrule(lr){2-13}
& 400 & 10  & 25 & 99.4 & 99.4 & 99.5 & 99.4 & 99.3 & 99.4 & 99.1 & 99.0 & 99.3 \\
&&     & 51 & 99.6 & 99.5 & 99.7 & 99.6 & 99.6 & 99.7 & 99.5 & 99.4 & 99.7 \\
&& 50  & 25 & 96.8 & 96.8 & 97.1 & 94.1 & 94.2 & 94.8 & 92.5 & 91.9 & 92.9 \\
&&     & 51 & 97.6 & 97.7 & 97.9 & 96.0 & 95.5 & 96.4 & 94.6 & 94.4 & 95.2 \\
&& 100 & 25 & 94.3 & 93.9 & 95.0 & 91.3 & 91.6 & 92.2 & 89.3 & 88.9 & 90.2 \\
&&     & 51 & 94.6 & 94.9 & 95.2 & 92.0 & 91.2 & 92.6 & 90.0 & 89.2 & 90.4 \\
\midrule
Model 6 & 200 & 10  & 25 & 91.7 & 90.6 & 94.6 & 90.3 & 88.7 & 93.5 & 89.6 & 87.1 & 92.7 \\
&&     & 51 & 92.4 & 91.3 & 95.2 & 90.8 & 89.9 & 94.3 & 89.6 & 88.3 & 93.0 \\
&& 50  & 25 & 97.2 & 96.4 & 98.6 & 95.4 & 94.5 & 96.6 & 94.4 & 93.9 & 96.0 \\
&&     & 51 & 97.7 & 96.4 & 99.6 & 96.6 & 94.7 & 99.1 & 95.4 & 94.7 & 98.3 \\
&& 100 & 25 & 96.5 & 96.1 & 98.1 & 94.5 & 93.8 & 96.9 & 92.8 & 92.1 & 94.9 \\
&&     & 51 & 98.2 & 96.9 & 99.0 & 97.1 & 97.0 & 98.7 & 96.2 & 94.9 & 98.2 \\
\cmidrule(lr){2-13}
& 400 & 10  & 25 & 99.9 & 99.9 & 100 & 99.9 & 99.9 & 100 & 100 & 99.9 & 100 \\
&&     & 51 & 100 & 100 & 100 & 100 & 100 & 100 & 100 & 100 & 100 \\
&& 50  & 25 & 99.9 & 99.8 & 100 & 99.9 & 99.9 & 100 & 100 & 99.9 & 100 \\
&&     & 51 & 100 & 99.9 & 100 & 100 & 100 & 100 & 99.9 & 99.9 & 100 \\
&& 100 & 25 & 99.8 & 99.8 & 99.9 & 99.9 & 99.8 & 99.9 & 99.5 & 99.5 & 99.6 \\
&&     & 51 & 99.9 & 99.9 & 100 & 99.8 & 99.8 & 99.9 & 99.9 & 99.8 & 99.9 \\
\bottomrule
\end{tabular}
\label{tab:dis.HA_supp}
\end{table}
 
\begin{table}[htp]
\caption{Rejection rates of our proposed test $T_n$ for  Model 7 based on the three kernels  at the 0.05 nominal level. All numbers are multiplied by 100.} 
%\vspace{-0.8em}
\footnotesize
\centering
\vspace{-0.2cm}
\renewcommand\arraystretch{0.6}
\begin{tabular}{cccccccccccc}
\toprule
& \multirow{2}{*}{$n$} & \multirow{2}{*}{$p$}
& \multicolumn{3}{c}{$L=2$}
& \multicolumn{3}{c}{$L=4$}
& \multicolumn{3}{c}{$L=6$} \\
\cmidrule(lr){4-6}
\cmidrule(lr){7-9}
\cmidrule(lr){10-12}
& & &
${\rm QS}$ & ${\rm Pz}$ & ${\rm Bt}$
& ${\rm QS}$ & ${\rm Pz}$ & ${\rm Bt}$
& ${\rm QS}$ & ${\rm Pz}$ & ${\rm Bt}$ \\
\midrule
% \toprule
%  &$n$ & $p$ & $T_{2,{\rm QS}}$ & $T_{2,\rm Pz}$ & $T_{2,{\rm Bt}}$ & $T_{
% \rm 4,QS}$ & $T_{\rm 4,Pz}$ & $T_{\rm 4,Bt}$ & $T_{\rm 6,QS}$ & $T_{\rm 6,Pz}$ & $T_{\rm 6,Bt}$ \\
% \midrule
%\multicolumn{11}{c}{$r=3$ (empirical sizes)} \\
$r=3$& 200 & 10  & 3.2  & 2.9  & 3.2  & 2.0 & 2.1  & 2.1  & 2.6  & 2.4   & 2.4   \\
(size)&& 50  & 1.2  & 1.1  & 1.5  & 1.2  & 1.2  & 1.3  & 0.8  & 0.6  & 1.2   \\
&& 100 & 1.1  & 1.2  & 1.3  & 1.0   & 1.1  & 1.3  & 0.5  & 0.7  & 0.9   \\
\cmidrule(lr){2-12}
&400 & 10  & 3.9  & 3.8  & 4.1  & 3.8   & 3.9  & 4.0  & 3.6  & 4.0   & 3.5   \\
&& 50  & 2.6  & 2.4  & 2.6  & 1.9  & 1.6  & 2.2  & 1.8  & 1.7  & 1.9   \\
&& 100 & 2.0   & 1.7  & 2.0   & 1.3  & 1.3  & 1.6  & 1.1  & 1.2  & 1.4   \\
\midrule
%\multicolumn{11}{c}{$r=1$ (empirical powers)} \\
$r=1$& 200 & 10  &  96.3  & 96.3  & 96.5  & 94.6  & 94.5  & 95.0  & 92.9  & 92.9   & 93.8   \\
(power)&& 50  & 99.4  & 99.5  & 99.8  & 98.7  & 98.3  & 98.8  & 97.3  & 97.2  & 98.4   \\
&& 100 & 99.4  & 99.4  & 99.7  & 98.4  & 98.7  & 99.0  & 97.7  & 97.7  & 98.3   \\
\cmidrule(lr){2-12}
&400 & 10  & 100  & 100  & 100  & 100   & 100  & 100  & 99.9  & 100   & 100   \\
&& 50  & 100  & 100  & 100  & 100  & 100  & 100  & 100  & 100  & 100   \\
&& 100 & 100  & 100 & 100  & 99.5   & 100  & 100  & 99.5  & 99.5  & 99.5  \\
\bottomrule
\end{tabular}
\label{tab:factor_supp}
\end{table}

\clearpage
\subsection{List of countries for  yield data}  
Table \ref{tab.list} presents the list of $p=22$  inclusive countries in the yield data under our study.

\renewcommand{\arraystretch}{0.7}

\begin{table}[h]
\centering
\caption{ISO Alpha-3 codes and corresponding countries for yield data.}
\vspace{-0.3cm}
\begin{tabular}{p{1.5cm}p{2.5cm} p{1.5cm}p{2.5cm} p{1.5cm}p{2.5cm}}
\toprule
{Code} & {Country} & {Code} & {Country} & {Code} & {Country} \\
\midrule
AUS & Australia      & ESP & Spain           & NLD & Netherlands      \\
AUT & Austria        & FIN & Finland         & SGP & Singapore        \\
BEL & Belgium        & FRA & France          & SVK & Slovakia         \\
CAN & Canada         & GBR & Great Britain   & SVN & Slovenia         \\
CHE & Switzerland    & ITA & Italy           & SWE & Sweden           \\
CHN & China          & JPN & Japan           & USA & United States    \\
DEU & Germany        & KOR & South Korea     & ZAF & South Africa     \\
DNK & Denmark        &     &                 &     &                  \\
\bottomrule
\end{tabular}
\label{tab.list}
\end{table}

\subsection{Details for age-specific mortality data}

\subsubsection{List of countries for mortality data}
Table \ref{tab.morlist} presents the list of $p=24$ countries included in the mortality data.

\begin{table}[htbp]
\centering
\caption{ISO Alpha-3 codes and corresponding countries for mortality data.}
\vspace{-0.3cm}
\begin{tabular}{p{1.5cm}l p{1.5cm}l p{1.5cm}l}
\toprule
Code & Country & Code & Country & Code & Country \\
\midrule
AUS & Australia      & ESP & Spain           & NOR & Norway          \\
AUT & Austria        & EST & Estonia         & NZL & New Zealand     \\
BEL & Belgium        & FIN & Finland         & POL & Poland          \\
BGR & Bulgaria       & FRA & France          & RUS & Russia          \\
BLR & Belarus        & GBR & Great Britain   & SVK & Slovakia        \\
CAN & Canada         & HUN & Hungary         & SWE & Sweden          \\
CZE & Czech Republic & IRL & Ireland         & UKR & Ukraine         \\
DNK & Denmark        & ITA & Italy           & USA & United States   \\
\bottomrule
\end{tabular}
\label{tab.morlist}
\end{table}

\subsubsection{Estimation procedure for the vector functional autoregressive  model}\label{threesteps}
To estimate the vector functional autoregressive model,  we first center all the series of $\bX_t(\cdot)$ about their empirical means, then apply
the three-step approach in \cite{Guo2023} (detailed as follows), and finally compute the residuals $\hat{\bvep}_t(\cdot)$ based on $\widehat{\bA}$.
\begin{itemize}
	\item{\bf Step 1.}  Perform functional principal component analysis (FPCA) based on the centered $\{X_{t,j}(\cdot)\}$ for each $j \in [p]$, and thus obtain estimated eigenfunctions and FPC scores. The numbers of FPCs are chosen to explain at least 98\% of the total variation.
	\item{\bf Step 2.} Implement a regularized least squares approach with group lasso penalization to the vector time series of estimated FPC scores and thus obtain the estimated block sparse coefficient matrices, where the regularization parameters are selected by BIC. 
	\item{\bf Step 3.} Recover the functional sparse estimates of entries in $\bA$ based on the block sparse estimates obtained in Step 2 and estimated eigenfunctions obtained in Step~1. 
\end{itemize}

%\subsubsection{Figure of the sparsity patterns}

Figure \ref{fig.heatmap} plots the sparsity patterns in $\widehat{\bA}$  for $24$ countries, where black and white correspond to non-zero and zero functional entries in  $\widehat{\bA}$, respectively.

\begin{figure}[htbp]
     \centering
     \includegraphics[scale=0.36]{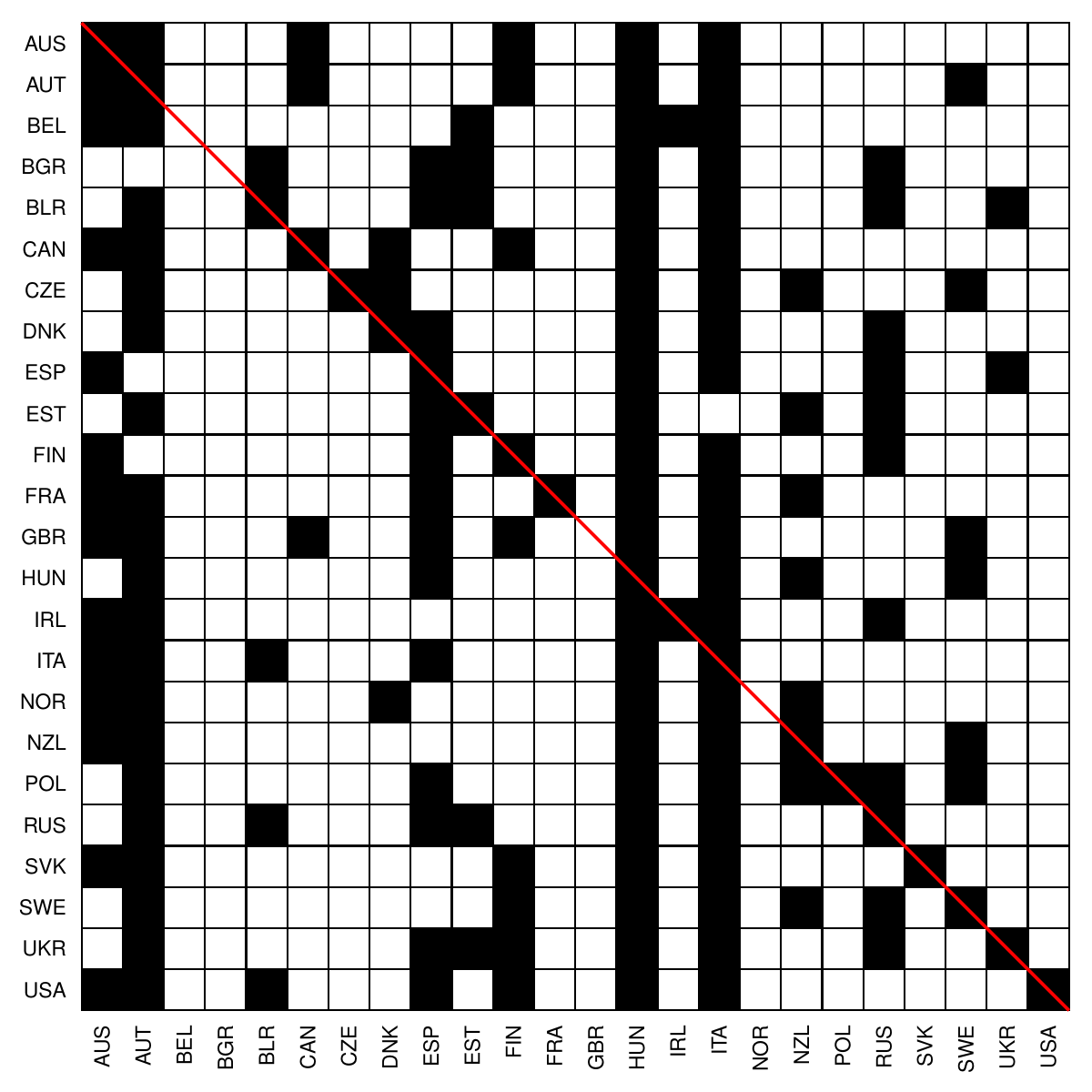}
     \caption{The binary heatmap for the sparsity structure in  $\widehat{\bA}$ for 24 countries.} 
     %Black and white correspond to non-zero and zero functional entries in  $\widehat{\bA}$, respectively.}
     \label{fig.heatmap}
 \end{figure}

\spacingset{1.2}

\end{document}